\documentclass[prx,aps,floatfix,amsmath,superscriptaddress,twocolumn,longbibliography,nofootinbib]{revtex4-2}

\usepackage{amssymb,enumerate}
\usepackage{graphicx}
\usepackage{graphics}
\usepackage{amsmath}
\usepackage{amsthm,bbm}
\usepackage{color}
\usepackage{dsfont}
\usepackage{hyperref}
\usepackage{natbib}
\usepackage{multirow}
\usepackage{algorithm}
\usepackage{algpseudocode}
\usepackage[english]{babel}

\usepackage{wrapfig, blindtext}
\usepackage[most]{tcolorbox}
\tcbset{colback=yellow!10!white, colframe=red!50!black, 
  highlight math style= {enhanced, 
    colframe=red,colback=red!10!white,boxsep=0pt}
}
\usepackage{amsmath}
\usepackage{tikz-cd}
\usepackage{empheq}
\usepackage{graphicx}

\usepackage{anyfontsize}
\usepackage{hyperref}
\usepackage[capitalise]{cleveref}
\usepackage{nameref}
\usepackage{xpatch}
\makeatletter
\xpatchcmd{\@ssect@ltx}{\@xsect}{\protected@edef\@currentlabelname{#8}\@xsect}{}{}
\xpatchcmd{\@sect@ltx}{\@xsect}{\protected@edef\@currentlabelname{#8}\@xsect}{}{}
\makeatother
\usepackage{youngtab}

\usepackage{tikz}
\usetikzlibrary{decorations.pathreplacing}

\usepackage{tikz-cd}

\usepackage{makecell}

\makeatletter 
\renewcommand\onecolumngrid{%
  \do@columngrid{one}{\@ne}%
  \def\set@footnotewidth{\onecolumngrid}%
  \def\footnoterule{\kern-6pt\hrule width 1.5in\kern6pt}%
}
\makeatother 
\newlength\dlf 

\usepackage{amsmath,amssymb}
\usepackage{amssymb,enumerate}
\usepackage{graphicx}
\usepackage{graphics}
\usepackage{amsmath}
\usepackage{amsthm,bbm}
\usepackage{color}
\usepackage{dsfont}
\usepackage{bbm}

\usepackage{lipsum}
\usepackage{lmodern}
\usepackage{tcolorbox}

\usepackage{amsfonts}
\usepackage{graphicx,graphics,epsfig,times,bm,bbm,amssymb,amsmath,amsfonts,mathrsfs}
\usepackage[normalem]{ulem}
\usepackage{subfigure}
\usepackage{dsfont}
\usepackage{braket}
\usepackage{upgreek }
\usepackage{tikz}
\usepackage{natbib}
\usepackage{chngcntr}

\usepackage{amssymb,enumerate}
\usepackage{graphicx}
\usepackage{graphics}
\usepackage{amsmath}
\usepackage{amsthm,bbm}
\usepackage{color}
\usepackage{dsfont}
\usepackage{hyperref}

\usepackage{float}
\usepackage{qcircuit}

\newtheorem{lemma}{Lemma}
\newtheorem*{lemma*}{Lemma}

\newtheorem{proposition}{Proposition}

\theoremstyle{remark}

\newcommand{\bes} {\begin{subequations}}
\newcommand{\ees} {\end{subequations}}
\newcommand{\bea} {\begin{eqnarray}}
\newcommand{\eea} {\end{eqnarray}}
\newcommand{\be} {\begin{equation}}
\newcommand{\ee} {\end{equation}}

\def\>{\rangle}
\def\<{\langle}
\def\Tr{\operatorname{Tr}}

\newcommand{\ketbra}[2]{|{#1}\>\!\<#2|}

\newcommand{\ignore}[1]{}

\usepackage{amsmath,amssymb}

\usepackage{mathtools}
\usepackage{microtype}

\usepackage{ytableau}
\ytableausetup{smalltableaux}

\crefname{section}{Sec.}{Secs.}
\crefname{claim}{Claim}{Claims}

\begin{document}	

\begin{abstract}

We study quantum networks with tree structures, in which information propagates from a root to leaves. 
  At each node in the network, the received qubit unitarily interacts with fresh ancilla qubits, after which each qubit is sent through a noisy channel to a different node in the next level. 
 Therefore,  as the tree depth grows, there is a competition between the irreversible effect of noise and the protection against such noise achieved by delocalization of information. In the classical setting, where each node simply copies the input bit into multiple output bits, this model has been studied as the broadcasting or reconstruction problem on trees, which has broad applications. In this work, we study the quantum version of this problem. We consider a Clifford encoder at each node that encodes the input qubit in a stabilizer code, along with a single qubit Pauli noise channel at each edge. Such noisy quantum trees describe a scenario in which one has access to a stream of fresh (low-entropy) ancilla qubits, but cannot perform error correction. Therefore, they provide a different perspective on quantum fault tolerance. Furthermore, they provide a useful model for describing the effect of noise within the encoders of concatenated codes.  We prove that above certain noise thresholds, which depend on the properties of the code such as its distance, as well as the properties of the encoder, information decays exponentially with the depth of the tree. On the other hand, by studying certain efficient decoders, we prove that for codes with distance $d\geq2$ and for sufficiently small (but non-zero) noise, classical information and entanglement propagate over a noisy tree with infinite depth. Indeed, we find that this remains true even for binary trees  with certain 2-qubit encoders at each node, which encode the received qubit in the binary repetition code with distance $d=1$.
\end{abstract}

\title{   Noisy Quantum Trees: Infinite Protection Without  Correction 
}







\author{Shiv Akshar Yadavalli}\affiliation{Department of Physics, Duke University, Durham, NC 27708, USA}
\author{Iman Marvian}
\affiliation{Department of Physics, Duke University, Durham, NC 27708, USA}
\affiliation{Department of Electrical and Computer Engineering, Duke University, Durham, NC 27708, USA}

\maketitle
\section{Introduction}

Overcoming noise and decoherence in quantum systems is the biggest challenge for quantum information technology. Understanding how these 
 undesirable  phenomena affect storage, communication and processing of quantum information is a problem with broad interest \cite{Preskill2018Nisq}. Beyond these applications, understanding the competition between entanglement generation and noise in quantum circuits is a fundamental problem in many-body physics that 
 has attracted significant attention in  recent years \cite{gullans_huse_measinduced, gullans_huse_2_measinduced, skinner_measinduced, fischer_measinduced}.

\begin{figure}[ht]
\centering
\includegraphics[width=0.4\textwidth]{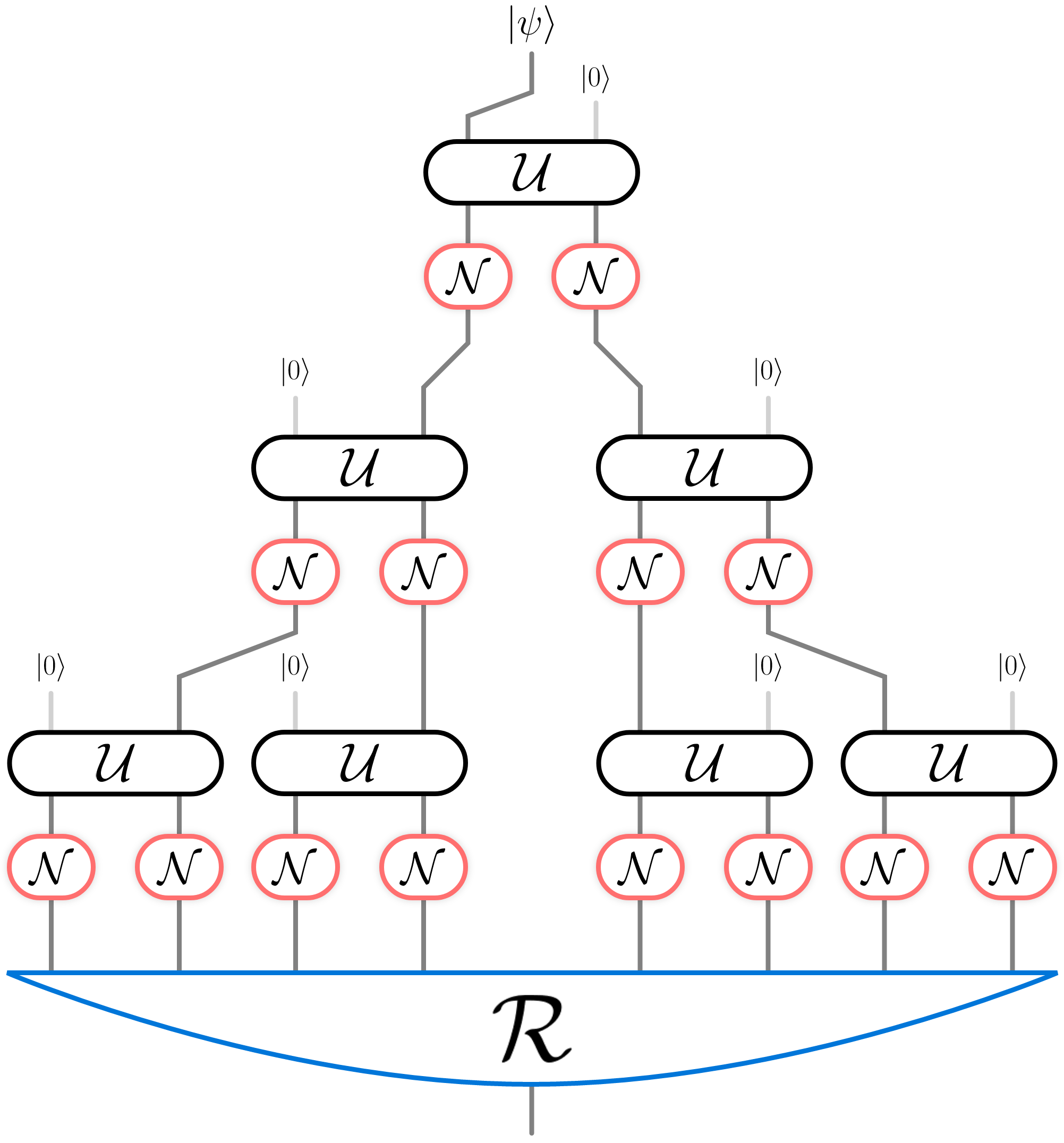}
\caption{\textbf{Noisy Quantum Tree--}  Here,  the input qubit at the root of the tree is in unknown state $\ket{\psi}$,  each $\mathcal{U}$ is a unitary transformation that couples an input qubit to a fresh ancilla qubit  initially prepared in state $|0\rangle$,   and $\mathcal{N}$ is a noise process
 (Time flows from top to bottom). 
  $\mathcal{R}$ represents a decoding process that takes \textit{all} the leaf qubits and approximately reconstructs the root state. The illustrated quantum tree has depth $T=3$. 
 } 
\label{QnCtree}
\end{figure}


Classically, the study of noisy circuits traces back to von Neumann when he proved the first threshold theorem for fault-tolerant computation using noisy elements \cite{von1956probabilistic}. In subsequent decades, 
Pippenger and others further developed this into a theory of noisy circuits \cite{pippenger1985networks,pippenger1988reliable, percolation, kesten_stigum}. A particularly interesting and useful model for the propagation of
classical information is \emph{broadcasting} on noisy trees (see, e.g.,  \cite{evans2000, LyonsPerestextbook}). A  broad class of classical processes -- both natural and artificial -- can be modeled as such trees.  Applications include  broadcasting networks, reconstruction in genetics, and Ising models in statistical physics \cite{evans2000, Mossel_second_eigenvalue, Mossel_Peres, LyonsPerestextbook}.

Here we consider a basic version of this problem:  at a distinguished node, called the \emph{root} of the tree, a bit of information is received and then it is sent to $b$ other nodes in the next level. Then, each node in the next level sends the received bit to $b$ other nodes in the next level, and so on.  Now suppose each link between the nodes at two successive levels is noisy, so that the bit is flipped with probability $p$. Suppose the tree has depth $T$. Let $p(\textbf{x}_T|\textbf{x}_0)$ be the probability that at level $T$ the leaves are in bit-string $\textbf{x}_T\in\{0,1\}^{b^T}$, given that the input bit at the root is $\textbf{x}_0\in\{0,1\}$. A basic interesting question here is as $T\rightarrow \infty$, whether the leaves remain correlated with the input bit or not.
We can formulate this in terms of the total variation distance of the distributions associated to inputs $\textbf{x}_0 =0$ and $\textbf{x}_0 =1$, i.e.,
\begin{align}\label{TVD}
\frac{1}{2}\sum_{\textbf{x}_T\in\{0,1\}^{b^T}}  \Big|p(\textbf{x}_T|\textbf{x}_0 =0) - p(\textbf{x}_T|\textbf{x}_0=1)\Big|\ .
\end{align}
In particular,  does this distance vanish in the limit $T\rightarrow \infty$? If so, then it will be  impossible to  determine the value of the input bit at the root by looking at the output bits at the leaves.


In  a paper \cite{evans2000} published in 2000, Evans \textit{et. al} showed  that 
for infinite $b$-ary trees and, more generally, for trees with branching number $b$ and with bit-flip probability $p$,  if  $b<(1-2 p)^{-2}$ then the total variation distance in Eq.(\ref{TVD}) vanishes exponentially fast with $T$, whereas it  remains non-zero for $b>(1-2 p)^{-2}$, even in the limit $T\rightarrow\infty$. In other words, there is a critical noise threshold 
\be\label{intro threshold}
p_{\text{th}}=\frac{1}{2}\Big(1-\frac{1}{\sqrt{b}}\Big)\ .
\ee
For $ p<p_{\text{th}}$, information about the input never fully disappears in the output, even as $T\rightarrow\infty$. Whereas for $p_{\text{th}}<p\le 1/2$, the output of the tree is asymptotically uncorrelated with the input.\\

In this paper, we study the quantum version of this problem. Clearly, in the quantum setting, due to the no-cloning \cite{wootters1982single} and   
no-broadcasting \cite{nobroadcasting} theorems, it is not possible to copy quantum information. Then, a natural way to generalize the classical problem to the quantum setting is to assume that at each node the received qubit unitarily interacts with one or multiple ancilla qubits initially prepared in a fixed pure state $|0\rangle$, and then  each qubit is transmitted to a node in the  next layer through a noisy single-qubit channel (see Fig. \ref{QnCtree}).  The unitaries at different nodes can be identical or different, and possibly random.

Due to the noise in the circuit, as information propagates from the root to the leaves of the tree, at each step it partially decays. At the same time, the unitary transformation at each node ``delocalizes" information in a single qubit into multiple qubits. Even though this unitary transformation may not necessarily be the encoder of an error-correcting code, the intuition from the theory of quantum error correction suggests that this delocalization of information can protect information against \emph{local} noise. Hence, as the tree depth grows, there is a competition between the effect of noise and the protection against such noise achieved by further delocalization of information. 

\begin{figure}[ht]
\centering
\[
\Qcircuit @C=1em @R=1em {
\lstick{\ket{\psi}} & \qw & \qw & \gate{H} & \ctrl{+2} & \qw & \gate{\mathcal{N}} & \gate{H}  & \ctrl{+1}  & \gate{\mathcal{N}} & \qw\\
& & & & & & & \lstick{\ket{0}}    & \qw   \oplus & \gate{\mathcal{N}} & \qw \\
& & & \lstick{\ket{0}}   & \qw \oplus & \qw & \gate{\mathcal{N}} & \gate{H}  & \ctrl{+1}  & \gate{\mathcal{N}} & \qw\\
& & & & & & &  \lstick{\ket{0}}    & \qw   \oplus  & \gate{\mathcal{N}} & \qw \gategroup{1}{3}{3}{5}{1.6em}{--} \\}\large
\]\\

\includegraphics[width=0.5\textwidth]{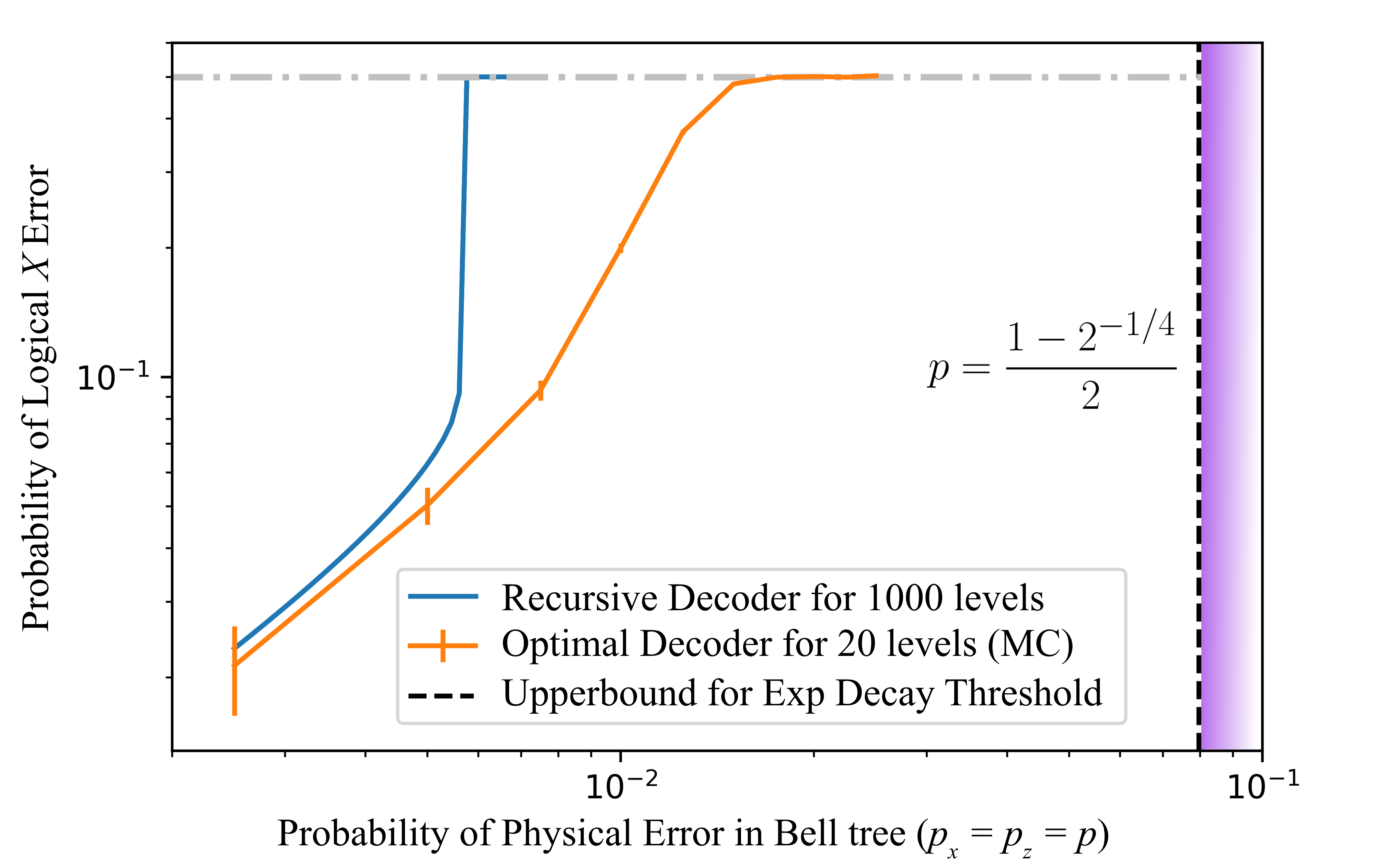}
\caption{\textbf{The Bell Tree -- } The top figure presents a binary tree, which we call a \textit{Bell tree}, with depth $T=2$.     
Here, we assume that the single-qubit channel $\mathcal{N}$
consists of independent bit-flip and phase-flip channels  which apply $X$ and $Z$ errors with probability $p_x=p_z=p$. The plot describes the probability of error after decoding, i.e., the probability of logical error.  Using the fact that the Bell tree has logical subtrees with
branching number 
$\sqrt{2}$, we show that for $p>(1-2^{-1/4})/2\sim8\%$, information decays exponentially with depth $T$ (see Sec. \ref{upperbound sec}).  Therefore, the probability of $X$ logical error for the infinite tree is $1/2$ (the same holds for the probability of logical $Z$ errors). The blue curve corresponds to the probability of logical $X$ error for a Bell tree of depth $T=1000$ after applying an efficient decoder introduced in Sec.\ref{bell tree sec}, which uses a recursive scheme with 2 \emph{reliability} bits.  The   threshold for this decoder is around $p_x=p_z\sim 0.5\%$. The orange curve corresponds to the logical $X$ error after applying the optimal decoder for the Bell tree of depth $T=20$, which is implemented via the belief propagation algorithm presented in Sec.\ref{Optimal section}. This plot suggests that the actual noise threshold for infinite propagation is below $1.7\%$. }\label{BinaryvsBell}
\end{figure}

Fig. \ref{BinaryvsBell} presents a simple example of this problem for  a tree of depth $T=2$, where at each node the received qubit 
goes through a Hadamard 
gate and then interacts with a fresh ancilla qubit via a CNOT. We call this the  \textit{Bell} tree of depth $T=2$, because the unitary transformation  at each node (the dotted box in  Fig. \ref{BinaryvsBell}) transforms the computational basis of two qubits to the so-called Bell basis 
  (note that $T=2$ Bell tree can be interpreted as the encoder of the generalized Shor code [[4,1,2]]). The ``quantum code" defined by this two-qubit encoder is the repetition code and has distance $d=1$, which means it cannot correct or detect general single-qubit errors.

Nevertheless, we prove that for noise below a certain threshold, classical information and entanglement propagate over the Bell tree with depth $T\rightarrow \infty$  (see Fig. \ref{BinaryvsBell}). In particular, we show that if  $Z$ and $X$ errors occur independently with probability $p_z=p_x<\sim 0.5\%$, then  entanglement with a hypothetical reference qubit that is initially entangled with the input, survives over the infinite tree.
 To achieve this we introduce and analyze a novel decoder that during decoding augments every qubit with two \emph{reliability} bits to recursively decode and recover the input state (see the decoder in Fig.\ref{bell tree logic circuit}). It is worth noting that in the classical broadcasting problem discussed above, the threshold in Eq.(\ref{intro threshold}) can be achieved via majority voting on the leaves of the tree \cite{evans2000}\footnote{Since majority voting  ignores the tree structure, its performance is not optimal for finite trees. Yet, it remarkably achieves the threshold in Eq.(\ref{intro threshold})}. In the quantum setting, on the other hand, the decoder is more complicated. Indeed, in contrast to the classical case, even in the absence of noise, it is impossible to obtain any information about the input state without having access to $\approx\sqrt{2^T}$ qubits (see Fig.\ref{Bell} and 
 Sec.\ref{Sec:gen}).

We also rigorously prove that  
 for
 \be
p_z=p_x> \frac{1-2^{-1/4}}{2}\approx  8\%\ ,
 \ee
 classical information, e.g., as quantified by the trace distance or mutual information, decays exponentially in the tree, and in the limit $T\rightarrow\infty$ the output of the tree channel becomes independent of the input.  Numerical analysis of the Bell tree of depth $T=20$, which is performed via a belief propagation algorithm discussed in Sec.\ref{Optimal section},  suggests that the actual threshold is  below $\sim 1.7\%.$\footnote{ 
 It should be noted that entanglement among different output qubits can  survive even for noise stronger than this threshold.}

\subsection*{A simplified model of fault tolerant computation}

The study of noisy quantum circuits goes back to the pioneering work of Peter Shor in 1995 \cite{Shor95} who discovered the first quantum error correcting code and showed how it can be used to suppress noise in  quantum computers. 
Shortly after, the theory of quantum error-correcting codes was developed \cite{CladerbankShor96, Steane96, gottesman1997stabilizer} and the threshold theorem for fault-tolerant quantum computing was established \cite{aharonov2008fault,knill1996threshold}. In particular, Knill \textit{et al.} \cite{knill1996threshold} used codes obtained by concatenating codes with distance  $d \ge 3$ to establish the presence of a non-zero noise threshold below which arbitrarily long quantum computation is possible (see also  \cite{gottesman2014faulttolerant, chamberland2016threshold} for more recent works).

In general, to achieve protection against noise, the standard fault-tolerant protocols involve regular error correction modules that discard entropy from the system. This, in particular, requires encoding into fresh ancilla, measurements \& classical communication, and decoding operations mid-circuit. To better understand the theory of fault tolerance and the importance of its underlying assumptions, it is useful to consider fault tolerance under other assumptions about available resources.

In particular, one can consider a scenario in which one has access to a stream of fresh (low-entropy) ancilla qubits, but cannot perform error correction. Then, is it possible to slow down -- or fully stop -- the irreversible information loss caused by noise and protect quantum information? Noisy quantum trees provide a simple formulation of this problem, and hence a simplified model of fault tolerance.

\subsection*{Comparison with the standard fault tolerance setup}

We make these differences with standard fault tolerance explicit: on one hand, in this simplified model we assume we \textit{have} access to 
\begin{itemize}
\item  Noiseless (low entropy) ancilla qubits, and 
\item A noiseless decoder acting on all the output qubits collectively. 
\end{itemize}
These resources are not available in the standard setup of fault-tolerant quantum computation. On the other hand,  we do \textit{not} allow
\begin{itemize}
\item Measurements,  classical communication, and error correction in the middle of the (infinite) tree network.
\end{itemize}
Furthermore, because of the tree structure, we do \textit{not} allow
\begin{itemize}
\item Interactions (gates) between ``data" qubits.  
\end{itemize}

That is, data qubits can be coupled only to fresh ancilla qubits prepared in a fixed state. In conclusion, although non-zero noise thresholds appear in both settings, the exact value of these thresholds are not necessarily related.

Here, we mention a few other related works in the context of fault-tolerant quantum computation. Aharonov studied noise-induced phase transitions in the propagation of entanglement in noisy quantum circuits \cite{aharonovphasetransition}. In this context, various authors have established upperbounds on noise thresholds, none of which can be directly applied to quantum trees due to the differences noted earlier. For instance, Harrow and Nielsen placed an upper bound of $0.74$ for the fault-tolerance threshold  \cite{HarrowRobust2003} for circuits with 2-local gates. Razbarov \cite{Razborov} improved this upper bound to $1-1/k$ for $k$-local gates. Kempe \textit{et al.} \cite{kempe} further improved this to $1-\Theta(1/\sqrt{k})$ by studying the evolution of the Hilbert-Schmidt distance between two inputs to a circuit, but limiting to the distinguishability of a single (or a few) output qubits. Similarly, the effect of noise in Haar-random circuits has been recently studied 
\cite{dalzell2021random,deshpande2022tight,PRXQuantum.3.010333}.


\subsection*{Other Applications}

In addition to offering a simplified theoretical model for quantum fault tolerance, similar to the classical case, the exploration of noisy quantum trees' properties is expected to have broad applications.
Specifically, such networks are relevant in the context of quantum computing and quantum communications with faulty components; e.g., they describe the effect of noise within the
\begin{itemize}
 \item  Encoders of concatenated codes \cite{dohertyconcatenated, poulin2006optimal}.
\item State preparation circuit for preparing tree tensor network states  \cite{TTN, barthel2022closedness} on quantum computers. 
\item  Binary switching tree of Quantum Random Access Memory (QRAM)
\cite{QRAM1, QRAM2}\ .
\end{itemize} 

Besides these applications, noisy tree networks also provide a natural model for the physical process in which a particle enters a \emph{cold} environment and randomly interacts with other particles in the environment which are initially in a pure state. Other possible applications of this  model can be in, e.g., the study of quantum chaos and scrambling \cite{roberts2013memory, almheiri2015bulk, pastawski2015holographic}.


\subsection*{Overview of Results}
In this paper, we will mainly focus on \textit{stabilizer trees} whose vertices correspond to the same encoder of a stabilizer code, and whose edges are subject to the same single-qubit Pauli noise (see Sec. \ref{stab tree setup}). In the following, we summarize the results presented in each section:
\begin{itemize}

    \item In Sec. \ref{upperbound sec} we prove an upperbound on the distinguishability of the output states at the leaves of the tree, as quantified by the trace distance,  and show that for noise above a certain threshold it decays exponentially fast with the tree depth $T$ (see Propositions \ref{Prop 1}, \ref{prop2} and \ref{prop3}). 


    \item In Sec. \ref{local rec section} we study a  decoding strategy based on  local recovery, where one recursively  decodes  blocks of qubits corresponding to the encoder  at each node of the  tree. While it is sub-optimal, using this strategy we can rigorously prove the existence of a non-zero noise threshold in the case of codes with distance  $d\geq3$. 
    
      \item In Sec. \ref{local recovery 1 bit section} we consider codes with distance $d=2$, which  cannot correct  single-qubit errors  in unknown locations. In this case  the local recursive recovery approach of Sec. \ref{local rec section}, does not yield 
      a non-zero threshold for the infinite tree.  To overcome this, 
 we consider a modification of this scheme where  at each step one passes a single classical ``reliability" bit to the next level. We rigorously prove and numerically demonstrate that this approach  achieves a non-zero noise threshold for the infinite tree.  
       
    \item  In Sec. \ref{bell tree sec} we study the  \textit{Bell} tree (see Fig. \ref{BinaryvsBell}), that corresponds to a 2-qubit code with distance $d=1$. Such codes cannot even detect general single-qubit errors. Nevertheless, we rigorously prove and numerically demonstrate the existence of a non-zero threshold. To achieve this we introduce and analyze an efficient decoder, which is a simple modification of the recursive local recovery where at each step \textit{two} reliability bits are sent to the next level.

    \item  In Sec. \ref{Optimal section} we describe an {optimal and efficient recovery} strategy for stabilizer trees that involves a belief propagation algorithm on classical syndrome data. This allows us to numerically study the decay of information in any stabilizer tree with Pauli noise.
    
     \item Finally, in Sec. \ref{deph tree mapping sec} we show that by fully dephasing qubits at all levels, the stabilizer tree problem can be mapped to an equivalent fully classical problem about propagation of information on a classical tree with correlated noise.

    
\end{itemize}

\section{The setup: Noisy Quantum Trees} \label{stab tree setup}

\subsection{General Case}

Suppose at each node of a full $b$-ary tree  a qubit arrives and interacts with $b-1$ ancillary qubits initially prepared in a fixed state $|0\rangle$ via a unitary transformation $U$, and then each qubit is sent  to a different node in the next level. The overall process at each node can be described as an \emph{isometry}  $V: \mathbb{C}^2\rightarrow (\mathbb{C}^{2})^{\otimes b}$  defined by
\begin{equation}\label{Eq:encoder}
V|\psi\rangle=U(|\psi\rangle|0\rangle^{\otimes (b-1)})\ .
\end{equation}
The image of $V$ defines a 2-dimensional \emph{code} subspace in a $2^b$-dimensional Hilbert space of $b$ qubits. Applying the above process recursively $T$ times, we obtain a chain of encoded states
\be\label{chain}
|\psi\rangle \rightarrow |\psi_1\rangle \rightarrow \cdots  \rightarrow  |\psi_T\rangle \ ,
\ee
where the state at level $k+1$ is obtained via the relation 
\be
|\psi_{k+1}\rangle=V^{\otimes b^k}  |\psi_{k}\rangle\ .
\ee
Then, the overall process can be described by the  isometry 
\be\label{noisless isometry}
V_{T}=\prod_{j=0}^{T-1} V^{\otimes {b}^j} = V^{\otimes b^{T-1}} \cdots V^{\otimes {b}^2} V^{\otimes b} V\ ,
\ee
that encodes 1 qubit in  
\be
N=b^T
\ee 
 qubits. We denote the corresponding quantum channel by $\mathcal{V}_{T}$, where  $\mathcal{V}_{T}(\rho)= 
V_{T} \rho V_{T}$. 

In the language of quantum error-correcting codes, the above process defines a concatenated code \cite{gottesman1997stabilizer, poulin2006optimal}. In this context, one often ignores the noise within the \emph{encoder}. However, we are interested to understand how such noise would affect the output state. Therefore, we assume that after each encoder $V$, the output qubits go through noisy channels (see Fig.\ref{QnCtree}). Furthermore, we assume the noise is independent and identically distributed (i.i.d.) on all qubits and the noise on each qubit is described by the single-qubit channel $\mathcal{N}$. Then, the noisy tree can be described by the quantum channel
\be\label{noisy tree definition -- no root noise}
\mathcal{E}_{T}=\prod_{j=0}^{T-1}    \mathcal{N}^{\otimes b^{j+1}} \circ \mathcal{V}^{\otimes b^j}\ .
\ee
For example, the circuit in Figure \ref{QnCtree} corresponds to $\mathcal{E}_3$. Channel $\mathcal{E}_{T}$ can also be defined recursively as 
\be
\mathcal{E}_{T}=\mathcal{E}^{\otimes b}_{T-1}\circ \mathcal{N}^{\otimes b}\circ \mathcal{V}=\widetilde{\mathcal{E}}^{\otimes b}_{T-1} \circ \mathcal{V}\ ,
\ee
where 
\begin{align}\label{noisy tree definition with root noise}
\widetilde{\mathcal{E}}_{T}=\mathcal{E}_T\circ \mathcal{N}\ ,
\end{align}
corresponds to the noisy tree where the noise is also applied on the input qubit prior to the first encoder. Note that the presence of this single-qubit channel at the root does not affect the noise threshold for  transmitting classical information in the infinite tree. 


\subsection{Stabilizer Trees with Pauli Noise}
In this paper, we mainly assume that the single-qubit noise channel $\mathcal{N}$ is a Pauli channel. 
Specifically, we will be interested in the case of independent $X$ and $Z$ errors, i.e.,
\begin{equation}
\mathcal{N}=\mathcal{N}_x \circ \mathcal{N}_z=\mathcal{N}_z \circ \mathcal{N}_x\ , 
\end{equation}
where
\bes
\begin{align}
\mathcal{N}_x(\rho) &= (1-p_x)\rho + p_x X\rho X,\\
\mathcal{N}_z(\rho) &= (1-p_z)\rho + p_z Z\rho Z
\end{align}
\ees
are, respectively,  bit-flip and phase-flip channels. 

We also assume the encoder unitary $U$ in Eq.(\ref{Eq:encoder}) is a \textit{Clifford} unitary, such that for all  $P \in \mathcal{P}_b=\{\eta I, \eta X , \eta Y, \eta Z: \eta=\pm1 , \pm i \}^{\otimes b}$, $UPU^\dagger \in \mathcal{P}_b$.   Recall that any Clifford unitary can be realized by composing CNOT, Hadamard and Phase gates  \cite{gottesman1997stabilizer,neilsenandchuang}.
The code defined by this encoder is a stabilizer code  with the stabilizer generators  
\be
UZ_j U^\dag\ \ \  : j=2,\cdots, b\ ,
\ee
where $Z_j$ denotes the Pauli $Z$ operator on qubit $j$ tensor product with identity operators on the other qubits \cite{gottesman1997stabilizer,neilsenandchuang}. Given any stabilizer code, there exists an encoder $U$ satisfying the additional property that it has a  \emph{logical} operator $Z_L$, such that $Z_L V=V Z$ and 
 $Z_L\in   \langle i I,Z \rangle^{\otimes b}$,  i.e., $Z_L$ can be written as a tensor product of the identity and Pauli $Z$ operators, up to a global phase. 
  Following \cite{gottesman1997stabilizer, neilsenandchuang} we refer to such an encoder as a \emph{standard} encoder.

 A nice feature of stabilizer codes is the existence of a simple error-correction scheme for correcting Pauli errors: suppose after encoding the qubits are 
subjected to Pauli errors, which in general can be correlated for different qubits. Then, the optimal recovery of the (unknown) input state $|\psi\rangle$ can be achieved by measuring the stabilizers of the code, which can be realized by first applying the inverse of the Clifford unitary $U$,  and then  measuring all the ancilla qubits in the  $Z$ basis.  Then, based on the  outcomes of these measurements, one applies one of the Pauli operators $X, Y, Z$, or the identity operator on the data qubit to correct the error and recover the state $|\psi\rangle$ (the choice of the Pauli operator depends on the inferred distribution of Pauli errors given the syndrome information).

We study the propagation of information through these noisy stabilizer trees by considering recovery channels (decoders) that process all the $b^T$ leaf qubits as the input, and then output a single (approximately) recovered qubit. Note that the optimal recovery channel $\mathcal{R}^{\rm opt}_T$ for the  stabilizer tree $\mathcal{E}_T$ is the standard stabilizer recovery procedure noted earlier, but for the entire concatenated encoder unitary, $V_T$.

After performing the optimal recovery for the stabilizer codes with Pauli errors, 
the entire process can be described  by a Pauli channel, as 
\begin{align}\label{special}
    \mathcal{R}^{\rm opt}_T\circ \mathcal{E}_T(\rho) = r^I_T \rho + r^x_T X\rho X + r^y_T Y \rho Y + r^z_T Z \rho Z\ ,
\end{align}
for arbitrary density operator $\rho$, where $r^I_T, r^X_T, r^Y_T, r^Z_T \ge 0$. 



\subsubsection{CSS Codes with standard encoders}\label{css intro}
As an important special case, we consider the case of CSS (Calderbank-Shor-Steane)  codes with \emph{standard} encoders. CSS stabilizer codes on $b$ qubits are those whose stabilizer generators $UZ_j U^\dag\  : j=2,\cdots, b$  belong to either $ \langle i I,Z \rangle ^ {\otimes b}$, 
 or $ \langle i I,X \rangle ^ {\otimes b}$, i.e., can be written as the tensor product of the identity operators with only Pauli $Z$, or only Pauli $X$ operators \cite{CladerbankShor96, Steane96, GF4}.

  Every CSS code has a standard encoder $V$, defined by the property that there exist logical operators $Z_L$ and $X_L$, such that $Z_L V=V Z$ and $X_L V=V X$, and
 $Z_L\in   \langle i I,Z \rangle^{\otimes b}$, $X_L\in   \langle i I,X \rangle^{\otimes b}$.   This property, in particular, implies that the concatenated code defined by the encoder $V_T$ in Eq.(\ref{noisless isometry}) is also a CSS code.

  The fact that the  stabilizer generators of a CSS code can be partitioned into two types containing only $Z$ or only $X$ Pauli operators simplifies the analysis of error correction. In particular,  if 
  $Z$ and $X$ errors are independent, such that     
$\mathcal{N}=\mathcal{N}_x \circ \mathcal{N}_z$, then after optimal error correction the overall channel can also be decomposed as a composition  of a bit-flip and phase-flip channel as,
\begin{align}\label{setup eqn: css}
    \mathcal{R}^{\text{opt}}_T \circ \mathcal{E}_T=\mathcal{Q}_z\circ \mathcal{Q}_x =\mathcal{Q}_x\circ \mathcal{Q}_z \ ,
\end{align}
where 
\bes
\begin{align}
\mathcal{Q}_z(\rho)&=(1-q^z_T)\rho + q^z_T Z\rho Z\\
\mathcal{Q}_x(\rho)&=(1-q^x_T) \rho + q^x_T X\rho X\ .
\end{align}
\ees
Note that this is  a special case of Eq.(\ref{special}) corresponding to $r_T^x=q_T^x(1-q_T^z)$, $r_T^z=q_T^z(1-q_T^x)$, and $r_T^y=q_T^x q_T^z$. 
We refer to $q^x_T, q^z_T$ as probabilities of  logical $X$ and $Z$ errors for the tree of depth $T$, respectively.
  These probabilities grow  monotonically with the tree depth $T$ and are bounded by 1/2. Therefore, the limits of $T\rightarrow \infty$   of $q^z_T$ and  $q^x_T$ exist and  are denoted by $q^z_\infty$ and $q^x_\infty$, respectively.

\section{
The noise threshold
for  
exponential decay of information}\label{upperbound sec}

In this section, we 
establish that if the noise is stronger than a certain threshold in noisy quantum trees, then information decays exponentially fast with the depth of the tree. First, we start with the special case of standard encoders and then in Sec.\ref{Sec:gen} we present the general result that applies to general (Clifford) encoders. 
  The  results presented in Sec. \ref{std enc sec}, \ref{std css enc sec}, and \ref{anti standard subsec} are proved in Sec. \ref{Sec:mainbound} and \ref{sec:proofs}.

\subsection{Stabilizer Trees with Standard Encoders}\label{std enc sec}

Consider trees with general stabilizer codes  with standard encoders.   
Recall that the weight of an operator $P\in \mathcal{P}_b$,   denoted by ${\rm weight}(P)$,   is the number of qubits on which $P$ acts non-trivially, i.e., is not a multiple of the identity operator. In the following,  $\|\cdot\|_\diamond$ denotes the diamond norm (see Appendix \ref{appendix: diamond pauli} for the definition). We prove that,

\begin{proposition}\label{Prop 1}
Let $V: \mathbb{C}^2\rightarrow (\mathbb{C}^2)^{\otimes b}$
be the encoder of a stabilizer code that 
encodes one qubit into $b$ qubits. Suppose  for this encoder there exists a logical $Z$ operator  $Z_L$, satisfying  $V Z=Z_L V$, 
which can be written 
as the tensor product of Pauli $Z$ and identity  operators $I$, i.e., $Z_L\in \langle i I,Z \rangle ^ {\otimes b}$  (any stabilizer code has an encoder with this property). Let 
\be
b_z={\rm weight}(Z_L)
\ee
be the weight of $Z_L$. Suppose the noise channel $\mathcal{N}$ that defines the tree channel $\mathcal{E}_T$ in Eq.(\ref{noisy tree definition -- no root noise})  is $\mathcal{N}=\mathcal{N}_x\circ \mathcal{N}_z$, where 
$\mathcal{N}_x$ and $\mathcal{N}_z$ are, respectively, the bit-flip and phase-flip channels. Let $p_z$ be the probability of $Z$ error in the phase-flip channel $\mathcal{N}_z$.  
  Then,
\be\label{mainbound}
\big\|\mathcal{E}_{T}-\mathcal{E}_{T}\circ \mathcal{D}_{z}\big\|_\diamond\le {\sqrt{2}} \times \Big[\sqrt{b_z} \times |1-2p_z|\Big]^{T}\ ,
\ee
where $\mathcal{D}_{z}(\rho)=(\rho+Z\rho Z)/2$ is the fully dephasing channel. 
\end{proposition}
Therefore, for $p_z$ in the interval $(1-{b_z}^{-\frac12})/2<p_z< (1+{b_z}^{-\frac12})/2$, in the limit $T\rightarrow\infty$ the channel $\mathcal{E}_T$ becomes a classical-quantum channel that transfers input information only in the $Z$ basis. Note that, in general,  the logical operator $Z_L$ is not  unique and the strongest bound is obtained for the logical operator  with the minimum weight $b_z$. As we further explain in Sec.\ref{Sec:mainbound}, the main idea for proving this result, and the other similar bounds found in this section, is to consider the \emph{logical subtree} defined by the sequence of logical operators in the tree (see Fig.\ref{restricted_subtree}).

\begin{figure}[ht]
\centering
\includegraphics[width=0.4\textwidth]{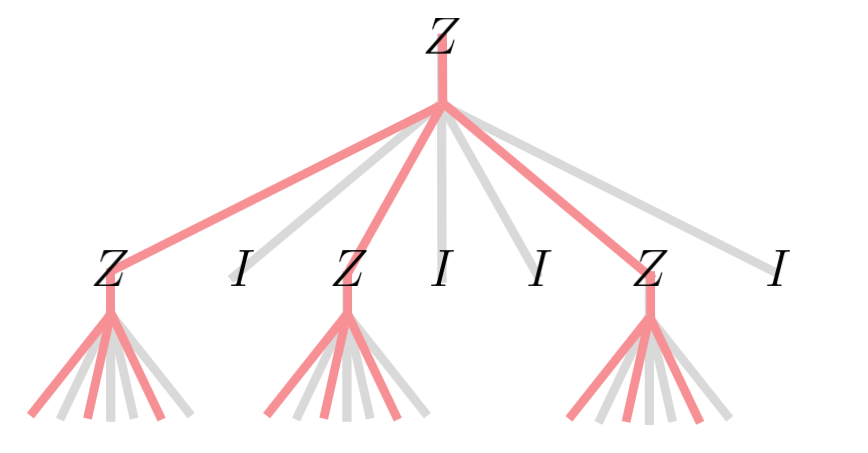}
\caption{\textbf{Illustration of the logical subtree --} Here we consider a tree defined by a standard encoder of  Steane-7 code. The highlighted subtree corresponds to the  logical operator $Z_L=ZIZIIZI$, and is called a ``logical subtree" in this paper. In this case, the logical subtree is a full 3-ary tree, whereas the original tree is a full 7-ary tree.} 
\label{restricted_subtree}
\end{figure}

It is also worth emphasizing that in Eq.(\ref{mainbound}) we are not comparing channel $\mathcal{E}_T$ with its noiseless version $\mathcal{V}_T$, which is more common in the context of noisy quantum circuits. Rather, we are comparing  $\mathcal{E}_T$ with $\mathcal{E}_{T}\circ \mathcal{D}_{z}$, or equivalently with  $\mathcal{E}_{T}\circ \mathcal{Z}$, where $\mathcal{Z}(\cdot)=Z(\cdot)Z$ is the channel that applies Pauli $Z$ on the input qubit. In particular, note that
\be\label{b11}
\big\|\mathcal{E}_{T}-\mathcal{E}_{T}\circ \mathcal{D}_{z}\big\|_\diamond=\frac{1}{2} \big\|\mathcal{E}_{T}-\mathcal{E}_{T}\circ \mathcal{Z}\big\|_\diamond\  .
\ee

In the light of  Helstrom's theorem \cite{Helstrom}, this result can be understood in terms of the decay of  distinguishability of states: suppose at the input of the tree we have one of the orthogonal states  $\ket{\pm}=(\ket{0}\pm \ket{1})/\sqrt{2}$ with equal probabilities (this corresponds to sending a  bit of classical information in the $X$ basis). Then, by looking at the output state at the tree's leaves, we can successfully distinguish these two cases with the maximum success  probability
\begin{align}
    P_{\text{success}} &= \frac{1}{2} +\frac{||\mathcal{E}_{T}(\ketbra{+}{+})-\mathcal{E}_{T}(\ketbra{-}{-})||_1}{4} \nonumber \\
    &\leq  \frac{1}{2} + \frac{(\sqrt{b_z} |1-2p_z|)^{T}}{\sqrt{2}}\ ,
\end{align}
where $\|\cdot\|_1$ denotes the $l$-1 norm, the first equality follows from Helstrom's  theorem, and the inequality follows from the above proposition by noting that $Z|+\rangle=|-\rangle$ and the $l$-1 norm for any particular state is bounded by the diamond norm.  
 We conclude that for input states corresponding to $X$ eigenstates, the distinguishability of states decays exponentially with the depth $T$ of the tree (the same holds true for $Y$ eigenstates as well). 

\subsection{CSS code trees with Standard Encoders}\label{std css enc sec}

As mentioned in Sec.\ref{css intro} every CSS code has encoders that are standard for both $Z$ and $X$ directions. Therefore, for trees constructed from such encoders, we can apply this bound to both $Z$ and $X$ directions. Let $d_z$ and $d_x$ be the minimum weight of logical $Z$ and $X$ operators, respectively. Then, applying the triangle inequality, we find\footnote{See Eq.(\ref{rt10})} that the trace distance of the outputs of the channel $\mathcal{E}_T$ for an arbitrary input density operator $\rho$ and the maximally-mixed state is bounded by 
\begin{align}\label{stand}
&\big\|\mathcal{E}_T(\rho)-\mathcal{E}_T(\frac{I}{2}) \big\|_1\le\nonumber \\  &{\sqrt{2}} \times \Big([\sqrt{d_z}|1-2p_z|]^{T}+[\sqrt{d_x}  |1-2p_x|]^{T}\Big)\ .
\end{align}
Another useful way of characterizing the error in channel $\mathcal{E}_T$ is in terms of the probability of logical errors. Recall that for CSS codes with standard encoders and independent $Z$ and $X$ errors, the optimal error correction can be performed independently for $Z$ and $X$ errors, and after optimal error correction, the overall channel is
\begin{align}\nonumber
    \mathcal{R}^{\text{opt}}_T  \circ \mathcal{E}_T=\mathcal{Q}_z\circ \mathcal{Q}_x =\mathcal{Q}_x\circ \mathcal{Q}_z \ ,
\end{align}
where $\mathcal{Q}_z(\rho)=q^z_T\rho + (1-q^z_T)Z\rho Z$ is a phase-flip channel, and  $\mathcal{Q}_x(\rho)=q^x_T\rho + (1-q^x_T)X\rho X$ is a bit-flip channel.   
Then, the data-processing inequality for diamond norm distance implies that
\begin{align}
\|\mathcal{Q}_z\circ \mathcal{Q}_x-\mathcal{Q}_z \circ \mathcal{Q}_x\circ\mathcal{D}_z\|_\diamond &= \|\mathcal{R}_T\circ \mathcal{E}_T-\mathcal{R}_T\circ\mathcal{E}_T\circ\mathcal{D}_z\|_\diamond\nonumber\\
&\leq \| \mathcal{E}_T-\mathcal{E}_T\circ\mathcal{D}_z\|_\diamond.
\end{align}
The left-hand side is equal to $|1-2q^z_T|$ (see Appendix \ref{ent depth uncorrelated XZ}) and the right-hand side is  bounded by  Eq.(\ref{mainbound}) in  proposition \ref{Prop 1}. Therefore, 
\be\label{mainbound css}
|1-2q^z_T|\leq {\sqrt{2}} \times \Big[\sqrt{d_z} \times |1-2p_z|\Big]^{T}\ ,
\ee
and a similar bound holds for logical $X$ probability as well. 
We conclude that when  $p_x=p_z=p$, if
\begin{align}\label{b1}
|1-2p|^{2}\times \max\{d_x,d_z\}<1\ ,
\end{align}
then the infinite tree does not transfer any information,  whereas 
if 
\begin{align}\label{EB-cond}
|1-2p|^2\times \min\{d_x,d_z\}<1\ ,
\end{align}
then it 
is entanglement-breaking but it may still transfer classical information in either $Z$ or $X$ input basis. Here, 
\be
d=\min\{d_x,d_z\}
\ee
is the code distance, which according to  the quantum singleton bound  \cite{neilsenandchuang}  satisfies $d \leq (b+1)/2$. Therefore, if $|1-2p|\le \sqrt{2/(b+1)}$ then the  infinite tree does not transmit entanglement.  On the other hand,
using 
 the classical result of \cite{evans2000} discussed in the introduction,  we know that for the repetition code, which is a CSS code, classical information is transmitted  for $|1-2p|> \sqrt{{1}/{b}}$. We conclude that in the noise regime 
 \begin{align}
     \sqrt{\frac{1}{b}}<|1-2p|< \sqrt{\frac{2}{b+1}} \ ,
 \end{align}
there are CSS codes with standard encoders transmitting  classical information to any depth, whereas entanglement can not be transmitted by such encoders to infinite depth.

Note that even if the noise is stronger than the threshold set by Eq.(\ref{EB-cond}), the channel $\mathcal{E}_T$ may still transmit entanglement for a finite depth $T$. More precisely, the single-qubit channel  $\mathcal{R}^{\text{opt}}_T\circ\mathcal{E}_T$  is not necessarily entanglement-breaking. Equivalently,  for a maximally-entangled  state $|\Phi\rangle$ of a pair of qubits, the two-qubit state 
\be\nonumber
(\mathcal{R}^{\text{opt}}_T\circ\mathcal{E}_T)\otimes \text{id}(|\Phi\rangle\langle\Phi|)
\ee
can be entangled for a finite $T$, where $\text{id}$ denotes the identity channel on a \emph{reference} qubit. 
  For concreteness,  assume that the probability of $Z$ errors is above the threshold, i.e., 
 $|1-2p_z|^2\times d_z<1 $. Then, as we show in Appendix \ref{ent depth appendix}, for   
\begin{align}\label{ent depth final}
T> \frac{c}{-\ln(\sqrt{d_z}|1-2p_z|)} \ ,
\end{align}
the channel $\mathcal{E}_T$ is entanglement-breaking, where $c$ is a constant that only depends on the probability of logical $X$ error, namely   
$c=\ln\big(\frac{\sqrt{2}(1-q^x_1)}{q^x_1}\big)$, which is finite for $p_x>0$ (recall that $q^x_1$ is the probability of logical $X$ error for a tree of depth $1$).

\begin{figure}
    \centering
    \includegraphics[width=0.5\textwidth]{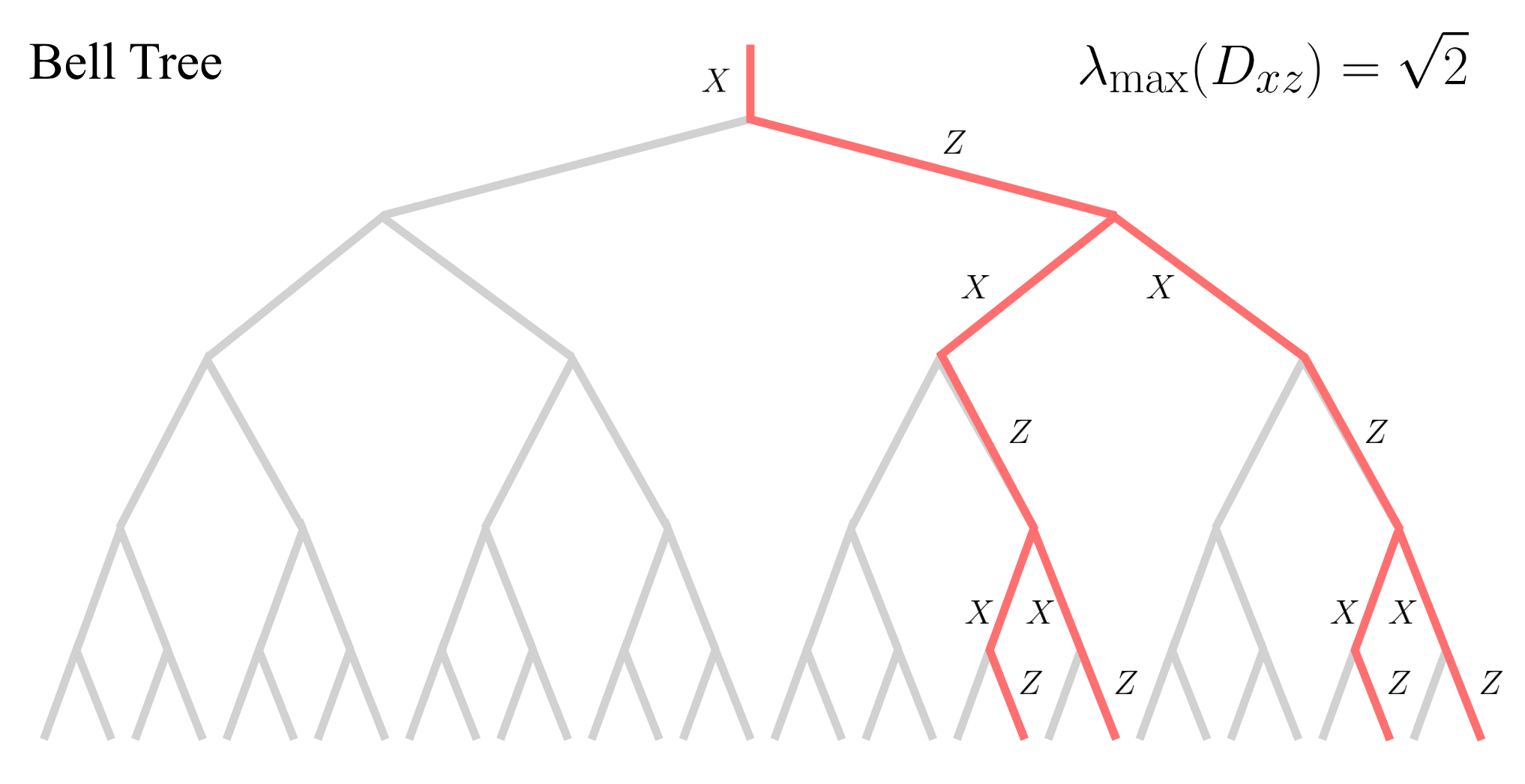}
    \caption{\textbf{A logical subtree of the Bell Tree:} This figure illustrates the $X$ logical subtree associated with the Bell tree, 
 defined in Fig. \ref{BinaryvsBell}.  Recall that  this encoder is an anti-standard encoder of a CSS code, namely the binary repetition code.   
    Clearly, in alternate levels, the tree branches into  either $1$ or $2$ (i.e., $d_x$ or $d_z$) children, resulting in an effective branching factor of $\sqrt{1 \times 2}$.}  \label{Bell}
\end{figure}

\subsection{CSS code trees with `Anti-standard' encoders: Bell Tree}\label{anti standard subsec}
Consider a standard encoder for a CSS code, described by an isometry $V$. Now 
suppose before applying this encoder,  we apply the Hadamard gate $H$ on the qubit. The resulting encoder, described by the isometry $V H$ has logical $Z$ operators 
$Z_L\in \langle i I,X \rangle ^ {\otimes b}$  and the logical $X$ operator  $X_L\in \langle i I,Z \rangle ^ {\otimes b}$ with weights $d_z={\rm weight}(Z_L)$ and $d_x={\rm weight}(X_L)$, respectively. We refer to this type of encoder as an ``anti-standard" encoder.

We show that in this 
case Eq.(\ref{mainbound css}) will be modified to
\be
|1-2q^x_T|\leq {\sqrt{2}} \times \sqrt{d_x^{\lceil T/2 \rceil}\times d_z^{\lfloor T/2 \rfloor} } \times |1-2p|^{T}\ ,
\ee
where for simplicity we have assumed $p=p_x=p_z$  (see lemma \ref{lem}). 
Note that the quantity $d_x^{\lceil T/2 \rceil}\times d_z^{\lfloor T/2 \rfloor} $ is the weight of a logical $X$ operator for the encoder $V_T$.  
A similar bound can be obtained for $q^z_T$ by exchanging  $d_z$ and $d_x$ in the right-hand side of this equation.

 We conclude that for $p$ satisfying 
\begin{align}
|1-2p|^{2}\times \sqrt{d_x\times d_z}<1\ ,
\end{align}
the infinite tree does not transfer any information (note that  this lower bound is the geometric mean of the lower bounds in Eq.(\ref{b1}) and Eq.(\ref{EB-cond}) for standard encoders).

As a simple example, 
consider the repetition code with the isometry $V|c\rangle=|c\rangle^{\otimes b}: c=0,1$, for which $d_x=b$ and $d_z=1$. From the classical result of \cite{evans2000} discussed in the introduction, one can show that the infinite tree constructed from this code does not transmit classical information in  $|0\rangle, |1\rangle$ basis if, and only if 
 \be
 |1-2p_x|^{2}\times b<1\ ,
 \ee
 whereas if $p_z \neq 0,1$,  it does not transfer any information in $X$ (or, $Y$) basis.       
 On the other hand, when we add the Hadamard gate and convert the standard encoder to an anti-standard encoder, no information is transmitted over the infinite tree if
 \be
 |1-2p|^{2} \times  \sqrt{b}<1 \ ,
 \ee
 where we assume $p_x=p_z=p$. Therefore,  adding the Hadamard gates lowers the noise threshold for transferring classical information encoded in the input $Z$ basis. However, as we will show in the example of the Bell tree, which corresponds to $b=2$, this allows transmission of information encoded in the input $X$ basis and entanglement, even for non-zero $p_z>0$.  It is worth noting that 
 in the absence of noise, the channel $\mathcal{V}_2$ obtained from two layers of the anti-standard encoder is indeed the encoder of the generalized Shor code (see Appendix Sec.\ref{subsec: shor rec} for further discussion).

\subsection{Stabilizer Trees with General Encoders}\label{Sec:gen}
In this section, we establish the exponential decay of information for stabilizer trees  constructed from general encoders that are not necessarily standard or anti-standard. However, before presenting the most general case in proposition \ref{prop3}, first, we consider encoders that are constructed only from CNOTs and Hadmarad gates. The main relevant property of such encoders is that they have logical $X$ and  $Z$ operators  $L_x, L_z \in \langle i I,\sigma_z , \sigma_x\rangle ^ {\otimes b}$, such that $L_x$ and $L_z$ do not act as  $\sigma_y$ operator on any qubits (in this section, for convenience, we use $\sigma_x , \sigma_y,$ and $\sigma_z$, to denote Pauli $X$, $Y$, and $Z$ operators, respectively). 
 Then, we show that in this case, a modification of the bound in Eq.(\ref{stand}) holds.

\begin{proposition}\label{prop2}
Let $L_x, L_z \in \langle i I,\sigma_x , \sigma_z\rangle ^ {\otimes b}$
be logical $\sigma_x$ and logical $\sigma_z$ operators for the encoder $V:\mathbb{C}^2\rightarrow(\mathbb{C}^2)^{\otimes b}$, such that $L_w V=V \sigma_w\ : w=x,z$.  
For $v,w\in\{x,z\}$ let $n(v\rightarrow w)$ be the number of qubits on which logical $\sigma_v$ acts as $\sigma_w$, and  
 \be\nonumber
 b_\text{max}=\max\{{\rm weight}(L_x), {\rm weight}(L_z) \}\ ,
 \ee
 be the maximum weight of these two logical operators.  
 Let $\lambda_\text{max}(D_{xz})$ be the maximum eigenvalue (spectral radius) of  the weight transition matrix
\be\label{transition}
D_{xz}=
\left(
\begin{array}{cc}
n(x\rightarrow x)\   &   n(z\rightarrow x)\    \\
  n(x\rightarrow z)\   & n(z\rightarrow z)\  
\end{array}\right)
\ .
\ee
Assume the noise channel $\mathcal{N}$ that defines the tree channel $\mathcal{E}_T$ in Eq.(\ref{noisy tree definition -- no root noise})  is $\mathcal{N}=\mathcal{N}_x\circ \mathcal{N}_z$, where 
$\mathcal{N}_x$ and $\mathcal{N}_z$ are, respectively, the bit-flip and phase-flip channels with probability $p_x=p_z=p$. If 
\be\label{D_xz upper}
(1-2p)^{2} \times \min\{\lambda_\text{max}(D_{xz}) ,  b_\text{max}\} <1 \ ,
\ee
then in the limit $T\rightarrow\infty $, 
the output of channel $\mathcal{E}_T$ becomes independent of the input.  More precisely, for any single-qubit density operator $\rho$, it holds that
\begin{align}\label{xz}
&\big\|\mathcal{E}_T(\rho)-\mathcal{E}_T(\frac{I}{2}) \big\|_1\le {2\sqrt{2}} \times \sqrt{g_{xz}(T)}\times |1-2p|^{T}\ .
\end{align}
Here,
\be
g_{xz}(T)=\max_{w\in\{x,z\}} 
(1,1) D^T_{xz} e_w \le b^T_\text{max}\ ,
\ee
is an upper bound on the weights of logical  $\sigma_x$ and $\sigma_z$ operators at level $T$, where  $e_x=(1 , \ 0 )^{\rm T}$, $e_z=(0 , \ 1 )^{\rm T}$.
\end{proposition}

\begin{figure}
    \centering
    \includegraphics[width=0.46\textwidth]{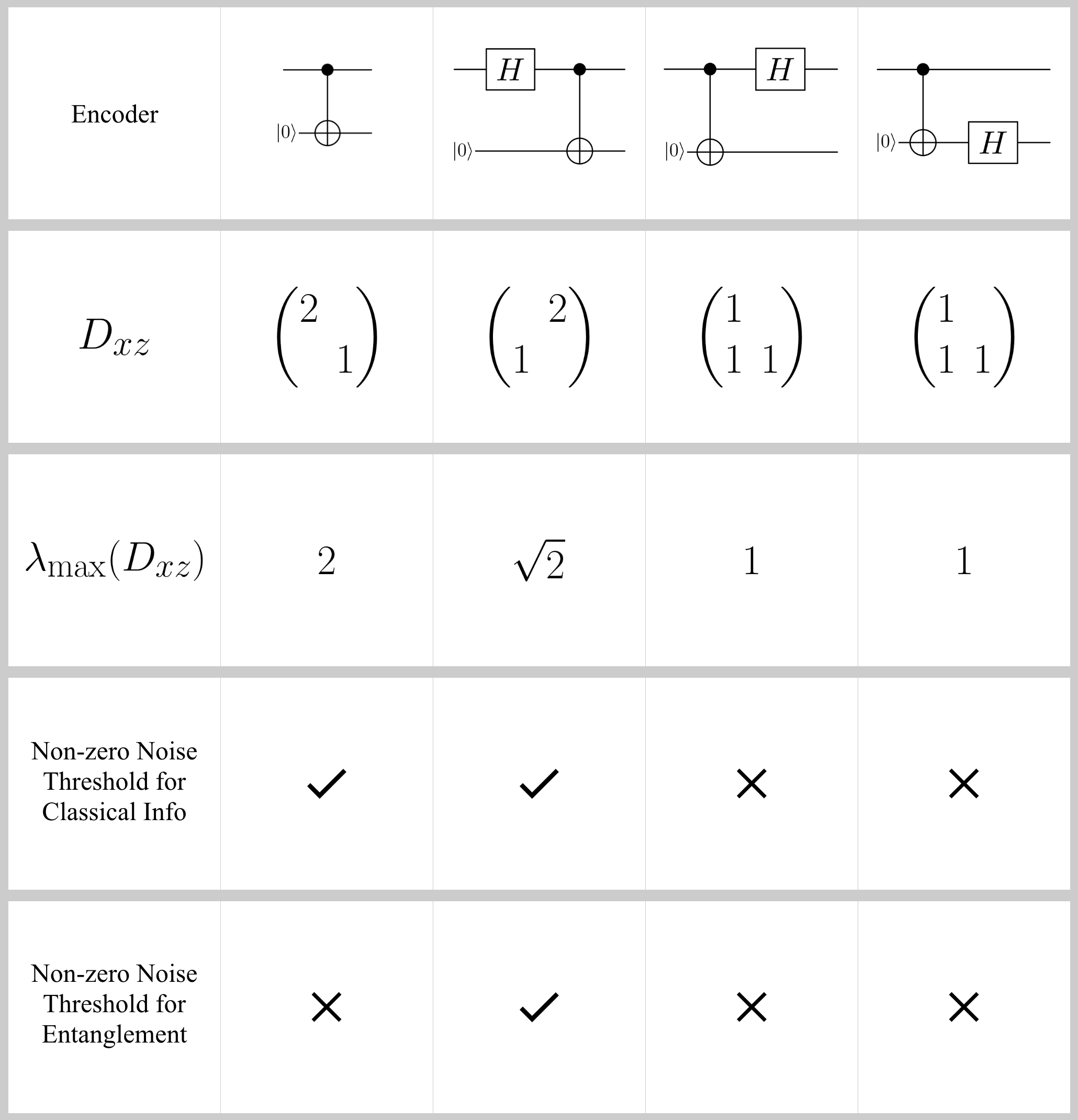}
    \caption{The  encoder of the  binary repetition code and 3 of its variations  obtained by adding a single Hadamard gate to its inputs/outputs. In every circuit, the top qubit is the input data qubit and the bottom qubit is the ancilla initialized in state $\ket{0}$. In addition to these 3 variations, there is a $4^{\rm th}$ variation in which the Hadamard acts on the ancilla before applying the CNOT gate. However, in that case the encoder acts trivially on the input state. For each of these encoders, we indicate their weight transition matrix $D_{xz}$ defined in Eq.(\ref{transition}) and its largest eigenvalue $\lambda_\text{max}(D_{xz})$. Furthermore, we indicate whether there exists a non-zero noise threshold below which the tree transmits classical information and/or entanglement to any depth. 
    Eq.(\ref{D_xz upper}) puts an upper bound on the noise threshold for classical information in terms of  $\lambda_\text{max}(D_{xz})$.  In particular, for $\lambda_\text{max}(D_{xz})=1$ information does not propagate over the infinite tree; this is the case for the $3^{\rm rd}$ and $4^{\rm th}$ encoders. The second column, which is the encoder of the Bell tree in Fig.\ref{BinaryvsBell} yields a non-zero noise threshold for the propagation of  classical information and entanglement   on an infinite tree. 
 It is worth  noting that  the transition matrix $D_{xz}$ is not unique. For instance, for the third encoder one can choose $Z_L=XI$ or $Z_L=IZ$. With $X_L=ZX$, the former choice yields $\lambda_{\rm max}(D_{xz})=(1+\sqrt{5})/2$, whereas for the latter $\lambda_{\rm max}(D_{xz})=1$; we pick the lower value for the tighter upperbound.}
    \label{Fig-Bell-variation}
\end{figure}

Note that the maximum eigenvalue of $D_{xz}$ is 
\be
\lambda_\text{max}(D_{xz})=n_{\text{avg}}+
\sqrt{n_{\text{dif}}^2+n_{\text{cross}}^2}\ ,
\ee
where
\bes
\begin{align}
n_{\text{avg}}&=\frac{1}{2}\big[n(x\rightarrow x)+n(z\rightarrow z)\big]\\ n_{\text{dif}}&=\frac{1}{2}\big[n(x\rightarrow x)-n(z\rightarrow z)\big]\\n_{\text{cross}}&=\sqrt{n(x\rightarrow z)\times n(z\rightarrow x) }\ .
\end{align}
\ees
Roughly speaking, this quantity determines the maximum rate of   growth of the logical subtrees 
in the regime $T\rightarrow\infty$
(see Fig.\ref{restricted_subtree} and Fig.\ref{Bell}). 
To obtain the strongest bound on the decay of information, one should consider the logical $X$ and $Z$ operators for which $\lambda_\text{max}(D_{xz})$ is minimized.

It is worth considering  the two special cases of standard and anti-standard encoders of CSS codes.  For standard encoders, there exist logical operators with  
$$n(z\rightarrow z)=d_z\ , \ \ \ \  n(z\rightarrow x)=0 $$ and 
$$n(x\rightarrow x)=d_x\ ,\ \  \ \  n(x\rightarrow z)=0\ ,$$
which implies 
\be
\lambda_\text{max}(D_{xz})=\frac{1}{2}\big(d_x+d_z\big)+\frac{1}{2}|d_x-d_z|=\max\{d_x,d_z\}\ .
\ee
Next, suppose we add a Hadamard to this encoder to obtain an anti-standard encoder with
$$n(z\rightarrow z)=0\ , \ \   \ \ n(z\rightarrow x)=d_x $$ and 
$$n(x\rightarrow x)=0\ ,\ \  \ \  n(x\rightarrow z)=d_z\ ,$$
which implies 
\be
\lambda_\text{max}(D_{xz})=\sqrt{d_x\times d_z}\ .
\ee
Therefore, the bound in Eq.(\ref{xz}) generalizes the special bounds  we have previously seen in the case of standard and anti-standard encoders.

Finally, we consider the most general stabilizer encoder which can be an arbitrary Clifford unitary. In this case, the logical $\sigma_x$ and $\sigma_z$ operators contain $\sigma_x$, $\sigma_y$, and $\sigma_z$. Then, in this situation, it is more natural to assume a noise model where $X$, $Y$, $Z$ errors happen independently, each with probability $p_x=p_y=p_z=p\le 1/2$. This corresponds to the depolarizing channel 
 \be\label{depol}
\mathcal{M}_\epsilon(\rho)=\mathcal{N}_x\circ\mathcal{N}_y\circ\mathcal{N}_z(\rho)=(1-\epsilon)\rho+ \epsilon \frac{I}{2}\ ,
 \ee
where $\epsilon=4p(1-p)$,    and  $\mathcal{N}_w(\rho)=p\sigma_w \rho \sigma_w+(1-p)\rho$ for $w=x,y,z$.

\begin{proposition}\label{prop3}
Let $L_x, L_y, L_z \in \langle i I,\sigma_x , \sigma_y, \sigma_z\rangle ^ {\otimes b}$
be logical $\sigma_x$, $\sigma_y$, and  $\sigma_z$ operators for the encoder $V:\mathbb{C}^2\rightarrow(\mathbb{C}^2)^{\otimes b}$, such that $L_w V=V \sigma_w\ : w=x, y ,z$.  
Let $n(v\rightarrow w)$ be the number of qubits on which logical $\sigma_v$ acts as $\sigma_w$
and consider the corresponding  $3\times 3$ matrix
\be
D_{xyz}=
\left(
\begin{array}{ccc}
n(x\rightarrow x)\   & n(y\rightarrow x)\    &  n(z\rightarrow x)\   \\n(x\rightarrow y)\   & n(y\rightarrow y)\    &  n(z\rightarrow y)\   \\
  n(x\rightarrow z)\   & n(y\rightarrow z)\    &  n(z\rightarrow z)\  
\end{array}\right)
\ .
\ee
Let $\lambda_\text{max}(D_{xyz})$ be the spectral radius of $D_{xyz}$, i.e., the largest absolute value of the eigenvalues of $D_{xyz}$, and 
\be\nonumber
 b_\text{max}=\max\{{\rm weight}(L_x), {\rm weight}(L_y),  {\rm weight}(L_z) \}\ ,
 \ee
be the maximum weight of these logical operators. 
Assume the noise channel $\mathcal{N}$ that defines the tree channel $\mathcal{E}_T$ in Eq.(\ref{noisy tree definition -- no root noise}) is the depolarizing channel $\mathcal{M}_\epsilon$ in Eq.(\ref{depol}).   If 
\be \label{perc_like}
(1-\epsilon)\times \min\{\lambda_\text{max}(D_{xyz}) ,  b_\text{max}\}<1\ ,
\ee 
then in the limit $T\rightarrow\infty $, 
the output of channel $\mathcal{E}_T$ becomes independent of the input. More precisely, for any single-qubit density operator $\rho$, it holds that
\begin{align}\label{xyz}
&\big\|\mathcal{E}_T(\rho)-\mathcal{E}_T(\frac{I}{2}) \big\|_1\le {2\sqrt{2}} \times \sqrt{g_{xyz}(T)\times (1-\epsilon)^T} \ .
\end{align}
Here, 
\be
g_{xyz}(T)=\max_{w\in \{x,y,z\}}\begin{pmatrix}
    1 & 
    1 &
    1
    \end{pmatrix} D^T_{xyz} e_w\le  b_\text{max}^T\ ,
\ee
is an upper bound on the weights of logical operators in level $T$, 
where $e_x=(1, 0 , 0)^{\rm T}$, and $e_y$ and $e_z$, are defined in a similar fashion.
\end{proposition}

To prove this result, we first show the following lemma, which is a generalization of proposition \ref{Prop 1}, and is  of independent interest. Here, $\mathcal{D}_w$ is the single-qubit channel that fully dephases the input qubit in the eigenbasis of $\sigma_w$, i.e., 
$\mathcal{D}_w(\rho)=(\rho+\sigma_w\rho \sigma_w)/2$.

\begin{lemma}\label{lem}
Let ${L}_w[T]\in \langle iI, \sigma_x, \sigma_y, \sigma_z\rangle^{\otimes b^T}$ be a logical $\sigma_w$ operator, such that $ {L}_w[T] V_T=V_T\sigma_w$.  Let 
$\mathcal{N}$ be the noise channel  that defines the tree channel $\mathcal{E}_T$ in Eq.(\ref{noisy tree definition -- no root noise}). Then,
\be
\big\|\mathcal{E}_{T}-\mathcal{E}_{T}\circ \mathcal{D}_{w}\big\|^2_\diamond\le 2\times r(T)  \times {\rm weight} ({L}_w[T]) \ ,
\ee  
where 
\begin{itemize}
\item[(i)] $w=x,y,z$ and  $r(T)=(1-\epsilon)^{T}$,  provided that $\mathcal{N}$ is the depolarizing channel $\mathcal{M}_\epsilon$ in Eq.(\ref{depol}).
\item[(ii)] $w=x, z$, and  
 $r(T)=(1-2 p)^{2 T}$, provided that $\mathcal{N}=\mathcal{N}_z\circ \mathcal{N}_x$,  with the probability of bit-flip and phase-flip errors $p_x=p_z=p$,  and  provided that there exists a sequence of operators
${L}_w[0]\rightarrow L_w[1]\rightarrow \cdots \rightarrow  L_w[T]$ 
 such that $L_w[0]=\sigma_w$, $L_w[t]\in  \langle i I,\sigma_x , \sigma_z\rangle ^ {\otimes b^{t}}$, and  
${L}_w[t+1]V^{\otimes b^t}= V^{\otimes b^t}{L}_w[t]$ for all $t\geq0$.
\end{itemize}
\end{lemma}
In particular, in the first case of this lemma, if there exists a sequence of logical operators ${L}_w[T]$ such that
\be
\lim_{T\rightarrow \infty} (1-\epsilon)^{T} \times {\rm weight}({L}_w[T])=0\ ,
\ee
then the infinite tree is entanglement-breaking and might only transfer classical information encoded in the eigenbasis of $\sigma_w$. Note that this lemma does not assume the restriction that all encoders in the tree must be the same. Indeed, it can also be extended to stabilizer trees where the Clifford encoders are sampled \textit{randomly}.

\subsection*{Comparison with bond percolation bound on noise threshold}
The depolarizing channel $\mathcal{M}_\epsilon(\rho)=(1-\epsilon)\rho+ \epsilon \frac{I}{2}$ has a simple probabilistic interpretation: it discards the input qubit with probability $\epsilon$, and replaces it with the maximally mixed state.  Therefore,
in a quantum tree with this  depolarizing noise channel on every edge, with probability $\epsilon$, \textit{no} information about the root is transmitted through each edge. This immediately allows us to apply the standard results from \textit{percolation} theory \cite{LyonsPerestextbook}, where each edge (or `bond') of the tree is independently removed from the tree with probability $\epsilon$. It is well-known that for a tree with branching number $b$, if $(1-\epsilon)\times b<1$,  the probability of the existence of a continuous connected path from the root to the leaves in an infinite tree will approach zero \cite{evans2000, Evans-Schulman}.
Note that the existence of such path  from the root to leaves is only a necessary, but not a sufficient requirement for non-zero information propagation down to the leaves of a tree network.  

We conclude that a quantum tree with depolarizing channel $\mathcal{M}_\epsilon$ on each edge will not transmit information to infinite depth when $(1-\epsilon)\times b<1$. Thus, the bond percolation threshold upperbound, i.e., $1-1/b$, is an upperbound for the noise threshold for the propagation of information over infinite depth for any noisy quantum tree (not necessarily stabilizer trees), including trees where encoders are sampled randomly.  
  
It is worth comparing this bound with the bound in Eq.(\ref{perc_like}), which is established for stabilizer trees: in general, $\min\{\lambda_\text{max}(D_{xyz}) ,  b_\text{max}\}\leq b$. This means the latter bound will be  stronger than the percolation bound (that is, it restricts the noise level $\epsilon$ that allows transmission of information over infinite tree to a smaller range $\epsilon \le 1-{1}/{\min\{\lambda_\text{max}(D_{xyz}) ,  b_\text{max}\}}$).

\subsection{Proof of Proposition \ref{Prop 1}}\label{Sec:mainbound}

First, recall that the effect of noise inside the tree is equivalent to correlated Pauli noise on the leaves. That is
\be\label{Eqt}
\mathcal{E}_T=\mathcal{F}_T\circ \mathcal{V}_T\ ,
\ee
 where $\mathcal{V}_T$ defined in Eq.(\ref{noisless isometry}) is the ideal encoder of the concatenated code associated with the tree, and $\mathcal{F}_T$ are correlated Pauli errors on the leaves. 
 In the following, we prove the proposition  for the special case of $p_x=0$, i.e., when there are no $X$ errors. Clearly, in the presence of $X$ errors, due to the 
 contractivity of the diamond norm under data processing, the upper bound remains valid. Then, for the standard encoders, $\mathcal{F}_T$ can be thought of as a channel that randomly applies Pauli $Z$ operators on some qubits. To express this we use the following notation: for any bit string $\textbf{z}=z_1\cdots z_n\in\{0,1\}^n$ define
 \be
\textbf{Z}^{\bf{z}}:=Z^{z_1}\otimes \cdots \otimes Z^{z_{n}}\ ,
\ee
which means a Pauli $Z$ is applied on those qubits where $\bf{z}$ takes value 1. Then, for the  standard encoders, $\mathcal{F}_T$ can be written as
 \be\label{def1}
\mathcal{F}_T(\cdot)=\sum_{\textbf{z}\in\{0,1\}^{b^T}} f_T(\textbf{z})\ \textbf{Z}^{\textbf{z}} (\cdot) \textbf{Z}^{\textbf{z}}\ ,
\ee
where $f_T$ is a probability distribution over $\{0,1\}^{b^T}$.  It is worth emphasizing that  because the output of $\mathcal{V}_T$ is restricted to a 2D code subspace, 
the channel
$\mathcal{F}_T$, and hence the  probability distribution $f_T$, 
are not uniquely defined (as we further explain below, this is related to the freedom of choosing logical operators for a code).

The goal is to bound
$\big\|\mathcal{E}_{T}-\mathcal{E}_{T}\circ \mathcal{D}_{z}\big\|_\diamond=\frac{1}{2} \big\|\mathcal{E}_{T}-\mathcal{E}_{T}\circ \mathcal{Z}\big\|_\diamond$,  
 where $\mathcal{Z}(\cdot)=Z(\cdot)Z$ is the channel that applies Pauli $Z$ on the input qubit.
We note that
 \be
\mathcal{E}_{T}\circ \mathcal{Z}=\mathcal{F}_T\circ \mathcal{V}_{T}\circ \mathcal{Z}=\mathcal{F}_T\circ \mathcal{Z}_L[T]\circ\mathcal{V}_{T} \ ,
 \ee
where $\mathcal{Z}_L[T]$ is the quantum channel corresponding to a logical $Z$ operator at the leaves. More precisely, 
\be
\mathcal{Z}_L[T](\cdot)=\textbf{Z}^{\textbf{z}_L[T]}(\cdot)\textbf{Z}^{\textbf{z}_L[T]}\ ,
\ee
where $\textbf{z}_L[T]\in \{0,1\}^{b^T}$ is a bit string of length $b^T$ satisfying
\be\label{logic}
\eta V_T Z=\textbf{Z}^{\textbf{z}_L[T]}V_T \ ,
\ee
where $\eta\in\{\pm 1, \pm i\}$ is a global phase. In other words, $\textbf{z}_L[T]$
determines on which qubits  the logical operator $\eta\textbf{Z}^{\textbf{z}_L[T]}$ acts as the identity operator and on which qubits it acts as  Pauli $Z$.

We conclude that
\begin{align}
 \big\|\mathcal{E}_{T}-\mathcal{E}_{T}\circ \mathcal{Z}\big\|_\diamond \nonumber &= \big\|\mathcal{F}_T\circ \mathcal{V}_{T}-\mathcal{F}_T\circ\mathcal{Z}_L[T]\circ \mathcal{V}_{T}\big\|_\diamond \nonumber\\ &\le \big\|\mathcal{F}_T-\mathcal{F}_T\circ\mathcal{Z}_L[T]\big\|_\diamond \ ,
\end{align}
 where the inequality follows from the contractivity of the diamond norm. 

Next, using Eq.(\ref{def1}) we note that 
\be
\mathcal{F}_T\circ\mathcal{Z}_L[T](\cdot)=\sum_{\textbf{z}\in\{0,1\}^{N}} f_T(\textbf{z}+\textbf{z}_{L}[T])\ \textbf{Z}^{\textbf{z}} (\cdot) \textbf{Z}^{\textbf{z}}\ ,
\ee
where $\textbf{z}+\textbf{z}_{L}[T]$ is the bitwise addition of $\textbf{z}$ and $\textbf{z}_{L}[T]$ mod 2. Then, using properties of the diamond norm, this  implies
\begin{align}
&\big\|\mathcal{F}_T-\mathcal{Z}_L[T]\circ\mathcal{F}_T\big\|_\diamond\nonumber\\ &=\big\|\sum_{\textbf{z}\in\{0,1\}^{N}} \big[f_T(\textbf{z})- f_T(\textbf{z}+\textbf{z}_{L}[T])\big]\ \textbf{Z}^{\textbf{z}} (\cdot) \textbf{Z}^{\textbf{z}}\big\|_\diamond
\nonumber \\ &=
\sum_{\textbf{z}\in\{0,1\}^{N}} \big|f_T(\textbf{z}+\textbf{z}_{L}[T])-f_T(\textbf{z})\big|\ ,
\end{align}
where the last step is explained in Appendix \ref{appendix: diamond pauli} (this can be seen, e.g., using  the triangle inequality for the diamond norm, which implies the left-hand side is less than or equal to the right-hand, together with the fact that the equality is achieved when all the qubits are prepared in state $|+\rangle$). 

Putting everything together we conclude that 
\begin{align}\label{52}
\big\|\mathcal{E}_{T}-\mathcal{E}_{T}\circ \mathcal{D}_{z}\big\|_\diamond&\le \frac{1}{2}\big\|\mathcal{F}_T-\mathcal{Z}_L[T]\circ\mathcal{F}_T\big\|_\diamond \\ &=\frac{1}{2}\sum_{\textbf{z}\in\{0,1\}^{N}} \big|f_T(\textbf{z}+\textbf{z}_{L}[T])-f_T(\textbf{z})\big|\ .\nonumber
\end{align}
The second line is the total variation distance between the distribution 
$f_T$ and the distribution obtained from $f_T$ by flipping all those bits where $\textbf{z}_{L}[T]$ is 1. As mentioned before, neither  $f_T$ nor $\textbf{z}_{L}[T]$ are unique. One is free to choose any distribution $f_T$ such that the corresponding channel $\mathcal{F}_T$  defined in Eq.(\ref{def1}) satisfies  Eq.(\ref{Eqt}). Similarly, one can choose any logical operator 
$\textbf{Z}^{\textbf{z}_L[T]}$ that satisfies  Eq.(\ref{logic}). But, clearly, given the structure of the tree, it is natural to choose the logical operator 
$\textbf{Z}^{\textbf{z}_L[T]}$  and the distribution $f_T$ in a way that is defined by the recursive application of the same rule on all nodes of the tree.

In particular, let $\eta \textbf{Z}^{\textbf{z}_L[1]}$ be a logical $Z$ operator for the original code that is applied at each node of the tree, i.e.,   
\be
VZ=\eta \textbf{Z}^{\textbf{z}_L[1]}V\ ,
\ee
where $\eta \in\{\pm1,\pm i\}$ is a global phase. Then, we can define a chain of bit strings
\be\label{seq}
1 \rightarrow \textbf{z}_L[1]\rightarrow \cdots  \rightarrow \textbf{z}_L[T]\ ,
\ee
where $\textbf{z}_L[t+1]\in\{0,1\}^{b^{t+1}}$ is 
obtained from $\textbf{z}_L[t]\in\{0,1\}^{b^{t}}$ by replacing any 0 in $\textbf{z}_L[t]$  with $0^b$ and any 1 with $\textbf{z}_L[1]$.

Then, for all $t$ it holds that
\be\label{atr}
V_t Z=\eta  \textbf{Z}^{\mathbf{z}_L[t]} V_t\ ,
\ee
for some $\eta\in\{\pm 1, \pm i\}$.
For each level $t\le T$,   consider the nodes  where $\textbf{z}_L[t]$ takes value 1. The set of such nodes 
 defines a logical subtree.  For instance, in the example in Figure \ref{restricted_subtree}, this corresponds to the highlighted 3-ary subtree.

Next, we discuss  the distribution $f_T$. Recall that this distribution should be chosen such that Eq.(\ref{Eqt}) is satisfied by $\mathcal{F}_T$. 
We choose to define this distribution recursively in  a way that is consistent  with the above chain of logical operators. 
  In particular, we define $f_T$ via the recursive equation
\begin{align}\label{tr}
\mathcal{F}_{t+1}(\cdot)&=\sum_{\textbf{z}\in\{0,1\}^{b^{t+1}}} f_{t+1}(\textbf{z})\ \textbf{Z}^{\textbf{z}} (\cdot) \textbf{Z}^{\textbf{z}}\ ,
\\ &=\mathcal{N}_z^{\otimes b^{t+1}}\circ \sum_{\textbf{w}\in\{0,1\}^{b^{t}}} f_t(\textbf{w})\ \textbf{Z}_L^{\textbf{w}}  (\cdot) \textbf{Z}_L^{\textbf{w}}\ \nonumber ,
\end{align}
where $f_0(0)=1$ and $f_0(1)=0$, and for any $\textbf{w}=w_1\cdots w_t$ $\textbf{Z}^\textbf{w}_L$ denotes
\begin{align}
\textbf{Z}_L^{\textbf{w}}&=\textbf{Z}^{w_1\times \textbf{z}_L[1]}\otimes \cdots \cdots \otimes \textbf{Z}^{w_t\times \textbf{z}_L[1]}
=\textbf{Z}^{\textbf{w}'}\ ,
\end{align}
where ${\textbf{w}'}$
is the bit string of length $b^{t+1}$ obtained from $\textbf{w}$ by replacing all $0$ and $1$  in $\textbf{w}$ by $0^b$ and $\textbf{z}_L[1]$, respectively.   This definition implies 
\be\label{rtre}
\textbf{Z}^{\textbf{w}'}V^{\otimes b^t}=\textbf{Z}_L^{\textbf{w}}V^{\otimes b^t}=V^{\otimes b^t}\textbf{Z}^{\textbf{w}}\ . 
\ee
Composing both sides of Eq.(\ref{tr}) from the right-hand side with $\mathcal{V}^{\otimes b^t}$ and applying Eq.(\ref{rtre}) implies
\be
\mathcal{F}_{t+1}\circ \mathcal{V}^{\otimes b^t}=\mathcal{N}_z^{\otimes b^{t+1}}\circ \mathcal{V}^{\otimes b^t}\circ \mathcal{F}_{t}\ , 
\ee
for all $t$, which in turn implies 
\be
\mathcal{F}_{T}\circ \mathcal{V}_T=\prod_{j=0}^{T-1}    \mathcal{N}_z^{\otimes b^{j+1}} \circ \mathcal{V}^{\otimes b^j}=\mathcal{E}_T\  , 
\ee
as required.

Note that Eq.(\ref{tr}) has a simple interpretation: to 
sample a bit string  $\textbf{z}\in \{0,1\}^{b^{t+1}}$ according to the distribution $f_{t+1}$, one first samples a bit string  $\textbf{w}\in\{0,1\}^{b^t}$ according to the distribution $f_{t}$, then  replace  any 0 in $\textbf{w}$ with $0^b$ and any 1 in $\textbf{w}$ with $\textbf{z}_L[1]$, and finally flips  each bit randomly and independently with probability $p_z$. 
This definition has an  important implication for the distribution $f_T$:
the bits inside the logical subtree defined by the  sequence of logical operators in  Eq.(\ref{seq}) 
are uncorrelated with the bits outside this subtree.  
\footnote{It is worth noting that the qubits inside the logical subtree are, in general,  correlated with qubits outside the subtree.}
Furthermore, $\textbf{z}_L[T]$ only acts inside this subtree. It follows that discarding the bits outside the logical subtree does not change  the total  variation in the second line of Eq.(\ref{52}). That is
\begin{align}
&\frac{1}{2}\sum_{\textbf{z}\in\{0,1\}^{b^T}} \big|f_T(\textbf{z}+\textbf{z}_{L}[T])-f_T(\textbf{z})\big|\nonumber \\ &=\frac{1}{2}\sum_{\textbf{w}\in\{0,1\}^{b_z^{T}}} \big|g_T(\textbf{w}+\textbf{1}^{b_z^T})-g_T(\textbf{w})\big|\ ,\label{tt}
\end{align}
where $g_T$ is the marginal distribution associated to the $b_z^T$ bits on which $\textbf{z}_{L}[T]$ is 1, and $\textbf{1}^{b_z^T}$ is a bit string with $b_z^T$ 1's. The distribution $g_T$ has a simple interpretation in terms of the classical broadcasting problem discussed in the introduction: 
consider the 
distribution associated with the leaves of a full $b_z$-ary tree with depth $T$, where at each edge the bit is flipped with probability $p_z$. Then, for the input 0 at the root of the tree, the probability that leaves are in the bit string $\textbf{w}\in \{0,1\}^{b_z^T}$  is $g_T(\textbf{w})$ and for the input 1, this probability is $g_T(\textbf{w}+\textbf{1}^{b_z^T})$.

Therefore, we can apply the results on classical broadcasting on trees and, in particular, the results of \cite{evans2000} to bound the total variation distance in Eq.(\ref{tt})\footnote{The total variation distance in Eq.(\ref{tt}) can also be studied using standard results in percolation theory.}. In particular, as we explain in the Appendix \ref{classical tree discussion}, applying Theorem 1.3 of \cite{evans2000} we find
\be\label{pr2}
\frac{1}{2}\sum_{\textbf{w}\in\{0,1\}^{b_z^{T}}} \big|g_T(\textbf{w}+\textbf{1}^{b_z^T})-g_T(\textbf{w})\big|\le {\sqrt{2}}\Big[\sqrt{b_z}\times  |1-2p_z|\Big]^{T}
\ .
\ee
Combining this with Eq.(\ref{52}) and Eq.(\ref{tt}),   we arrive at 
$$
\big\|\mathcal{E}_{T}-\mathcal{E}_{T}\circ \mathcal{D}_{z}\big\|_\diamond\le {\sqrt{2}}\Big[\sqrt{b_z}\times  |1-2p_z|\Big]^{T}\ ,$$
which completes the proof of proposition \ref{Prop 1}.

\subsection{Proofs  of propositions  \ref{prop2} and  \ref{prop3} and lemma \ref{lem}}\label{sec:proofs}
Here, we  explain how proposition \ref{Prop 1}  can be generalized to propositions  \ref{prop2} and  \ref{prop3} and lemma \ref{lem} that apply to general stabilizer codes with general encoders. We mainly focus on proving proposition \ref{prop3} and lemma \ref{lem}. Proof of  proposition \ref{prop2}  follows in a similar fashion.

First, note that the  idea of a logical subtree can be straightforwardly generalized to  the trees in Proposition \ref{prop2} or \ref{prop3}. In particular, using the same approach that we used to define the sequence in Eq.(\ref{seq}), for $w\in\{x,y,z\}$ we  define the sequence of logical operators  
\begin{align}\label{sq13}
{L}_w[0]\rightarrow L_w[1]\rightarrow \cdots \rightarrow  L_w[T]\ ,
\end{align}
such that
\be\label{log}
{L}_w[t+1]V^{\otimes b^t}= V^{\otimes b^t}{L}_w[t]\ ,
\ee
where
$${L}_w[t]\in \langle i I, \sigma_x, \sigma_y,\sigma_z\rangle^{\otimes b^t}\ ,$$
and,  in particular, ${L}_w[0]=\sigma_w$ is a Pauli operator. It is also worth noting that given any  operator  ${L}_w[T]\in \langle iI, \sigma_x, \sigma_y, \sigma_z\rangle^{\otimes b^T}$ satisfying ${L}_w[T] V_T=V_T \sigma_w$ for 
$$V_{T}=\prod_{j=0}^{T-1} V^{\otimes {b}^j}\ , $$ 
one can find a sequence of logical operators in the form of Eq.(\ref{sq13}) satisfying Eq.(\ref{log}). 

Given this sequence, we can extract a logical subtree where at each level $t$, the nodes of this subtree correspond to those qubits on which  operator $L_w[t]$ acts non-trivially, i.e., acts as one of the Pauli operators.  
Furthermore, this definition also associates one of the  ``types" $x, y$, or $z$ to each node in the logical subtree:  
 namely, the type of each node at level $t$ is determined by the Pauli operator in $L_w[t]$  that acts on the qubit in that node. Similarly, we can  associate a type $x, y$, or $z$ to each edge in this subtree: an edge connecting a node at level $t$ to a  node at level $t+1$, has the same type as the node at level $t$, i.e., the parent node.

Now consider the noisy tree channel $\mathcal{E}_T$. Each edge in this tree corresponds to a single-qubit channel $\mathcal{N}$, where  $\mathcal{N}=\mathcal{N}_x\circ\mathcal{N}_z$ in proposition \ref{prop2}, and $\mathcal{N}=\mathcal{N}_x\circ\mathcal{N}_y\circ\mathcal{N}_z$ in proposition \ref{prop3}. In the proof of proposition \ref{Prop 1}, we began by ignoring the bit-flip noise channels $\mathcal{N}_x$ on all edges. Similarly, for the proof of proposition \ref{prop2} and \ref{prop3} we selectively   discard the noises that are \textit{not} consistent with an edge's type. In proposition \ref{prop2} this means discarding one of the channels $\mathcal{N}_x$ or $\mathcal{N}_z$ in $\mathcal{N}=\mathcal{N}_x\circ \mathcal{N}_z$, and in proposition \ref{prop3} we discard $\mathcal{N}_x\circ \mathcal{N}_y$ from
$\mathcal{N}_x\circ \mathcal{N}_y\circ \mathcal{N}_z$, if 
 the edge type is $z$, and similarly for $x$ and $y$.   
 By doing so, we ensure that 
 for each qubit in the logical subtree, the Pauli type associated with the logical operator acting on the qubit is the same as the type of Pauli noise on the qubit.  Note that here we again rely on the monotonicity of the diamond norm distance under data processing.

 Now recall that we had set all the flip probabilities in channels $\mathcal{N}_x$, $\mathcal{N}_y$, and $\mathcal{N}_z$  
 to be $p$, which in the case of proposition \ref{prop3} is $p=(1-\sqrt{1-\epsilon})/2$. 
  Thus, at this juncture, we may presume that we have a logical subtree of just one type of operator, say $Z$, and all its edges have only phase-flip errors with probability $p$. In other words, the logical subtree we obtain after discarding `inconsistent' noise differs from the tree obtained in the proof of proposition \ref{Prop 1}  by a reference frame change on each edge.
 

Thus, we can apply the same proof technique as discussed for proposition \ref{Prop 1}, but for the logical subtree obtained for general encoders.
The only important  difference between
the two cases is that the logical subtree will not correspond to a full $b_z$-ary tree, as in Eq.(\ref{pr2}). However, fortunately, as we explain in Appendix \ref{classical tree discussion}, Theorem 1.3' of \cite{evans2000} 
 allows us to upper bound the total variation distance in Eq.(\ref{pr2}), for arbitrary trees. In particular, for the logical subtree obtained from the sequence in Eq.(\ref{sq13}), the same bound in Eq.(\ref{pr2}) remains valid, provided that we replace $\sqrt{b_z^T}$ in the right-hand side of Eq.(\ref{pr2}) with the more general expression $\sqrt{  {\rm weight}(L_w[T])}$. Then, we conclude that   
\begin{align}
\big\|\mathcal{E}_{T}-\mathcal{E}_{T}\circ \mathcal{D}_{w}\big\|^2_\diamond&\le 2\times  {\rm weight}(L_w[T])\times  |1-2 p|^{2T}\ \nonumber , \\ &=2\times  {\rm weight}(L_w[T])\times (1-\epsilon)^{T}\nonumber \ ,
\end{align}
which proves the first part of lemma \ref{lem}. The second part of the lemma follows similarly.

Next, to prove proposition \ref{prop3} we note that
\begin{align}\label{rt10}
&\big\|\mathcal{E}_T(\rho)-\mathcal{E}_T(\frac{I}{2}) \big\|_1= \big\|\mathcal{E}_T(\rho)-\mathcal{E}_T\circ \mathcal{D}_x\circ \mathcal{D}_z(\rho)  \big\|_1\nonumber\ \\ &\le \big\|\mathcal{E}_T-\mathcal{E}_T\circ \mathcal{D}_x\circ \mathcal{D}_z  \big\|_\diamond \nonumber\\ &\le \big\|\mathcal{E}_T-\mathcal{E}_T\circ \mathcal{D}_x  \big\|_\diamond+ \big\|\mathcal{E}_T\circ \mathcal{D}_x-\mathcal{E}_T\circ \mathcal{D}_x\circ \mathcal{D}_z   \big\|_\diamond\nonumber\\ &\le \big\|\mathcal{E}_T-\mathcal{E}_T\circ \mathcal{D}_x  \big\|_\diamond+ \big\|\mathcal{E}_T-\mathcal{E}_T\circ \mathcal{D}_z   \big\|_\diamond \nonumber\\ &\le\sqrt{2 (1-\epsilon)^T}\times\Big(\sqrt{{\rm weight}(L_x[T])}+\sqrt{{\rm weight}(L_z[T])}\Big) ,
\end{align}
where the second line follows immediately from the definition of the diamond norm, the third line follows from the triangle inequality, the fourth line follows from the contractivity of  the diamond norm and the last line follows from lemma \ref{lem}, which is proven above.

Now to prove proposition \ref{prop3},  suppose the sequence of logical operators in Eq.(\ref{sq13}) is obtained by concatenating the same logical operators $L_w[1]\in \langle i\sigma_0, \sigma_x, \sigma_y,\sigma_z\rangle^{\otimes b}$ for $w\in \{0,x,y,z\}$, where $\sigma_0=I$  and 
$L_0[1]=\sigma_0^{\otimes b}= I^{\otimes b}$ are the identity  and logical identity operators, respectively. 
   That is, for 
$${L}_w[t]=\sigma_{v_1}\otimes \cdots \otimes \sigma_{v_{b^t}}\ , $$
with $v_1, \cdots, v_{b^t}\in \{0, x,y,z\}$, we define  
$${L}_w[t+1]={L}_{v_1}[1]\otimes \cdots \otimes {L}_{v_{b^{t}}}[1] \ , $$ 
which satisfies
 Eq.(\ref{log}).

Next, we   determine the number of different Pauli operators in ${L}_{w}[T]$.  For a fixed $w\in \{x,y,z\}$,  suppose the number of Pauli operators $\sigma_x, \sigma_y$, and $\sigma_z$ in the logical  operator ${L}_w[t]$ are denoted as $n^x_{t}$, $n^y_{t}$, and $n^z_{t}$, respectively. Then, at level $t+1$, we obtain 
\begin{align}
    \begin{pmatrix}
    n^x_{t+1}\\
    n^y_{t+1}\\
    n^z_{t+1}
    \end{pmatrix} = 
    \begin{pmatrix}
    n(x\rightarrow x) & n(y\rightarrow x) & n(z\rightarrow x)\\
    n(x\rightarrow y) & n(y\rightarrow y) & n(z\rightarrow y)\\
    n(x\rightarrow z) & n(y\rightarrow z) & n(z\rightarrow z)
    \end{pmatrix} 
    \begin{pmatrix}
    n^x_{t}\\
    n^y_{t}\\
    n^z_{t}
    \end{pmatrix} \ ,
\end{align}
where this $3\times 3$ transition matrix is denoted as $D_{xyz}$. Then, 
\begin{align}
    {\rm weight}(L_w[T])&=n^x_T + n^y_T + n^z_T\nonumber\\ &= \begin{pmatrix}
    1 & 
    1 &
    1
    \end{pmatrix} D^T_{xyz} e_w\ \nonumber \\ &\le g_{xyz}(T)\nonumber \\ 
    &\le b^T_\text{max} \nonumber\ .
\end{align}
Here, the second line is the sum of the matrix elements of matrix $D^T_{xyz}$ in the column associated with $w\in\{x,y,z\}$, the upper bound in the third line corresponds to the maximum between these 3 columns, and the last line follows  
from the fact that ${\rm weight}(L_w[t+1])\le b_\text{max}\times {\rm weight}(L_w[t])$. 
Combining the above bound on ${\rm weight}(L_w[T])$ with Eq.(\ref{rt10}) proves Eq.(\ref{xyz}), i.e.,
\begin{align}\label{bvbv}
&\big\|\mathcal{E}_T(\rho)-\mathcal{E}_T(\frac{I}{2}) \big\|_1\\ &\le {2\sqrt{2}} \max_{w\in\{x,y,z\}} \sqrt{\begin{pmatrix}
    1 & 
    1 &
    1
    \end{pmatrix} \Big[(1-\epsilon)   D_{xyz}\Big]^T e_w }  \nonumber\ \\ &\le {2\sqrt{2}}\times \big(\sqrt{b_\text{max}\times (1-\epsilon)} \big)^T  \nonumber\ .
\end{align}

Finally, recall that for any square matrix $A$,  $\lim_{n\rightarrow\infty} A^n=0$ if, and only if, its spectral radius $\lambda_\text{max}(A)<1$. 
It follows that if $(1-\epsilon)\times \lambda_\text{max}(D_{xyz})<1$, or  if $(1-\epsilon) \times b_\text{max}<1$, then 
in the limit $T\rightarrow \infty$ the right-hand side of the above bound goes to zero, implying that the output of channel $\mathcal{E}_T$ becomes independent of the input. This completes the proof of proposition \ref{prop3}. Proposition \ref{prop2} can be shown in a similar fashion.

\section{Recursive Decoding of Trees with $d\geq 3$ codes} \label{local rec section}


In the previous section, we showed that above certain noise thresholds classical information 
and entanglement decay exponentially with depth $T$ and vanish  in the limit $T\rightarrow\infty$. 
But, is there any non-zero noise level that allows the propagation of 
classical information and entanglement over the infinite tree? In this section, we answer in the affirmative by considering a simple recursive local decoding strategy for stabilizer trees with distance $d\geq 3$.


We first explain the idea for general encoders and noise channels and then consider the special case of stabilizer codes with Pauli channels.   
Consider one layer of the tree channel, the channel $\mathcal{N}^{\otimes b}\circ \mathcal{V}$ that
first encodes one input qubit in $b$ qubits and then sends each qubit through the single-qubit noise channel $\mathcal{N}$.  Let $\mathcal{R}: \mathcal{L}(({\mathbb{C}^{2}})^{\otimes  b})\rightarrow \mathcal{L}(\mathbb{C}^2)$ be a channel that (approximately) reverses this process,  such that  the concatenated channel 
$f[\mathcal{N}]=\mathcal{R}\circ\mathcal{N}^{\otimes b}\circ \mathcal{V}$ is as close as possible to the identity channel (e.g., with respect to the diamond norm, or the entanglement fidelity), where  for any single-qubit channel $\mathcal{M}$, we have defined 
\be\label{func}
f[\mathcal{M}]\equiv\mathcal{R}\circ\mathcal{M}^{\otimes b}\circ \mathcal{V}\ .
\ee
 As depicted in Figure \ref{recursivetree_circ}, by applying this recovery recursively, we find  a recovery for the channel corresponding to the entire tree of depth $T$, denoted by ${\mathcal{E}}_{T}=\prod_{j=0}^{T-1}    \mathcal{N}^{\otimes b^{j+1}} \circ \mathcal{V}^{\otimes b^j}$, or $\widetilde{\mathcal{E}}_{T}={\mathcal{E}}_{T}\circ\mathcal{N}$, if we consider the noise at the input qubit. Namely, the recovery  is 
\begin{align}\label{rec tree eqn}
    \mathcal{R}^{\rm loc}_T = \prod^{T-1}_{i=0} \mathcal{R}^{\otimes b^{T-i-1}} \ .
\end{align}
As we will see in the following, in general, this is a sub-optimal recovery strategy. However, the advantage of this approach is that implementing the recovery does not require long-range interactions between distant qubits. More precisely, the error corrections are decided locally based on the observed syndromes from only one block with $b$ qubits. Hence, we sometimes  refer to this approach as \emph{local recovery}, as opposed to the optimal recovery that would make a `global' decision based on \textit{all} observed syndromes together (see Sec.\ref{Optimal section}). Furthermore, as we discuss below, analyzing the performance of this recovery is relatively easy. We note that the recursive recovery approach has been previously studied  in the context of concatenated codes with noiseless encoders  \cite{dohertyconcatenated} (in the language of this paper, this corresponds to the tree in which uncorrelated noise channels are applied to the qubits in the leaves, but not inside the tree).

\begin{figure}[ht]
\centering
\includegraphics[width=0.5\textwidth]{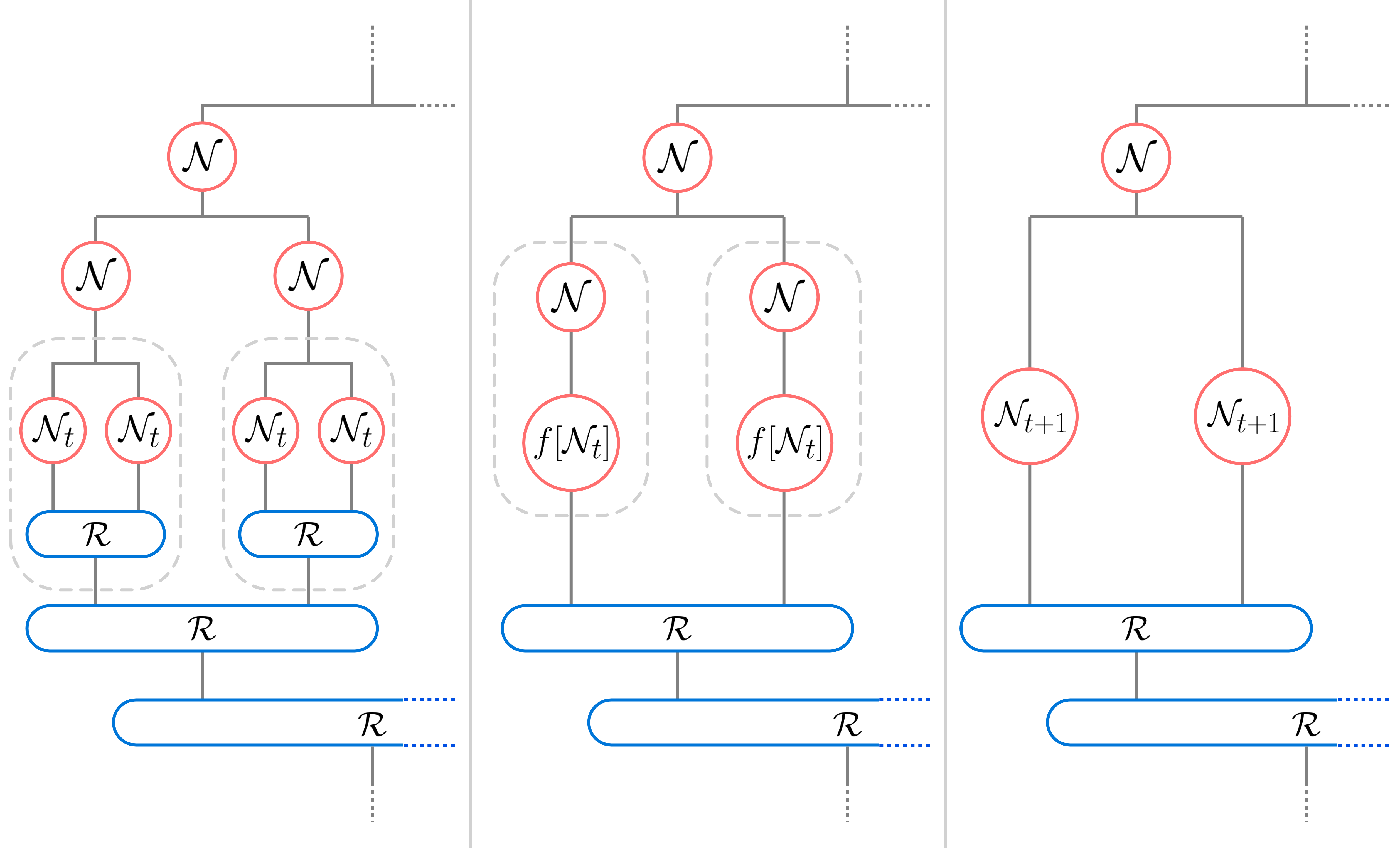}
\caption{ \textbf{Derivation of the recursive equation \ref{gsj}:} After recovery upto level $t$, the effective noise channel is  described by the single-qubit noise channel $\mathcal{N}_t$ defined in Eq.(\ref{def45}) (see the \textbf{left} image). In the \textbf{middle} image, we consider the single qubit channel obtained after applying recovery to a single block, i.e., $f[\mathcal{N}_t]=\mathcal{R}\circ \mathcal{N}_t^{\otimes b}\circ \mathcal{V}$ ($b=2$ in this diagram). In the \textbf{right} image, we compose this with the physical noise  in the tree, i.e., $\mathcal{N}_{t+1}=f[\mathcal{N}_t] \circ \mathcal{N}$. Thus, the effective noise at the next level $t+1$ is denoted by $\mathcal{N}_{t+1}$. } 
\label{recursivetree_circ}
\end{figure}


\subsection{Fixed-point equation for infinite tree}

For a tree with depth $t$, let $\mathcal{N}_t$ be the overall noise after applying the above recovery process to channel $\widetilde{\mathcal{E}}_t$, i.e.,
\be\label{def45}
\mathcal{N}_t=\mathcal{R}^{\rm loc}_t\circ \widetilde{\mathcal{E}}_t=\mathcal{R}^{\rm loc}_t\circ {\mathcal{E}}_t\circ \mathcal{N} \ .
\ee
As seen in Figure \ref{recursivetree}, in a tree with depth $t+1$ there are $b$ subtrees each with depth $t$. For each subtree the effective error after error correction is $\mathcal{N}_t$. Then, ignoring the single-qubit channel $\mathcal{N} $ at the root of the tree with depth $t+1$, the overall noise is 
\be
f(\mathcal{N}_t)=\mathcal{R}\circ\mathcal{N}_t^{\otimes b}\circ \mathcal{V}\ .
\ee

\begin{figure}[ht]
\centering
\includegraphics[width=0.4\textwidth]{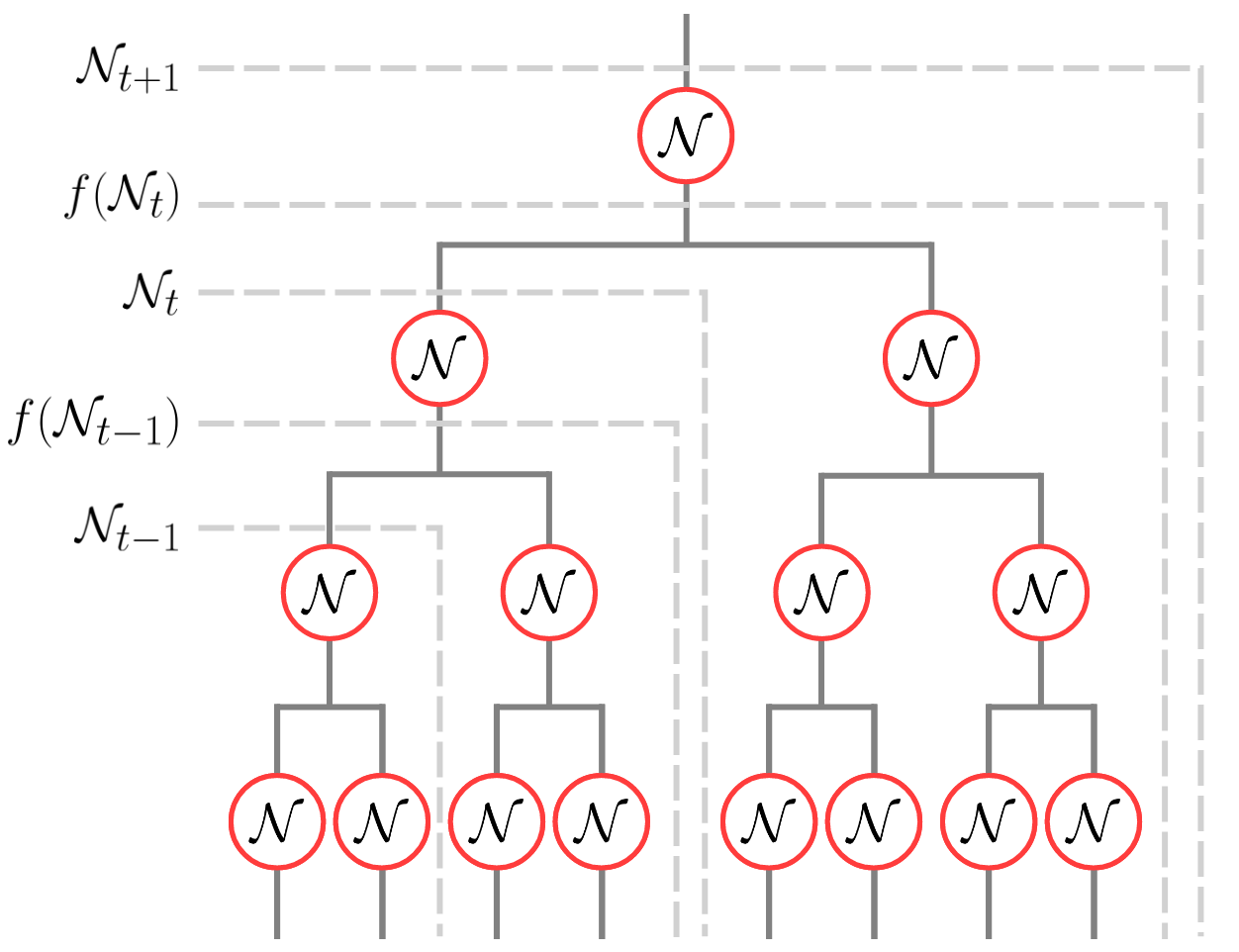}
\caption{This figure illustrates the steps of recursive decoding of a tree by indicating the subtree considered in each step.  $\mathcal{N}_{t-1}$ and $\mathcal{N}_{t}$ are single-qubit channels obtained after the local recovery of  subtrees of depth $t-1$ and $t$, respectively.  $\mathcal{N}_{t}$ can be expressed as a function of $\mathcal{N}_{t-1}$ and $\mathcal{N}$, as summarized in Eq.(\ref{gsj}). This way of indexing local recovery levels will  be followed in the rest of this paper.} 
\label{recursivetree}
\end{figure}

Finally, adding the effect of the single-qubit channel $\mathcal{N}$ at the root, we obtain the overall noise channel 
\be\label{gsj}
\mathcal{N}_{t+1}= f(\mathcal{N}_t)\circ \mathcal{N} =  \mathcal{R} \circ\mathcal{N}^{\otimes b}_t\circ \mathcal{V}\circ \mathcal{N}\ .
\ee
Fig. \ref{recursivetree_circ} illustrates Eq.(\ref{gsj}). With the initial condition $\mathcal{N}_0=\mathcal{N}$, this recursive equation fully determines $\mathcal{N}_t$ for arbitrary $t$. Then, in the limit $t\rightarrow\infty$, we obtain the single-qubit channel 
\be
\mathcal{N}_\infty= \lim_{t\rightarrow\infty}\mathcal{N}_t= \lim_{t\rightarrow\infty} \mathcal{R}^{\rm loc}_t\circ \widetilde{\mathcal{E}}_t\ ,
\ee
where we assume that the recovery channel is chosen properly such that this limit exists. 
This channel satisfies the fixed-point equation
\be
\mathcal{N}_{\infty}=f(\mathcal{N}_\infty)\circ \mathcal{N}  =\mathcal{R}\circ\mathcal{N}_\infty^{\otimes b}\circ \mathcal{V} \circ \mathcal{N}    \ .
\ee
 Note that the single-qubit channel $\mathcal{N}_T$ describes the output of the recursive recovery applied to the  channel $\widetilde{\mathcal{E}}_T={\mathcal{E}}_T\circ \mathcal{N}$, which includes  channel $\mathcal{N}$ at the input.  On the other hand, if one does not include the single-qubit noise at the input of the tree, then the overall channel is described by $f(\mathcal{N}_{T-1})$, and in the limit $T\rightarrow\infty$, one obtains the channel $f(\mathcal{N}_{\infty})$.

Depending on the strength of the noise in the single-qubit channel $\mathcal{N}$ and the properties of the encoder $\mathcal{V}$, the channel $\mathcal{N}_{\infty}$(or, $f(\mathcal{N}_{\infty})$) might be a channel with constant output 
which does not transmit any information, i.e., with zero classical capacity, an entanglement-breaking channel with nonzero classical capacity, or a channel that transmits entanglement (and hence even classical information).

\begin{figure}[ht]
\centering
\includegraphics[width=0.5\textwidth]{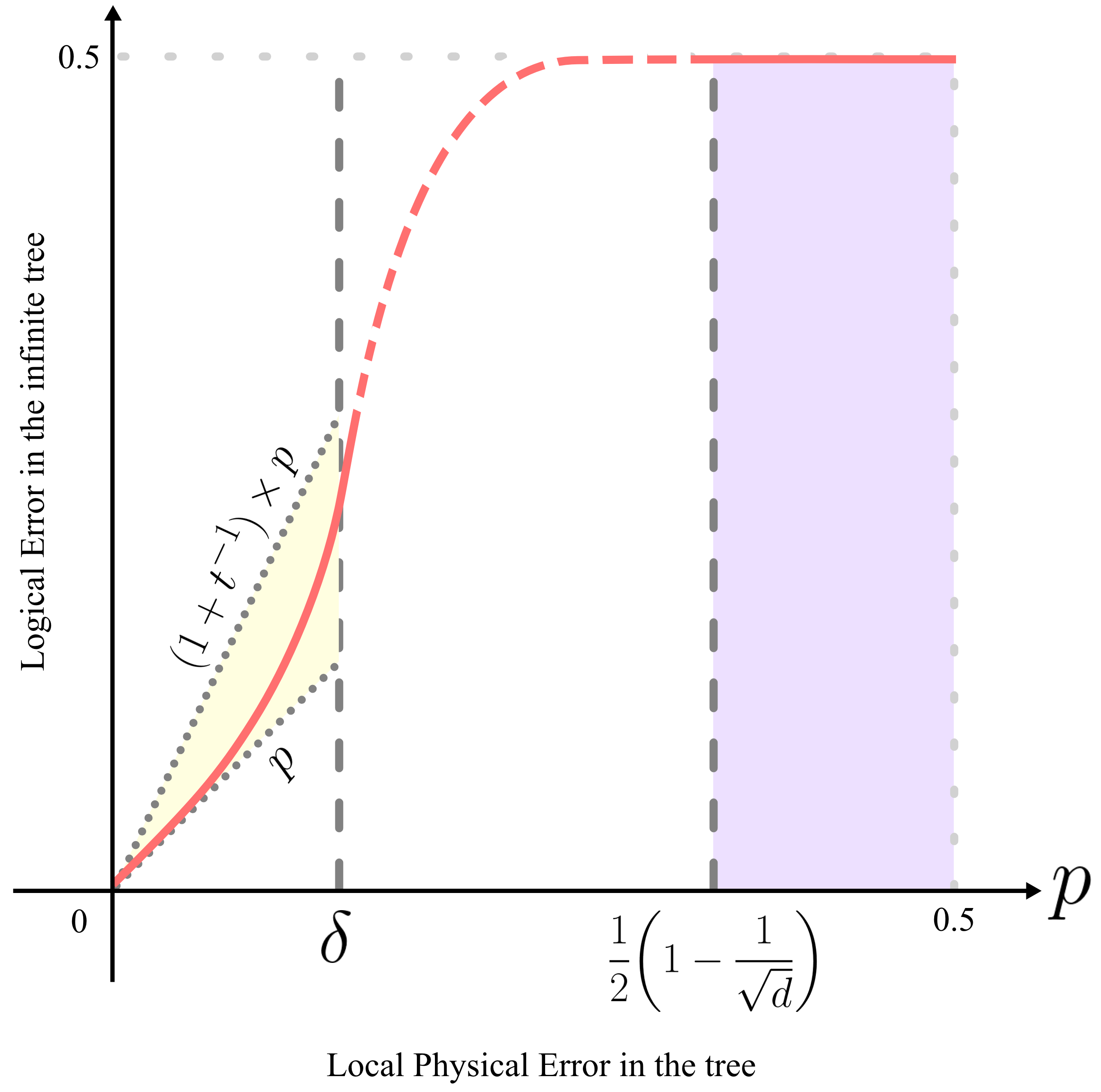}
\caption{\textbf{Summary of results for CSS code trees with distance $d\geq 3$ -- } 
This plot illustrates the  general behavior  of the logical $Z$ error $q_T^z$ in the channel 
$\widetilde{\mathcal{E}}_T=\mathcal{E}_T\circ \mathcal{N}$, in the limit $T\rightarrow \infty$ (i.e., the infinite tree
\textit{with} noise on the root)  as a function of the  physical $Z$ error $p$ on the edges of the tree. Here, $d$ is the code distance, which in particular, means upto $t_z=\lfloor d_z-1\rfloor/2 \geq \lfloor d-1\rfloor/2$, $Z$ errors can be corrected. We characterize three regions, $p<\delta \equiv \frac{t_z}{1+t_z}\big(\frac{1}{c(1+t_z)}\big)^{\frac{1}{t_z}}$ (region I),  $p>\frac{1}{2}\big(1-\frac{1}{\sqrt{d}}\big)$ (region III), and the region in between (region II). From proposition \ref{prop: css lowerbound} on  local recovery analysis, we see that in region I, the logical error is bounded between $p$ and $(1+t_z^{-1})\times p$, and it is $p+\mathcal{O}(p^2)$ for $p\ll 1$. From our upperbound on the noise threshold, we see that in region III, the logical error in the infinite tree saturates to $1/2$. The form of the curve in region II remains unknown, but it can be estimated for specific trees using the optimal recovery strategy in Sec \ref{Optimal section}.} 
\label{css result summary fig}
\end{figure}

\subsection{Recursive decoding of CSS code trees \\ (A  lower bound on the noise threshold)}
In the following, we 
further study this approach for the case of CSS trees with standard encoders. In Appendix  \ref{gen stab threshold} we explain how these results can be extended to general stabilizer trees.

 We consider the recovery channel $\mathcal{R}^{\rm loc}_T$  defined in Eq.(\ref{rec tree eqn})  applied to a CSS code tree $\widetilde{\mathcal{E}}_T$ which acts on $X$ and $Z$ errors separately. This ensures,
\begin{align}\label{correc}
    \mathcal{R}^{\rm loc}_T \circ \widetilde{\mathcal{E}}_T=\mathcal{Q}_z\circ \mathcal{Q}_x=\mathcal{Q}_x\circ \mathcal{Q}_z\ ,  \end{align}
where $\mathcal{Q}_z=q^z_T\rho + (1-q^z_T)Z\rho Z$ is a phase flip channel, and  $\mathcal{Q}_x=q^x_T\rho + (1-q^x_T)X\rho X$ is a bit flip channel. Since the cases of $Z$ and $X$ errors are similar, we focus on the recovery of $Z$ errors and prove a  lower bound on the noise threshold.

\begin{proposition}\label{prop: css lowerbound}
Consider the tree channel $\widetilde{\mathcal{E}}_T=\mathcal{E}_T\circ \mathcal{N}$ constructed from a standard encoder of a  CSS code, as defined  in Eq.(\ref{noisy tree definition with root noise}). Suppose the single-qubit channel $\mathcal{N}=\mathcal{N}_z\circ \mathcal{N}_x$ is composed of bit-flip and phase-flip channels $\mathcal{N}_x$ and $\mathcal{N}_z$ that apply $X$ and $Z$ errors with probabilities $p_x$ and $p_z$, respectively. Suppose the minimum weight of the logical $Z$ operator for this code is $d_z$, which means the  code corrects up to $t_z=\lfloor (d_z-1)/2\rfloor$,  $Z$ errors.  For sufficiently small
probability of error $p_z$, e.g., for 
\be\label{delta defn}
p_z \le    \frac{t_z}{t_z+1} \bigg(\frac{1}{(t_z+1)c}\bigg)^{\frac{1}{t_z}}\equiv\delta \ ,
\ee
the probability of logical $Z$ error after local recovery is bounded by
\be
\forall T:\ \ \ \ q^z_T \le (1+\frac{1}{t_z}) p_z\ ,
\ee
where $c=\sum_{k=t_z+1}^b {b \choose k}$ is a positive constant bounded as $2^{t_z} \leq c \leq 2^b$. Furthermore,
\begin{align}
    \lim_{p_z \rightarrow 0} \frac{q^z_\infty}{p_z} = 1.
\end{align}
\end{proposition}
 In particular, if $t_z\ge 1$, then for $p_z \le \delta$, the error $q^z_T<1/2$, which means that the tree of any depth $T$ transfers non-zero classical information in the input $X$ basis. Thus, the logical single-qubit error in the infinite tree after local recovery is  
\begin{align}\label{pz qz reln}
    q^z_\infty \leq  (1+\frac{1}{t_z})p_z\ .
\end{align}
Also, note that $q^z_T$ is the probability of logical error in a tree with depth $T$ \textit{with} noise at the root (i.e., $\widetilde{\mathcal{E}}_T$). Clearly, without noise at the input channel, the overall noise is less. 

 Recall that after recursive decoding, the effective single qubit channel is a concatenation of a bit-flip and phase-flip channel with error probabilities $q^x_T$ and $q^z_T$, respectively. As we explain in Appendix \ref{ent depth uncorrelated XZ}, this channel is entanglement-breaking if, and only if, $|1-q^z_T|\times |1-q^x_T|<1/2$. Proposition \ref{prop: css lowerbound} ensures that we can make $q^x_T$ and $q^z_T$ arbitrarily small,  by reducing $p_x$ and $p_z$ respectively. Therefore, there exists a non-zero value of $p_x$ and $p_z$ for which any infinite CSS code tree transmits entanglement.

\begin{figure}[ht!]
\centering
\includegraphics[width=0.313\textwidth]{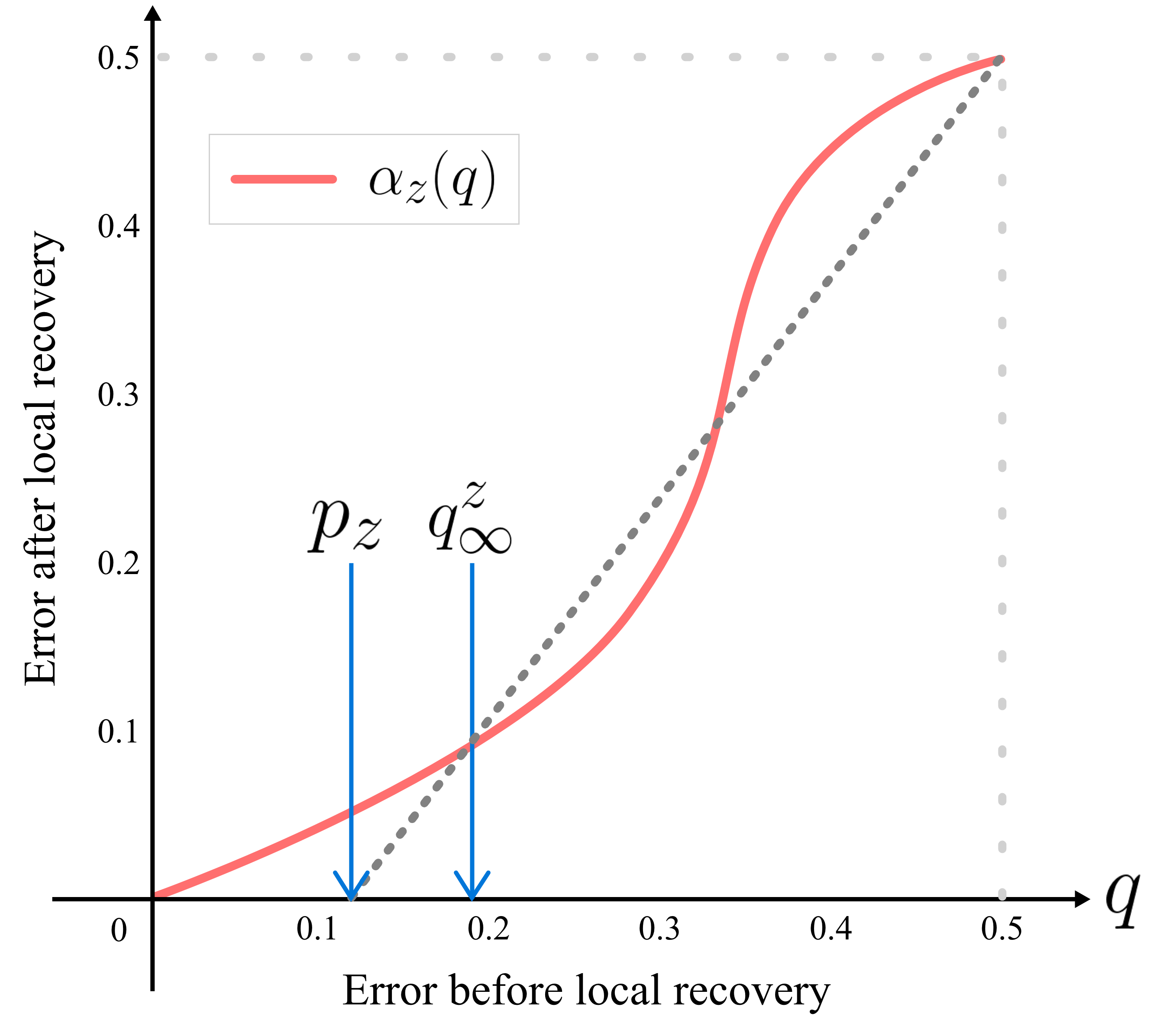}
\includegraphics[width=0.313\textwidth]{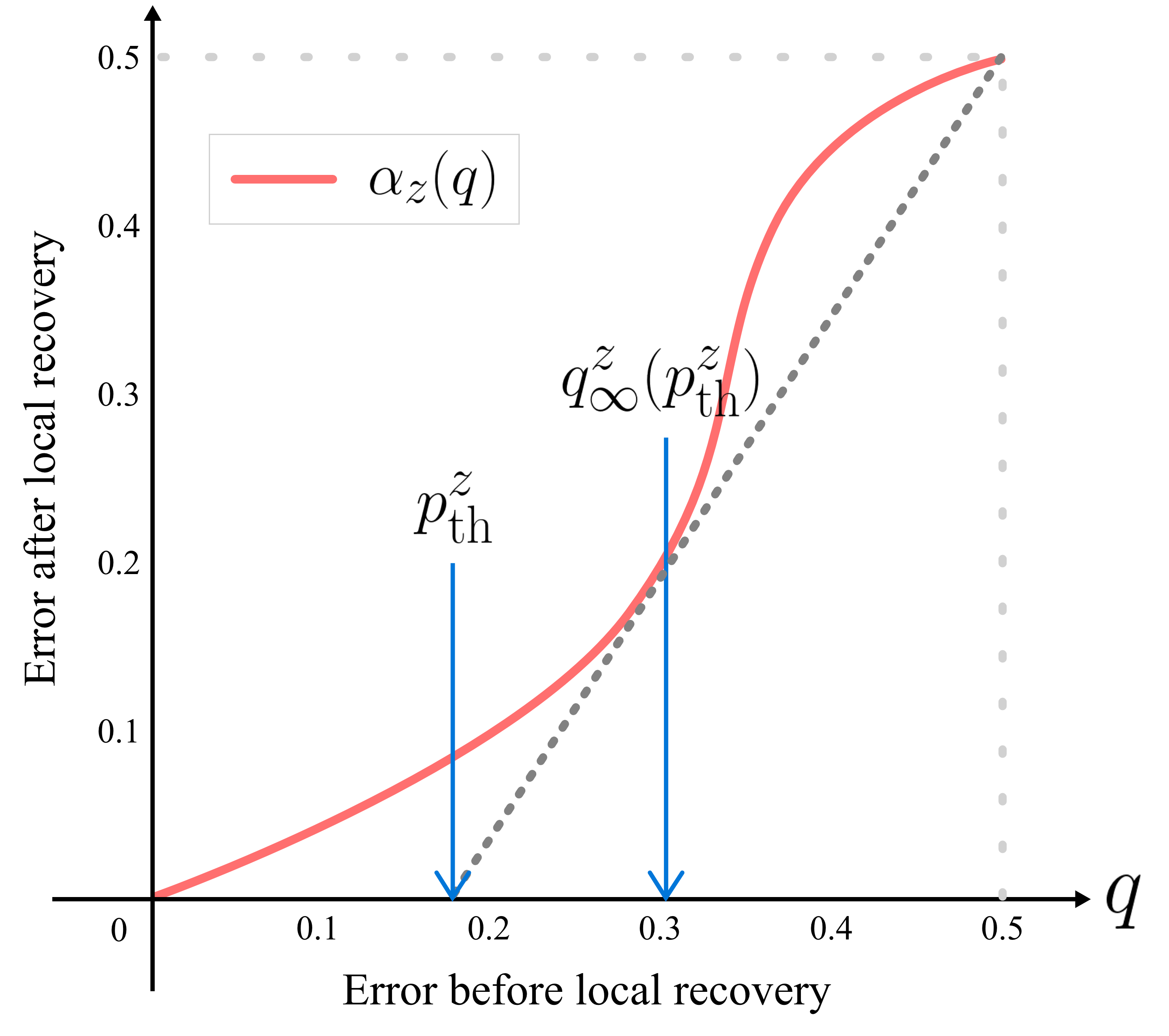}
\includegraphics[width=0.313\textwidth]{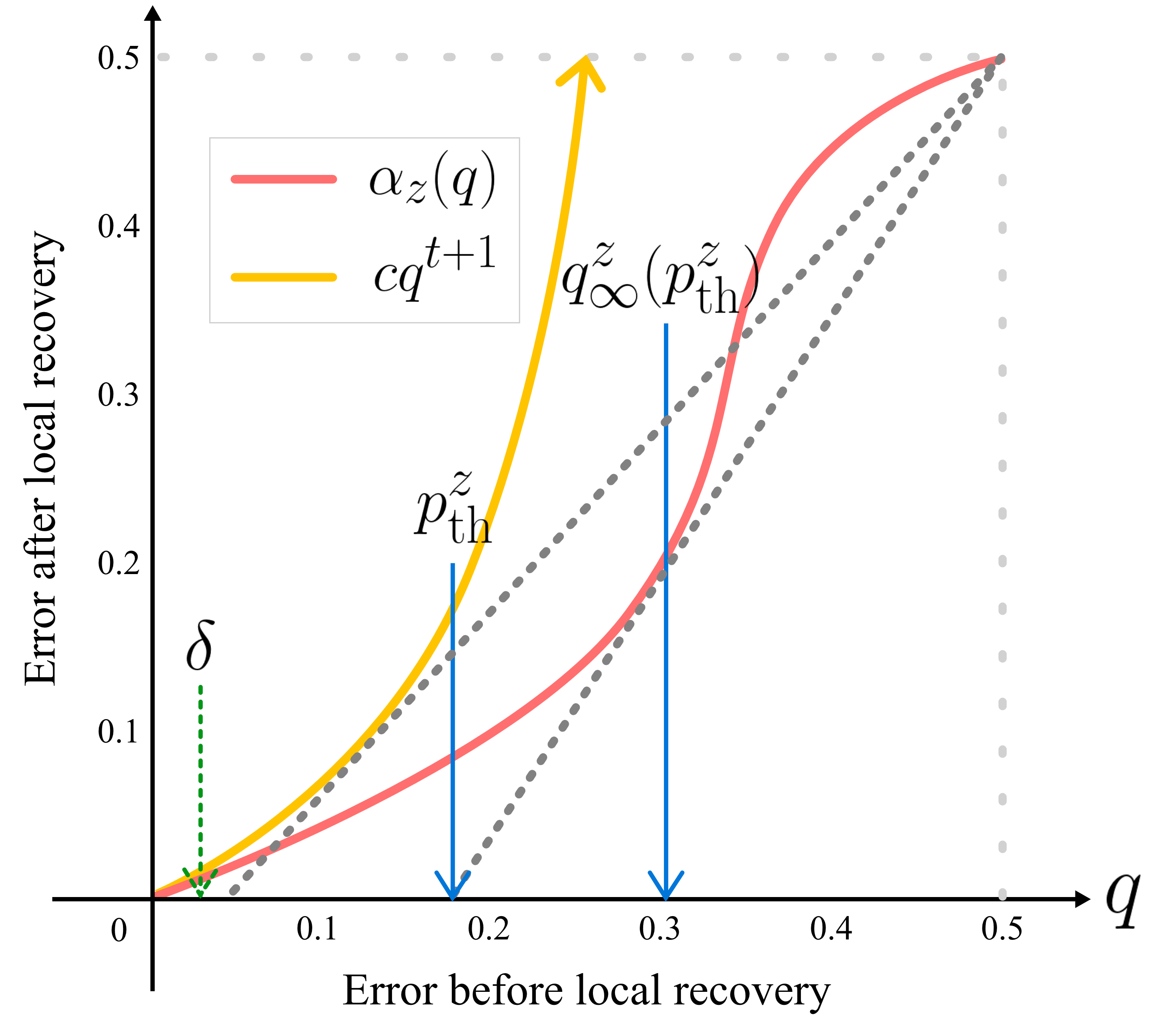}
\caption{\textbf{Schematic for recursive decoding threshold computation --} The \textbf{top} figure shows the intersection of two curves $\alpha(q)$ and $(q-p_z)/(1-2p_z)$, where $p_z$ is the error within the tree and the line's $x$-intercept. According to Eq.(\ref{fixed23}), for the infinite tree the overall error after the recovery, denoted by $q^z_\infty$, is determined by the the intersection of these two curves. The smallest $x$-value of the intersections is $q^z_\infty$.  For large enough $p_z$, only the $(0.5,0.5)$ intersection persists, entailing that the  infinite tree upon local recovery transmits no information. The \textbf{middle} figure shows $p^z_{\rm th}$, i.e. the maximum value of $p_z$ for which $q^z_\infty$ is not $0.5$. The \textbf{bottom} figure shows that every $\alpha(q)$ for a $[[b,1,d]]$ CSS code is upper bounded by $c q^{\lfloor (d_z-1)/2\rfloor+1}$, for some constant $2^{t_z}\leq c\leq2^b$, where $d_z\ge d$ is the distance for $Z$ errors. 
Consider the tangent line to this curve that passes the point $(0.5,0.5)$. The $x$-intercept of this tangent line is a lower bound on $p_{\rm th}^z$. Graphically, we see that this lowerbound is always greater than zero.  Using this observation, in Appendix \ref{local rec appendix}, we show that $p_{\rm th}^z$ is lower bounded  by $\delta:=\frac{t_z}{t_z+1} \big(\frac{1}{(t_z+1)c}\big)^{\frac{1}{t_z}}$,  which is strictly positive when $d_z\geq3$.  
We also show that for $p_z$ less than this value,  the logical error in the channel $\widetilde{\mathcal{E}}_T=\mathcal{E}_T\circ \mathcal{N}$, is bounded by $p_z (1+\frac{1}{t_z})$.}
\label{Fig:tanget}
\end{figure}

\subsubsection{Analysis of recursive decoding of CSS code trees\\ (Proof of Proposition \ref{prop: css lowerbound})}\label{css local rec}

To analyze the performance of the local recursive strategy for a CSS code tree, we first note that for CSS codes with uncorrelated $X$ and $Z$ errors, error correction for these  errors can be performed independently, and therefore after error correction, the $X$ and $Z$ errors remain uncorrelated.   In particular, there exists a  recovery process $\mathcal{R}$ that corrects $X$ and $Z$ errors independently, i.e., a recovery process  $\mathcal{R}$  such that 
\be
f[\mathcal{M}^x\circ \mathcal{M}^z]=f[\mathcal{M}^x]\circ f[\mathcal{M}^z]\ ,
\ee
where $\mathcal{M}^x$ and $\mathcal{M}^z$ are bit-flip and phase-flip channels, respectively, and function $f$ defined in Eq.(\ref{func}) denotes the effective single-qubit channel after recovery.  In particular, if   $\mathcal{M}^z(\rho)=(1-p)\rho+p Z\rho Z$, then 
\be
f[\mathcal{M}^z](\rho)= (1-\alpha_z(p)) \rho+\alpha_z(p) Z \rho Z\ ,
\ee
where $\alpha_z(p)$ determines the probability of logical error after recovery. Note that $\alpha_z({1}/{2})={1}/{2}$, which corresponds to the cases of maximal $Z$ error.  For instance, for the Steane-7 code \begin{align}
    \alpha_z(p)=& \ 21p^2 - 98p^3 + 210p^4- 252p^5 + 168p^6 -48p^7\ .\nonumber
\end{align}
In Appendix \ref{CSS appendix} we explain how this function can be calculated for general CSS codes (see also \cite{dohertyconcatenated}).

Now we consider the recursive Eq.(\ref{gsj}) for the phase-flip channel  $\mathcal{Q}_z$ in Eq.(\ref{correc}), which applies $Z$ error with probability $q_T^z$. Writing Eq.(\ref{gsj}) in terms of function $\alpha_z$, we obtain
\be\label{CSS rec eqn}
q_{t+1}^z= (1-p_z)\times \alpha_z(q_{t}^z)+ p_z\times [1-\alpha_z(q_{t}^z)]\ ,
\ee
with the initial condition $q_{0}^z=p_z$. The fixed point of this equation, denoted by  $q^z_\infty$, satisfies
\begin{align}\label{fixed23}
     {\alpha_z}(q^z_\infty)=\frac{q^z_\infty-p_z}{1-2p_z}\ .
\end{align}
Therefore, to determine $q^z_\infty$ we study the intersection of the curve ${\alpha}_z(q)$ with the line $(q-p_z)/(1-2p_z)$  as functions of $q$  (see Fig. \ref{Fig:tanget}).  Since $\alpha_z(1/2)=1/2$,  $q^z_\infty=1/2$ is always a solution. However, for  a sufficiently small value of $p_z$, the two curves have more intersections in the interval $0<q^z_\infty<0.5$. Let $p^{z}_{\rm th}$ be the largest value of $p_z$ for which they  have an intersection in this interval. 
As it can be seen in Fig. \ref{Fig:tanget} this point is determined by the tangent of curve $\alpha_z$ that passes through the point $(0.5,0.5)$. More specifically, $p^{z}_{\rm th}$  is the $x$-intercept for this line. Given $\alpha_z$, $p^z_\text{\rm th}$ can be estimated numerically.


While the exact threshold $p^z_\text{\rm th}$ depends on the detail of the polynomial $\alpha_z$, its value is strictly positive, provided that the lowest order of $q$ in polynomial  $\alpha_z(q)$ is 2 or higher, which means the recovery corrects all single-qubit $Z$ errors. This is always possible for codes with  distance $d\ge 3$. In particular, suppose the code distance for $Z$ errors is $d_z\ge d$, i.e., the code  corrects $Z$ errors with weight  $t_z=\lfloor(d_z-1)/2\rfloor$ or less. Then, the lowest order of $q$ in the polynomial $\alpha_z(q)$ is 
$t_z+1$.  Therefore, there exists a positive constant $c$, such that  $\alpha_z(q)$ is upper bounded as  
\be
\alpha_z(q)\le c\times q^{t_z+1}\ ,
\ee
where $2^{t_z} \leq c\leq 2^b$ (see Appendix \ref{local rec appendix}).

This, in turn, implies that the $x$-intercept of the tangent to $cq^{t_z+1}$ lower bounds the $x$-intercept of the tangent to $\alpha_z(q)$. Then, using an idea explained in Fig. \ref{Fig:tanget}, in Appendix \ref{local rec appendix} we show that 
\be\label{th lowerb}
p^z_\text{th}\geq\frac{t_z}{t_z+1} \bigg(\frac{1}{(t_z+1)c}\bigg)^{\frac{1}{t_z}}\equiv \delta\ ,
\ee
thus, proving the first part of Proposition \ref{prop: css lowerbound}. Furthermore, from Fig. \ref{Fig:tanget} it is clear that for $p_z< p^z_\text{th}$ the overall probability of $Z$ error in the infinite tree, denoted by $q_\infty^z$ can be arbitrarily small. In particular, we show that    $q^z_\infty \leq p_z\times (1+\frac{1}{t_z})$. 

Finally, to show that $\lim_{p_z \rightarrow 0} {q^z_\infty}/{p_z} = 1$, consider the Taylor expansion of  $q^z_\infty$ as a function of $p_z$ in the limit $p_z\rightarrow 0$, as $q^z_\infty=a_1 p_z+ \mathcal{O}(p^2_z)$. Then, on the left-hand side of
Eq.(\ref{fixed23}), the lowest  degree of 
 the polynomial $\alpha_z(q^z_\infty)$ in $p_z$ is larger than or equal to 2, whereas on the right-hand side the lowest-degree term in $p_z$ is $(a_1-1)\times  p_z$. Therefore,   Eq.(\ref{fixed23}) implies $a_1=1$, hence $\lim_{p_z \rightarrow 0} {q^z_\infty}/{p_z} = 1$.

\subsection{Example: Noise threshold for recursive decoding of the repetition code tree}\label{subsec: local rec}

\begin{figure}[ht!]
\centering
\includegraphics[width=0.48\textwidth]{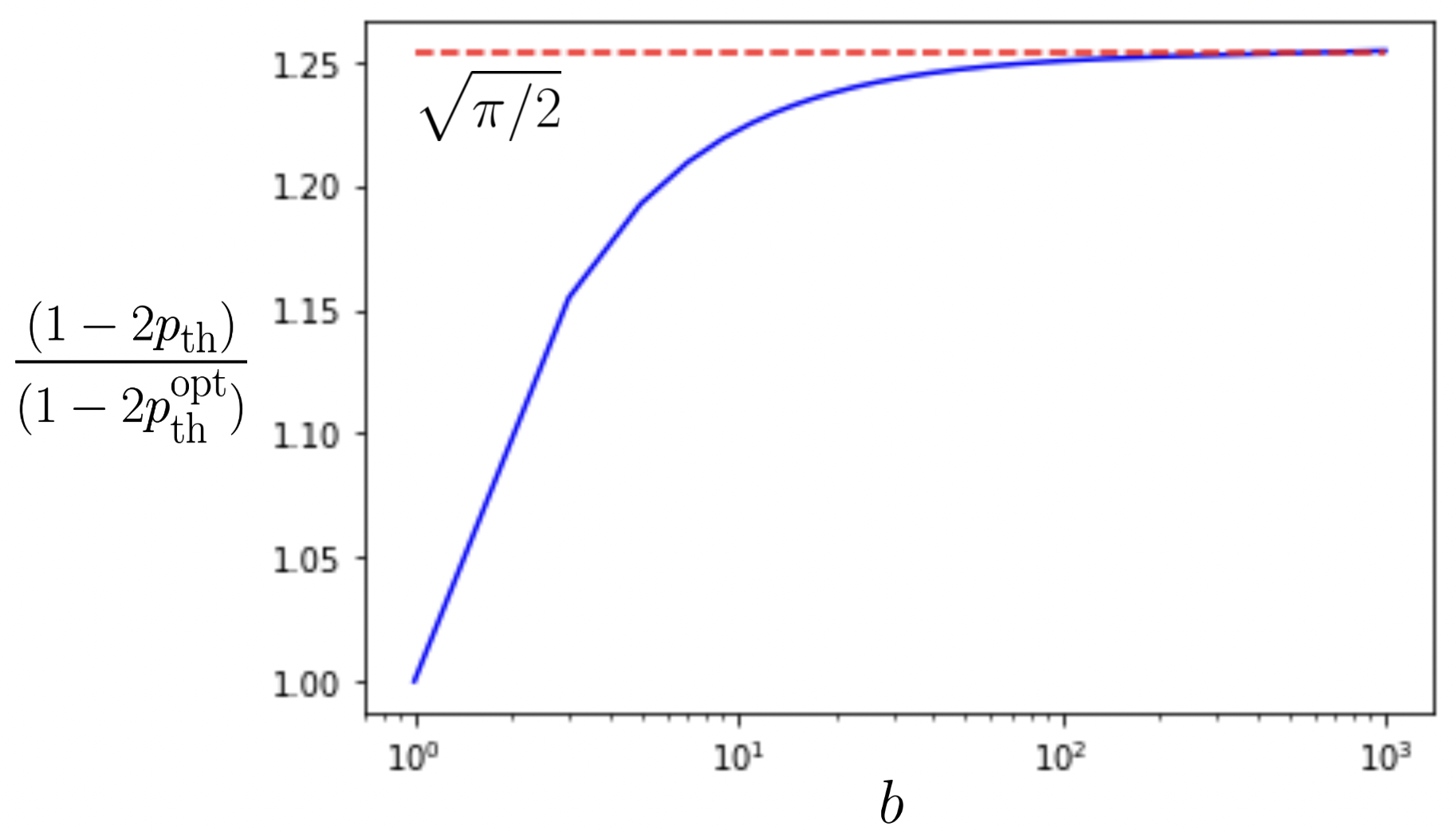} 
\caption{\textbf{Comparing the threshold for the recursive and optimal decoders-- } 
Consider a full $b$-ary tree, where 
 each node has the encoder $\ket{c}\mapsto \ket{c}^{\otimes b}: c=0,1$,  for odd $b$. Also, assume each edge has a bit-flip channel, which applies Pauli $X$ with probability $p_x$. The result of \cite{evans2000} implies that  for the infinite tree, the noise threshold for transmission of classical information 
 in the input $Z$ basis 
   is $p^{\rm opt}_{\text{th}}=\frac{1}{2}(1-\frac{1}{\sqrt{b}})$. Meanwhile, applying the method described in the previous section to the function $\alpha_{\rm maj}$ in Eq.(\ref{alpha cl}), we numerically estimate the local recovery threshold $p_{\rm th}$. The plot shows the ratio $(1-2p_{\rm th})/(1-2p^{\rm opt}_{\rm th})$. For large $b$, this ratio approaches $\sqrt{\pi/2}\approx 1.253$, matching the analytical result in  Eq.(\ref{loc-threshold}).}
\label{class opt v local}
\end{figure}

As a simple example, 
we consider the CSS code tree constructed from the standard encoder for the repetition code.
Specifically, suppose the encoder at each node is  
\begin{equation}
|c\rangle\ \longrightarrow \ V |c\rangle=|c\rangle_L=|c\rangle^{\otimes b}\ \ \ :\  c\in\{0,1\}\ ,
\end{equation}
for odd integer $b$. 
Classically the optimal error correction strategy for the repetition code can be realized via majority voting. In the quantum case measuring qubits in $\{|0\rangle,|1\rangle\}$ basis and then performing majority voting destroys coherence between the logical states $|0\rangle_L$ and $|1\rangle_L$.  However, one can perform majority voting in a way that does not destroy this coherence, namely by measuring stabilizers $Z_jZ_{j+1}: j=1,\cdots, b-1$.  Then, one determines the number of bit-flips (i.e., Pauli $X$) needed to make the eigenvalues of all stabilizers $+1$. There are only two patterns of bit flips that make all the eigenvalues +1. Then, the decoder chooses the one that requires the minimum number of flips. We shall denote this $b\rightarrow 1$ qubit recovery procedure with $\mathcal{R}_{\rm maj}$.  

Assuming each qubit is subjected to independent bit-flip error, which flips the bit with probability $p_x$, the probability that the decoder makes a wrong guess is given by
\begin{align}\label{alpha cl}
    \alpha_{\rm maj}(p_x)=\sum_{k={\lfloor b/2\rfloor}+1}^{b} {b\choose k} p_x^k (1-p_x)^{b-k}\ .
\end{align}
 We can compute the $p_{\rm th}$ as detailed in Sec. \ref{css local rec}. In particular, in Appendix \ref{local rec appendix} we show that asymptotically in the limit of large $b$
\be\label{loc-threshold}
p_{\rm th}= \frac12\Big[1-\sqrt{\frac{\pi}{2}}\frac{1}{\sqrt{b}}\Big]+\mathcal{O}(\frac{1}{b})\ .
\ee
On the other hand, applying the  result of Evans et al. \cite{evans2000} on broadcasting over classical trees  discussed in the introduction, we know that for the optimal recovery  
the exact threshold is  
\be
p^{\text{opt}}_{\rm th}= \frac12\Big[1-\frac{1}{\sqrt{b}}\Big]\ .
\ee
Interestingly, this threshold can also be achieved via similar majority voting, albeit when it is performed \emph{globally} on all qubits at the output of the tree, whereas the threshold in Eq.(\ref{loc-threshold}) is obtained via recursive local majority voting on groups of $b$ qubits. In Fig. \ref{class opt v local}  we compare the actual thresholds for the local and global majority voting approaches for finite values of $b$. See Appendix \ref{CSS appendix} for further examples.

\section{Trees with distance $d=2$ codes:\\ Recursive decoding with one reliability bit}\label{local recovery 1 bit section}

In this section, we consider trees where at each node the received qubit is encoded in a code with distance $d=2$. Since such codes can not correct single-qubit errors in unknown locations, the local recovery approach discussed in the previous section fails to yield a non-zero threshold for infinite trees: under local recovery, the probability of logical error as a function of the probability of physical error  has a linear term, which will then accumulate as the tree depth grows.  

To overcome this issue, one may consider a modified version of the recursive recovery that combines several layers together and performs recovery on these combined layers, which corresponds to a concatenated code with distance  $d>2$. However, this modification will not solve the aforementioned issue because the noise between the layers still produces uncorrectable logical errors with a probability that is  linear in the physical error probability.

To fix this problem, we introduce a new scheme, which takes advantage of the fact  that although $d=2$ codes cannot correct errors in unknown locations, they can  (i) detect single-qubit errors, and (ii) correct single-qubit errors in known locations. Based on these observations, it is natural to consider  a 
modification of local recovery where each layer sends a \textit{reliability bit} to the next layer that determines if the decoded qubit is reliable or not. At each layer, the value of this bit is determined based on the value of the reliability 
bit received from the previous layers together with the syndromes observed in this layer. See  Fig. \ref{recursivetree2} and Sec.\ref{subsec: one bit local recovery}
for further details.

\begin{figure}[ht]
\centering
\includegraphics[width=0.5\textwidth]{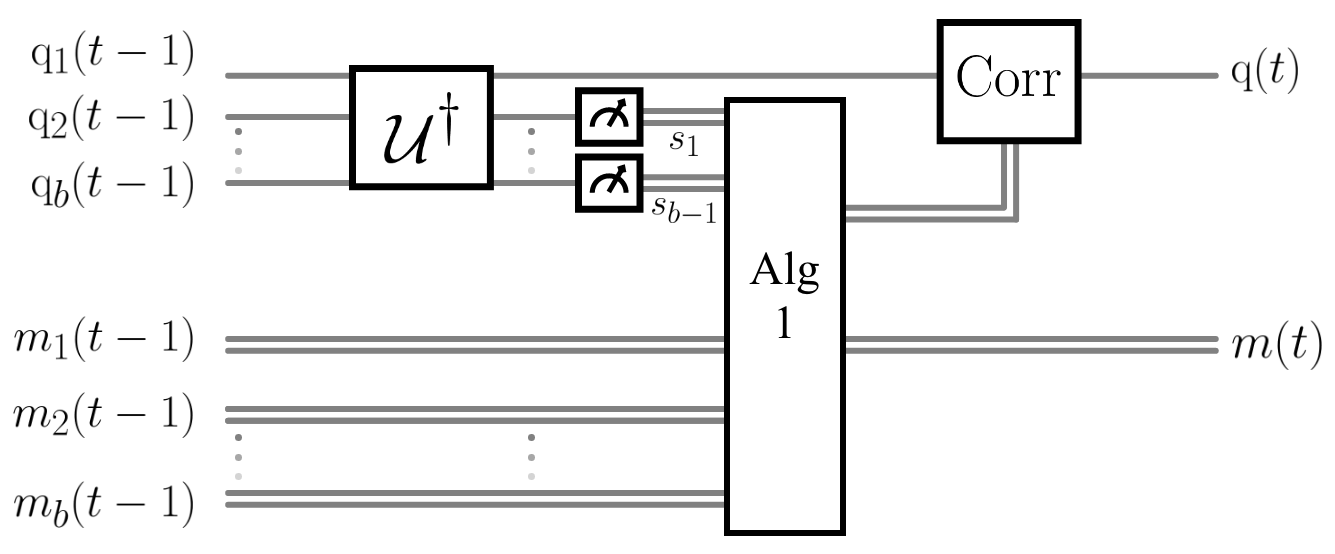}
\includegraphics[width=0.45\textwidth]{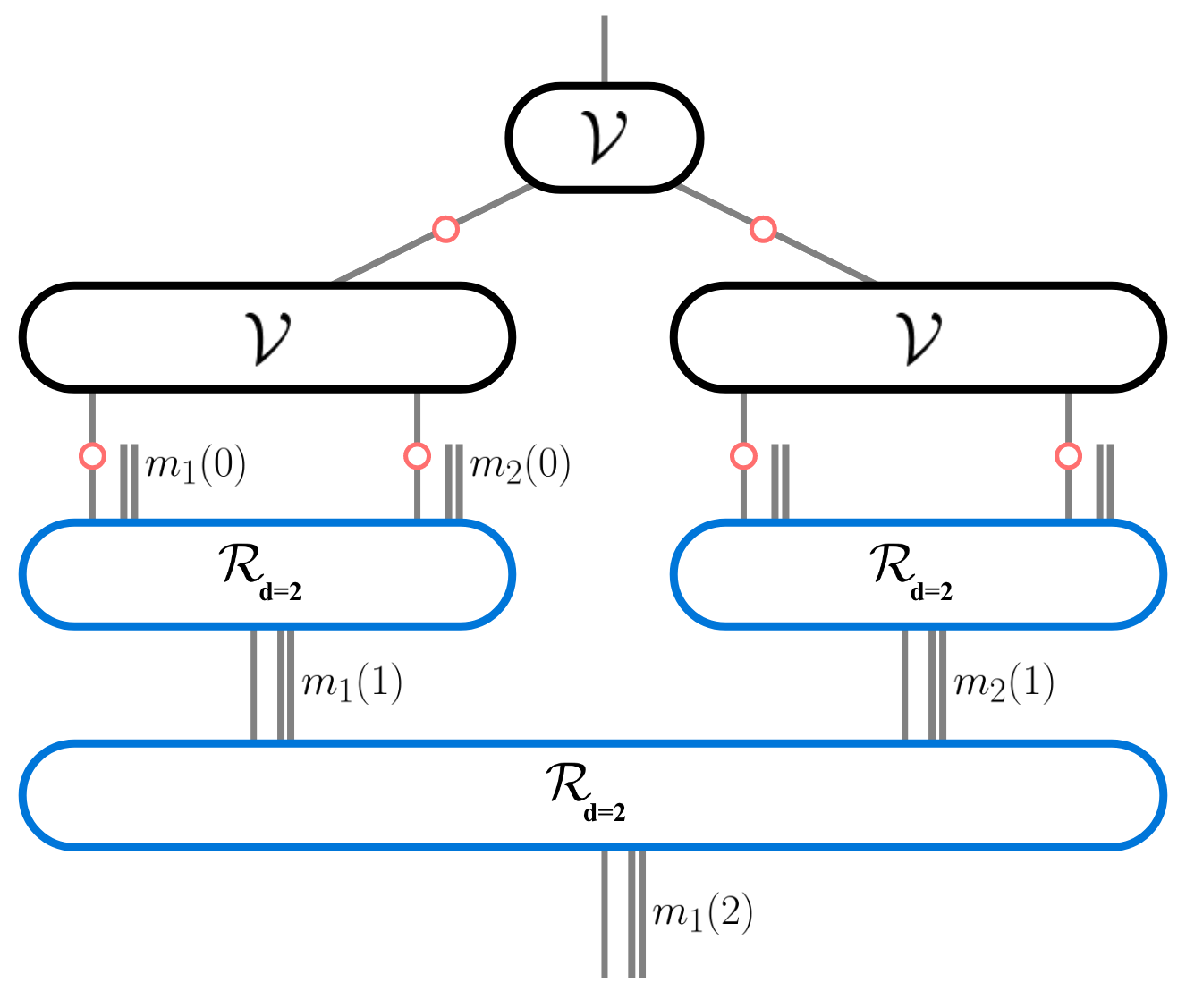}
\caption{The 
\textbf{top} figure shows the 
 decoding module for $d=2$ codes with one reliability bit. It  first applies the the unitary $\mathcal{U}^\dagger$, i.e., the inverse of encoder,  on $b$ input qubits $\{{\rm q}_i(t-1)\}^b_{i=1}$, and  then measures $b-1$ ancilla qubits to obtain the syndromes $\{s_j\}^{b-1}_{j=1}$. Then, it corrects the logical qubit with a unitary ${\rm Corr}$  controlled by the classical output of Algorithm \ref{alg:local rec with bit} that classically processes $\{s_j\}^{b-1}_{j=1}$ and $\{m_i(t-1)\}^b_{i=1}$. The exact dictionary between the classical message+syndrome data and the correction unitary is determined by the encoding unitary $\mathcal{U}$. Furthermore, this algorithm also generates the reliability bit $m(t)$ for the subsequent level $t$. The \textbf{bottom} figure presents the action of this decoder for $T=2$ layers.  
} 
\label{recursivetree2}
\end{figure}

We rigorously prove and numerically demonstrate that using this recovery it is possible to achieve a noise threshold $p_\text{th}$ strictly larger than zero for codes with distance $d=2$. In particular, in Appendix \ref{d=2 appendix} we show that

\begin{proposition}\label{prop: d=2 tree}
Consider the tree channel $\mathcal{E}_T$ defined in Eq.(\ref{noisy tree definition -- no root noise}), for  a $b$-ary tree of depth $T$. Suppose 
the noise at each edge is described by the Pauli channel 
$$\mathcal{N}(\rho)= (1-p_{\text{tot}}) \rho+p_x X\rho X+p_y Y\rho Y + p_z Z\rho Z\ , $$
where $p_{\text{tot}}=p_x+p_y+p_z$  is the total probability of error.  Assume the image of encoder    $V:\mathbb{C}^2\rightarrow \mathbb{C}^{\otimes b}$ 
is  a stabilizer error-correcting code $[[b,1,2]]$,  i.e., a code with distance $2$. Then, 
for $p_{\text{tot}}< 1/(16b^4+4b^2)$, after applying
the  recovery channel $\mathcal{R}_T$ in Sec.\ref{subsec: one bit local recovery}, we obtain the Pauli channel 
$$\mathcal{R}_T\circ{\mathcal{E}}_T(\rho)= (1-r^{\text{tot}}_T) \rho+r^x_T X\rho X+r^y_T Y\rho Y + r^z_T Z\rho Z\ ,$$
with  the total probability of logical error $r^{\text{tot}}_T=r^x_T+r^y_T+r^z_T$, bounded by 
\be\label{br3}
r^{\text{tot}}_T\le (1+8b^2) \times p_{\text{tot}} \ , 
\ee
for all $T$.
\end{proposition}
Indeed, in Appendix \ref{d=2 appendix} we prove a slightly stronger result: recall that the final output of the recursive decoder in Fig.\ref{recursivetree2} is the decoded qubit along with the  corresponding reliability bit $m(T)$ whose value 1 indicates the likelihood of error in the decoded qubit (see subsection \ref{subsec: one bit local recovery} below for further details).
We show that if the total physical error is $p_{\text{tot}}< 1/(16b^4+4b^2)$ then for all $T$ the probability that the reliability bit $m(T)=1$ remains bounded by $8b^2\times p_{\text{tot}}$, and if $m(T)=0$, then the probability of error on the decoded qubit is bounded by  $(16b^4+4b^2)\times p_{\text{tot}}^2$, hence a quadratic suppression of undetected errors.  Then, ignoring the value of the reliability bit $m(T)$, the total probability of error is bounded as Eq.(\ref{br3}).

In conclusion, for $p_{\text{tot}}< 1/(16b^4+4b^2)$,  the probability of logical error $r^{\text{tot}}_T$ is bounded by $(1+8b^2)/(16b^4+4b^2)$, even for the infinite tree. Since $b\ge 2$, this means
$r^{\text{tot}}_T\le \sim 12\%$, which implies 
that both classical information and entanglement propagate over the infinite tree  (see Appendix \ref{d=2 appendix}).

As a simple example, in Fig. \ref{classical d=2 plot} we consider the tree constructed from the binary repetition encoder $|c\rangle\rightarrow|c\rangle|c\rangle\ : c=0,1$, with only bit-flip errors ($p_z=p_y=0$). Note that technically this code has distance $d=1$.  Nevertheless, it detects single-qubit $X$ errors, i.e., it has distance 2 for $X$ errors. 
 Our numerical analysis demonstrates the threshold is $p_x \sim 0.125$ which is consistent with the (albeit weak)  lowerbound  estimate from Proposition \ref{prop: d=2 tree}, $\sim 0.3\%$.  Note that in this case the exact threshold can be determined from  the result of \cite{evans2000} on  the classical broadcasting problem,  namely 
 Eq.(\ref{intro threshold})  which implies
$p^x_{\text{th}}=(1-1/\sqrt{2})/2\approx 14\%$.

Finally, it is worth noting that for trees with standard or anti-standard  encoders of CSS codes with distance 2 with  independent $X$ and $Z$ errors on each edge, one can apply the recursive decoding strategy using \textit{two} reliability bits -- one for each type of error -- to get improved performance.


\begin{figure}[ht!]
\centering
\includegraphics[width=0.5\textwidth]{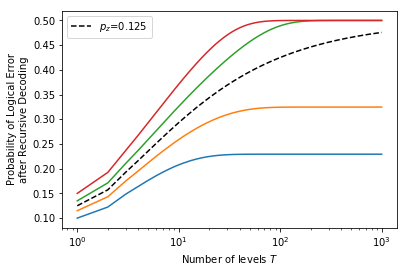}
\caption{\textbf{Performance of local recovery with one reliability bit on binary repetition code tree --} Suppose each node has  the encoder of the binary repetition code, namely $\ket{c}\rightarrow \ket{c}\ket{c}: c=0,1$, and each edge has a bit-flip channel, which applies Pauli $X$ with probability $p_x$. The decoder with one reliability bit achieves  a non-zero noise threshold for information propagation in infinite trees: when $p_x<\sim 0.125$, the logical $X$ error after local recovery of a tree with depth $T\rightarrow \infty$ saturates to strictly below $0.5$, while $p_x>\sim 0.125$ ensures the logical $X$ error after recovery saturates to $0.5$.}
\label{classical d=2 plot}
\end{figure}


\subsection{Recursive decoding with one reliability bit}\label{subsec: one bit local recovery}


Here, we present a full description 
of the above scheme.  Recall that at level $t$ of recursive decoding of a tree of depth $T$, there are 
$b^{T-t}$ identical blocks each with $b$ qubits inside. Since noise affects all blocks in the same way and the same decoder is applied to all blocks, the following analysis applies to all the blocks at this level.    
Therefore, to simplify the notation, we suppress the indices that determine the block under consideration, and just label qubits inside a block with index $i=1,\cdots, b$. Let $m_i(t-1)$ be the corresponding reliability bits that are received from the previous level (see Fig. \ref{recursivetree2}).   In the following, we assume the reliability bit takes the value 1 when  the qubit is \emph{unreliable} or \emph{marked}.     For the first level of recovery, i.e., at the leaves of the tree, we assume the reliability bits are all initiated at $0$. 

Now, similar to the original local recovery in the previous section, the  recovery starts by applying the inverse of the encoder circuit and then measuring $b-1$ ancilla qubits, which determine the value of syndromes $s_j: j=1,\cdots, b-1$. For codes with distance $d=2$ this information is not enough to determine the location of the error. Therefore, unlike the original scheme, the Pauli correction applied to the logical qubit not only depends  on these syndromes, but also depends on the value of reliability bits $m_i(t-1): i=1,\cdots, b$. More specifically,  
\begin{enumerate}

\item If no qubit is marked, i.e., $m_i(t-1)=0$ for $i=1,\cdots, b$,  and no syndrome is observed, i.e., $s_j=0$ for all $j=1,\cdots, b-1$, then decode (without applying any correction) and set $m(t)=0$, indicating the reliability of this qubit.
\item Similarly, if exactly one qubit is marked and no syndrome is observed then again decode and set $m(t)=0$.
\item If exactly one qubit is marked and the observed syndromes
are consistent with a single-qubit error in that location, then assume there is an error in that location,  correct the error, and set $m(t)=0$. 

\item In all the other cases, decode ignoring the reliability bits and set $m(t)=1$.   
\end{enumerate}
This procedure  is also described more precisely in Algorithm \ref{alg:local rec with bit}. Two important remarks are in order here: First, note that only in case 3, the value of the reliability bits affect the applied correction. This is exactly where we take advantage of the fact that codes with distance 2 can correct single-qubit errors in a known location. Secondly, note that if all reliability bits $m_i{(t-1)}=0: i=1,\cdots b$ and some error syndromes are observed, while the local recovery cannot correct the error, it sets $m(t)=1$  signaling to the next layer that an error has been detected in this level. In this case, we take advantage of the error detection property of distance 2 codes. 

Therefore, assuming the probability that a received qubit at this level has an error with probability $p\ll 1$, the probability that the reliability bit is set to $m(t)=1$ is of order $b\times p$ in the leading order. On the other hand, the probability of an undetected (and thus, unmarked)  error is of order $ (b\times p)^2$. This suppression of the probability of error from order $p$ to $p^2$, will make it possible to propagate information over an infinite tree, provided that the probability of physical noise in the single-qubit channels is sufficiently small. In Appendix \ref{d=2 appendix}, we present a rigorous error analysis for this scheme and prove Proposition \ref{prop: d=2 tree}.   

\begin{algorithm}[H]
\caption{Decoding module with one reliability bit}\label{alg:local rec with bit}
\begin{algorithmic}
\State \textbf{begin} with $b$ qubits $i=1,...,b$ and their corresponding \textit{reliability} bits $m_i(t-1)$ received from level $t-1$. \\
\State \textbf{apply $\mathcal{U}^\dagger$ to the $b$ qubits} and \textbf{measure the $b-1$ ancilla qubits} to obtain $\{s_j\}_j$ 
\If{$m_i(t-1)=0, \ \forall i$}
    \If{$s_j=0, \forall j$ } \State $m(t)\leftarrow0$
    \Else \State $m(t)\leftarrow1$
    \EndIf
\ElsIf{$m_i(t-1)=1$ for exactly one $i\in\{1,\cdots, b\}$, denoted as $k$ }
    \If{$s_j=0, \forall j$ } \State $m(t)\leftarrow0$ 
    \ElsIf{some $s_j\neq0$, but $\{s_j\}_j$ are consistent with error on $k^{\rm th}$ qubit}
    \State \textbf{correct} logical qubit appropriately
    \State $m(t)\leftarrow0$
    \Else \State $m(t)\leftarrow1$
    \EndIf
\Else \State $m(t)\leftarrow1$
\EndIf\\
\Return logical qubit and $m(t)$ to $t$ level.
\end{algorithmic}
\end{algorithm} 

\section{\textit{Bell} Tree -- A $d=1$ code tree}\label{bell tree sec}

Next, we study the Bell tree, introduced in Fig. \ref{BinaryvsBell} in the introduction. Here, the encoder is an (anti-standard) encoder of the binary repetition code, which  has distance $d=1$: while it can detect a single  $X$ error (and correct none), it is fully insensitive to $Z$ errors. Nevertheless, we show that if the probability of $X$ and $Z$ errors are sufficiently small (but non-zero) it is still possible to transmit both classical information and entanglement through the infinite tree. We show this using two different decoders: (i) the optimal decoder which is realized using a belief propagation algorithm discussed in the next section, and (ii) a sub-optimal decoder 
introduced in this section,  which recursively applies a decoding module with \textit{two} reliability bits (see Fig. \ref{Bell tree circuit}). 

The relatively simple structure of this  decoder makes it amenable to both analytical and numerical study. For instance, Fig. \ref{d=2 XZ numerics}
presents the results of a numerical study of  trees  with depth up to $T= 1000$, which involves $2^{1000}$ qubits. For this plot,  the physical single-qubit noise in the tree channel $\mathcal{E}_T$ is  $\mathcal{N}=\mathcal{N}_x\circ \mathcal{N}_z$, where $\mathcal{N}_x$ 
and $\mathcal{N}_z$ are bit-flip and phase-flip channels, respectively,   which apply  $X$ and  $Z$ errors, 
 each with probability $p_x=p_z=p$.  This numerical study shows that in the infinite tree  
for 
 $p<\sim 0.5\%$ the probabilities  of logical  $Z$ error $q^z_\infty$ and logical $X$ error $q^x_\infty$ remain strictly smaller than 0.5; namely
  $q^x_\infty \le 7\%$ and $q^z_\infty \le 3\%$. This means that for $p_x=p_z\lesssim 0.5\%$ both classical information and entanglement propagate to any depth in the Bell tree. Note that although the probability of physical $X$ and $Z$ errors are  equal, the probability of logical errors $q^x_T$ and $q^z_T$ are not generally equal (see below for further discussion).  
  In addition to these numerical results,  in Appendix \ref{appendix: bell tree} we also formally prove the existence of a non-zero threshold for this decoder, and estimate the noise threshold to be lowerbounded by $\sim 0.245\%$, which is consistent with the numerical study.

We also note that 
the optimal decoder performance for $T=20$ suggests that the actual threshold is lower than $p_x=p_z\approx1.7\%$. It is also worth noting that the noise threshold for the decoder studied in this section is comparable with the actual noise threshold estimated by the optimal decoder ($\sim 0.5\% $ versus $\sim 1.7\%$), and the 
performance of the two decoders are similar in the very small noise regime, namely $p_x=p_z\ll 0.5\%$. Finally, recall that our results in Sec. \ref{anti standard subsec} imply exponential decay of information for  $p_x=p_z> (1-2^{-1/4})/2 \approx 8\%$.

\begin{figure}[ht!]
\centering
\includegraphics[width=0.5\textwidth]{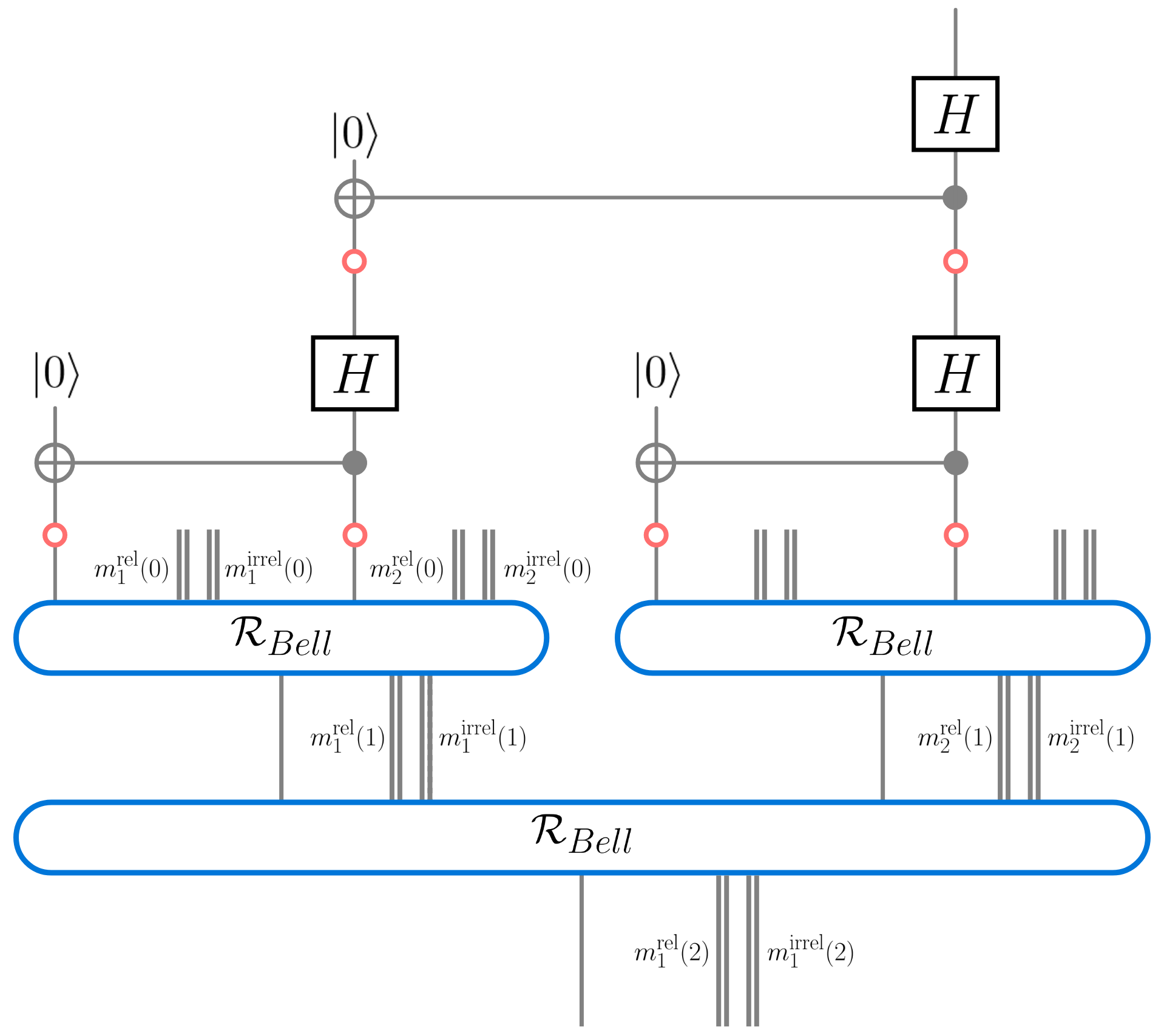}
\caption{
\textbf{2 levels of Bell tree and the corresponding recursive decoder: --} $\mathcal{R}_{Bell}$ is  the decoder module described in Fig.\ref{bell tree logic circuit}. Each qubit is augmented  with two reliability bits $m^{\text{rel}}(t)$ and $m^{\text{irrel}}(t)$, represented by double lines. 
} 
\label{Bell tree circuit}
\end{figure}

\begin{figure}[ht!]
\centering
\includegraphics[width=0.5\textwidth]{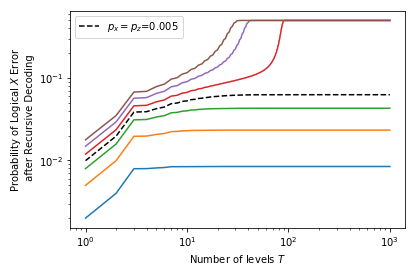}
\caption{\textbf{Performance of the recursive decoder with two reliability  bits on the Bell trees --} We consider the Bell tree with depth $T\rightarrow \infty$ with physical error $p_x=p_z=p$. For $p<\sim 0.005$, the asymptotic logical $X$  errors after decoding saturates to  $q^x_\infty\le \sim 0.07$
(a similar result holds for $q^z_\infty$). On the other hand, for $p>\sim 0.005$ the asymptotic logical $X$ and $Z$ error saturate to $0.5$.}
\label{d=2 XZ numerics}
\end{figure}

\subsection{Recursive decoding with two reliability bits}
\label{subsec: 2 bit local rec}
Here, we present the full description of the recursive decoder with \textit{two} reliability bits 
for the Bell tree.  This decoder recursively applies the decoding module with the circuit diagram in Fig.\ref{bell tree logic circuit}.   The example in Fig.\ref{Bell tree circuit} presents two layers of this decoder.  Recall that at level $t$ of recursive decoding of a Bell tree of depth $T$, there are $2^{T-t}$ identical blocks each with $2$ qubits inside. Since the noise affects all the blocks in the same way and the same decoder is applied to all of them,  similar to the previous section,   we suppress the indices that determine the block under consideration and just label the two qubits inside a block with index $i=1, 2$.

Similar to Algorithm \ref{alg:local rec with bit}, here too the value of the reliability bits at each level is decided based on the observed syndrome and the reliability bits received from the previous level.   Each  
reliability bit determines the likelihood of one type of error, namely $X$ and $Z$ errors, on the decoded qubit. 
 At the first level (i.e., the leaves of the tree) all reliability bits are initiated at 0.  

\begin{figure}[ht!]
\centering
\includegraphics[width=0.45\textwidth]{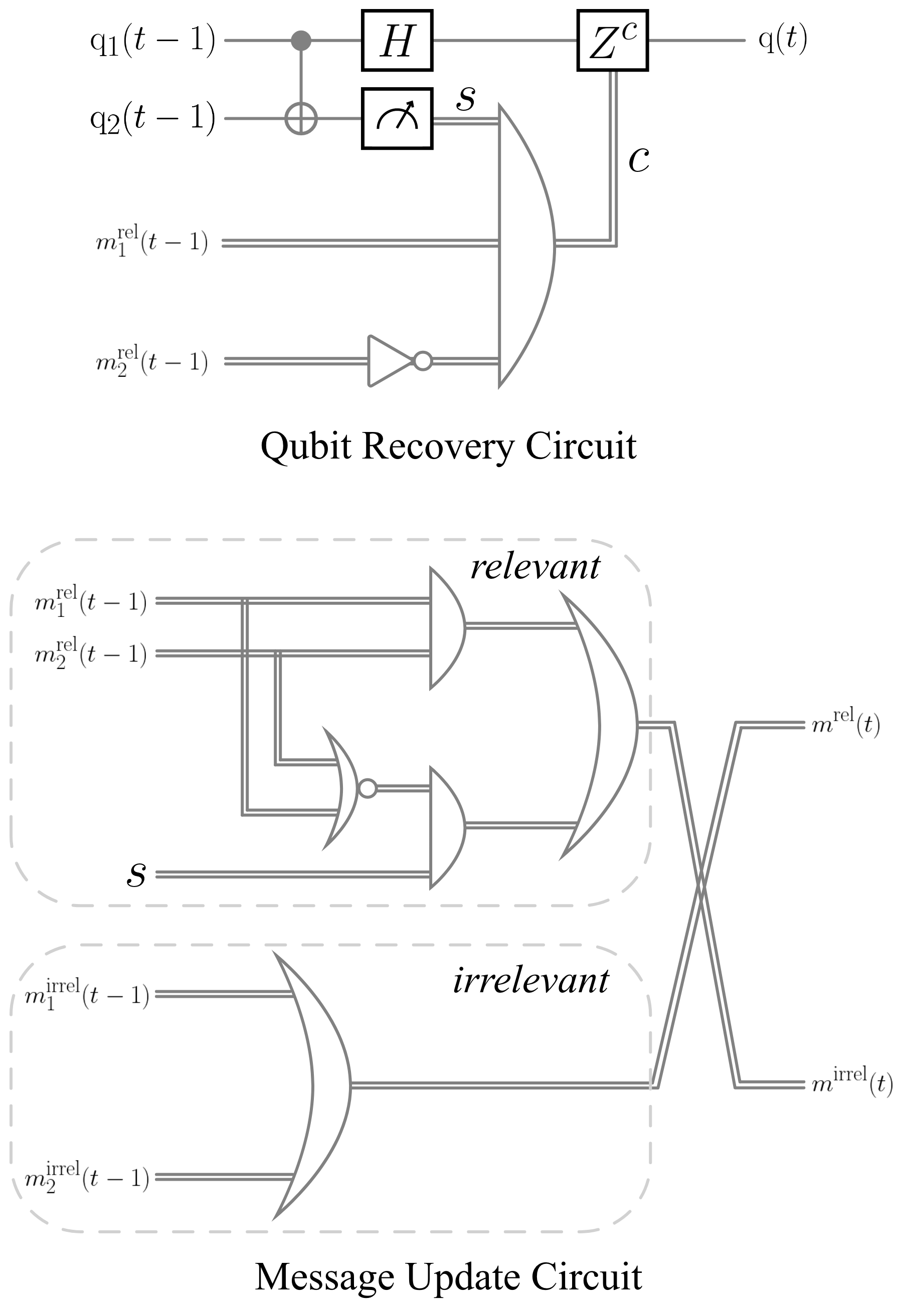}
\caption{\textbf{Decoding module for the recursive decoder of the Bell tree -- }This figure illustrates the three circuits involved in the Bell Tree decoding module   from level $t-1$ to $t$. The input to the module are two qubits ${\rm q}_1(t-1)$ and ${\rm q}_2(t-1)$ along with their ordered reliability bit pairs, labeled as  
$\{(m^{\text{rel}}_i(t-1),m^{\text{irrel}}_i(t-1))\}_i$ for $i=1,2$. The first figure indicates the decoding  circuit that detects the syndrome $s$ and corrects the error based on $m^{\text{rel}}_i(t-1): i=1,2$ and $s$. The bottom circuit generates the reliability bits sent to level $t$. It is segmented into relevant and irrelevant parts. The relevant part takes $m^{\text{rel}}_i(t-1): i=1,2$  \textit{and} the syndrome $s$ as the input, whereas the irrelevant part simply is an ${\rm OR}$ function on $m^{\text{irrel}}_i(t-1): i=1,2$.  Note that the `relevant' and `irrelevant' parts of the decoding module remain separate until the swapping of the bits at the end. This final swap accounts for the fact that the Bell encoder is an anti-standard encoder (see \ref{anti standard subsec}). }
\label{bell tree logic circuit}
\end{figure}

Consider the action of the decoding module in  Fig.\ref{bell tree logic circuit}  on the leaves of the tree. The observed syndrome $s$ is determined by the outcome of measuring $Z\otimes Z$ stabilizer on the pair of qubits in a block. On the other hand, since the decoding module applies a Hadamard gate on the decoded qubit, the measurement performed in the next level determines the outcome of $X$ stabilizers measured  on the leaf qubits. At each level, the decoding module should combine the information obtained from the syndrome measurement with the right type of reliability bits ($X$ or $Z$). This can be achieved, e.g.,  by adding a third bit that keeps track of the parity of the current level, i.e., whether it is odd or even.

However, in the decoding module in Fig.\ref{bell tree logic circuit}, we use a slightly different approach that allows us to avoid the use of this additional bit. Namely, we  keep this information in the order of the reliability bits:  when qubits $i=1,2$ are received from the level $t-1$, they each come with two reliability bits labeled as \emph{relevant} and \emph{irrelevant} bits. The relevant bit is determined by the values of the syndromes at levels $t-2, t-4, \cdots$, which corresponds to the same type of stabilizer that is going to be measured at level $t$ (i.e., $X$ type or $Z$ type).

More precisely, the decoding module in Fig.\ref{bell tree logic circuit} is designed to realize the following rules:

\begin{enumerate}
\item Upon receiving qubits 1 and 2, measure the stabilizer $Z_1\otimes Z_2$ and denote the outcome by $s$.

\item For the pair of received qubits if both  relevant reliability bits are zero, i.e.,  
$m_1^{\text{rel}}(t-1) \text{ OR } m_2^{\text{rel}}(t-1)=0$, and no syndrome is observed ($s=0$), then apply the inverse encoder without applying any correction and set the relevant reliability bit to zero, $m^{\text{rel}}(t)=0$.
\item Similarly, if exactly one relevant reliability bit is 1, i.e., $m_1^{\text{rel}}(t-1) \text{ XOR } m_2^{\text{rel}}(t-1)=1$, and no syndrome
is observed ($s=0$), then  apply the inverse encoder without applying any correction and set the relevant reliability bit to 0, i.e., $m^{\text{rel}}(t)=0$.
\item If exactly one relevant reliability bit is 1 and the syndrome is observed (i.e., $s=1$), then 
assume an error has happened on the marked qubit, correct the error, then apply the inverse encoder and set the relevant reliability bit to 0, i.e.,  $m^{\text{rel}}(t)=0$.
\item In all the other cases, apply the inverse encoder ignoring the reliability bits and set $m^{\text{rel}}(t)=1$. 
\item Set the irrelevant reliability bit to 1 if any of the received irrelevant reliability bits are 1, i.e., $m^{\text{irrel}}(t)=m_1^{\text{irrel}}(t-1)\text{ OR } m_2^{\text{irrel}}(t-1) $, and then swap the relevant and irrelevant bits.
\end{enumerate}
This procedure is  described more precisely in Algorithm \ref{alg:local rec with 2 bits} in Appendix \ref{appendix: bell tree}, where we also present a rigorous error analysis for this algorithm.  
It is worth noting that although the circuit in Fig.\ref{bell tree logic circuit} implements the above rules, the order of implementation is slightly different: 
 it first applies the inverse of the encoder and measures the ancilla qubit, and then applies the correction.   It is also worth noting that only in case  4, a correction is performed, and  after that the reliability bit  is set to $m^{\text{rel}}(t)=0$.  It is interesting to consider a more ``conservative" version of this decoder where in case 4, after the correction the relevant reliability bit is set to $m^{\text{rel}}(t)=1$, indicating unreliability of the decoded qubit. Numerical results presented in Appendix \ref{appendix: bell tree} suggests that this modification reduces the threshold to around 
$\sim 0.3\%$, and generally results in a higher probability of logical error.

\section{Optimal Decoding with Belief Propagation}\label{Optimal section}

Similar to any stabilizer code, the optimal decoding for the concatenated code defined by  the tree encoder $V_T=\prod_{j=0}^{T-1} V^{\otimes {b}^j}$ can be performed by  measuring the stabilizer generators. This can be realized by implementing the inverse of this encoder unitary and then measuring all the ancilla qubits in the $Z$ basis. For a tree of depth $T$, this  yields a syndrome bitstring $\mathbf{s}_T$ of length $b^T-1$. In the absence of noise, one obtains the all-zero bit string and the remaining decoded qubit is equal to the input state of the tree.  
On the other hand, in the presence of an error, the output qubit is equal to the input, upto a logical error  $L_T=\{I,X, Y, Z\}$. Then, any  decoder uses the syndrome $\mathbf{s}_T$ to infer $L_T$ and correct it by applying $L_T$ on the decoded qubit. The optimal decoder needs to find $L_T$ 
that maximizes the conditional probability  
$p(L_T|\mathbf{s}_T)$.   
\textit{Prima facie},  this would require an inefficient double-exponential search over $2^{b^T-1}$ distinct values of the syndromes.

However, thanks to the tree structure, 
 it turns out that in the case that is of interest in this paper, the decoding can be performed exponentially more efficiently, in time $\mathcal{O}(2^T)$.
  In 2006, Poulin \cite{poulin2006optimal} developed an efficient algorithm for decoding concatenated codes via belief propagation (BP). This algorithm assumes the encoder is noiseless for the entire concatenated code and the qubits at the leaves of the tree are affected by Pauli errors that are uncorrelated between different qubits. On the other hand, errors between the encoders in the tree, that are considered in this paper,  are equivalent to correlated errors on the leaves of the tree.  
Although the efficient algorithm of \cite{poulin2006optimal} cannot be extended to general correlated errors, 
as we explain below, it can be extended to the type of correlated errors that are caused by local errors in between encoders. 
 Roughly speaking, this is possible because such correlations are consistent with the tree structure.


\subsection{Belief Propagation Update Rule}
\begin{figure}[ht!]
\centering
\includegraphics[width=0.4\textwidth]{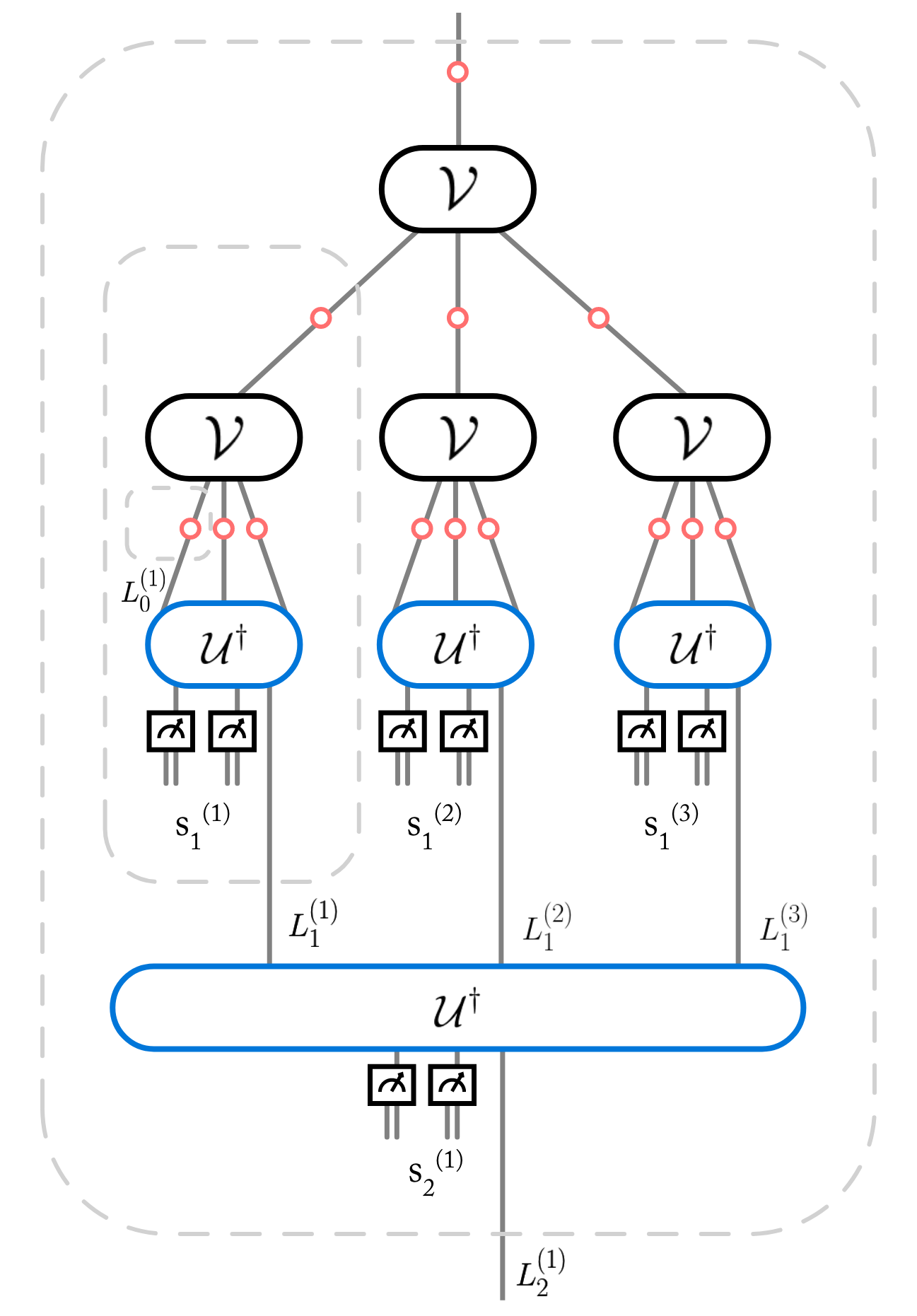} 
\caption{The upper half of the figure is a tree with stabilizer encoder $\mathcal{V}$ along with Pauli noise $\mathcal{N}$ on each edge, which are denoted by red circles. The lower half is the application of the inverse unitaries $U^\dag$ on the output of the tree process (Note that $V|\psi\rangle=U|\psi\rangle|0\rangle^{\otimes (b-1)}$).  Upon inversion, each ancillary qubit is measured in the $Z$-basis to obtain the syndromes.  $L^{(i)}_t\in\{I, X, Y, Z\}$ indicates the overall logical error 
caused by the noise channels inside the dashed bubble above it, and  ${s}^{(i)}_t$ is  the corresponding syndrome string.}
\label{BP_Circuit23}
\end{figure}

We explain how the algorithm works for the 3-ary tree of depth $T=2$, presented in the example in Fig.\ref{BP_Circuit23}. Each red dot in this figure corresponds to the Pauli noise channel $$\mathcal{N}(.)=r_I(.)+r_xX(.)X + r_yY(.)Y + r_zZ(.)Z\ .$$
To simplify the presentation, we have considered a tree with noise at the root.

Recall that the optimal decoder needs to determine $p(L_T|\textbf{s}_T)$ for the observed value of $\textbf{s}_T$.  Using the labeling introduced in Fig.\ref{BP_Circuit23}. The goal is to determine 
$$p(L_2^{(1)}|s_2^{(1)}, s_1^{(1)}, s_1^{(2)},  s_1^{(3)})\ .$$ 
Consider the action of the decoder circuit  on the first level. Let error $E$ be an arbitrary tensor product of Pauli operators and the identity operator. Then, 
\begin{align}
    U^\dagger E U=\mathcal{L}(E) \otimes \mathbf{X}^{{\rm Synd}(E)} \mathbf{Z}^\mathbf{z},
\end{align}
where $\mathcal{L}(E)$ is the logical operator acting on the logical qubit and $\mathbf{X}^{{\rm Synd}(E)}$ is the syndrome acting on the ancilla qubits. Since the ancilla is initially prepared in $\ket{0}^{\otimes b-1}$, it remains unchanged under $\mathbf{Z}^\mathbf{z}$. We measure and record the syndrome ${\rm Synd}(E)$ in each step of decoding. 

The key point is $p(L^{(1)}_2|s_2^{(1)}, s_1^{(1)}, s_1^{(2)},  s_1^{(3)})$ has a decomposition as 
\begin{align}\label{BP 3}
p&(L^{(1)}_2|s_2^{(1)}, s_1^{(1)}, s_1^{(2)},  s_1^{(3)})\nonumber\\ &= \frac{1}{p(s^{(1)}_2|\mathbf{s}_1)}\sum_{\textbf{L}_1}f_{s^{(1)}_2}(L^{(1)}_2, \textbf{L}_1)\times \prod^3_{i=1} p(L^{(i)}_1|{s}^{(i)}_1)\ ,
\end{align}
where $\textbf{L}_1=L^{(1)}_1 L^{(2)}_1 L^{(3)}_1$, $\mathbf{s}_1=s^{(1)}_1s^{(2)}_1s^{(3)}_1$ and 
\begin{align}
    f_{s^{(1)}_2}(L^{(1)}_2,\textbf{L}_1):= \delta[s^{(1)}_2={\rm Synd}(\textbf{L}_1)]\times N(L^{(1)}_2|\mathcal{L}(\textbf{L}_1)),\nonumber
\end{align}
where the conditional probability $N$ describes the  transition probability 
associated to Pauli channel $\mathcal{N}$, such that  for any pair of $E_1$ and $E_2$ in the Pauli group, 
\begin{align}\label{pauli noise transition matrix}
N(E_1|E_2)=r_{E_1E_2}\ ,
\end{align}
where $E_1E_2$ is the matrix product of Pauli operators $E_1$ and $E_2$ upto a global phase.

Notice that conditioned on $s^{(i)}_1$, the logical error $L^{(i)}_1$ is independent of the rest of the syndrome bits, which is clear from the tree structure.  Now, this formula can be reapplied to $p(L^{(i)}_1|s^{(i)}_1)$ for each $i$. For instance, for the $i=1$ subtree  we have 
\begin{align}
p(L^{(1)}_1|s^{(1)}_1) = \sum_{\textbf{L}_0} f_{s^{(1)}_1}(L^{(1)}_1,\textbf{L}_0)\frac{1}{p(s^{(1)}_1)} \times \prod_{i=1}^3 p(L^{(i)}_0) 
\end{align}
where $\textbf{L}_0=L^{(1)}_0L^{(2)}_0L^{(3)}_0$, and the superscript labels the three leaves in the first subtree. We can repeat this for the other two trees as well. At this point, we have reached level $0$, i.e., the leaves of the tree. Therefore, when $p(L^{(i)}_0)$ appears in the expression, we have reached the terminal condition for recursion and we set $p(L^{(i)}_0)=N(L^{(i)}_0|I)=r_{L^{(i)}_0}$.  See Appendix \ref{BP Appendix} for further details of this algorithm for general stabilizer trees. In the following, we present some numerical examples.


\begin{figure}[ht!]
\centering
\includegraphics[width=0.5\textwidth]{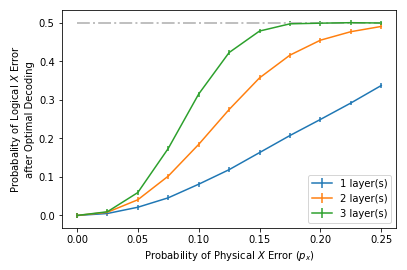} 
\includegraphics[width=0.5\textwidth]{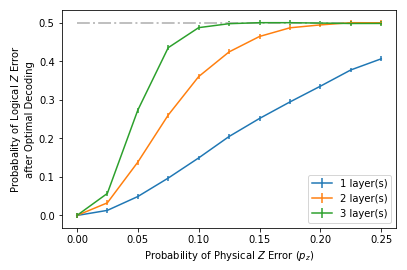} 

\caption{\textbf{Optimal recovery of Shor-9 tree.} Probability of logical error as a function of the probability of physical error in the tree for trees of depth $T=1, 2$ and $3$. In the \textbf{top} plot, the $y$-axis indicates the overall \textit{logical} $X$-error \textit{after} optimal decoding given the  physical $X$ noise within the tree has probability $p_x$ indicated on the $x$-axis. The \textbf{bottom} plot shows the same quantities for $Z$ errors.  This provides improved upperbound estimates for  noise threshold for $X$ errors at $\sim 0.17$, and  for $Z$ errors at $\sim 0.13$. Each data point in the plots is computed via a Monte Carlo of $10^5$ samples explained in Appendix \ref{BP Appendix}.}
\label{Shor 9 optimal numerics}
\end{figure}

\begin{figure}[ht!]
\centering
\includegraphics[width=0.5\textwidth]{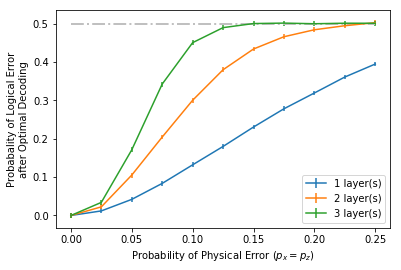}
\caption{\textbf{Optimal recovery of Steane-7 tree.} Probability of logical error as a function of the probability of physical error for trees of depth $T=1,2,3$. Note that since Steane-7 code is self-dual  $X$ and $Z$ errors have the same behavior, and therefore
this plot presents both cases. In particular, $x$ axis is the  probability $p_x=p_z$ of physical error in the channel. This provides  upperbound estimates for asymptotic noise threshold for $X$ (or, $Z$) error at $p\sim 0.15$. Each data point in the plots is computed via a Monte Carlo of $10^5$ samples explained in Appendix \ref{BP Appendix}.}
\label{Steane 7 optimal numerics}
\end{figure}

\subsection{Examples}\label{BP numerics}
\subsubsection*{Shor-9  Code} Fig. \ref{Shor 9 optimal numerics} plots the logical error for the tree as defined in Eq.(\ref{noisy tree definition -- no root noise}) constructed from a Shor-9 code encoder that is standard in both $X$ and $Z$.  We consider trees of depths $T=1,2,3$. 
The specifics of constructing this plot are explained in Appendix \ref{BP Appendix}. We see that the asymptotic noise threshold for $X$ errors is below $\sim 0.17$, whereas for $Z$ errors it is below $\sim 0.13$ (since the number of $Z$ stabilizers in the Shor-9 code is higher than $X$ stabilizers and therefore this code is better equipped to correct $X$ errors than $Z$ errors). Recall that since this code has distance 3, our result in Proposition \ref{Prop 1} establishes an upper bound $\sim 0.21$; thus, these numerics provide evidence for a tighter bound for the noise threshold.

After optimal recovery for channel $\mathcal{E}_T$, the overall channel is a composition of a bit-flip and a phase-flip channel with the probability of $X$ error and $Z$ errors $q^x_T, q_T^z\le 1/2$, respectively. Such channels are entanglement-breaking if, and only if, $|1-q_T^z|\times |1-q_T^x|<1/2$ (see Appendix \ref{ent depth uncorrelated XZ}).  From numerical results in Fig.\ref{Shor 9 optimal numerics}, it follows that for depth $T=3$, this condition is satisfied for $p>\sim 0.07$. Note that since the error monotonically increases with $T$, the same also holds for $T\ge 3$.

\subsubsection*{Steane-7 Code}
Fig. \ref{Steane 7 optimal numerics} plots the optimal recovery performance for a tree as defined in Eq.(\ref{noisy tree definition -- no root noise}) composed of Steane-7 encoders that are standard in both $X$ and $Z$, and is similarly subject to independent $X$ and $Z$ errors for depths $T=1,2,3$. Specifics of this recovery procedure are detailed in Appendix \ref{BP Appendix}. Here,  we see that the asymptotic noise threshold for $X$ errors is below $\sim 0.15$. Because the Steane-7 code is self-dual, the same threshold holds true for the $Z$ errors as well. Using an argument similar to the Shor-9 code, one can show that for  $p>\sim 0.07$, entanglement cannot be transmitted beyond depth $T\ge 3$.

\section{Mapping Quantum trees to classical trees with correlated errors}\label{deph tree mapping sec}
In this section, we argue that stabilizer trees can be understood in terms of an equivalent classical problem, which is a variant of the standard broadcasting problem on trees discussed in the introduction. For simplicity,  we focus on trees constructed from a standard encoder of a CSS code with independent phase-flip and bit-flip channels.

Suppose we modify the original channel $\mathcal{E}_T$ defined in Eq.(\ref{noisy tree definition -- no root noise}), by adding a fully dephasing channel $\mathcal{D}_z$ before and after each encoder. 
That is, at each node we modify the encoder $\mathcal{V}$ to
\be\label{deph}
\mathcal{V}\ \ \ \longrightarrow \ \ \  \overline{\mathcal{V}}\equiv \mathcal{D}_z^{\otimes b}\circ\mathcal{V}\circ \mathcal{D}_z\ .
\ee
Then, the channel $\mathcal{E}_{T}$ defined in Eq.(\ref{noisy tree definition -- no root noise}) will be modified to 
\be
\overline{\mathcal{E}}_{T}(\rho) =\prod_{j=0}^{T-1}   {\overline{\mathcal{V}}}^{\otimes b^j}\circ \mathcal{N}^{\otimes b^j}(\rho)\ ,
\ee
which will be called the \emph{dephased tree} in the following. Equivalently, rather than inserting dephasing channels $\mathcal{D}_z$, we can assume the probability of $Z$ error $p_z$ is increased to 1/2. 

Then, at any level of the tree the density operator of qubits will be diagonal in the computational basis.  
In particular, the action of the dephased encoder $\overline{\mathcal{V}}$ is fully characterized by the classical channel with the conditional probability distribution
\be\label{class-enc}
P_z(\textbf{z}|c)=\langle \textbf{z} |\mathcal{V}(|c\rangle\langle c|)|\textbf{z}\rangle\ \ \ \ \  : \ c=0,1\ ,\ \textbf{z}\in\{0,1\}^b\ \ .
\ee
Note that the property that the density operator is diagonal in the computational basis remains valid under Pauli errors. In particular, $Z$ errors act trivially on such states, whereas $X$ errors, i.e., bit-flip channels, become a binary symmetric channel, which can described by the conditional probability
\be
N_z(j|i)=p_x+(1-2p_x)\delta_{i,j}\ \ \  : \ i, j=0,1\ .
\ee
Therefore, by adding dephasing before and after each encoder, as described in Eq.(\ref{deph}), we obtain a fully classical problem: at each node, a bit $c\in\{0,1\}$ enters the encoder, and at the output of the encoder  a bit string $\textbf{z}\in\{0,1\}^b$ with probability $P(\textbf{z}|c)$ is generated. Then, each bit goes into a binary symmetric channel $Q$, which flips the input bit with probability $p_x$ and leaves it unchanged with probability $1-p_x$.  Finally, each bit goes to the next level and enters another encoder $\overline{\mathcal{V}}$.


The assumption that the encoder at each node is a standard encoder implies that the effect of inserted dephasing channels $\mathcal{D}_z$ in the middle of the tree is equivalent to correlated $Z$ errors on the leaves.  Furthermore, because for CSS codes $Z$ and $X$ errors can be corrected independently, we conclude that the probability of logical $X$ error, $q_T^x$ for the dephased channel $\overline{\mathcal{E}}_{T}$ and the original channel ${\mathcal{E}}_{T}$  are equal. Another way to phrase this observation is to say that 
the distinguishability of the output states $\mathcal{E}_T(|0\rangle\langle 0|)$ and  $\mathcal{E}_T(|1\rangle\langle 1|)$ 
with respect to any measure of distinguishability, such as the trace distance,  is the same as the distinguishability of 
states 
$\overline{\mathcal{E}}_{T}(|0\rangle\langle 0|)$ and $ \overline{\mathcal{E}}_{T}(|1\rangle\langle 1|)$. To prove this it suffices to show that there exists a channel $\mathcal{T}_T$
such that 
\be
\mathcal{T}_T\circ \overline{\mathcal{E}}_T(|c\rangle\langle c|)=\mathcal{E}_T(|c\rangle\langle c|)\ \ \ : c=0,1\ .
\ee
A channel $\mathcal{T}_T$ that satisfies the above equation is the error correction of $Z$ errors, which requires measuring $X$ stabilizers and correcting the $Z$ errors based on the outcomes of the measurement,  followed by adding certain (correlated) $Z$ errors to reproduce the effect of $Z$ errors in $\mathcal{E}_T$  (see proposition \ref{prop} in Appendix \ref{deph appendix}).   Note that instead of the $Z$ basis, we can dephase qubits in the $X$ basis. This results in  another classical tree with depth $T$, which can be used to determine the probability of logical $Z$ error $q_T^z$.  Similarly, in the case of CSS codes with anti-standard encoders the equivalent classical tree can be obtained  by measuring the qubits in the $Z$ and $X$ bases, alternating between different levels (see the example below).

We conclude that  any CSS tree with standard or anti-standard encoders and independent $X$ and $Z$ errors  
can be fully characterized in terms  of a modification of the  classical broadcasting problem which involves correlated noise on the  edges that leave the same node. As an example, here we consider the classical tree corresponding to  the Shor 9-qubit code. See Appendix \ref{shor dephased} for further examples.

\subsection{Example: Bell Tree}
Consider the Bell tree discussed in Sec.(\ref{bell tree sec}) with noise rate $p_x=p_z=p$. Recall that the Bell encoder is an anti-standard encoder. Thus, as mentioned above, we must dephase alternatingly in the $X$ and $Z$ bases. Suppose we are interested in the propagation of information encoded in $Z$ eigenstates (i.e., $\{\ket{0},\ket{1}\}$) of the input. Then, we shall dephase the root of the tree in the $Z$ direction, then the next level must be dephased in the $X$ direction, and so on. This implies that the dephased versions of two levels of a Bell tree is the concatenation of two kinds of classical encoders
\begin{align}
    \mathbb{M}_1:&\ P_z(00|0)=1\nonumber\\ 
    &\ P_z(11|1)=1\\
    \mathbb{M}_2:&\   P_z(\mathbf{w}|0)=1/2,\ {\rm when}\ \mathbf{w} \in \{00, 11\}\nonumber\\
  &\   P_z(\mathbf{w}|1)=1/2,\ {\rm when}\ \mathbf{w} \in \{01, 10\}
\end{align}
that are present on alternate levels of the tree.  For instance, the dephased version of a $T=2$ Bell tree (as seen in Fig. (\ref{BinaryvsBell}) is,
\begin{align}
    \mathbb{N}^{\otimes 4} \circ \mathbb{M}_1^{\otimes 2} \circ \mathbb{N}^{\otimes 2} \circ \mathbb{M}_2\ ,
\end{align}
where $\mathbb{N}$ represents the classical bit-flip channel that flips the input bit with probability $p$.

\section{Future Directions}
This paper opens up several lines of inquiry. For instance, finding the exact noise threshold for the propagation of classical information and entanglement in an infinite tree remains an open question. 
A natural extension of this work is to consider other (non-Clifford) encoders and noise processes. In particular, the case of Haar-random (independent) unitaries \cite{dalzell2021random,deshpande2022tight} at different nodes is an interesting model for natural processes. Other noise models such as erasure noise or coherent noise can also be studied. While the current study assumes that the decoder has access to \textit{all} the exponentially many qubits at the end of the tree, more pragmatically one can study how classical and quantum correlations are affected when we consider only a subset of the leaves instead of all of them. A thorough study of classical trees obtained by dephasing quantum trees (as explained in Sec.\ref{open q}) is an open question not just of independent interest within classical network theory, but with direct applications to CSS code trees. Finally, some ideas presented in this paper can have applications for more general quantum networks described by directed acyclic graphs.\\


 \section*{Acknowledgments}
 This work is supported by a collaboration between the US DOE and other Agencies. This material is based upon work supported by the U.S. Department of Energy, Office of Science, National Quantum Information Science Research Centers, Quantum Systems Accelerator. Additional support is acknowledged from  NSF QLCI grant OMA-2120757,  NSF grant FET-1910571, NSF FET-2106448.  
  We acknowledge helpful discussions with Thomas Barthel, Ken Brown, Daniel Gottesman, Michael Gullans, David Huse, Jianfeng Lu, Henry Pfister, and Grace Sommers. SAY sincerely thanks Saathwik Yadavalli for all the time and effort he volunteered for making figures for the paper despite his busy schedule.

\bibliography{main}

\begin{thebibliography}{49}%
\makeatletter
\providecommand \@ifxundefined [1]{%
 \@ifx{#1\undefined}
}%
\providecommand \@ifnum [1]{%
 \ifnum #1\expandafter \@firstoftwo
 \else \expandafter \@secondoftwo
 \fi
}%
\providecommand \@ifx [1]{%
 \ifx #1\expandafter \@firstoftwo
 \else \expandafter \@secondoftwo
 \fi
}%
\providecommand \natexlab [1]{#1}%
\providecommand \enquote  [1]{``#1''}%
\providecommand \bibnamefont  [1]{#1}%
\providecommand \bibfnamefont [1]{#1}%
\providecommand \citenamefont [1]{#1}%
\providecommand \href@noop [0]{\@secondoftwo}%
\providecommand \href [0]{\begingroup \@sanitize@url \@href}%
\providecommand \@href[1]{\@@startlink{#1}\@@href}%
\providecommand \@@href[1]{\endgroup#1\@@endlink}%
\providecommand \@sanitize@url [0]{\catcode `\\12\catcode `\$12\catcode `\&12\catcode `\#12\catcode `\^12\catcode `\_12\catcode `\%12\relax}%
\providecommand \@@startlink[1]{}%
\providecommand \@@endlink[0]{}%
\providecommand \url  [0]{\begingroup\@sanitize@url \@url }%
\providecommand \@url [1]{\endgroup\@href {#1}{\urlprefix }}%
\providecommand \urlprefix  [0]{URL }%
\providecommand \Eprint [0]{\href }%
\providecommand \doibase [0]{https://doi.org/}%
\providecommand \selectlanguage [0]{\@gobble}%
\providecommand \bibinfo  [0]{\@secondoftwo}%
\providecommand \bibfield  [0]{\@secondoftwo}%
\providecommand \translation [1]{[#1]}%
\providecommand \BibitemOpen [0]{}%
\providecommand \bibitemStop [0]{}%
\providecommand \bibitemNoStop [0]{.\EOS\space}%
\providecommand \EOS [0]{\spacefactor3000\relax}%
\providecommand \BibitemShut  [1]{\csname bibitem#1\endcsname}%
\let\auto@bib@innerbib\@empty
\bibitem [{\citenamefont {Preskill}(2018)}]{Preskill2018Nisq}%
  \BibitemOpen
  \bibfield  {author} {\bibinfo {author} {\bibfnamefont {J.}~\bibnamefont {Preskill}},\ }\bibfield  {title} {\bibinfo {title} {Quantum {C}omputing in the {NISQ} era and beyond},\ }\href {https://doi.org/10.22331/q-2018-08-06-79} {\bibfield  {journal} {\bibinfo  {journal} {{Quantum}}\ }\textbf {\bibinfo {volume} {2}},\ \bibinfo {pages} {79} (\bibinfo {year} {2018})}\BibitemShut {NoStop}%
\bibitem [{\citenamefont {Gullans}\ and\ \citenamefont {Huse}(2020{\natexlab{a}})}]{gullans_huse_measinduced}%
  \BibitemOpen
  \bibfield  {author} {\bibinfo {author} {\bibfnamefont {M.~J.}\ \bibnamefont {Gullans}}\ and\ \bibinfo {author} {\bibfnamefont {D.~A.}\ \bibnamefont {Huse}},\ }\bibfield  {title} {\bibinfo {title} {Dynamical purification phase transition induced by quantum measurements},\ }\href {https://doi.org/10.1103/PhysRevX.10.041020} {\bibfield  {journal} {\bibinfo  {journal} {Phys. Rev. X}\ }\textbf {\bibinfo {volume} {10}},\ \bibinfo {pages} {041020} (\bibinfo {year} {2020}{\natexlab{a}})}\BibitemShut {NoStop}%
\bibitem [{\citenamefont {Gullans}\ and\ \citenamefont {Huse}(2020{\natexlab{b}})}]{gullans_huse_2_measinduced}%
  \BibitemOpen
  \bibfield  {author} {\bibinfo {author} {\bibfnamefont {M.~J.}\ \bibnamefont {Gullans}}\ and\ \bibinfo {author} {\bibfnamefont {D.~A.}\ \bibnamefont {Huse}},\ }\bibfield  {title} {\bibinfo {title} {Scalable probes of measurement-induced criticality},\ }\href {https://doi.org/10.1103/PhysRevLett.125.070606} {\bibfield  {journal} {\bibinfo  {journal} {Phys. Rev. Lett.}\ }\textbf {\bibinfo {volume} {125}},\ \bibinfo {pages} {070606} (\bibinfo {year} {2020}{\natexlab{b}})}\BibitemShut {NoStop}%
\bibitem [{\citenamefont {Skinner}\ \emph {et~al.}(2019)\citenamefont {Skinner}, \citenamefont {Ruhman},\ and\ \citenamefont {Nahum}}]{skinner_measinduced}%
  \BibitemOpen
  \bibfield  {author} {\bibinfo {author} {\bibfnamefont {B.}~\bibnamefont {Skinner}}, \bibinfo {author} {\bibfnamefont {J.}~\bibnamefont {Ruhman}},\ and\ \bibinfo {author} {\bibfnamefont {A.}~\bibnamefont {Nahum}},\ }\bibfield  {title} {\bibinfo {title} {Measurement-induced phase transitions in the dynamics of entanglement},\ }\href {https://doi.org/10.1103/PhysRevX.9.031009} {\bibfield  {journal} {\bibinfo  {journal} {Phys. Rev. X}\ }\textbf {\bibinfo {volume} {9}},\ \bibinfo {pages} {031009} (\bibinfo {year} {2019})}\BibitemShut {NoStop}%
\bibitem [{\citenamefont {Li}\ \emph {et~al.}(2019)\citenamefont {Li}, \citenamefont {Chen},\ and\ \citenamefont {Fisher}}]{fischer_measinduced}%
  \BibitemOpen
  \bibfield  {author} {\bibinfo {author} {\bibfnamefont {Y.}~\bibnamefont {Li}}, \bibinfo {author} {\bibfnamefont {X.}~\bibnamefont {Chen}},\ and\ \bibinfo {author} {\bibfnamefont {M.~P.~A.}\ \bibnamefont {Fisher}},\ }\bibfield  {title} {\bibinfo {title} {Measurement-driven entanglement transition in hybrid quantum circuits},\ }\href {https://doi.org/10.1103/PhysRevB.100.134306} {\bibfield  {journal} {\bibinfo  {journal} {Phys. Rev. B}\ }\textbf {\bibinfo {volume} {100}},\ \bibinfo {pages} {134306} (\bibinfo {year} {2019})}\BibitemShut {NoStop}%
\bibitem [{\citenamefont {Von~Neumann}(1956)}]{von1956probabilistic}%
  \BibitemOpen
  \bibfield  {author} {\bibinfo {author} {\bibfnamefont {J.}~\bibnamefont {Von~Neumann}},\ }\bibfield  {title} {\bibinfo {title} {Probabilistic logics and the synthesis of reliable organisms from unreliable components},\ }\href@noop {} {\bibfield  {journal} {\bibinfo  {journal} {Automata studies}\ }\textbf {\bibinfo {volume} {34}},\ \bibinfo {pages} {43} (\bibinfo {year} {1956})}\BibitemShut {NoStop}%
\bibitem [{\citenamefont {Pippenger}(1985)}]{pippenger1985networks}%
  \BibitemOpen
  \bibfield  {author} {\bibinfo {author} {\bibfnamefont {N.}~\bibnamefont {Pippenger}},\ }\bibfield  {title} {\bibinfo {title} {On networks of noisy gates},\ }in\ \href@noop {} {\emph {\bibinfo {booktitle} {26th Annual Symposium on Foundations of Computer Science (sfcs 1985)}}}\ (\bibinfo {organization} {IEEE},\ \bibinfo {year} {1985})\ pp.\ \bibinfo {pages} {30--38}\BibitemShut {NoStop}%
\bibitem [{\citenamefont {Pippenger}(1988)}]{pippenger1988reliable}%
  \BibitemOpen
  \bibfield  {author} {\bibinfo {author} {\bibfnamefont {N.}~\bibnamefont {Pippenger}},\ }\bibfield  {title} {\bibinfo {title} {Reliable computation by formulas in the presence of noise},\ }\href@noop {} {\bibfield  {journal} {\bibinfo  {journal} {IEEE Transactions on Information Theory}\ }\textbf {\bibinfo {volume} {34}},\ \bibinfo {pages} {194} (\bibinfo {year} {1988})}\BibitemShut {NoStop}%
\bibitem [{\citenamefont {Benjamini}\ \emph {et~al.}(1998)\citenamefont {Benjamini}, \citenamefont {Pemantle},\ and\ \citenamefont {Peres}}]{percolation}%
  \BibitemOpen
  \bibfield  {author} {\bibinfo {author} {\bibfnamefont {I.}~\bibnamefont {Benjamini}}, \bibinfo {author} {\bibfnamefont {R.}~\bibnamefont {Pemantle}},\ and\ \bibinfo {author} {\bibfnamefont {Y.}~\bibnamefont {Peres}},\ }\bibfield  {title} {\bibinfo {title} {{Unpredictable paths and percolation}},\ }\href {https://doi.org/10.1214/aop/1022855749} {\bibfield  {journal} {\bibinfo  {journal} {The Annals of Probability}\ }\textbf {\bibinfo {volume} {26}},\ \bibinfo {pages} {1198 } (\bibinfo {year} {1998})}\BibitemShut {NoStop}%
\bibitem [{\citenamefont {Kesten}\ and\ \citenamefont {Stigum}(1966)}]{kesten_stigum}%
  \BibitemOpen
  \bibfield  {author} {\bibinfo {author} {\bibfnamefont {H.}~\bibnamefont {Kesten}}\ and\ \bibinfo {author} {\bibfnamefont {B.~P.}\ \bibnamefont {Stigum}},\ }\bibfield  {title} {\bibinfo {title} {A limit theorem for multidimensional galton-watson processes},\ }\href {http://www.jstor.org/stable/2239076} {\bibfield  {journal} {\bibinfo  {journal} {The Annals of Mathematical Statistics}\ }\textbf {\bibinfo {volume} {37}},\ \bibinfo {pages} {1211} (\bibinfo {year} {1966})}\BibitemShut {NoStop}%
\bibitem [{\citenamefont {Evans}\ \emph {et~al.}(2000)\citenamefont {Evans}, \citenamefont {Kenyon}, \citenamefont {Peres},\ and\ \citenamefont {Schulman}}]{evans2000}%
  \BibitemOpen
  \bibfield  {author} {\bibinfo {author} {\bibfnamefont {W.}~\bibnamefont {Evans}}, \bibinfo {author} {\bibfnamefont {C.}~\bibnamefont {Kenyon}}, \bibinfo {author} {\bibfnamefont {Y.}~\bibnamefont {Peres}},\ and\ \bibinfo {author} {\bibfnamefont {L.~J.}\ \bibnamefont {Schulman}},\ }\bibfield  {title} {\bibinfo {title} {Broadcasting on trees and the ising model},\ }\href {https://doi.org/10.1214/aoap/1019487349} {\bibfield  {journal} {\bibinfo  {journal} {Ann. Appl. Probab.}\ }\textbf {\bibinfo {volume} {10}},\ \bibinfo {pages} {410} (\bibinfo {year} {2000})}\BibitemShut {NoStop}%
\bibitem [{\citenamefont {Lyons}\ and\ \citenamefont {Peres}(2016)}]{LyonsPerestextbook}%
  \BibitemOpen
  \bibfield  {author} {\bibinfo {author} {\bibfnamefont {R.}~\bibnamefont {Lyons}}\ and\ \bibinfo {author} {\bibfnamefont {Y.}~\bibnamefont {Peres}},\ }\href {https://doi.org/10.1017/9781316672815} {\emph {\bibinfo {title} {Probability on Trees and Networks}}},\ \bibinfo {series} {Cambridge Series in Statistical and Probabilistic Mathematics}, Vol.~\bibinfo {volume} {42}\ (\bibinfo  {publisher} {Cambridge University Press, New York},\ \bibinfo {year} {2016})\ pp.\ \bibinfo {pages} {xv+699},\ \bibinfo {note} {available at \url{https://rdlyons.pages.iu.edu/}}\BibitemShut {NoStop}%
\bibitem [{\citenamefont {Mossel}(2001)}]{Mossel_second_eigenvalue}%
  \BibitemOpen
  \bibfield  {author} {\bibinfo {author} {\bibfnamefont {E.}~\bibnamefont {Mossel}},\ }\bibfield  {title} {\bibinfo {title} {Reconstruction on trees: Beating the second eigenvalue},\ }\href {http://www.jstor.org/stable/2667270} {\bibfield  {journal} {\bibinfo  {journal} {The Annals of Applied Probability}\ }\textbf {\bibinfo {volume} {11}},\ \bibinfo {pages} {285} (\bibinfo {year} {2001})}\BibitemShut {NoStop}%
\bibitem [{\citenamefont {Mossel}\ and\ \citenamefont {Peres}(2003)}]{Mossel_Peres}%
  \BibitemOpen
  \bibfield  {author} {\bibinfo {author} {\bibfnamefont {E.}~\bibnamefont {Mossel}}\ and\ \bibinfo {author} {\bibfnamefont {Y.}~\bibnamefont {Peres}},\ }\bibfield  {title} {\bibinfo {title} {Information flow on trees},\ }\href {http://www.jstor.org/stable/1193228} {\bibfield  {journal} {\bibinfo  {journal} {The Annals of Applied Probability}\ }\textbf {\bibinfo {volume} {13}},\ \bibinfo {pages} {817} (\bibinfo {year} {2003})}\BibitemShut {NoStop}%
\bibitem [{\citenamefont {Wootters}\ and\ \citenamefont {Zurek}(1982)}]{wootters1982single}%
  \BibitemOpen
  \bibfield  {author} {\bibinfo {author} {\bibfnamefont {W.~K.}\ \bibnamefont {Wootters}}\ and\ \bibinfo {author} {\bibfnamefont {W.~H.}\ \bibnamefont {Zurek}},\ }\bibfield  {title} {\bibinfo {title} {A single quantum cannot be cloned},\ }\href@noop {} {\bibfield  {journal} {\bibinfo  {journal} {Nature}\ }\textbf {\bibinfo {volume} {299}},\ \bibinfo {pages} {802} (\bibinfo {year} {1982})}\BibitemShut {NoStop}%
\bibitem [{\citenamefont {Barnum}\ \emph {et~al.}(1996)\citenamefont {Barnum}, \citenamefont {Caves}, \citenamefont {Fuchs}, \citenamefont {Jozsa},\ and\ \citenamefont {Schumacher}}]{nobroadcasting}%
  \BibitemOpen
  \bibfield  {author} {\bibinfo {author} {\bibfnamefont {H.}~\bibnamefont {Barnum}}, \bibinfo {author} {\bibfnamefont {C.~M.}\ \bibnamefont {Caves}}, \bibinfo {author} {\bibfnamefont {C.~A.}\ \bibnamefont {Fuchs}}, \bibinfo {author} {\bibfnamefont {R.}~\bibnamefont {Jozsa}},\ and\ \bibinfo {author} {\bibfnamefont {B.}~\bibnamefont {Schumacher}},\ }\bibfield  {title} {\bibinfo {title} {Noncommuting mixed states cannot be broadcast},\ }\href {https://doi.org/10.1103/PhysRevLett.76.2818} {\bibfield  {journal} {\bibinfo  {journal} {Phys. Rev. Lett.}\ }\textbf {\bibinfo {volume} {76}},\ \bibinfo {pages} {2818} (\bibinfo {year} {1996})}\BibitemShut {NoStop}%
\bibitem [{\citenamefont {Shor}(1995)}]{Shor95}%
  \BibitemOpen
  \bibfield  {author} {\bibinfo {author} {\bibfnamefont {P.~W.}\ \bibnamefont {Shor}},\ }\bibfield  {title} {\bibinfo {title} {Scheme for reducing decoherence in quantum computer memory},\ }\href {https://doi.org/10.1103/PhysRevA.52.R2493} {\bibfield  {journal} {\bibinfo  {journal} {Phys. Rev. A}\ }\textbf {\bibinfo {volume} {52}},\ \bibinfo {pages} {R2493} (\bibinfo {year} {1995})}\BibitemShut {NoStop}%
\bibitem [{\citenamefont {Calderbank}\ and\ \citenamefont {Shor}(1996)}]{CladerbankShor96}%
  \BibitemOpen
  \bibfield  {author} {\bibinfo {author} {\bibfnamefont {A.~R.}\ \bibnamefont {Calderbank}}\ and\ \bibinfo {author} {\bibfnamefont {P.~W.}\ \bibnamefont {Shor}},\ }\bibfield  {title} {\bibinfo {title} {Good quantum error-correcting codes exist},\ }\href {https://doi.org/10.1103/PhysRevA.54.1098} {\bibfield  {journal} {\bibinfo  {journal} {Phys. Rev. A}\ }\textbf {\bibinfo {volume} {54}},\ \bibinfo {pages} {1098} (\bibinfo {year} {1996})}\BibitemShut {NoStop}%
\bibitem [{\citenamefont {Steane}(1996)}]{Steane96}%
  \BibitemOpen
  \bibfield  {author} {\bibinfo {author} {\bibfnamefont {A.~M.}\ \bibnamefont {Steane}},\ }\bibfield  {title} {\bibinfo {title} {Error correcting codes in quantum theory},\ }\href {https://doi.org/10.1103/PhysRevLett.77.793} {\bibfield  {journal} {\bibinfo  {journal} {Phys. Rev. Lett.}\ }\textbf {\bibinfo {volume} {77}},\ \bibinfo {pages} {793} (\bibinfo {year} {1996})}\BibitemShut {NoStop}%
\bibitem [{\citenamefont {Gottesman}(1997)}]{gottesman1997stabilizer}%
  \BibitemOpen
  \bibfield  {author} {\bibinfo {author} {\bibfnamefont {D.}~\bibnamefont {Gottesman}},\ }{\selectlanguage {English}\emph {\bibinfo {title} {Stabilizer codes and quantum error correction}}},\ \href {https://login.proxy.lib.duke.edu/login?url=https://www.proquest.com/dissertations-theses/stabilizer-codes-quantum-error-correction/docview/304364982/se-2} {Ph.D. thesis} (\bibinfo {year} {1997}),\ \bibinfo {note} {copyright - Database copyright ProQuest LLC; ProQuest does not claim copyright in the individual underlying works; Last updated - 2021-09-28}\BibitemShut {NoStop}%
\bibitem [{\citenamefont {Aharonov}\ and\ \citenamefont {Ben-Or}(2008)}]{aharonov2008fault}%
  \BibitemOpen
  \bibfield  {author} {\bibinfo {author} {\bibfnamefont {D.}~\bibnamefont {Aharonov}}\ and\ \bibinfo {author} {\bibfnamefont {M.}~\bibnamefont {Ben-Or}},\ }\bibfield  {title} {\bibinfo {title} {Fault-tolerant quantum computation with constant error rate},\ }\href@noop {} {\bibfield  {journal} {\bibinfo  {journal} {SIAM Journal on Computing}\ } (\bibinfo {year} {2008})}\BibitemShut {NoStop}%
\bibitem [{\citenamefont {Knill}\ \emph {et~al.}(1996)\citenamefont {Knill}, \citenamefont {Laflamme},\ and\ \citenamefont {Zurek}}]{knill1996threshold}%
  \BibitemOpen
  \bibfield  {author} {\bibinfo {author} {\bibfnamefont {E.}~\bibnamefont {Knill}}, \bibinfo {author} {\bibfnamefont {R.}~\bibnamefont {Laflamme}},\ and\ \bibinfo {author} {\bibfnamefont {W.}~\bibnamefont {Zurek}},\ }\bibfield  {title} {\bibinfo {title} {Threshold accuracy for quantum computation},\ }\href@noop {} {\bibfield  {journal} {\bibinfo  {journal} {arXiv preprint quant-ph/9610011}\ } (\bibinfo {year} {1996})}\BibitemShut {NoStop}%
\bibitem [{\citenamefont {Gottesman}(2014)}]{gottesman2014faulttolerant}%
  \BibitemOpen
  \bibfield  {author} {\bibinfo {author} {\bibfnamefont {D.}~\bibnamefont {Gottesman}},\ }\href@noop {} {\bibinfo {title} {Fault-tolerant quantum computation with constant overhead}} (\bibinfo {year} {2014}),\ \Eprint {https://arxiv.org/abs/1310.2984} {arXiv:1310.2984 [quant-ph]} \BibitemShut {NoStop}%
\bibitem [{\citenamefont {Chamberland}\ \emph {et~al.}(2016)\citenamefont {Chamberland}, \citenamefont {Jochym-O'Connor},\ and\ \citenamefont {Laflamme}}]{chamberland2016threshold}%
  \BibitemOpen
  \bibfield  {author} {\bibinfo {author} {\bibfnamefont {C.}~\bibnamefont {Chamberland}}, \bibinfo {author} {\bibfnamefont {T.}~\bibnamefont {Jochym-O'Connor}},\ and\ \bibinfo {author} {\bibfnamefont {R.}~\bibnamefont {Laflamme}},\ }\bibfield  {title} {\bibinfo {title} {Thresholds for universal concatenated quantum codes},\ }\href {https://doi.org/10.1103/PhysRevLett.117.010501} {\bibfield  {journal} {\bibinfo  {journal} {Phys. Rev. Lett.}\ }\textbf {\bibinfo {volume} {117}},\ \bibinfo {pages} {010501} (\bibinfo {year} {2016})}\BibitemShut {NoStop}%
\bibitem [{\citenamefont {Aharonov}(2000)}]{aharonovphasetransition}%
  \BibitemOpen
  \bibfield  {author} {\bibinfo {author} {\bibfnamefont {D.}~\bibnamefont {Aharonov}},\ }\bibfield  {title} {\bibinfo {title} {Quantum to classical phase transition in noisy quantum computers},\ }\href {https://doi.org/10.1103/PhysRevA.62.062311} {\bibfield  {journal} {\bibinfo  {journal} {Phys. Rev. A}\ }\textbf {\bibinfo {volume} {62}},\ \bibinfo {pages} {062311} (\bibinfo {year} {2000})}\BibitemShut {NoStop}%
\bibitem [{\citenamefont {Harrow}\ and\ \citenamefont {Nielsen}(2003)}]{HarrowRobust2003}%
  \BibitemOpen
  \bibfield  {author} {\bibinfo {author} {\bibfnamefont {A.~W.}\ \bibnamefont {Harrow}}\ and\ \bibinfo {author} {\bibfnamefont {M.~A.}\ \bibnamefont {Nielsen}},\ }\bibfield  {title} {\bibinfo {title} {Robustness of quantum gates in the presence of noise},\ }\href {https://doi.org/10.1103/PhysRevA.68.012308} {\bibfield  {journal} {\bibinfo  {journal} {Phys. Rev. A}\ }\textbf {\bibinfo {volume} {68}},\ \bibinfo {pages} {012308} (\bibinfo {year} {2003})}\BibitemShut {NoStop}%
\bibitem [{\citenamefont {Razborov}(2004)}]{Razborov}%
  \BibitemOpen
  \bibfield  {author} {\bibinfo {author} {\bibfnamefont {A.~A.}\ \bibnamefont {Razborov}},\ }\bibfield  {title} {\bibinfo {title} {An upper bound on the threshold quantum decoherence rate},\ }\href@noop {} {\bibfield  {journal} {\bibinfo  {journal} {Quantum Info. Comput.}\ }\textbf {\bibinfo {volume} {4}},\ \bibinfo {pages} {222–228} (\bibinfo {year} {2004})}\BibitemShut {NoStop}%
\bibitem [{\citenamefont {Kempe}\ \emph {et~al.}(2010)\citenamefont {Kempe}, \citenamefont {Regev}, \citenamefont {Uunger},\ and\ \citenamefont {de~Wolf}}]{kempe}%
  \BibitemOpen
  \bibfield  {author} {\bibinfo {author} {\bibfnamefont {J.}~\bibnamefont {Kempe}}, \bibinfo {author} {\bibfnamefont {O.}~\bibnamefont {Regev}}, \bibinfo {author} {\bibfnamefont {F.}~\bibnamefont {Uunger}},\ and\ \bibinfo {author} {\bibfnamefont {R.}~\bibnamefont {de~Wolf}},\ }\bibfield  {title} {\bibinfo {title} {Upper bounds on the noise threshold for fault-tolerant quantum computing},\ }\href@noop {} {\bibfield  {journal} {\bibinfo  {journal} {Quantum Info. Comput.}\ }\textbf {\bibinfo {volume} {10}},\ \bibinfo {pages} {361–376} (\bibinfo {year} {2010})}\BibitemShut {NoStop}%
\bibitem [{\citenamefont {Dalzell}\ \emph {et~al.}(2021)\citenamefont {Dalzell}, \citenamefont {Hunter-Jones},\ and\ \citenamefont {Brandão}}]{dalzell2021random}%
  \BibitemOpen
  \bibfield  {author} {\bibinfo {author} {\bibfnamefont {A.~M.}\ \bibnamefont {Dalzell}}, \bibinfo {author} {\bibfnamefont {N.}~\bibnamefont {Hunter-Jones}},\ and\ \bibinfo {author} {\bibfnamefont {F.~G. S.~L.}\ \bibnamefont {Brandão}},\ }\href@noop {} {\bibinfo {title} {Random quantum circuits transform local noise into global white noise}} (\bibinfo {year} {2021}),\ \Eprint {https://arxiv.org/abs/2111.14907} {arXiv:2111.14907 [quant-ph]} \BibitemShut {NoStop}%
\bibitem [{\citenamefont {Deshpande}\ \emph {et~al.}(2022)\citenamefont {Deshpande}, \citenamefont {Niroula}, \citenamefont {Shtanko}, \citenamefont {Gorshkov}, \citenamefont {Fefferman},\ and\ \citenamefont {Gullans}}]{deshpande2022tight}%
  \BibitemOpen
  \bibfield  {author} {\bibinfo {author} {\bibfnamefont {A.}~\bibnamefont {Deshpande}}, \bibinfo {author} {\bibfnamefont {P.}~\bibnamefont {Niroula}}, \bibinfo {author} {\bibfnamefont {O.}~\bibnamefont {Shtanko}}, \bibinfo {author} {\bibfnamefont {A.~V.}\ \bibnamefont {Gorshkov}}, \bibinfo {author} {\bibfnamefont {B.}~\bibnamefont {Fefferman}},\ and\ \bibinfo {author} {\bibfnamefont {M.~J.}\ \bibnamefont {Gullans}},\ }\href@noop {} {\bibinfo {title} {Tight bounds on the convergence of noisy random circuits to the uniform distribution}} (\bibinfo {year} {2022}),\ \Eprint {https://arxiv.org/abs/2112.00716} {arXiv:2112.00716 [quant-ph]} \BibitemShut {NoStop}%
\bibitem [{\citenamefont {Dalzell}\ \emph {et~al.}(2022)\citenamefont {Dalzell}, \citenamefont {Hunter-Jones},\ and\ \citenamefont {Brand\~ao}}]{PRXQuantum.3.010333}%
  \BibitemOpen
  \bibfield  {author} {\bibinfo {author} {\bibfnamefont {A.~M.}\ \bibnamefont {Dalzell}}, \bibinfo {author} {\bibfnamefont {N.}~\bibnamefont {Hunter-Jones}},\ and\ \bibinfo {author} {\bibfnamefont {F.~G. S.~L.}\ \bibnamefont {Brand\~ao}},\ }\bibfield  {title} {\bibinfo {title} {Random quantum circuits anticoncentrate in log depth},\ }\href {https://doi.org/10.1103/PRXQuantum.3.010333} {\bibfield  {journal} {\bibinfo  {journal} {PRX Quantum}\ }\textbf {\bibinfo {volume} {3}},\ \bibinfo {pages} {010333} (\bibinfo {year} {2022})}\BibitemShut {NoStop}%
\bibitem [{\citenamefont {Rahn}\ \emph {et~al.}(2002)\citenamefont {Rahn}, \citenamefont {Doherty},\ and\ \citenamefont {Mabuchi}}]{dohertyconcatenated}%
  \BibitemOpen
  \bibfield  {author} {\bibinfo {author} {\bibfnamefont {B.}~\bibnamefont {Rahn}}, \bibinfo {author} {\bibfnamefont {A.~C.}\ \bibnamefont {Doherty}},\ and\ \bibinfo {author} {\bibfnamefont {H.}~\bibnamefont {Mabuchi}},\ }\bibfield  {title} {\bibinfo {title} {Exact performance of concatenated quantum codes},\ }\href {https://doi.org/10.1103/PhysRevA.66.032304} {\bibfield  {journal} {\bibinfo  {journal} {Phys. Rev. A}\ }\textbf {\bibinfo {volume} {66}},\ \bibinfo {pages} {032304} (\bibinfo {year} {2002})}\BibitemShut {NoStop}%
\bibitem [{\citenamefont {Poulin}(2006)}]{poulin2006optimal}%
  \BibitemOpen
  \bibfield  {author} {\bibinfo {author} {\bibfnamefont {D.}~\bibnamefont {Poulin}},\ }\bibfield  {title} {\bibinfo {title} {Optimal and efficient decoding of concatenated quantum block codes},\ }\href@noop {} {\bibfield  {journal} {\bibinfo  {journal} {Physical Review A}\ }\textbf {\bibinfo {volume} {74}},\ \bibinfo {pages} {052333} (\bibinfo {year} {2006})}\BibitemShut {NoStop}%
\bibitem [{\citenamefont {Shi}\ \emph {et~al.}(2006)\citenamefont {Shi}, \citenamefont {Duan},\ and\ \citenamefont {Vidal}}]{TTN}%
  \BibitemOpen
  \bibfield  {author} {\bibinfo {author} {\bibfnamefont {Y.-Y.}\ \bibnamefont {Shi}}, \bibinfo {author} {\bibfnamefont {L.-M.}\ \bibnamefont {Duan}},\ and\ \bibinfo {author} {\bibfnamefont {G.}~\bibnamefont {Vidal}},\ }\bibfield  {title} {\bibinfo {title} {Classical simulation of quantum many-body systems with a tree tensor network},\ }\href {https://doi.org/10.1103/PhysRevA.74.022320} {\bibfield  {journal} {\bibinfo  {journal} {Phys. Rev. A}\ }\textbf {\bibinfo {volume} {74}},\ \bibinfo {pages} {022320} (\bibinfo {year} {2006})}\BibitemShut {NoStop}%
\bibitem [{\citenamefont {Barthel}\ \emph {et~al.}(2022)\citenamefont {Barthel}, \citenamefont {Lu},\ and\ \citenamefont {Friesecke}}]{barthel2022closedness}%
  \BibitemOpen
  \bibfield  {author} {\bibinfo {author} {\bibfnamefont {T.}~\bibnamefont {Barthel}}, \bibinfo {author} {\bibfnamefont {J.}~\bibnamefont {Lu}},\ and\ \bibinfo {author} {\bibfnamefont {G.}~\bibnamefont {Friesecke}},\ }\bibfield  {title} {\bibinfo {title} {On the closedness and geometry of tensor network state sets},\ }\href@noop {} {\bibfield  {journal} {\bibinfo  {journal} {Letters in Mathematical Physics}\ }\textbf {\bibinfo {volume} {112}},\ \bibinfo {pages} {72} (\bibinfo {year} {2022})}\BibitemShut {NoStop}%
\bibitem [{\citenamefont {Giovannetti}\ \emph {et~al.}(2008{\natexlab{a}})\citenamefont {Giovannetti}, \citenamefont {Lloyd},\ and\ \citenamefont {Maccone}}]{QRAM1}%
  \BibitemOpen
  \bibfield  {author} {\bibinfo {author} {\bibfnamefont {V.}~\bibnamefont {Giovannetti}}, \bibinfo {author} {\bibfnamefont {S.}~\bibnamefont {Lloyd}},\ and\ \bibinfo {author} {\bibfnamefont {L.}~\bibnamefont {Maccone}},\ }\bibfield  {title} {\bibinfo {title} {Quantum random access memory},\ }\href {https://doi.org/10.1103/PhysRevLett.100.160501} {\bibfield  {journal} {\bibinfo  {journal} {Phys. Rev. Lett.}\ }\textbf {\bibinfo {volume} {100}},\ \bibinfo {pages} {160501} (\bibinfo {year} {2008}{\natexlab{a}})}\BibitemShut {NoStop}%
\bibitem [{\citenamefont {Giovannetti}\ \emph {et~al.}(2008{\natexlab{b}})\citenamefont {Giovannetti}, \citenamefont {Lloyd},\ and\ \citenamefont {Maccone}}]{QRAM2}%
  \BibitemOpen
  \bibfield  {author} {\bibinfo {author} {\bibfnamefont {V.}~\bibnamefont {Giovannetti}}, \bibinfo {author} {\bibfnamefont {S.}~\bibnamefont {Lloyd}},\ and\ \bibinfo {author} {\bibfnamefont {L.}~\bibnamefont {Maccone}},\ }\bibfield  {title} {\bibinfo {title} {Architectures for a quantum random access memory},\ }\href {https://doi.org/10.1103/PhysRevA.78.052310} {\bibfield  {journal} {\bibinfo  {journal} {Phys. Rev. A}\ }\textbf {\bibinfo {volume} {78}},\ \bibinfo {pages} {052310} (\bibinfo {year} {2008}{\natexlab{b}})}\BibitemShut {NoStop}%
\bibitem [{\citenamefont {Roberts}\ and\ \citenamefont {Stanford}(2013)}]{roberts2013memory}%
  \BibitemOpen
  \bibfield  {author} {\bibinfo {author} {\bibfnamefont {D.~A.}\ \bibnamefont {Roberts}}\ and\ \bibinfo {author} {\bibfnamefont {D.}~\bibnamefont {Stanford}},\ }\bibfield  {title} {\bibinfo {title} {On memory in exponentially expanding spaces},\ }\href@noop {} {\bibfield  {journal} {\bibinfo  {journal} {Journal of High Energy Physics}\ }\textbf {\bibinfo {volume} {2013}},\ \bibinfo {pages} {1} (\bibinfo {year} {2013})}\BibitemShut {NoStop}%
\bibitem [{\citenamefont {Almheiri}\ \emph {et~al.}(2015)\citenamefont {Almheiri}, \citenamefont {Dong},\ and\ \citenamefont {Harlow}}]{almheiri2015bulk}%
  \BibitemOpen
  \bibfield  {author} {\bibinfo {author} {\bibfnamefont {A.}~\bibnamefont {Almheiri}}, \bibinfo {author} {\bibfnamefont {X.}~\bibnamefont {Dong}},\ and\ \bibinfo {author} {\bibfnamefont {D.}~\bibnamefont {Harlow}},\ }\bibfield  {title} {\bibinfo {title} {Bulk locality and quantum error correction in ads/cft},\ }\href@noop {} {\bibfield  {journal} {\bibinfo  {journal} {Journal of High Energy Physics}\ }\textbf {\bibinfo {volume} {2015}},\ \bibinfo {pages} {1} (\bibinfo {year} {2015})}\BibitemShut {NoStop}%
\bibitem [{\citenamefont {Pastawski}\ \emph {et~al.}(2015)\citenamefont {Pastawski}, \citenamefont {Yoshida}, \citenamefont {Harlow},\ and\ \citenamefont {Preskill}}]{pastawski2015holographic}%
  \BibitemOpen
  \bibfield  {author} {\bibinfo {author} {\bibfnamefont {F.}~\bibnamefont {Pastawski}}, \bibinfo {author} {\bibfnamefont {B.}~\bibnamefont {Yoshida}}, \bibinfo {author} {\bibfnamefont {D.}~\bibnamefont {Harlow}},\ and\ \bibinfo {author} {\bibfnamefont {J.}~\bibnamefont {Preskill}},\ }\bibfield  {title} {\bibinfo {title} {Holographic quantum error-correcting codes: Toy models for the bulk/boundary correspondence},\ }\href@noop {} {\bibfield  {journal} {\bibinfo  {journal} {Journal of High Energy Physics}\ }\textbf {\bibinfo {volume} {2015}},\ \bibinfo {pages} {1} (\bibinfo {year} {2015})}\BibitemShut {NoStop}%
\bibitem [{\citenamefont {Nielsen}\ and\ \citenamefont {Chuang}(2002)}]{neilsenandchuang}%
  \BibitemOpen
  \bibfield  {author} {\bibinfo {author} {\bibfnamefont {M.~A.}\ \bibnamefont {Nielsen}}\ and\ \bibinfo {author} {\bibfnamefont {I.}~\bibnamefont {Chuang}},\ }\bibfield  {title} {\bibinfo {title} {Quantum computation and quantum information},\ }\href {https://doi.org/10.1119/1.1463744} {\bibfield  {journal} {\bibinfo  {journal} {American Journal of Physics}\ }\textbf {\bibinfo {volume} {70}},\ \bibinfo {pages} {558} (\bibinfo {year} {2002})},\ \Eprint {https://arxiv.org/abs/https://doi.org/10.1119/1.1463744} {https://doi.org/10.1119/1.1463744} \BibitemShut {NoStop}%
\bibitem [{\citenamefont {Calderbank}\ \emph {et~al.}(1997)\citenamefont {Calderbank}, \citenamefont {Rains}, \citenamefont {Shor},\ and\ \citenamefont {Sloane}}]{GF4}%
  \BibitemOpen
  \bibfield  {author} {\bibinfo {author} {\bibfnamefont {A.~R.}\ \bibnamefont {Calderbank}}, \bibinfo {author} {\bibfnamefont {E.~M.}\ \bibnamefont {Rains}}, \bibinfo {author} {\bibfnamefont {P.~W.}\ \bibnamefont {Shor}},\ and\ \bibinfo {author} {\bibfnamefont {N.~J.~A.}\ \bibnamefont {Sloane}},\ }\bibfield  {title} {\bibinfo {title} {Quantum error correction and orthogonal geometry},\ }\href {https://doi.org/10.1103/PhysRevLett.78.405} {\bibfield  {journal} {\bibinfo  {journal} {Phys. Rev. Lett.}\ }\textbf {\bibinfo {volume} {78}},\ \bibinfo {pages} {405} (\bibinfo {year} {1997})}\BibitemShut {NoStop}%
\bibitem [{\citenamefont {{Helstrom}}(1969)}]{Helstrom}%
  \BibitemOpen
  \bibfield  {author} {\bibinfo {author} {\bibfnamefont {C.~W.}\ \bibnamefont {{Helstrom}}},\ }\bibfield  {title} {\bibinfo {title} {{Quantum detection and estimation theory}},\ }\href {https://doi.org/10.1007/BF01007479} {\bibfield  {journal} {\bibinfo  {journal} {Journal of Statistical Physics}\ }\textbf {\bibinfo {volume} {1}},\ \bibinfo {pages} {231} (\bibinfo {year} {1969})}\BibitemShut {NoStop}%
\bibitem [{\citenamefont {{Evans}}\ and\ \citenamefont {{Schulman}}(1999)}]{Evans-Schulman}%
  \BibitemOpen
  \bibfield  {author} {\bibinfo {author} {\bibfnamefont {W.~S.}\ \bibnamefont {{Evans}}}\ and\ \bibinfo {author} {\bibfnamefont {L.~J.}\ \bibnamefont {{Schulman}}},\ }\bibfield  {title} {\bibinfo {title} {Signal propagation and noisy circuits},\ }\href@noop {} {\bibfield  {journal} {\bibinfo  {journal} {IEEE Transactions on Information Theory}\ }\textbf {\bibinfo {volume} {45}},\ \bibinfo {pages} {2367} (\bibinfo {year} {1999})}\BibitemShut {NoStop}%
\bibitem [{\citenamefont {Wilde}(2013)}]{wilde2013quantum}%
  \BibitemOpen
  \bibfield  {author} {\bibinfo {author} {\bibfnamefont {M.~M.}\ \bibnamefont {Wilde}},\ }\href@noop {} {\emph {\bibinfo {title} {Quantum information theory}}}\ (\bibinfo  {publisher} {Cambridge University Press},\ \bibinfo {year} {2013})\BibitemShut {NoStop}%
\bibitem [{\citenamefont {Horodecki}\ \emph {et~al.}(2003)\citenamefont {Horodecki}, \citenamefont {Shor},\ and\ \citenamefont {Ruskai}}]{horodecki2003entanglement}%
  \BibitemOpen
  \bibfield  {author} {\bibinfo {author} {\bibfnamefont {M.}~\bibnamefont {Horodecki}}, \bibinfo {author} {\bibfnamefont {P.~W.}\ \bibnamefont {Shor}},\ and\ \bibinfo {author} {\bibfnamefont {M.~B.}\ \bibnamefont {Ruskai}},\ }\bibfield  {title} {\bibinfo {title} {Entanglement breaking channels},\ }\href@noop {} {\bibfield  {journal} {\bibinfo  {journal} {Reviews in Mathematical Physics}\ }\textbf {\bibinfo {volume} {15}},\ \bibinfo {pages} {629} (\bibinfo {year} {2003})}\BibitemShut {NoStop}%
\bibitem [{\citenamefont {Peres}(1996)}]{Horo-Peres1}%
  \BibitemOpen
  \bibfield  {author} {\bibinfo {author} {\bibfnamefont {A.}~\bibnamefont {Peres}},\ }\bibfield  {title} {\bibinfo {title} {Separability criterion for density matrices},\ }\href {https://doi.org/10.1103/PhysRevLett.77.1413} {\bibfield  {journal} {\bibinfo  {journal} {Phys. Rev. Lett.}\ }\textbf {\bibinfo {volume} {77}},\ \bibinfo {pages} {1413} (\bibinfo {year} {1996})}\BibitemShut {NoStop}%
\bibitem [{\citenamefont {Bryan}\ and\ \citenamefont {Wadsworth}(1960)}]{gp1960bryan}%
  \BibitemOpen
  \bibfield  {author} {\bibinfo {author} {\bibfnamefont {J.~C.}\ \bibnamefont {Bryan}}\ and\ \bibinfo {author} {\bibfnamefont {G.~P.}\ \bibnamefont {Wadsworth}},\ }\href@noop {} {\emph {\bibinfo {title} {Introduction to Probability and Random Variables}}}\ (\bibinfo  {publisher} {New York. McGraw Hill Book Co.},\ \bibinfo {year} {1960})\ p.~\bibinfo {pages} {52}\BibitemShut {NoStop}%
\bibitem [{\citenamefont {Mezard}\ and\ \citenamefont {Montanari}(2009)}]{mezard2009information}%
  \BibitemOpen
  \bibfield  {author} {\bibinfo {author} {\bibfnamefont {M.}~\bibnamefont {Mezard}}\ and\ \bibinfo {author} {\bibfnamefont {A.}~\bibnamefont {Montanari}},\ }\href@noop {} {\emph {\bibinfo {title} {Information, physics, and computation}}}\ (\bibinfo  {publisher} {Oxford University Press},\ \bibinfo {year} {2009})\BibitemShut {NoStop}%
\end{thebibliography}%

\onecolumngrid

\newpage

\maketitle
\vspace{-5in}
\begin{center}

\Large{Appendix}
\end{center}
\appendix

\section{Quick review of the relevant results on classical trees}\label{classical tree discussion}

In this appendix, we briefly review a result presented in \cite{evans2000} and \cite{Mossel_Peres} that upperbounds the information propagation down noisy trees.

Consider a tree network where a certain node/vertex has been chosen to be the root. This choice allows us to view the tree as a directed graph where bits stream from the root to the leaves. 
 Every vertex $\mathbf{v}$ is a bit-copier process where $0 \mapsto \mathbf{0}^b$ and $1 \mapsto \mathbf{1}^b$, where $b$ is the number of children at the vertex $\mathbf{v}$, and $\mathbf{0}^b$ is the string $00..0$ repeated $b$ times (similarly for $\mathbf{1}^b$). This defines a \textit{$b$-ary} tree. Additionally, each edge $\mathbf{e}$ is associated to a single bit-flip with probability $p_\mathbf{e}$. We also define $\theta_\mathbf{e}:= 1-2 p_\mathbf{e}$.

Let us assume the bit at the root of the tree is $\textbf{x}_0$.   At the end of this tree process, the probability of obtaining bitstring $\textbf{x}_T$ at level $T$ is ${p}(\textbf{x}_T|\textbf{x}_0)$ where we conditioned by the root value $\textbf{x}_0\in\{0,1\}$. The information propagated via the tree is quantified by the total variation distance,
\begin{align}
    \Delta \equiv \frac12||{p}(\textbf{x}_T|\textbf{x}_0=0) - {p}(\textbf{x}_T|\textbf{x}_0=1)||_1 = \frac12 \sum_{\textbf{x}_T}|{p}(\textbf{x}_T|\textbf{x}_0=0) - {p}(\textbf{x}_T|\textbf{x}_0=1)|.
\end{align}

To upperbound this quantity, we first define,
\begin{align}
    \Theta_{\textbf{e}}:=\prod_{\textbf{e}' \in [\textbf{e}]} \theta_{\textbf{e}'}
\end{align}
where $[\textbf{e}]$ is the set of all edges in the path between the root and edge $\textbf{e}$, including the ends. Because of the tree structure, we are ensured a unique path between the root and any edge $\textbf{e}$.  
Theorem 1.3' in \cite{evans2000} states that,
\begin{align}
    \Delta^2 \leq 2 \sum_{\textbf{e} \in \textbf{W}_T} \Theta^2_\textbf{e}
\end{align}
where the summation is over the set $\textbf{W}_T$ that contains all edges at level $T$ of the tree. 
\footnote{This statement is proved in \cite{evans2000} using a \textit{domination} argument where they construct a `stringy' version of the original tree which can be post-processed by appropriate stochastic maps to simulate the original tree process (i.e., the original tree is \textit{dominated} by the `stringy' tree). Thus, benchmarking this `stringy' tree gives an upperbound on the original tree's performance due to the data processing inequality. And more importantly, this `stringy' tree is constructed to be one for whom the total variation distance is easily computed, thus yielding the above upperbound.}

Suppose the error probability on each edge is $p$ and there are $n_T$ edges in level $T$. Then this formula is simplified to
\begin{align}
    \Delta \leq \sqrt{2} \times \sqrt{n_T} \times (1-2p)^T \ .
\end{align}

Furthermore, consider the important special case of a \textit{full} $b$-ary tree where at level $T$, $n_T=b^T$, and all bit-flip probabilities across edges are the same, i.e., $p$. So,
\begin{align}
\Delta \leq  \sqrt{2}\times (\sqrt{b} (1-2p))^{T}.
\end{align}
This is the result presented in Eq.(13) of \cite{evans2000}. This upperbound provides an asymptotic noise threshold upperbound of $p_{\rm th} \leq \frac12 (1-\frac{1}{\sqrt{b}})$ (This square-root scaling upperbound is sometimes referred to as the Kesten-Stigum bound). Remarkably,   it turns out that this bound holds as equality. In particular, 
Refs. \cite{evans2000, Mossel_Peres} show that if $b \times |1-2p|^2>1$, then the output of the majority voting decoder remains correlated with the input of the tree, even in the limit $T\rightarrow \infty$. 

\newpage

\section{Review of some useful facts about Pauli channels}

Here, we briefly review some useful facts about the Diamond norm distance of Pauli channels and the condition for entanglement-breaking Pauli channels. See, e.g., \cite{wilde2013quantum, horodecki2003entanglement} for further discussion and references.

\subsection{Diamond norm distance between Pauli Channels}\label{appendix: diamond pauli}

Recall that the \textit{diamond norm} of a super-operator ${\Phi}$ that acts on operators in the system $S$ with Hilbert space $\mathcal{H}_S$ is defined as 
\begin{align}
    \|{\Phi}\|_\diamond=\max_{\rho_{SR}} \|{\Phi} \otimes {\rm id}_R(\rho_{SR})\|_1\ ,
\end{align}
where $R$ is a reference system with Hilbert space $\mathcal{H}_R$ such that ${\rm dim}(\mathcal{H}_R)={\rm dim}(\mathcal{H}_S)$, and ${\rm id}_R$ is the identity super-operator on $\mathcal{H}_R$. The maximization is over all states in the composite system $SR$, i.e., $\mathcal{H}_S \otimes \mathcal{H}_R$. We use the diamond norm to quantify the distance between two channels. 
\begin{lemma}\label{Pauli diamond}
Consider two n-qubit Pauli channels, namely
\begin{align}
    \mathcal{P}(\cdot)&=\sum_{\mathbf{x}, \mathbf{z}\in\{0,1\}^n} p(\mathbf{x}, \mathbf{z})\ \mathbf{X}^\mathbf{x} \mathbf{Z}^\mathbf{z} (.) \mathbf{Z}^\mathbf{z} \mathbf{X}^\mathbf{x}\nonumber \\
    \mathcal{Q}(\cdot)&=\sum_{\mathbf{x}, \mathbf{z}\in\{0,1\}^n} q(\mathbf{x}, \mathbf{z})\ \mathbf{X}^\mathbf{x} \mathbf{Z}^\mathbf{z} (.) \mathbf{Z}^\mathbf{z} \mathbf{X}^\mathbf{x}.\nonumber
\end{align}
The diamond distance of the channels is equal to the total variation distance of the corresponding probability distributions
$$d_\diamond(\mathcal{P},\mathcal{Q})=\frac12\|\mathcal{P}-\mathcal{Q}\|_\diamond = \frac12\sum_{\mathbf{x}, \mathbf{z}} |p(\mathbf{x},\mathbf{z}) - q(\mathbf{x},\mathbf{z})|\ .$$
\end{lemma}
\begin{proof}
We  prove this statement by lower and upperbounding by the TVD. The upperbound follows from triangle inequality for diamond norm, namely
\begin{align}
    &||\mathcal{P}-\mathcal{Q}||_\diamond
=\|\sum_{\mathbf{x},\mathbf{z}} [p(\mathbf{x},\mathbf{z}) - q(\mathbf{x},\mathbf{z})] \mathbf{X}^\mathbf{x} \mathbf{Z}^\mathbf{z}(\cdot) \mathbf{X}^\mathbf{x} \mathbf{Z}^\mathbf{z} \|_\diamond\le \sum_{\mathbf{x}, \mathbf{z}} |p(\mathbf{x},\mathbf{z}) - q(\mathbf{x},\mathbf{z})| \times \|\mathbf{X}^\mathbf{x} \mathbf{Z}^\mathbf{z}(\cdot) \mathbf{X}^\mathbf{x} \mathbf{Z}^\mathbf{z} \|_\diamond=\sum_{\mathbf{x}, \mathbf{z}} |p(\mathbf{x},\mathbf{z}) - q(\mathbf{x},\mathbf{z})|\ .
\end{align}

To show that this bound holds as an equality, it suffices to show that there exists a state $\rho_{RS}$ such that $\|(\mathcal{P}-\mathcal{Q})\otimes {\rm id}_R(\rho_{SR})\|_1$ is equal to the right-hand side. 
Inspired by the super-dense coding protocol, which allows perfect discrimination of Pauli operators using maximally entangled states, we consider the maximally entangled state $\ket{\Psi}=\frac{1}{\sqrt{2^n}}\sum_{\mathbf{i}} \ket{\mathbf{i}}\otimes \ket{\mathbf{i}}$, where the summation $\mathbf{i}$ is over all bitstrings of length $n$. Then we have,
\begin{align}
    ||\mathcal{P}-\mathcal{Q}||_\diamond
    =&\max_{\rho_{SR}} ||\sum_{\mathbf{x},\mathbf{z}} [p(\mathbf{x},\mathbf{z}) - q(\mathbf{x},\mathbf{z})]\ \mathbf{X}^\mathbf{x} \mathbf{Z}^\mathbf{z} \otimes I_R (\rho_{SR}) \mathbf{Z}^\mathbf{z} \mathbf{X}^\mathbf{x} \otimes I_R||_1\\
    \geq&  ||\sum_{\mathbf{x},\mathbf{z}} [p(\mathbf{x},\mathbf{z}) - q(\mathbf{x},\mathbf{z})]\ \mathbf{X}^\mathbf{x} \mathbf{Z}^\mathbf{z} \otimes I_R (\ketbra{\Psi}{\Psi}) \mathbf{Z}^\mathbf{z} \mathbf{X}^\mathbf{x} \otimes I_R||_1
\end{align}
As we expect from super-dense coding, the set of states, $\{\mathbf{X}^\mathbf{x} \mathbf{Z}^\mathbf{z} \otimes I_R \ket{\Psi}\}_{\mathbf{x},\mathbf{z}}$ indexed by $\mathbf{x}, \mathbf{z}$ are orthonormal. To see this, consider $\mathbf{x}\neq\mathbf{x}'$ and/or $\mathbf{z}\neq\mathbf{z}'$,
\begin{align}
    &\bra{\Psi}[\mathbf{X}^\mathbf{x} \mathbf{Z}^\mathbf{z} \otimes I_R]\ [\mathbf{Z}^\mathbf{z'} \mathbf{X}^\mathbf{x'}  \otimes I_R]\ket{\Psi}\\
    =& \pm \bra{\Psi} \mathbf{X}^\mathbf{x + x'} \mathbf{Z}^\mathbf{z +z'} \otimes I_R \ket{\Psi}\\
    =& \pm \frac{1}{2^n} \sum_{\mathbf{i}, \mathbf{i}'} \bra{\mathbf{i}} \mathbf{X}^\mathbf{x+x'} \mathbf{Z}^\mathbf{z+z'} \ket{\mathbf{i}'} \braket{\mathbf{i}|\mathbf{i}'}\\
    =& \pm \frac{1}{2^n} \sum_{\mathbf{i}} \bra{\mathbf{i}} \mathbf{X}^\mathbf{x+x'} \mathbf{Z}^\mathbf{z+z'} \ket{\mathbf{i}}\\ 
    =& \pm \frac{1}{2^n} {\rm Tr}(\mathbf{X}^\mathbf{x+x'} \mathbf{Z}^\mathbf{z+z'})\\ 
    =&\ 0
\end{align}
where in the last step we used that $\mathbf{x}+\mathbf{x}'\neq\mathbf{0}$ when $\mathbf{x} \neq \mathbf{x}'$ (similarly for $\mathbf{z}$). Thus, $\sum_{\mathbf{x},\mathbf{z}} [p(\mathbf{x},\mathbf{z}) - q(\mathbf{x},\mathbf{z})]\ [\mathbf{X}^\mathbf{x} \mathbf{Z}^\mathbf{z} \otimes I_R \ket{\Psi}][\bra{\Psi} \mathbf{Z}^\mathbf{z} \mathbf{X}^\mathbf{x} \otimes I_R]$ is a convex-mixture of orthonormal states. So the 1-norm evaluates to,
\begin{align}
    &||\sum_{\mathbf{x},\mathbf{z}} [p(\mathbf{x},\mathbf{z}) - q(\mathbf{x},\mathbf{z})]\ \mathbf{X}^\mathbf{x} \mathbf{Z}^\mathbf{z} \otimes I_R (\ketbra{\Psi}{\Psi})  \mathbf{Z}^\mathbf{z} \mathbf{X}^\mathbf{x} \otimes I_R||_1
    = \sum_{\mathbf{x},\mathbf{z}} |p(\mathbf{x},\mathbf{z}) - q(\mathbf{x},\mathbf{z})|\ .
\end{align}

This proves the lemma.
\end{proof}

Now we compute a useful example using the above lemma.
\begin{lemma}\label{Z channel diamond norm}
    Consider the channel $\mathcal{Q}_z \circ \mathcal{Q}_x$ that is a concatenation of a bit-flip and phase-flip channel, i.e., $\mathcal{Q}_x(.)=(1-p_x)(.)+p_xX(.)X$  and $\mathcal{Q}_z(.)=(1-p_z)(.)+p_zZ(.)Z$ respectively, where $p_x\leq1$ and $p_z\leq1$.  Let $\mathcal{D}_z$ be the dephasing channel defined by $\mathcal{D}_z(\rho)=(\rho+Z\rho Z)/2$. Then,
    \begin{align}
        \|\mathcal{Q}_z \circ \mathcal{Q}_x- \mathcal{Q}_z \circ \mathcal{Q}_x\circ \mathcal{D}_z\|_\diamond= |1-2p_z|.
    \end{align}
\end{lemma}
\begin{proof}
The commutativity of Pauli channels together with contractivity of the diamond norm under data processing implies
\be
\|\mathcal{Q}_z \circ \mathcal{Q}_x- \mathcal{Q}_z \circ \mathcal{Q}_x\circ \mathcal{D}_z\|_\diamond=\|\mathcal{Q}_z \circ \mathcal{Q}_x- \mathcal{Q}_z \circ  \mathcal{D}_z\circ\mathcal{Q}_x\|_\diamond\le \|\mathcal{Q}_z  - \mathcal{Q}_z \circ \mathcal{D}_z\|_\diamond=\frac{1}{2} \|\mathcal{Q}_z  - \mathcal{Q}_z \circ \mathcal{Z}\|_\diamond = |1-2p_z|\ ,
\ee
where $\mathcal{Z}(.)=Z(.)Z$,  and  the last step  follows from lemma \ref{Pauli diamond}. To finish the proof we note that the left-hand side is lower bounded by
\be
\|\mathcal{Q}_z \circ \mathcal{Q}_x- \mathcal{Q}_z \circ \mathcal{Q}_x\circ \mathcal{D}_z\|_\diamond\ge \|\mathcal{Q}_z \circ \mathcal{Q}_x(|+\rangle\langle +|)- \mathcal{Q}_z \circ  \mathcal{D}_z\circ\mathcal{Q}_x(|+\rangle\langle +|)\|_1=\|\mathcal{Q}_z (|+\rangle\langle +|)-   \mathcal{D}_z\circ\mathcal{Q}_z(|+\rangle\langle +|)\|_1=|1-2p_z|\ ,
\ee
where we have used the fact that $\mathcal{Q}_x(|+\rangle\langle +|)=|+\rangle\langle +|$. This proves the lemma.
\end{proof}

\subsection{Entanglement Breaking  Pauli channels}\label{ent depth uncorrelated XZ}
We shall first state the following lemmas regarding entanglement-breaking properties of single-qubit Pauli channels \cite{wilde2013quantum, horodecki2003entanglement}.

\begin{lemma}\label{pauli ent breaking}
Consider the Pauli channel $\mathcal{N}(\rho)=p_I \rho+ p_x X\rho X+p_y Y\rho Y+ p_z Z\rho Z$, where $p_I,p_x, p_y, p_z\ge 0$ and $p_I+p_x+p_y+p_z=1$. This channel $\mathcal{N}$ is entanglement-breaking if and only if $\max\{p_I,p_x, p_y, p_z\} \leq 1/2$. 
\end{lemma}

\begin{proof}
    Let $\mathcal{N}$ act on system $A$, and $R$ be a reference system $R$. Now, consider a maximally entangled input state on $RA$, whose density operator can be written as $$\ketbra{\Phi}{\Phi}_{RA}=\frac14 (I\otimes I +X\otimes X - Y\otimes Y +Z\otimes Z).$$
Then, 
\begin{align}
    ({\rm id} \otimes \mathcal{N})&[\ketbra{\Phi}{\Phi}_{RA}] =\frac14 \Big[I\otimes I + (p_I + p_x - p_y - p_z) X\otimes X  - (p_I - p_x + p_y - p_z)Y\otimes Y + (p_I -p_x -p_y + p_z)Z\otimes Z\Big].
\end{align}
Partial-transpose of this state is 
$$\frac14 \Big[I\otimes I + (p_I + p_x - p_y - p_z) X\otimes X  + (p_I - p_x + p_y - p_z)Y\otimes Y + (p_I -p_x -p_y + p_z)Z\otimes Z\Big].$$
The eigenvalues of this matrix are $\{1/2 -p_I,1/2 -p_x,1/2 -p_y,1/2 -p_z\}$. From the Peres-Horodecki criterion \cite{Horo-Peres1}, all these eignevalues are positive if and only if $({\rm id} \otimes \mathcal{N}) [\ketbra{\Phi}{\Phi}_{RA}]$ is separable. Thus, for $\mathcal{N}$ to be entanglement-breaking, $\min\{1/2 -p_I,1/2 -p_x,1/2 -p_y,1/2 -p_z\} \geq 0$. This is equivalent to $\max\{p_I,p_x, p_y, p_z\} \leq 1/2$.  
\end{proof}
\begin{lemma}\label{ind XZ ent break}
    Define $\mathcal{N}_x(\rho) = (1-p_x)\rho + p_x X\rho X$ and $\mathcal{N}_z(\rho) = (1-p_z)\rho + p_z Z\rho Z$, where $p_x, p_z \leq 1/2$. Consider the special case of independent $X$ and $Z$ errors, that is $\mathcal{N}=\mathcal{N}_x \circ \mathcal{N}_z=\mathcal{N}_z \circ \mathcal{N}_x$. Such a channel $\mathcal{N}$ is entanglement-breaking if and only if $(1-p_x)(1-p_z)< 1/2$. 
\end{lemma}
\begin{proof}
The channel in the Lemma has the form
\begin{align}
    \mathcal{N}(\rho)=(1-p_x)(1-p_z)\rho + p_x(1-p_z)X\rho X + p_x p_z Y \rho Y + (1-p_x) p_z Z\rho Z.
\end{align}
Observe that when $p_x, p_z \leq 1/2$, the maximum coefficient is $(1-p_x)(1-p_z)$. From the aforementioned Lemma \ref{pauli ent breaking},
\begin{align}
    (1-p_x)(1-p_z)\leq 1/2 \iff \mathcal{N}\ {\rm is \ entanglement-breaking.}
\end{align}
\end{proof}

\subsection{Entanglement depth for CSS code trees}\label{ent depth appendix}
Consider a CSS code tree $\mathcal{E}_T$ where each edge has a single qubit noise channel which is an independent bit-flip and phase-flip channel, $\mathcal{N}=\mathcal{N}_x \circ \mathcal{N}_z$, with probabilities $p_x$ and $p_z$,  respectively. We noted in Eq.(\ref{setup eqn: css}) that for CSS code trees with optimal recovery, $\mathcal{R}_T\circ \mathcal{E}_T = \mathcal{Q}_z \circ \mathcal{Q}_x$, where $\mathcal{Q}_x(.)=(1-q^x_T)(.)+q^x_TX(.)X$  and $\mathcal{Q}_z(.)=(1-q^z_T)(.)+q^z_TZ(.)Z$. Note that for the optimal decoder $q^x_T, q^z_T\le 1/2$. Therefore, from lemma \ref{ind XZ ent break}, we find that channel $\mathcal{E}_T$ is entanglement-breaking if
\begin{align}\label{EB weak}
    (1-q^z_T)\times (1-q^x_1)<\frac12.
\end{align}
If the probability of physical $X$ error is not zero or one, i.e.,  $p_x \neq 0,1 $, then the probability of logical $X$ error is also non-zero, i.e.,  $0<q^x_1\leq q^x_T\leq 1/2$, where we have used the fact that $q^x_T$ grows monotonically with $T$.

Now for $(1-d_z^{-1/2})/2<p_z\leq 1/2$ proposition \ref{Prop 1} implies that information encoded in $X$ and $Y$ input bases decays exponentially. In particular, Eq.(\ref{mainbound css}) can be equivalently stated as
\be
1-q^z_T \leq \frac12 \Big[1+\sqrt{2}(\sqrt{d_z}\times |1-2p_z|)^T\Big]\ .
\ee
By substituting this into the LHS of Eq.(\ref{EB weak}) and simplifying, we get that $\mathcal{E}_T$ is entanglement-breaking when
\begin{align}
    T> \ln\Big(\frac{\sqrt{2}(1-q^x_1)}{q^x_1}\Big) \times \frac{1}{-\ln(\sqrt{d_z}(1-2p_z))},
\end{align}
which is Eq.(\ref{ent depth final}).
\newpage

\section{Local Recovery of CSS code trees}\label{local rec appendix}

In this section, we revisit the local recovery strategy discussed in Sec. \ref{local rec section} for CSS code trees whose encoders have distance $d\geq 3$ and each edge is subject to independent $X$ and $Z$ errors with probability $p_x$ and $p_z$,  respectively. Recall from Eq.(\ref{CSS rec eqn}) in Sec. \ref{css local rec} that the asymptotic logical $Z$ error $q^z_\infty$  satisfies the equality,
\begin{align}
    q^z_\infty= p_z\times (1-\alpha_z(q^z_\infty))+(1-p_z)\times \alpha_z(q^z_\infty)\ .
\end{align}
We aim to find the noise threshold $p^z_{\rm th}$, i.e., the largest value of $p_z\in[0,0.5]$ such that the above equation has a solution when $q^z_\infty\in (0,0.5)$. To find this threshold, first we rewrite the above equation as 
\begin{align}\label{fixed point eqn}
     {\alpha_z}(q^z_\infty)= \frac{q^z_\infty-p_z}{1-2p_z}\ ,
\end{align}
which means $q^z_\infty$ can be determined by finding the intersection of curve $\alpha_z(q)$ and the line $\frac{q-p_z}{1-2p_z}$. More precisely, $p^z_{\rm th}$ is  the x-intercept of the tangent line to curve $\alpha_z$ that passes through the point $(0.5,0.5)$ (see Fig. \ref{stab local rec plot}).

\begin{figure}[ht!]
\centering
\includegraphics[width=0.5\textwidth]{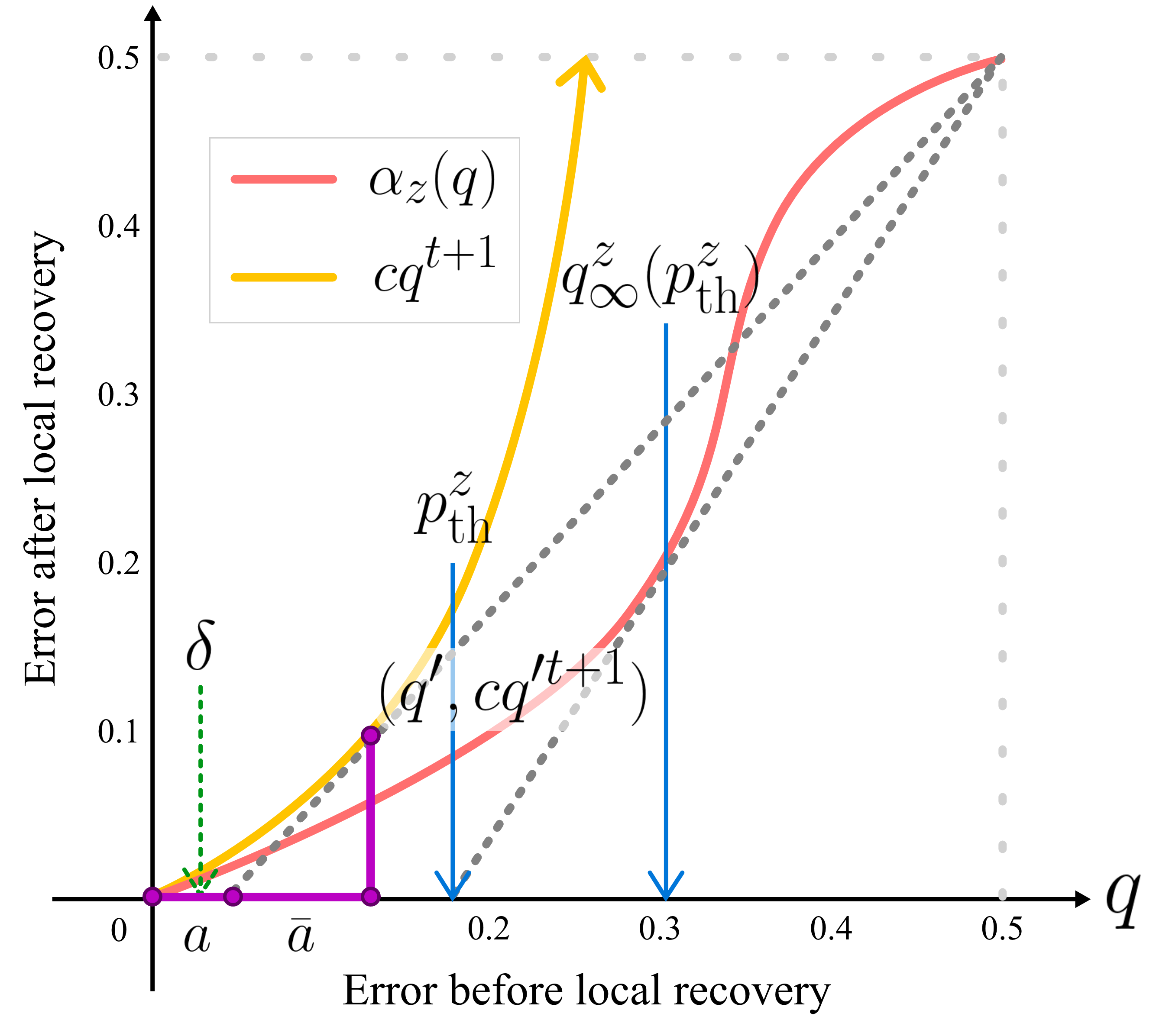}
\caption{This plot illustrates the tangent line to $c q^{t_z+1}$ which passes through the point (0.5,0.5). The x-intercept of this line, denoted by $a$, is  a lower bound on $p^{z}_{\rm th}$, i.e., $p^{z}_{\rm th}\geq a$. Suppose the tangent line intersects the curve $c q^{t_z+1}$ at point $(q',cq'^{t_z+1})$. Define $\bar{a}=q'-a$. Because the slope of the tangent line at $q'$ is equal to the derivative of curve $c q^{t_z+1}$ at this point, we find $\frac{cq'^{t_z+1}}{\bar{a}}=(t_z+1)cq'^t_z$, which implies $\bar{a}=\frac{q'}{t_z+1}$, or equivalently,    $a=\frac{t_z}{t_z+1}q'.$ Finally, note that  the slope of the tangent is greater than or equal to one, i.e., 
$(t_z+1)cq'^t_z \ge 1$, which means $ q'\geq \big(\frac{1}{(t_z+1)c}\big)^{{1}/{t_z}}$. We conclude that $a = \frac{t_z}{t_z+1}q' \geq \frac{t_z}{t_z+1} \big(\frac{1}{(t_z+1)c}\big)^{{1}/{t_z}}$, which in turn implies $p^{z}_{\rm th}\ge a\ge \frac{t_z}{t_z+1} \big(\frac{1}{(t_z+1)c}\big)^{{1}/{t_z}}$.  Thus, for $p_z\leq \delta$, the logical error upon local recovery of an infinite tree $q^z_\infty<1/2$, because $\delta \leq p^z_{\rm th}$.}
\label{stab local rec plot}
\end{figure}

\subsection{Lower bound on $p^z_{\rm th}$}\label{css lowerbound appendix}

In the following, first we derive an upper bound on
$\alpha_z$ and then as we describe in Fig. \ref{stab local rec plot}, based on this upper bound, we derive a  lower bound on $p^z_{\rm th}$.

Suppose a CSS code corrects all $Z$ errors with weight less than or equal to $t_z$. Assuming each qubit is subjected to independent $Z$ error with probability $q$, then, after the error correction, the probability of $Z$ error is bounded by  
\begin{align}
    \alpha_z(q)\leq\sum_{k=t_z+1}^b {b \choose k} q^k (1-q)^{b-k} \leq \Bigg[\sum_{k=t_z+1}^b {b \choose k}\Bigg] q^{t_z+1} \equiv c q^{t_z+1}\ ,
\end{align}
where $\sum_{k=t_z+1}^b {b \choose k} q^k (1-q)^{b-k}$ is the probability of $Z$ errors with weight $t_z+1$ or larger, and
\be
 c\equiv  \sum_{k=t_z+1}^b {b \choose k}\le 2^b\ .
 \ee
 As it can be seen from Fig. \ref{stab local rec plot}, since   $c q^{t_z+1}\ge \alpha_z(q)$ for $q\in[0,0.5]$, the  x-intercept of this line is smaller than or equal to the  x-intercept of curve $\alpha_z(q)$ that passes through the point $(0.5,0.5)$. As we explain in the caption of Fig. \ref{stab local rec plot}, this implies
 that for
 \be
p_z \le   \frac{t_z}{t_z+1} \bigg(\frac{1}{(t_z+1)c}\bigg)^{\frac{1}{t_z}}\equiv \delta\ ,
\ee
$q^z_\infty$ is strictly less than $1/2$.  

This proves the existence of a non-zero local recovery threshold for any CSS code tree with $d_z\geq3$. Furthermore, this also lower bounds the optimal recovery threshold for $d_z\geq 3$ CSS code tree.
\subsection{Upper bound on  $q_\infty^z$}\label{css error upperbound appendix}

\begin{figure}[ht!]
\centering
\includegraphics[width=0.5\textwidth]{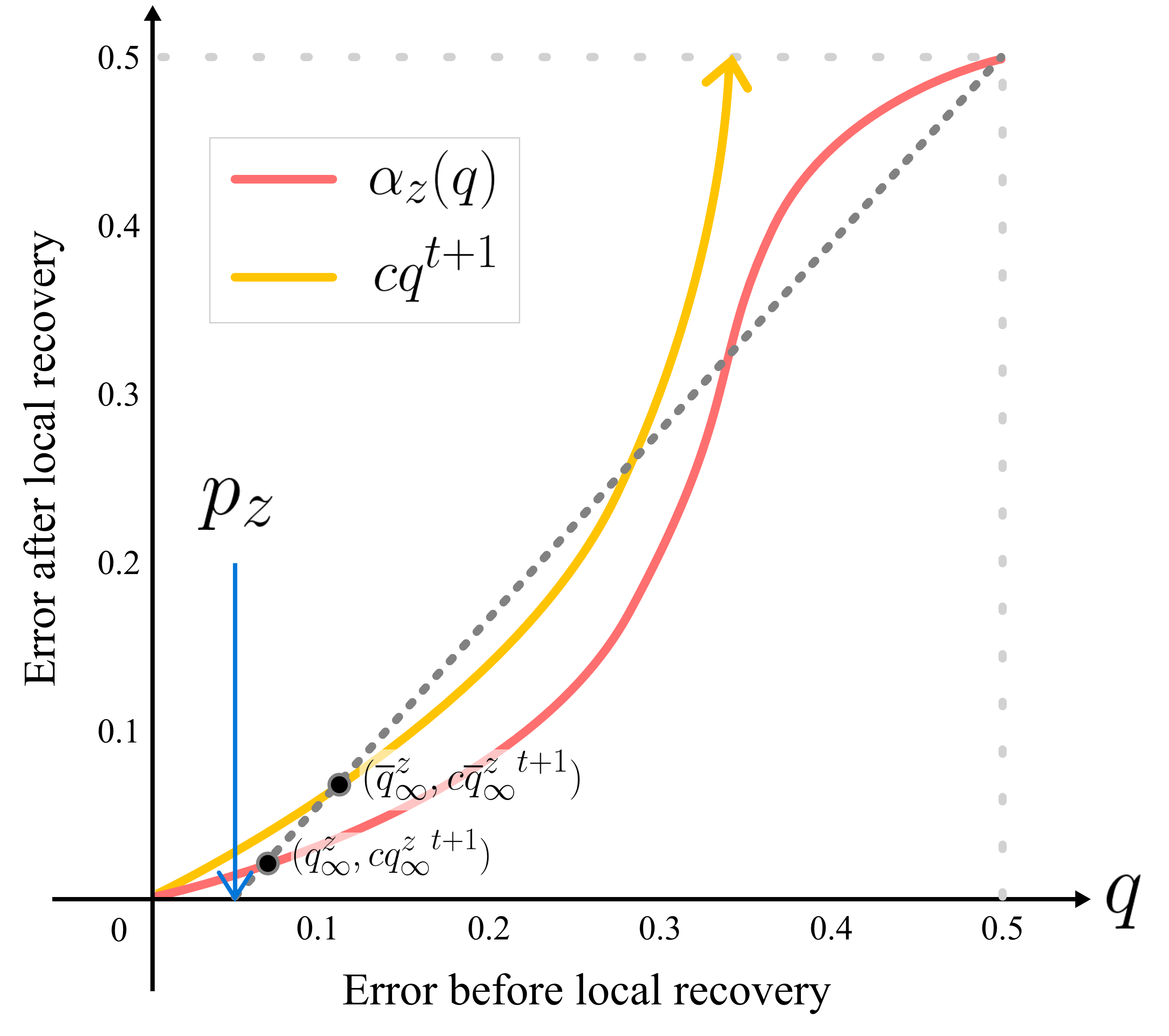}
\caption{This figure illustrates how to estimate the logical error in an infinite CSS code tree after local recovery given that the $Z$ error within the tree is set at $p_z$. The line that passes through (0.5,0.5) and has $x$-intercept $p_z$ intersects $\alpha(q)$ at the point $(q^z_\infty,\alpha(q^z_\infty))$. We show that the $q^z_\infty$ is the logical error in the infinite tree after local recovery. It is clear that $q^z_\infty$ depends on the specifics of $\alpha$, but we can upperbound $\alpha(q)\leq c \times q^{t_z+1}$, which allows an upperbound for $q^z_\infty$, which we denote by $\overline{q}^z_\infty$.}
\label{upperbound for q inf}
\end{figure}

By definition, for $p_z< p^z_\text{th}$, we have $q_\infty^z<1/2$, which means  after error correction  the total error  in the infinite tree is less than 1/2. In the following we establish an  upper bound on $q_\infty^z$. We again use the fact that  $\alpha(q)\leq cq^{t_z}$. As we saw above for sufficiently small $p_z$, e.g., 
\be
p_z<\frac{t_z}{t_z+1} \bigg(\frac{1}{(t_z+1)c}\bigg)^{{1}/{t_z}}=\delta\ ,
\ee
 the line $(q-p_z)/(1-2p_z)$ has two intersections with the curve $c q^{t_z+1}$. Let $(\overline{q}, c \overline{q}^{t_z+1})$ be the coordinates of the intersection with smaller $x$ coordinates (see Fig. \ref{upperbound for q inf}). More precisely $\overline{q}$ is the smallest positive real number satisfying 
\be\label{po2}
 \frac{\overline{q}-p_z}{1-2p_z}=c \overline{q}^{t_z+1}\ .
\ee 
At this point,  the slope of the curve 
 $c q^{t_z+1}$ is $(t_z+1) c \overline{q}^{t_z}$. Furthermore, because for $q$ in the interval  $0<q<\overline{q}$, we have  $ \frac{q-p_z}{1-2p_z}< c q^{t_z+1}$, at the intersection the slope of  line $(q-p_z)/(1-2p_z)$ is larger than the slope of  the curve 
$c q^{t_z+1}$ (this can also be seen from Fig. ), that is,  
 \be\label{tte}
(t_z+1) c \overline{q}^{t_z} \le \frac{c \overline{q}^{t_z+1}}{\overline{q}-p_z}\ ,
 \ee
 where we have used the fact that the $x$-intercept of line  $(q-p_z)/(1-2p_z)$  is $p_z$. Another way to justify this inequality is to notice that for point $q'$ (as defined in Fig. \ref{stab local rec plot}),
 \be
(t_z+1) c {q'}^{t_z} = \frac{c {q'}^{t_z+1}}{{q'}-a}.
\ee
Because $\overline{q}$ is the first point of intersection, we necessarily have that $\overline{q}<q'$ and $a<p_z$. 
Thus,
\be
(t_z+1) c \overline{q}^{t_z} \leq (t_z+1) c {q'}^{t_z} = \frac{c {q'}^{t_z+1}}{{q'}-a} \leq \frac{c \overline{q}^{t_z+1}}{\overline{q}-p_z}.
\ee
We simplify Eq.(\ref{tte}) to,
 \be
\overline{q}\le (1+\frac{1}{t_z}){p_z}\ .
 \ee
 Finally, we note that 
 \be
 q^z_\infty\le \overline{q}\ ,
 \ee
which can be seen from the Fig.\ref{upperbound for q inf},  or by noting that (i) $\overline{q}$  is the smallest positive real number satisfying Eq.(\ref{po2}), and   (ii) $q^z_\infty$ is the smallest positive real number satisfying 
  \be
 \frac{q^z_\infty-p_z}{1-2p_z}=\alpha(q^z_\infty)\ .
\ee 
 The latter means that   for all $q$ in the interval $0 < q<  q^z_\infty$ it holds that
  \begin{align}
  \frac{q-p_z}{1-2p_z} < \alpha(q) < c q^{t_z+1}\ ,
 \end{align}
 and, therefore,  $\overline{q}$  is not in this interval. We conclude that
  \begin{align}
q^z_\infty\le \overline{q} \le (1+\frac{1}{t_z}){p_z}\ .
 \end{align}
It is worth noting that if the code corrects single-qubit errors, i.e., if $t_z\ge 1$, then  for $q=1/2$ we have
 \be\label{artq}
c {\frac{1}{2^{t_z+1}}}> \alpha(\frac{1}{2})=\frac{1}{2}\ ,
\ee
which implies
\be\label{artq2}
c^{1/t_z} > 2\ .
 \ee
This implies that for 
\be
q^z_\infty\le \overline{q} \le (1+\frac{1}{t_z}){p_z}   \le (1+\frac{1}{t_z})\delta =\bigg(\frac{1}{(t_z+1)c}\bigg)^{{1}/{t_z}}< \frac{1}{2}\times \bigg(\frac{1}{t_z+1}\bigg)^{{1}/{t_z}} < \frac{1}{2}\ .
 \ee
This gives another proof that for $p_z\le \delta$, $q^z_\infty$ is strictly less than $1/2$.

\subsection{Fixed-point equation for stabilizer codes with Pauli errors}

Optimal recovery of a single block of stabilizer codes with Pauli errors can be realized by applying a Clifford unitary, measuring ancilla qubits in $\{|0\rangle,|1\rangle\}$ basis and performing a Pauli operator on the data qubit.   It follows that if the noise channel $\mathcal{N}$ is a Pauli channel, then the sequence of channels defined by the recursive equation in Eq.(\ref{gsj}) will be all Pauli channels. Then, as long as $\mathcal{N}_\infty$ is not the fully depolarizing channel $\mathcal{N}_\infty(\rho)=I/2$, it can transmit classical information. Also, for Pauli channels there are simple necessary and sufficient conditions that determine whether the channel is entanglement-breaking, or not (see Appendix \ref{ent depth uncorrelated XZ}).  

Let $\textbf{q}=(q^{x}, q^y, q^z)$ be the probability of $X$, $Y$, and $Z$ errors in channel $\mathcal{M}$ in Eq.(\ref{func}). Then,  the channel $f[\mathcal{M}]$ will be a Pauli channel in the form
\be
f[\mathcal{M}](\rho)=\alpha_{I}(\textbf{q})\rho+\sum_{w=x,y,z} \alpha_{w}(\textbf{q})\  \sigma_w \rho \sigma_w\ , 
\ee
where $\alpha_{I}(\textbf{q})=1-\sum_{w=x,y,z} \alpha_{w}(\textbf{q})$. As is clear, if $\textbf{q}$ is the probability of single-qubit error on the physical qubits, $\{\alpha_w(\textbf{q})\}_w$ is the probability of single-qubit error on the logical qubit. Let  $\textbf{q}_t=({q}^{x}_t,{q}^{y}_t, {q}^{z}_t)$ denote the error probabilities for the Pauli channel $\mathcal{N}_t$. Then,  the recursive Eq.(\ref{gsj})  can be rewritten as
\begin{align}\label{kj}
\textbf{q}_{t+1}&=\left(\alpha_{x}(\textbf{q}_t),
\alpha_y(\textbf{q}_t) , 
\alpha_z(\textbf{q}_t)
\right)
 \left(
\begin{array}{ccc}
p_I  & p_Z  & p_Y  \\
 p_Z &  p_I  &  p_X  \\
p_Y  & p_X   &   p_I
\end{array}
\right)\nonumber
\\ &+\big[1-\sum_{w=x,y,z} \alpha_{w}(\textbf{q}_t)\big] (p_X, p_Y, p_Z) \ ,
\end{align}
where  $p_X, p_Y, p_Z$ are the error probabilities for channel $\mathcal{N}$, and $p_I=1-(p_X+p_Y+p_Z)$.

Since  the recovery $\mathcal{R}$ is a linear map, functions $\{\alpha_w(\textbf{q}): w=x,y,z\}$ are polynomials of, at most, degree  $b$ in each variable $q^{x}$, $q^{y}$, and $q^{z}$. Furthermore, for codes with distance $d$, the recovery $\mathcal{R}$ can be chosen to correct all the errors with weight less than or equal to $\lfloor (d-1)/2\rfloor$.   
In this case, the lowest degree of monomials in $q^{x}$, $q^{y}$, and $q^{z}$ is  $\lfloor (d-1)/2\rfloor+1$. 

As we show in the following subsection, this implies that for codes with distance $d\ge 3$ there are positive error thresholds $p_\text{th}>p_\text{ent}>0$  such that for $p_x, p_y, p_z< p_\text{th}$ the channel $\mathcal{N}_\infty$ is not the fully depolarizing channel, hence transmitting (noisy) classical information, and for  $p_x, p_y, p_z< p_\text{ent}$ the channel $\mathcal{N}_\infty$ is not entanglement-breaking, hence transmitting  entanglement.

\subsection{Non-zero noise threshold for general stabilizer code trees} \label{gen stab threshold}
Here, we show how the above argument can be generalized to the case of all stabilizer codes with Pauli errors. In particular, we show that if in Eq.(\ref{kj}) none of polynomials $\alpha_x$, $\alpha_y$ and $\alpha_z$ have linear terms, then there exists a positive threshold $p_{\text{th}}>0$ such that  for $p_x, p_y, p_z < p_{\text{th}}$, Eq.(\ref{kj}) has fixed-point channels other than the fully depolarizing channel. To this end, we consider the probability of \textit{some} error occurring $p=\sum_{w=x,y,z} p_w$, where $\mathbf{p}=(p_x,p_y,p_z)$ is the probabilities associated to Pauli error. To study recovery, we define $$\widetilde{\alpha}(p):=\sum_{p_x+p_y+p_z=p} \alpha(\mathbf{p}),$$that is the probability of some logical error as a function of some physical error. Let $q_t$ indicate the probability of some error after local recovery of a tree of depth $t$. Then Eq.(\ref{kj}) can be turned into the inequality,
\begin{align}\label{gen stab rec reln}
q_{t+1} \leq 1-(1-\widetilde{\alpha}(q_t))(1-p)
\end{align}
where we note the probability of no error in level $t+1$ is \textit{at least} the probability of no error after recovery and no local error. We focus on the upperbound in Eq.(\ref{gen stab rec reln}) by defining $\epsilon_t$ as $$\epsilon_{t+1}= 1-(1-\widetilde{\alpha}(\epsilon_t))(1-p)$$ and $\epsilon_0=0$. These would satisfy $q_t \leq \epsilon_t$ for all $t$. The fixed point of this recursive relation for $\epsilon_t$, i.e., $\epsilon_\infty$, satisfies the equation,
\begin{align}\label{gen stab tangent eqn}
    \widetilde{\alpha}(\epsilon_\infty)=\frac{\epsilon_\infty-p}{1-p}
\end{align}
It is clear that for stabilizer codes with distance $d\geq3$, all errors with weight atmost $\widetilde{t}:=\lfloor(d-1)/2\rfloor$ will be corrected. Thus, $\tilde\alpha$ will have no terms of degree less than $\widetilde{t}+1$ in $q$ because it is a sum of $\alpha_w$'s which each do not have upto degree-$\widetilde{t}$ monomials. More specifically,
    \begin{align}
        \widetilde{\alpha}(q)\leq \widetilde{c} \times q^{\widetilde{t}+1}
    \end{align}
for positive constant $\widetilde{c}\geq1$. 
Because it is an upperbound, the local recovery threshold estimate for $\widetilde{c} \times q^{\widetilde{t}+1}$ (denoted by $\widetilde{p}_{\rm th}$) will lowerbound the threshold estimate for $\widetilde{\alpha}(q)$, i.e., $\widetilde{p}_{\rm th}\leq{p}_{\rm th}$. Observe that the family of lines considered in the RHS of Eq.(\ref{gen stab tangent eqn}) are characterized as passing through $(1,1)$. Thus, the threshold is estimated by computing the $x$-intercept (i.e., $\widetilde{p}_{\rm th}$) of the line passing through $(1,1)$ that is tangent to $\widetilde{c} \times q^{\widetilde{t}+1}$. This is equivalent to the threshold computation for CSS codes detailed in Fig. \ref{stab local rec plot} in Appendix \ref{local rec appendix} up to rescaling. That is, consider the two curves $$y=\widetilde{c} \times q^{\widetilde{t}+1},\ {\rm and}\ y=\frac{q-\widetilde{p}_{\rm th}}{1-\widetilde{p}_{\rm th}},$$
where the latter is tangent to the former. Upon rescaling as $(q',y')=(q/2,y/2)$, the two curves become, $$y'=( \widetilde{c} 2^{\widetilde{t}}) \times {q'}^{\widetilde{t}+1}, \ {\rm and}\ y'=\frac{q'-2\widetilde{p}_{\rm th}}{1-2\widetilde{p}_{\rm th}},$$
where the latter remains tangent to the former. So, if we set $c=\widetilde{c}2^{\widetilde{t}}$ and note that the $x$-intercept is $2\widetilde{p}_{\rm th}$, we can invoke Eq.(\ref{th lowerb}) to conclude,
\begin{align}
2\widetilde{p}_\text{th}\geq\frac{\widetilde{t}}{\widetilde{t}+1} \bigg(\frac{1}{(\widetilde{t}+1)\widetilde{c}2^{\widetilde{t}}}\bigg)^{\frac{1}{\widetilde{t}}}.
\end{align}
Thus, noting that $\widetilde{p}_{\rm th}\leq p_{\rm th}$, we conclude the following extension of Proposition \ref{prop: css lowerbound} to general stabilizer codes: 
consider a $d\geq 3$ stabilizer tree composed of a stabilizer code that can correct physical errors of weight upto $\widetilde{t}=\lfloor (d-1)/2 \rfloor$, with noise on each edge with probability of some Pauli error being $p=p_x+p_y+p_z$. Then, for sufficiently small $p$, e.g.
$$p < \frac14 \frac{\widetilde{t}}{\widetilde{t}+1} \bigg(\frac{1}{(\widetilde{t}+1)\widetilde{c}}\bigg)^{\frac{1}{\widetilde{t}}},$$
the probability of \textit{some} logical error in the $d\geq 3$ stabilizer tree of any depth $T$ upon local recovery, i.e. $q_T$, is upperbounded as, 
$$q_T \leq p \times (1 + \frac{1}{\widetilde{t}}),$$
because $q_T\leq q_\infty$ for all $T$.


\newpage

\section{Examples -- Computations for Local Recovery for specific stabilizer trees}\label{CSS appendix}
\subsection{Local Recovery threshold for classical trees}\label{local rec threshold appendix}
Here we derive the expression in Eq.(\ref{loc-threshold}). Consider the $\alpha_{\rm maj}$ function for local recovery of classical trees. For simplicity, we shall focus on odd branching number. Then,
\begin{align}
    \alpha_{\rm maj}(p)=\sum_{k={{\lfloor b/2\rfloor}+1}}^b {b\choose k} p^k (1-p)^{b-k} = 1 - \sum_{k=0}^{\lfloor b/2\rfloor} {b\choose k} p^k (1-p)^{b-k}
\end{align}

As per this definition, $\alpha_{\rm maj}(0)=0$ and $\alpha_{\rm maj}(0.5)=0.5$ (the latter can be easily seen by noting that when $p=0.5$, the probability distribution is symmetric about $0-1$ exchange). From Sec. \ref{css local rec}, we know that the x-intercept of the tangent to $\alpha_{\rm maj}(p)$, that also passes through $(0.5,0.5)$, is $p_{\rm th}$. We prove that $\alpha_{\rm maj}(p)$ is convex, i.e. $\alpha''_{\rm maj}(p)\geq0$ for $0\leq p \leq 0.5$ in Section \ref{alpha_convex_proof}. Thus, the tangent to the $\alpha_{\rm maj}$ curve is the tangent line at $p=0.5$. Using the fact that the tangent must pass through $(0.5,0.5)$, we get that 
\begin{align}\label{class local threshold}
  p_{\rm th}=\frac{1}{2}\Big(1-\frac{1}{\alpha'_{cl}(0.5)}\Big).  
\end{align}
We compute $\alpha'_{\rm maj}(0.5)$ below:
\begin{align}
    \alpha'_{\rm maj}(p) = - \sum_{k=0}^{\lfloor b/2\rfloor} {b\choose k} kp^{k-1} (1-p)^{b-k} + \sum_{k=0}^{\lfloor b/2\rfloor} {b\choose k} (b-k)p^k (1-p)^{b-k-1}
\end{align}
At $p=0.5$ this evaluates to,
\begin{align}
    \alpha'_{\rm maj}(0.5) = 2 \sum_{k=0}^{\lfloor b/2\rfloor} \frac{1}{2^b}{b\choose k} (b-2k)
\end{align}
Notice that $\frac{1}{2^b}{b\choose k}$ is the binomial distribution in $k$ with $p=0.5$. Because we are considering the regime when $b\rightarrow\infty$, we use the normal approximation to the binomial distribution to get,
\begin{align}
    2 \sum_{k=0}^{\lfloor b/2\rfloor} \frac{1}{2^b}{b\choose k} (b-2k) &\approx - 4 \int_0^{k=b/2} \frac{1}{\sqrt{\pi b/2}} e^{-\frac{(k-b/2)^2}{b/2}} (k-b/2) dk \Bigg(1 + \mathcal{O}\Big(\frac{1}{\sqrt{b}}\Big)\Bigg)\\
    &= -4 \int_{-b/2}^0 \frac{1}{\sqrt{\pi b/2}} e^{-\frac{x^2}{b/2}} x dx \Bigg(1 + \mathcal{O}\Big(\frac{1}{\sqrt{b}}\Big)\Bigg)
\end{align}
where we did a change of variables as $x=k-b/2$ in the second equality. Because we are anyway in the regime where $b\rightarrow\infty$, we further approximate this as,
\begin{align}
\approx -4 \int_{-\infty}^0 \frac{1}{\sqrt{\pi b/2}} e^{-\frac{x^2}{b/2}} x dx \Bigg(1 + \mathcal{O}\Big(\frac{1}{\sqrt{b}}\Big)\Bigg)= \sqrt{\frac{2}{\pi}} \sqrt{b} \Bigg(1 + \mathcal{O}\Big(\frac{1}{\sqrt{b}}\Big)\Bigg)= \sqrt{\frac{2}{\pi}} \Big( \sqrt{b} + \mathcal{O}(1) \Big)
\end{align}
By substituting this into Eq.(\ref{class local threshold}), we get, 
\begin{align}
    p_{th}=\frac{1}{2}\Bigg(1-\sqrt{\frac{\pi}{2}}\frac{1}{\sqrt{b}+\mathcal{O}(1)}\Bigg)=\frac{1}{2}\Bigg(1-\sqrt{\frac{\pi}{2}}\frac{1}{\sqrt{b}}\Bigg) + \mathcal{O}\Big(\frac{1}{b}\Big)
\end{align}

\subsubsection{Convexity of $\alpha_{\rm maj}$}\label{alpha_convex_proof}
We start by rewriting $\alpha_{\rm maj}$ by noting that that the cumulative binomial distribution ${Pr}(X_{\rm bin}\leq k)$ can be written in terms of the \textit{regularized incomplete beta} function $F(k;n,p)$ (refer to Eq.3.3 in \cite{gp1960bryan}),
\begin{align}
    \alpha_{\rm maj}(p)&=1-\sum^{\lfloor{b}/{2}\rfloor}_{k=0} {b \choose k} p^k(1-p)^{b-k}\\
    &= 1 - {\rm Pr}(X_{\rm bin}\leq \frac{b-1}{2})\\
    &= 1 - F(\frac{b-1}{2};b,p)\\
    &= 1 - \Big[\frac{b+1}{2}\Big]{b \choose \frac{b-1}{2}} \times \int_0^{1-p} t^{\frac{b-1}{2}}(1-t)^{\frac{b-1}{2}} dt
\end{align}
where $b$ is odd. Now we differentiate this expression twice with respect to $p$ rather easily,
\begin{align}
    \alpha'_{\rm maj}(p)&=\Big[\frac{b+1}{2}\Big]{b \choose \frac{b-1}{2}} \times (1-p)^{\frac{b-1}{2}}p^{\frac{b-1}{2}}\\
    &=\Big[\frac{b+1}{2}\Big]{b \choose \frac{b-1}{2}} \times (p-p^2)^{\frac{b-1}{2}},\\
    \alpha''_{\rm maj}(p)&= \Big[\frac{b+1}{2}\Big]{b \choose \frac{b-1}{2}} \times \frac{b-1}{2} \times (p-p^2)^{\frac{b-3}{2}} \times (1-2p).
\end{align}
Thus, $\alpha''_{\rm maj}(p)\geq0$ for all $p\in[0,0.5]$.

\subsection{Example: Local recovery threshold for Generalized Shor code}\label{subsec: shor rec}
An important class of CSS codes are generalized Shor codes. The encoder of a Shor code of order $n$ is a serial concatenation of two repetition codes ($M_1$ and $M_2$) that are dual to each other, i.e., $M_1^{\otimes n} \circ M_2$, that are defined as
\begin{align}\label{shor encoder defn}
    &M_1 \ket{c}= \ket{c}^{\otimes b}: \ \ \ \  c\in \{0,1\},\\
    &M_2 = H^{\otimes n} \circ M_1 \circ H,
\end{align}
where $H$ is the Hadamard gate. Observe that $M_1$ is the classical copier in the $Z$ basis, while $M_2$ is the same in the dual $X$ basis. Fig. \ref{T=1_Shor} illustrates the encoder for $n=3$ that is the usual Shor-9 encoder, upto an additional Hadamard at the input to make it a standard encoder.  We consider a two-step recovery of the generalized Shor code with two decoders for $M_1$ and $M_2$ respectively. These are majority voting procedures in the $Z$ and $X$ basis respectively. We denote them as $\mathcal{R}_{\rm maj}$ and $\mathcal{R}_{\rm par}:=\mathcal{H}\circ\mathcal{R}_{\rm maj}\circ \mathcal{H}^{\otimes n}$ respectively, where $\mathcal{H}(.)=H(.)H$.

\begin{figure}[ht!]
\centering
\includegraphics[width=0.15\textwidth]{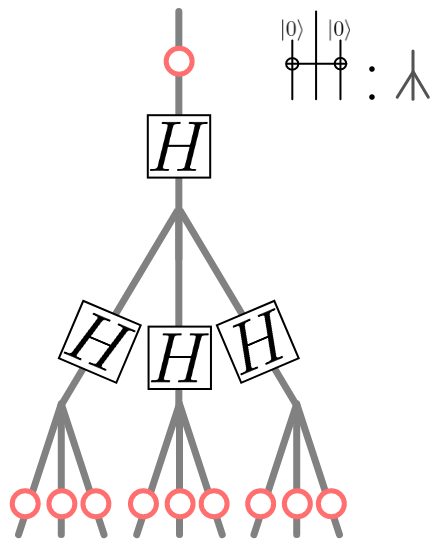} 
\caption{\textbf{Encoder of a standard generalized Shor encoder with $n=3$ --} This is the usual Shor-9 encoder. We have added a Hadamard gate to the root to ensure the encoder is standard with regards to $X$ and $Z$ logical operators.}
\label{T=1_Shor}
\end{figure}

We consider trees with a generalized Shor code encoder and independent single-qubit bit-flip and phase-flip errors, with probabilities $p_x$ and $p_z$ respectively, on each edge. Thus, the local recovery we apply to one level of the generalized Shor code is $\mathcal{R}_{\rm par}\circ \mathcal{R}_{\rm maj}^{\otimes n}$.

 Define $\mathcal{M}_i(.)=M_i(.)M^\dagger_i$ for $i=1,2$. For the encoder $\mathcal{M}_1$,
after applying single-qubit noise channels $\mathcal{N}_x$ on the $n$ physical qubits independently and then performing local recovery, we obtain the effective channel
$\mathcal{R}_{\rm maj}\circ \mathcal{N}^{\otimes n}_x \circ \mathcal{M}_1$. The corresponding 
  probability of logical $X$ error of this channel is $\alpha_{\rm maj}(p_x)$ (defined in Eq.(\ref{alpha cl})),  where $p_x$ is the probability of physical $X$ error sampled by $\mathcal{N}_x$. Similarly, in the case of the encoder $\mathcal{M}_2$, we obtain the effective channel  $\mathcal{R}_{\rm par}\circ \mathcal{N}^{\otimes n}_x \circ \mathcal{M}_2$ after local recovery, whose corresponding probability of logical $X$ error we denote by $\alpha_{\rm par}(p_x)$. Observe from the encoder definitions in Eq.(\ref{shor encoder defn}) that the relation of $\mathcal{M}_2$ to bit-flip errors is equivalent to the relation of $\mathcal{M}_1$ to phase-flip errors due to the Hadamards. Now recall that a single-qubit physical $Z$ error is a logical error for the repetition code.   So, logical error is determined by the parity of the weight of the physical error. Thus, the logical error rate associated with $\mathcal{R}_{\rm par}\circ \mathcal{N}^{\otimes n}_x \circ \mathcal{M}_2$ is
\begin{align}
    \alpha_{\rm par}(p_x)=\frac12 - \frac12|1-2p_x|^n.
\end{align}
Thus, the logical $X$ error in the full channel $$\mathcal{R}_{\rm par}\circ \mathcal{R}_{\rm maj}^{\otimes n}\circ\mathcal{N}^{\otimes n^2}_x \circ\mathcal{M}^{\otimes n}_1\circ \mathcal{M}_2$$
is given by the composition of the functions
\begin{align}
\alpha_x(p_x)=\alpha_{\rm par}\circ \alpha_{\rm maj}(p_x).
\end{align} 
It is easy to see that the above argument in the opposite order pertains to $Z$ errors. Thus,
\begin{align}
\alpha_z(p_z)=\alpha_{\rm maj}\circ\alpha_{\rm par}(p_z).
\end{align}

Using the procedure detailed in Sec. \ref{css local rec}, we can numerically estimate the local recovery threshold $p_{\rm th}^x$ (or, $p_{\rm th}^z$) from $\alpha_x$ (or, $\alpha_z$) for all generalized Shor codes $[[n^2,1,n]]$ as plotted in Fig. \ref{threshold_range}.\\

\begin{figure}[ht!]
\centering
\includegraphics[width=0.5\textwidth]{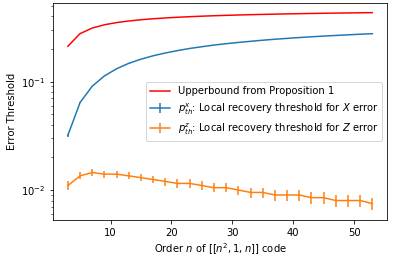} 
\caption{\textbf{Noise threshold range for $[[n^2,1,n]]$ Shor code tree for increasing $n$ -- } The upperbound $\frac12 (1-\frac{1}{\sqrt{n}})$ is obtained from Proposition \ref{Prop 1}. The local recovery thresholds, $p^x_{\rm th}$ and $p^z_{\rm th}$ (i.e., lowerbounds), are obtained from a numerical estimation procedure that benchmarks the local recovery strategy that is explained in Sec. \ref{css local rec}. The upperbound is the same for both $X$ and $Z$ errors, while the lowerbounds are different and are indicated. With increasing $n$, while the $X$-error threshold improves, the $Z$-error threshold suffers. The $x$-axis indicates the order $n$ of $[[n^2,1,n]]$ code for odd $n$, and the $y$-axis indicates the error threshold. The error bars of $0.001$ are determined by the resolution of the numerical estimation procedure.}
\label{threshold_range}
\end{figure}

\textit{Asymptotics for large $n$ --} When we study the asymptotic case of large $n$, the majority voting is equivalent to an asymptotic binomial hypothesis testing scenario for which we can use the exponential Chernoff's bound, that is, 
$$\alpha_{\rm maj}(p_x)\lesssim e^{-\frac{n}{3} |\frac12 - p_x|}.$$
This gives us that $$\alpha_x(p_x) \lesssim \frac{1-(1-2e^{-\frac{n}{3} |\frac12 - p_x|})^n}{2},$$ which clearly goes to $0$ for large $n$.  This shows that for large $n$, we have $p^z_{\rm th} \rightarrow 1/2$ because the probability of logical error after local recovery approaches zero.

When we study the large-$n$ asymptotics for $\alpha_z(p_z)$, observe that $\alpha_{\rm par}(p_z)=\frac12 - \frac{|1-2p_z|^n}{2} \rightarrow \frac12$ as $n\rightarrow \infty$. Thus,
$\alpha_z(p_z)\rightarrow1/2$ for large $n$. This implies that $p^z_{\rm th}\rightarrow 0$, i.e., an absence of a non-trivial local recovery threshold asymptotically.\\

\textit{Asymmetry in $X$ and $Z$ performance --} The worse performance of recovery for $X$-direction information as compared to $Z$-direction information is a consequence of the asymmetry in the number of $Z$-type and $X$-type stabilizers in the Shor code's stabilizer generators. Consider the $[[n^2,1,n]]$ code of order $n$. We then have $n(n-1)$ $Z$-type stabilizer generators, and $(n-1)$ $X$-type stabilizer generators. Thus, there are $n$ times more $Z$-type stabilizer generators than $X$-type stabilizer generators. This entails a much greater sensitivity to detect and correct $X$ errors (using $Z$-type stabilizers) than $Z$ errors. Furthermore, this discrepancy increases with increasing $n$. This reflects as lower local recovery thresholds for $X$-direction as compared to $Z$-direction, seen in Fig. \ref{threshold_range}.

\subsection{Probability of logical error for CSS codes}\label{css deph alpha}

Consider the scenario where the $n$ physical qubits that are used to encode a CSS code are subject to independent $X$ errors with probability $p_x$. To compute the logical $X$ error after optimal recovery of this CSS code, it suffices to compute how the distinguishability of encoded $Z$ eigenstates $\{\ket{0},\ket{1}\}$ reduces after being impacted by said i.i.d. bit-flip noise. To this end, recall that it suffices to consider the simpler dephased version of the CSS code (see Sec. \ref{deph tree mapping sec} and Appendix Sec. \ref{CSS dephased subsec} for details). 
So, to estimate the $\alpha_x$ function (i.e., the logical $X$ error), we shall first compute the total variation distance (TVD) between the output distributions of our $Z$-dephased CSS encoders along with local bit-flip noise of probability $p_x=p$, given inputs $0$ and $1$. This is because the optimum Helstrom measurement will achieve the least error probability $\alpha_x(p)=(1-{\rm TVD(p)})/2$.

Firstly, we indicate the distributions we obtain given input $0$ and $1$, in density matrix notation, as a convex mixture,
\begin{align}
    \rho_0 &= \frac{1}{|C_2|}\sum_{x \in C_2} \ketbra{x}{x}\\
    \rho_1 &= \frac{1}{|C_2|}\sum_{x \in C_2} \ketbra{x+\mathbf{x}_L}{x+\mathbf{x}_L}
\end{align}
where $\mathbf{x}_L$ is the bitstring form of $X_L$ operator.  
 These are the distributions obtained by fully dephasing the two logical states encoded in the $n$ physical qubits (refer to Sec. \ref{deph tree mapping sec} to see why dephasing does not affect their distinguishability).    At this point, these two distributions -- or states -- are fully distinguishable. Now we apply local bit-flip noise $\mathcal{N}_p(\rho)=(1-p)\rho + p X\rho X$, on the above two states and measure their distinguishability. That is, we compute,
\begin{align}
    &{\rm TVD}(p)=||\mathcal{N}_p^{\otimes n}(\rho_0)-\mathcal{N}_p^{\otimes n}(\rho_1)||_1
\end{align}
We compute one of these states,
\begin{align}
    &\mathcal{N}_p^{\otimes n}(\rho_0)\\
    =& \frac{1}{|C_2|} \sum_{x \in C_2} \sum_{y \in \{0,1\}^n} p^{|y|}(1-p)^{n-|y|} \ketbra{x+y}{x+y}\\
    =&  \frac{1}{|C_2|} \sum_{x,y \in \{0,1\}^n} p^{|y|}(1-p)^{n-|y|} \delta(x \in C_2) \ketbra{x+y}{x+y}\\
    =& \frac{1}{|C_2|} \sum_{z\in \{0,1\}^n} \Big(\sum_{y\in \{0,1\}^n} p^{|y|}(1-p)^{n-|y|} \delta(z-y \in C_2) \Big) \ketbra{z}{z} 
\end{align}
where in the last step we simply changed variables. Now, $\delta(z-y \in C_2)$ is the restriction that $z$ and $y$ belong to the same additive coset of $C_2$ in $\{0,1\}^n$. We shall indicate that coset by $[C_2+z]$ for a given $z$. Thus, we can re-write the above expression as,
\begin{align}
    = \frac{1}{|C_2|} \sum_{z\in \{0,1\}^n} \Big( \sum_{y\in [C_2+z]} p^{|y|}(1-p)^{n-|y|}\Big) \ketbra{z}{z} 
\end{align}
Define $\sigma_p([C_2+z]):=\sum_{y\in [C_2+z]} p^{|y|}(1-p)^{n-|y|}$, which is a non-negative number associated to each additive coset of $C_2$. This gives,
\begin{align}
    \rho_0 = \frac{1}{|C_2|} \sum_{z \in \{0,1\}^n} \sigma_p([C_2 +z])\ \ketbra{z}{z}
\end{align}

We obtain $\rho_1$ by noticing that the distribution is $\rho_0$ plus a shift by $\mathbf{x}_L$.
\begin{align}
    \rho_1 = \frac{1}{|C_2|} \sum_{z \in \{0,1\}^n} \sigma_{p}([C_2 +z+\mathbf{x}_L])\ \ketbra{z}{z}
\end{align}
Thus, we have,
\begin{align}\label{TVD z exp}
    &||\mathcal{N}_p^{\otimes n}(\rho_0)-\mathcal{N}_p^{\otimes n}(\rho_1)||_1
    = \frac{1}{2|C_2|} \sum_{z\in \{0,1\}^n}|\sigma_p([C_2+z]) - \sigma_{p}([C_2+z+\mathbf{x}_L])|.
\end{align}
We shall refine our notation to highlight only the important objects. Let $\mathcal{C}$ indicate the set of cosets of $C_2$ in $\{0,1\}^n$, and $K \in \mathcal{C}$ be a member of that set, i.e., a specific coset. Define, $$\sigma_p(K)=\sum_{y \in K} p^{|y|}(1-p)^{n-|y|},$$ where the summation is over all bitstrings belonging to the coset $K$. Noticing that $\sigma_p(.)$ is defined for entire cosets, we partition our $z$ summation in terms of cosets and sum over these cosets and see that,
\begin{align}
    &||\mathcal{N}_p^{\otimes n}(\rho_0)-\mathcal{N}_p^{\otimes n}(\rho_1)||_1
    =\frac{1}{2|C_2|} \sum_{K \in \mathcal{C}} |C_2| \cdot |\sigma_p(K) - \sigma_{p}(K+\mathbf{x}_L)|,
\end{align}
where we used the fact that every additive coset of $C_2$ has the same size $|C_2|$. So, in terms of cosets $K \in \mathcal{C}$, we have,
\begin{align}
    &||\mathcal{N}_p^{\otimes n}(\rho_0)-\mathcal{N}_p^{\otimes n}(\rho_1)||_1= \frac{1}{2}\sum_{K \in \mathcal{C}} |\sigma_p(K) - \sigma_{p}(K+\mathbf{x}_L)|,
\end{align}
where $\sigma_p(K)$ is understood as the probability that any bitstring in the coset $K$ is sampled by the bit-flip noise process with bit-flip probability $p$. For the Helstrom bound saturating measurement, the probability of error is (1-TVD)/2, which is the $\alpha$ value. Thus,
$$\alpha_x(p)=\frac12 - \frac{1}{4}\sum_{K \in \mathcal{C}} |\sigma_p(K) - \sigma_{p}(K+\mathbf{x}_L)|.$$

An important special case are CSS codes for whom $X^{\otimes n}$ is a logical $Z$ operator, i.e., $\mathbf{x}_L=\textbf{1}^n$. In that case, 
$$\alpha_x(p)=\frac12 - \frac{1}{4}\sum_{K \in \mathcal{C}} |\sigma_p(K) - \sigma_{p}(K+\mathbf{1}^n)| = \frac12 - \frac{1}{4}\sum_{K \in \mathcal{C}} |\sigma_{p}(K) - \sigma_{1-p}(K)|.$$

\subsection{Example: Steane-7 Code Coset structure and $\alpha$ function}
The Steane-7 code is self-dual. Thus, its performance with respect to $X$ and $Z$ errors is exactly the same when considering independent $X$ and $Z$ error models. So we drop the $x$ or $z$ subscript that labels the logical error functions $\alpha$. So we compute $\alpha_{\rm Steane-7}(p)$, where $p$ is the probability of (say) independent $X$ errors, in the following.

Computing $\sigma_p$ for different cosets comes down to knowing the distribution of weights in the different cosets. We know this weight distribution exactly for the Steane-7 code. We explicitly compute this in the following. 

The $C_2$ for the Steane-7 Code corresponds to the dual of the $[7,4,3]$ Hamming code that is used to construct the Steane-7 code. We see that the 3 generators of $C_2$ in {$\mathbb{F}_2^7$} are, 

\begin{center}
$
\begin{pmatrix}
1 & 0 & 1 & 0 & 1 & 0 & 1\\
0 & 1 & 1 & 0 & 0 & 1 & 1\\
0 & 0 & 0 & 1 & 1 & 1 & 1
\end{pmatrix}
$
\end{center}
{The subspace spanned by these vectors contain $2^3=8$ vectors}
\begin{center}
$
\begin{pmatrix}
0 & 0 & 0 & 0 & 0 & 0 & 0\\
1 & 0 & 1 & 0 & 1 & 0 & 1\\
0 & 1 & 1 & 0 & 0 & 1 & 1\\
0 & 0 & 0 & 1 & 1 & 1 & 1\\
0 & 1 & 1 & 1 & 1 & 0 & 0\\
1 & 1 & 0 & 0 & 1 & 1 & 0\\
1 & 0 & 1 & 1 & 0 & 1 & 0\\
1 & 1 & 0 & 1 & 0 & 0 & 1\\
\end{pmatrix}
$
\end{center}
We see the weight distribution of $C_2$ is $(0,4,4,4,4,4,4,4)$, where {each number is the weight of one of these eight vectors}. Next, we have $2^7/2^3 -1 = 2^4 -1=15$ cosets of $C_2$. One coset is simply the complement of $C_2$, i.e., $C_2 + \textbf{1}^7$, with the weights $(7,3,3,3,3,3,3)$. We can compute $7$ other cosets by adding vectors like $(1,0,0,0,0,0,0)$, $(0,1,0,0,0,0,0)$, and so on to the $8$ vectors in $C_2$. Adding such a unit vector to any vector can only change the weight by $\pm 1$. Additionally observe that each column has exactly four $1$'s. Thus, adding a unit vector will ensure the total weight remains the same. Thus, the weights are $(1,3,3,3,3,5,5,5)$. The remaining $7$ cosets are the complement of the above mentioned $7$; Thus, have weight distribution $(6,4,4,4,4,2,2,2)$. We can use this to compute $\sigma_p(K)$ for each coset as a function of $p$, and thus, its $\alpha$ function which is,
\begin{align}
    \alpha_{\rm Steane-7}(p)=21p^2 - 98p^3 + 210p^4 - 252p^5 + 168p^6 -48p^7.
\end{align}
This matches the result in Sec. VI in \cite{dohertyconcatenated}.
\newpage

\section{$d=2$ code tree error analysis}\label{d=2 appendix}

\subsection{Derivation of Proposition \ref{prop: d=2 tree}}\label{prop d=2 deriv}
In this section, we 
present an error analysis for 
the recursive decoder introduced in Sec.\ref{subsec: one bit local recovery} and prove proposition \ref{prop: d=2 tree}. Recall that in this proposition we focus on codes with distance 
$d=2$. In the statement of the proposition, we denote the total probability of error in the Pauli channel $\mathcal{N}$ by  $p_{\text{tot}}=p_x+p_y+p_z$.  For convenience, in this section, we use the more compact notation $p$ for $p_{\text{tot}}$.

Consider the decoding between level $t-1$ and $t$ (see Fig \ref{recursivetree2}). The decoder module (detailed in Sec.(\ref{subsec: one bit local recovery})) acts on $b$ qubits along with their reliability bits $\{m_i(t-1)\}_{i \in [b]}$. We refer to these $b$ input qubits as the \textit{block}. The output of the decoder is a single qubit along with its reliability bit $m(t)$. We  define a bit $e(t)$, which takes value  $e(t)=1$  when the output qubit has no error, and $e(t)=1$ if it does. In other words, $e(t)$ indicates whether there is a logical error after recursive decoding of $t$ levels.  Thus, we can define the probabilities,
\begin{align}
    {\text{Probability of a qubit being marked: }}\mu_t &= {\rm Pr}(m(t)=1),\\
    {\text{Probability of an unmarked error: }}\delta_t &= {\rm Pr}(e(t)=1\ {\rm and}\ m(t)=0).
\end{align}
Note that these probabilities are identical for all the blocks at the same level. We show in Sec. \ref{rec deriv} that these probabilities satisfy the coupled recursive relations
\begin{align}\label{mu rec}
    \mu_t &\leq b(\delta_{t-1}+p) + b^2\mu_{t-1}(\delta_{t-1}+p) + b^2\mu^2_{t-1}
\end{align}
and
\begin{align}\label{delta rec}
    \delta_t &\leq b^2 (\delta_{t-1}+p)^2 + b^2\mu_{t-1} (\delta_{t-1}+p)\ ,
\end{align}
with the  initial conditions $\mu_0=0$ and $\delta_0=0$. Recursive application of these coupled  equations allows us to upperbound the error after $t$ rounds of recursive decoding. We show in Sec \ref{fixed point setup} that when $p<1/\lambda\equiv 1/(16b^4+4b^2)$, these quantities satisfy
\begin{align}
    \mu_t &\leq 4b^2\times (p+ \lambda p^2) \leq 8b^2p,\\
    \delta_t &\leq \lambda p^2.
\end{align}
for all $t$. For the second upperbound for $\mu_t$ we simply notice that for $p<1/\lambda$, we have $p\lambda<1$.\\

Ignoring the reliability bit and just focusing on the total probability of logical error, one finds
\begin{align*}
{\rm Pr}(e(t)=1)
&= {\rm Pr}(e(t)=1\ {\rm and}\ m(t)=0) +  {\rm Pr}(e(t)=1| m(t)=1)\times {\rm Pr}(m(t)=1)\\
&\leq \delta_t + \mu_t,
\end{align*}
where we have applied Bayes' rule. Combining this with the tighter bounds above we conclude that for all $t$, 
\begin{align}\label{d=2 upperbound result}
    {\rm Pr}(e(t)=1) &\leq 4b^2p + (1+4b^2) \times \lambda p^2\\
    &\leq (1+8b^2) \times p
\end{align}
when $p<1/\lambda\equiv 1/(16b^4+4b^2)$. In the second line we used that $p\lambda<1$. This proves Proposition \ref{prop: d=2 tree}.\\

\subsection{Derivation of recursive relations Eq.(\ref{mu rec}) and Eq.(\ref{delta rec})}\label{rec deriv}

\begin{figure}
    \centering
    \includegraphics[width=0.2\textwidth]{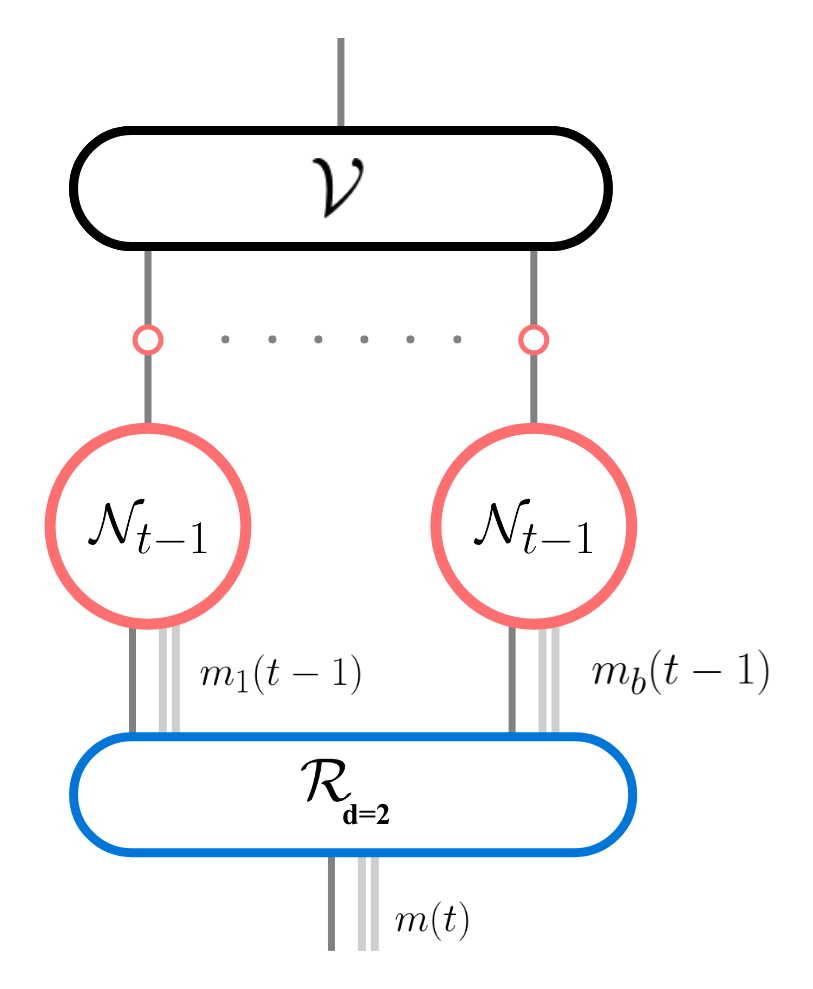}
    \caption{This figure illustrates the recovery of a single block at level $t$ during recursion. In this figure, let $\mathcal{N}_{t-1}$ denote the effective Pauli channel we obtain after recursive decoding of $t-1$ levels. The output of $\mathcal{N}_{t-1}$ is the qubit after $t-1$ levels of recovery along with its reliability bit $m_i(t-1)$, indexed $i$. A total of $b$ such channels, along with the physical error in the encoder (indicated as the small unlabeled red dots) form the input for the decoder $\mathcal{R}_{\rm d=2}$ at level $t$ of recovery.}
    \label{fig: d=2 recovery level t}
\end{figure}
In this section, we prove the recursive relations.  Since  all blocks of $b$ qubits in level $t$ (or $t-1$) are identical, we focus on the update of one particular block. 

We make some remarks regarding the error probabilities within a block: Consider a block of $b$ qubits received by the decoder module at level $t$ (see Fig.\ref{recursivetree2}). 
From the point of this decoder, each qubit may contain unmarked errors, caused by the following two random independent processes: (i) unmarked error in the output of the decoder at level $t-1$, which has probability $\delta_{t-1}$ or (ii) the physical error in the encoder at level $t$, which has probability $p$ (see Fig. \ref{fig: d=2 recovery level t}). Then, union bound implies that  the probability of an unmarked error on each input qubit received by the decoder module is, at most,  $\delta_{t-1}+p$.  Thanks to the tree structure, the event of the $i^{\rm th}$ qubit being marked (or containing unmarked errors) is independent of similar events for $j^{\rm th}$ qubit.  Finally, when we say a detected error, it means that the error produced a non-zero syndrome. We also repeatedly use the very useful bound $(1-x)^{n}\geq1-n x$ when $0\le x\le 1$ and $b\geq1$.
\subsubsection*{Derivation of Eq.(\ref{mu rec})}

Now we derive the bound for $\mu_t={\rm Pr}(m(t)=1)$. We note that if there is only one marked qubit  and none of the other qubits contain errors, either no syndrome will be observed, or the observed syndrome will be consistent with the marked qubit. In both cases, the algorithm sets $m(t)=0$. Therefore, the events that cause sending a message to the next level, i.e., set $m(t)=1$ are a subset of the union of the following events:

\begin{itemize}
\item (A) No qubit is marked and, at least, one of the qubits contains an error, which happens with probability bounded by $1-(1-[p+\delta_{t-1}])^b\le b\times [p+\delta_{t-1}] $.

\item (B) More than one qubit is marked, which happens with probability bounded by
\bes\label{FF}
\begin{align}
1 - b\mu_{t-1}(1-\mu_{t-1})^{b-1} - (1-\mu_{t-1})^{b} &=   1 - (1-\mu_{t-1})^{b-1}[1+(b-1)\mu_{t-1}]\\ &\le 1 - (1-(b-1)\mu_{t-1}) (1+(b-1)\mu_{t-1}) \\ &=(b-1)^2\times \mu_{t-1}^2\ ,
\end{align}
\ees
where to get the second line we have used $(1-x)^{b-1}\geq1-(b-1)x$.
\item (C) A qubit is marked and other unmarked qubits contain errors. To find this probability, consider the special case  in which qubit $r=1$ is marked whereas qubits $2,\cdots, b$ are unmarked and contain errors. The probability of this event is bounded by $\mu_{t-1}\times (b-1)\times (p+\delta_{t-1})$. This joint probability is the same for all other configurations $r=2,\cdots, b$. Therefore, the total probability of event C is bounded by $b(b-1)\times \mu_{t-1}\times (p+\delta_{t-1})$.   
\end{itemize}

Therefore, applying the union bound we find
\begin{align}
    \mu_t=\text{Pr}(m(t)=1)\le \text{Pr}(A\cup B\cup C)&\le \text{Pr}(A)+\text{Pr}(B)+\text{Pr}(C)\\ &\le  b[\delta_{t-1} +p] + (b-1)^2\times \mu^2_{t-1}+b(b-1)\times \mu_{t-1}\times (p+\delta_{t-1})\ .
\end{align}
Eq.(\ref{mu rec}) is a weaker version of this bound, obtained by replacing $b-1$ with $b$.

\subsubsection*{Derivation of Eq.(\ref{delta rec})}
Next, we derive the bound for $\delta_t={\rm Pr}(e(t)=1\ {\rm and}\ m(t)=0)$. Recall that the event where the algorithm sets  $m(t)=0$ and yet $e(t)=1$ are subsets of the following events:  (C) a qubit is marked and other unmarked qubits contain errors.  (D) no syndrome is observed, i.e., $\textbf{s}=\textbf{0}$, and none of the qubits is marked and yet there is an error. 
 We bound the probability of each event  separately.


\begin{itemize}
  \item (C) See above.
  \item (D) Since the code has distance $d=2$, it detects single-qubit errors. Therefore, to have an error $e(t)=1$ without observing any syndrome     ($\textbf{s}=\bf{0}$), there should be, at least, two errors. The probability that a qubit is unmarked and contains an error is bounded by $p+\delta_{t-1}$. Therefore, the probability of this event is bounded by 
\be
1- b[\delta_{t-1}+p](1-[\delta_{t-1}+p])^{b-1} - (1-[\delta_{t-1}+p])^{b}\le (b-1)^2\times (\delta_{t-1}+p)^2\ , 
\ee
where the derivation is similar to Eq.(\ref{FF}). 
\end{itemize}

Then, applying the union bound we find
\begin{align}
    \delta_t= {\rm Pr}(e(t)=1\ {\rm and}\ m(t)=0) \leq  \text{Pr}(C\cup D)\le \text{Pr}(C)+\text{Pr}(D)\le  (b-1)^2\times (\delta_{t-1}+p)^2+b(b-1)\times \mu_{t-1} (p+\delta_{t-1}) \ .
\end{align}
Eq.(\ref{delta rec}) is a weaker version of this bound, obtained by replacing $b-1$ with $b$.

\subsection{Analysis of recursive relations Eq.(\ref{mu rec}) and Eq.(\ref{delta rec})}\label{fixed point setup}

Here,  we analyze the recursive relations Eq.(\ref{mu rec}) and Eq.(\ref{delta rec}), namely
\begin{align}\nonumber
    \mu_t &\leq b(\delta_{t-1}+p) + b^2\mu_{t-1}(\delta_{t-1}+p) + b^2\mu^2_{t-1}\ ,\\
    \delta_t &\leq b^2 (\delta_{t-1}+p)^2 + b^2\mu_{t-1} (\delta_{t-1}+p)\ .\nonumber
\end{align}
In particular, the goal is to find upper bounds on $\mu_t$ and $\delta_t$ that hold for arbitrary large $t$. To simplify these equations, we substitute $\epsilon_t=\delta_{t}+p$ to obtain,
\begin{align}
    \mu_t &\leq b\epsilon_{t-1} + b^2\mu_{t-1}\epsilon_{t-1} + b^2\mu^2_{t-1},\\
    \epsilon_t &\leq p+ b^2 \epsilon_{t-1}^2 + b^2\mu_{t-1} \epsilon_{t-1}.
\end{align}
Also, for convenience, in the following we use the slightly weaker bounds
\bes
\begin{align}\label{coupled}
    \mu_t &\leq 2b^2\epsilon_{t-1} + b^2\mu^2_{t-1}=f(\mu_{t-1},\epsilon_{t-1}) ,\\
    \epsilon_t &\leq p+ b^2 \epsilon_{t-1}^2 + b^2\mu_{t-1} \epsilon_{t-1}=g(\mu_{t-1},\epsilon_{t-1})\ \ , 
\end{align}
\ees
where we have applied condition $\mu_t\leq1$, 
and defined
\begin{align}
   f(x, y)&:=   b^2x^2+2b^2 y\ ,\\
    g(x, y) &:= p+ b^2 y^2 + b^2 x y\ .
\end{align}
To understand the behavior of these recursive inequalities 
we focus on the following recursive equalities
\bes\label{couple2}
\begin{align}
    x_t &= f(x_{t-1}, y_{t-1})= 2b^2y_{t-1} +  b^2x^2_{t-1}\ , \\ 
    y_t &= g(x_{t-1}, y_{t-1}) = p+ b^2 y^2_{t-1} + b^2 x_{t-1} y_{t-1}\ ,
\end{align}
\ees
with the initial condition 
\be\label{initial}
x_0=\mu_0=0\ , \ \ \  y_0=\epsilon_0=p>0\ .
\ee
We note that because for $x,y\ge 0$  functions $f$ and $g$ are monotone in both variables, then applying bounds in Eq.(\ref{coupled}) recursively imply that 
\be
\mu_t\le x_t \ , \ \ \epsilon_t\le y_t \ .
\ee
Suppose $x_\infty, y_\infty >0$ are a  fixed point of Eq.(\ref{couple2}), i.e., they satisfy
\begin{align}\label{fixed pt 1}
    x_\infty &= f(x_\infty, y_\infty) = 2b^2y_\infty + b^2x^2_\infty,
\end{align}
and
\begin{align}\label{fixed pt 2}
    y_\infty&= g(x_\infty, y_\infty)= p + b^2 y^2_\infty + b^2 x_\infty y_\infty\ .
\end{align}
Then, again the monotonicity of functions $f$ and $g$ imply that if $x_0\leq x_\infty$ and $y_0\leq y_\infty$, then $x_t\leq x_\infty$ and $y_t \leq y_\infty$ for all $t$, i.e., 
\begin{align}
    x_t &= f(x_{t-1},y_{t-1}) \leq f(x_\infty,y_\infty) = x_\infty,\\
    y_t &= g(x_{t-1},y_{t-1}) \leq g(x_\infty,y_\infty) = y_\infty\ .
\end{align}
Therefore, if the initial condition in Eq.(\ref{initial}) satisfies the constraint  $  y_0=\epsilon_0=p\le y_\infty$, it is guaranteed that for all $t$, 
\begin{align}
\epsilon_t &\le y_t\le y_\infty \ ,\\ 
\mu_t &\le x_t\le x_\infty\ .
\end{align}

In the following subsection, we show that for $p<1/(16b^4+4b^2)$, there is a fixed point solution that satisfies 
\begin{align}
    2b^2p\leq & \ x_\infty \leq 4b^2p+ 4b^2(16b^4+4b^2)p^2,\\
    p \leq & \ y_\infty \leq p + (16b^4+4b^2)p^2.
\end{align}
The initial conditions, $x_0=0$ and $y_0=p$, clearly satisfy $x_0\leq x_\infty$ and $y_0\leq y_\infty$. So we conclude that $x_t\leq x_\infty$ and $y_t \leq y_\infty$ for all $t$ (this also justifies the use of the subscript $\infty$). Recalling that these were upperbounds, we get that when $p<1/(16b^4+4b^2)$,
\begin{align}
    \mu_t &\leq 4b^2p+ 4b^2(16b^4+4b^2)p^2,\\
    \epsilon_t &\leq p + (16b^4+4b^2)p^2.
\end{align}
for all $t$.   Undoing the substitution $\epsilon_t=\delta_t + p$, we conclude that when $p<1/\lambda \equiv 1/(16b^4+4b^2)$,
\begin{align}
    \mu_t &\leq 4b^2\times (p+ \lambda p^2),\\
    \delta_t &\leq \lambda p^2.
\end{align}
for all $t$.

\begin{figure}
    \centering
\includegraphics[width=0.5\textwidth]{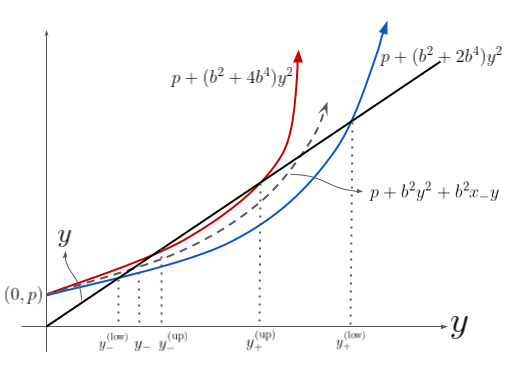}
    \caption{Schematic for intersection of the line $y$ (black) with the functions $p + (b^2+2b^4)y^2$ (blue), $p + b^2 y^2 + b^2 x_{-} y$ (dashed), and $p + (b^2+4b^4)y^2$ (red).}
    \label{fig parabola}
\end{figure}

\subsubsection{$x_\infty$ and $y_\infty$ computation}\label{fixed point comp}
Now we compute a specific solution for $x_\infty$ and $y_\infty$ that satisfies Eq.(\ref{fixed pt 1}) and Eq.(\ref{fixed pt 2}). We drop the $\infty$ subscript for clarity. Note that in general, these two equations will have 4 solutions. It suffices to estimate one real non-negative solution. We first solve Eq.(\ref{fixed pt 1}) for $x$ to get,
\begin{align}\label{x pm}
    x_{\pm}=\frac{1\pm \sqrt{1-8b^4y}}{2b^2}
\end{align}
For real solutions, we require $y<1/8b^4$. Then using the useful bounds
\begin{align}\label{sqrt bound}
    1-a \leq \sqrt{1-a} \leq 1-a/2\ \ {\rm when}\ 0\leq a \leq 1,
\end{align}
we get that the two solutions are bounded as,
\begin{align}
    2b^2y \leq & \ x_- \leq 4b^2y,\\
    \frac{1-4b^4y}{b^2} \leq & \ x_+ \leq \frac{1-2b^4y}{b^2}.
\end{align}
Observe that the `$+$' solution predicts that if $y=0$, then $x=1/b^2$. Whereas the `$-$' predicts that if $y=0$, then $x=0$. We argued earlier that any non-negative set of fixed points would provide a valid upperbound. So, we may as well choose the tighter upperbound, i.e., the `$-$' solution.
Now we solve Eq.(\ref{fixed pt 2}) after substituting $x_{-}$ in it. Using the bounds for $x_{-}$,  we get that the RHS of Eq.(\ref{fixed pt 2}) is bounded as
\begin{align}
     p + (b^2+2b^4)y^2 \leq p + b^2 y^2 + b^2 x_{-} y \leq p + (b^2+4b^4)y^2.
\end{align}
Eq.(\ref{fixed pt 2}) is guaranteed a solution if the line $y$ intersects the upperbound $p + (b^2+4b^4)y^2$. So, we are guaranteed an intersection (i.e., a solution) when $p<1/(16b^4+4b^2)$.

Now we study these solutions fixing such a $p<1/(16b^4+4b^2)$. Observe that the lower and upper bounds are convex parabolas that pass through $(0,p)$. So the line $y$ will intersect each parabola twice. We shall label the solutions for the lower parabola $y^{\rm (low)}_{\pm}$ and the solutions for the upper parabola $y^{\rm (up)}_{\pm}$ (Fig. (\ref{fig parabola}) clarifies this). Observe that there \textit{will} be an intersection between the line $y$ and $p + b^2 y^2 + b^2 x_{-} y$ (that we denote by $y_{-}$) between $y^{\rm (low)}_{-}$ and $y^{\rm (up)}_{-}$. Now we estimate these two:
\begin{align}
    y^{\rm (up)}_{-} &= \frac{1-\sqrt{1-(16b^4+4b^2)p}}{8b^4+2b^2} \leq p + (16b^4+4b^2)p^2\\
    y^{\rm (low)}_{-} &= \frac{1-\sqrt{1-(8b^4+4b^2)p}}{4b^4+2b^2} \geq p.
\end{align}

Thus, we have $p\leq y_{-} \leq p + (16b^4+4b^2)p^2$ when $p<1/(16b^4+4b^2)$. Now we take this value of $y_{-}$ and substitute it back into Eq.(\ref{x pm}) to compute $x_{-}$. From the bounds we saw for $x_{-}$, we conclude that $x_{-}\leq 4b^2y_{-} \leq 4b^2p+ 4b^2(16b^4+4b^2)p^2$ and $x_{-}\geq 2b^2y_{-} \geq 2b^2p$. Furthermore, for the regime of $p$ we are considering, this upperbound for $y_{-}$ (i.e., $2 \times 1/(16b^4+4b^2)$) is consistent with the condition for real solutions for $x_{-}$ in Eq.(\ref{x pm}) which was $y_{-}<1/8b^4$. 

Consolidating all our bounds (and re-introducing the $\infty$ subscript), we finally have that when $p<1/(16b^4+4b^2)$,
\begin{align}
    2b^2p\leq & \ x_\infty \leq 4b^2p+ 16(4b^6+b^4)p^2,\\
    p \leq & \ y_\infty \leq p + 4(4b^4+b^2)p^2.
\end{align}

\newpage

\section{Error Analysis of Bell Tree}\label{appendix: bell tree}

\begin{algorithm}[H]
\caption{Decoding Module with two reliability bits}\label{alg:local rec with 2 bits}
\begin{algorithmic}
\State \textbf{begin} with $2$ qubits and their ordered \textit{reliability} bit pairs $(m_i^{\rm rel}(t-1),m_i^{\rm irrel}(t-1))$ where $i=1,2$ in a block in $t-1$\\

\If{$m^{\rm irrel}_i(t-1)=0, \ \forall i$}
    \State $m^{\rm irrel}(t)\leftarrow0$
\Else \State $m^{\rm irrel}(t)\leftarrow1$
\EndIf\\

\State \textbf{measure} $Z_1Z_2$-\textbf{syndrome} $s$ on the $2$ qubits
\If{$m^{\rm rel}_i(t-1)=0, \ \forall i$}
    \If{$s=0$} \State $m^{\rm rel}(t)\leftarrow0$
    \Else \State $m(t)\leftarrow1$
    \EndIf
\ElsIf{$m^{\rm rel}_k(t-1)=1$ for only one of the two qubits indexed $i=1,2$ that is denoted $k$}
    \If{$s=0$ } \State $m^{\rm rel}(t)\leftarrow0$
    \Else \State \textbf{correct} $k^{\rm th}$ qubit
    \State $m^{\rm rel}(t)\leftarrow0$
    \EndIf
\Else \State $m^{\rm rel}(t)\leftarrow1$
\EndIf\\

\State \textbf{Swap} bit pair: $(m^{\rm rel}(t),m^{\rm irrel}(t)) \leftarrow (m^{\rm irrel}(t),m^{\rm rel}(t))$\\
\State output qubit $\leftarrow$ decoding of $2$ qubits in the block
\State \textbf{apply hadamard} to output qubit\\

\Return output qubit and ordered bit pair $(m^{\rm rel}(t),m^{\rm irrel}(t))$ to $t$ level.
\end{algorithmic}
\end{algorithm}

In this section, we analyze the performance of the recursive decoding procedure described in Algorithm \ref{alg:local rec with 2 bits} to analytically prove the existence of a noise threshold for the Bell tree by lowerbounding it. For simplicity, we assume $p_x=p_z=p$ for the physical independent $X$ and $Z$ errors in the tree. We show that when $p\leq 0.245\%$, then ${\rm Pr}(e^x(t)=1)\leq 53p$ and ${\rm Pr}(e^z(t)=1)\leq 25.5p$ for all even $t$.\\

\subsection{Decoding at level $t$}

Consider the decoding of a single block from level $t-1$ to $t$. The decoder module (detailed in Sec.(\ref{subsec: 2 bit local rec})) acts on $2$ qubits along with their reliability bit pairs $\{(m^{\rm rel}_i(t-1),m^{\rm irrel}_i(t-1))\}_{i \in [2]}$. We refer to these $2$ input qubits as the \textit{block}. The output of the decoder is a single qubit along with its reliability bit pair $(m^{\rm rel}(t),m^{\rm irrel}(t))$. We define a bit $e^x(t)$, which takes value $e^x(t)=0$ when the output qubit has no $X$ error, and $e^x(t)=1$ if it does. In other words, $e^x(t)$ indicates whether there is a logical $X$ error after recursive decoding of $t$ levels. We similarly define $e^z(t)$ for $Z$ errors.  
Note that because we have the same encoder, noise rate, and decoder at every level of encoding and decoding, the joint probability distribution ${\rm Pr}(e^x_i(t),m^{\rm rel}_i(t),e^z_i(t),m^{\rm irrel}_i(t))$ is identical for every qubit $i$ in any given level. Therefore, similar to the analysis of $d=2$ codes, we suppress the index for qubits.

During the recovery of level $t$, each qubit may contain $X$ and $Z$ errors that are caused by the following two random independent processes: (i) $X$ and $Z$ errors in the output of the decoder at level $t-1$, or (ii) the independent $X$ and $Z$ physical error in the encoder at level $t$ which we assume has probability $p_x=p_x=p$ respectively. 

The decoder of a single level of the Bell tree described in Fig \ref{bell tree logic circuit} and Algorithm \ref{alg:local rec with 2 bits} acts independently on $X$ and $Z$ errors by partitioning the decoder into a \textit{relevant} and \textit{irrelevant} part respectively. This is because (i) the decoder is equipped with only one syndrome $Z\otimes Z$ that is sensitive to only $X$ errors; thus, $X$ errors are detected and/or corrected (`relevant'), while $Z$ errors are left unattended (`irrelevant'), (ii) the relevant classical bits are updated independent of irrelevant classical bits. However, due to the application of Hadamard after each decoding level, all $X$ errors become $Z$ errors, and \textit{vice versa}. To respect this switching of reference frame, we also swap the order of the reliability bits in the last step of decoding (see Algorithm \ref{alg:local rec with 2 bits}). So, during the decoding of level $t$, the quantities $(e^z(t),m^{\rm irrel}(t))$ are a function of only $\{(e^x_i(t-1),m^{\rm rel}_i(t-1))\}_{i=1,2}$ where $i$ indexes the two qubits in the block. Similarly,
$(e^x(t),m^{\rm rel}(t))$ are a function of only $\{(e^z_i(t-1),m^{\rm irrel}_i(t-1))\}_{i=1,2}$. If we recurse these dependencies to the level $t=0$, we get the sequence of dependencies,
\begin{align}
    &(e^z(t),m^{\rm irrel}(t)) \leftarrow (e^x(t-1),m^{\rm rel}(t-1)) \leftarrow (e^z(t-2),m^{\rm irrel}(t-2)) \leftarrow ... \\
    &(e^x(t),m^{\rm rel}(t)) \leftarrow (e^z(t-1),m^{\rm irrel}(t-1)) \leftarrow (e^x(t-2),m^{\rm rel}(t-2)) \leftarrow ...
\end{align}
all the way down to $t=0$ quantities. And initially (at $t=0$), all the reliability bits are initiated as $0$, and errors are uncorrelated. So, any error or message of either type at \textit{any} level that effects $(e^z(t),m^{\rm irrel}(t))$ will never affect $(e^x(t),m^{\rm rel}(t))$, and \textit{vice versa}. Furthermore, physical $X$ and $Z$ errors within the tree are sampled independently with probabilities $p_x$ and $p_z$ respectively. These facts ensure that for all decoding levels $t$, the joint probability of errors and reliability bits values in layer $t$ decomposes as
\begin{align}\label{bell error decomp}
    {\rm Pr}\Big(e^x(t),e^z(t), m^{\rm rel}(t), m^{\rm irrel}(t)\Big) = {\rm Pr}\big(e^{x}(t), m^{\rm rel}(t)\big) \times {\rm Pr}\big(e^z(t), m^{\rm irrel}(t)\big)\ .
\end{align}
Due to Eq.(\ref{bell error decomp}), we shall study the the update of ${\rm Pr}(e^{x}(t), m^{\rm rel}(t))$ and ${\rm Pr}(e^z(t), m^{\rm irrel}(t))$ independently.\\

Just as done in the introductory section of Appendix section \ref{prop d=2 deriv}, we define,
\begin{align}
    &\mu^x_{t}:= {\rm Pr}(m^{\rm rel}(t)=1),\\
    &\delta^x_{t}:= {\rm Pr}(e^x(t) =1 \ {\rm and} \ m^{\rm rel}(t)=0)\ ,
\end{align}
and,
\begin{align}
    &\mu^z_{t}:= {\rm Pr}(m^{\rm irrel}(t)=1),\\
    &\delta^z_{t}:= {\rm Pr}(e^z(t) =1 \ {\rm and} \ m^{\rm irrel}(t)=0)\ .
\end{align}
Because the qubit index can be suppressed, this allows us to define $\mu^x_{t-1}, \delta^x_{t-1}$ and $\mu^z_{t-1},\delta^z_{t-1}$ similarly. And from Bayes' rule, these definitions satisfy for all $t$,
\begin{align}
    {\rm Pr}(e^x(t)=1)\leq \mu^x_{t}+ \delta^x_{t},\\
    {\rm Pr}(e^z(t)=1)\leq \mu^z_{t}+ \delta^z_{t}.
\end{align}

We study the evolution of these quantities upon decoding at level $t$ in Sec \ref{binary tree rec relns} and Sec \ref{no qec rec appendix}, and show that
\begin{align}\label{d=2 QEC rec 1}
    \mu^z_t \ &\leq 2(\delta^x_{t-1}+p) + (\mu_{t-1}^x)^2,
\end{align}
and
\begin{align}\label{d=2 QEC rec 2}
    \delta^z_t \ &\leq  (\delta^x_{t-1}+p)^2 + 2\mu^x_{t-1} (\delta^x_{t-1}+p),
\end{align} 
whereas,
\begin{align}\label{d=2 no QEC rec 1}
    \mu^x_{t} \leq 2 \mu^z_{t-1},
\end{align}
and
\begin{align}\label{d=2 no QEC rec 2}
    &\delta^x_{t} \leq 2(\delta^z_{t-1}+p).
\end{align}

We shall combine this update from level $t-1 \rightarrow t \rightarrow t+1$ for $X$ errors by composing Eq.(\ref{d=2 QEC rec 1}) and (\ref{d=2 QEC rec 2}) with Eq.(\ref{d=2 no QEC rec 1}) and (\ref{d=2 no QEC rec 2}) respectively. This gives,
\begin{align}
    \mu^x_{t+1} &\leq 4(\delta^x_{t-1}+p)+2{(\mu^x_{t-1})}^2,\\
    \delta^x_{t+1} & \leq 2p+ 2(\delta^x_{t-1}+p)^2 + 4\mu^x_{t-1}(\delta^x_{t-1}+p),
\end{align}
where $\mu^x_0=\delta^x_0=\mu^z_0=\delta^z_0=0$. We  can solve for the asymptotics of this system in exactly the same way as presented in Appendix Sec (\ref{fixed point comp}). We compute that when $p<1/408\approx 0.245\%$, then ${\rm Pr}(e^x(t)=1)\leq 53p$ for all even $t$.

The above analysis can be modified for $Z$ error evolution if we simply compose Eq.(\ref{d=2 no QEC rec 1}) and (\ref{d=2 no QEC rec 2}) with Eq.(\ref{d=2 QEC rec 1}) and (\ref{d=2 QEC rec 2}) (i.e., in the opposite order). Redoing the above asymptotic analysis for this new recursive system yields the same threshold bound that when $p<1/408\approx 0.245\%$, then ${\rm Pr}(e^z(t)=1)\leq 25.5p$ for all even $t$.

\subsection{Derivation of Eq.(\ref{d=2 QEC rec 1}) and Eq.(\ref{d=2 QEC rec 2}) update}\label{binary tree rec relns}

Recall that the bit $e^z(t)$ indicates whether  the qubit at the output of the decoder module at level $t$ has a $Z$ error. Furthermore, $m^{\rm irrel}(t)$ is the corresponding reliability bit, whose value 1 signals the presence of a $Z$ error on this qubit. 
 The \emph{relevant part} of the decoder module described in Algorithm \ref{alg:local rec with 2 bits} measures stabilizer  $Z \otimes Z$ on the received input qubits, which allows the decoder to detect/correct $X$ errors.    
 Since in the end, the decoder module applies a Hadamard gate, $X$ errors  become $Z$ errors. 
 In conclusion, the ``relevant" part of the decoding module determines how 
$e^x(t-1)$ and $m^{\rm rel}(t-1)$ for two input qubits will be transformed to $e^z(t)$ and $m^{\rm irrel}(t)$ for the output qubit. 

 
 As we saw before, the probability of an unmarked $X$ error on qubit at the input of a decoder  in level $t$ is at most $\delta^x_{t-1}+p$. The events that can result in  $e^z(t)=1$ or $m^{\rm irrel}(t)=1$ are a subset of the following 4 events which are specific cases of the events we considered in our analysis of algorithm  \ref{alg:local rec with bit}\footnote{It is worth noting the close connection with our error analysis of 
 algorithm \ref{alg:local rec with bit}. 
 Namely,  Algorithm \ref{alg:local rec with 2 bits} acts exactly as Algorithm \ref{alg:local rec with bit} on the relevant part of the decoding of a Bell tree. That is, we are decoding a binary repetition code that protects against $X$-type errors. This is a $d=2$ code for $X$ errors that detects single qubit $X$ errors using the only stabilizer $Z \otimes Z$. Therefore, we can specialize the computation presented in Sec \ref{d=2 appendix} to a binary repetition code. Recall that the two qubits in the code suffer a bit-flip error independently with probability $p$. So in line with the presentation in Sec \ref{d=2 appendix}, we shall first rewrite the events (A), (B), (C) and (D) mentioned there for the decoding of a binary repetition code. And then, we shall use these to estimate the probabilities for $\mu^z_t$ and $\delta^z_t$.}:

\begin{itemize}
    \item (A): No qubit is marked and exactly one of the two qubits contains an $X$ error. The probability of this  error is bounded by $1-(1-[\delta^x_{t-1}+p])^2 \leq 2(\delta^x_{t-1}+p)$.
    \item (B): Both qubits are marked, which  happens with probability $({\mu^x_{t-1}})^2$.
    \item (C): A qubit is marked and the other qubit is unmarked and has an $X$ error. There are two such events for two qubits. So, the total  probability bounded by $2 \times \mu^x_{t-1}\times (\delta^x_{t-1}+p)$.
    \item (D): No qubit is marked and  both qubits have (unmarked)  $X$ errors. This happens with probability $(\delta^x_{t-1}+p)^2$. 
\end{itemize}

In particular, we have excluded the event that (i) none of the qubits has errors, and (ii) the event that one of them has an error and is marked and the other one does not have an error and is not marked. In the latter case, the error correction performed by the module takes care of the error and sets the reliability bit to zero. Also, note that in the case of event C
the decoder makes a wrong guess about the location of the error, corrects it, and set the reliability bit to zero. Therefore, this event does not contribute to $\mu_z(t)$ but  contributes to $\delta_z(t)$.  
 
We conclude that

\begin{align}
    \mu^z_t = {\rm Pr}(m^{\rm irrel}(t)=1)\leq \text{Pr}(A \cup B) &\leq \text{Pr}(A) + \text{Pr}(B)\\
    &\leq 2(\delta^x_{t-1}+p) + ({\mu^x_{t-1}})^2\ .
\end{align}
Next, we focus on the probability of unmarked $Z$ errors, $\delta^z_t$.  Since events in $A$ and $B$ set the reliability bit to 1, they do not contribute to $\delta^z_t$.  
On the other hand, events in (C) will result in an unmarked $Z$ error and therefore do contribute to  $\delta^z_t$. We conclude that
\begin{align}
    \delta^z_t={\rm Pr}(e^z(t)=1\ {\rm and}\ m^{\rm irrel}(t)=0)\leq \text{Pr}(C \cup D) &\leq \text{Pr}(C) + \text{Pr}(D)\\
    &\leq 2\mu^x_{t-1} (\delta^x_{t-1}+p)+(\delta^x_{t-1}+p)^2.
\end{align}

So combining all these considerations, we get
\begin{align}
\mu^z_t \ &\leq 2(\delta^x_{t-1}+p) + ({\mu^x_{t-1}})^2\\
    \delta^z_t \ &\leq  (\delta^x_{t-1}+p)^2 + 2\mu^x_{t-1} (\delta^x_{t-1}+p).
\end{align}

\subsection{Derivation of Eq.(\ref{d=2 no QEC rec 1}) and Eq.(\ref{d=2 no QEC rec 2}) update}\label{no qec rec appendix}
In this section, we compute the update of the error pertaining to the irrelevant part of Algorithm \ref{alg:local rec with 2 bits}. This entails decoding of unprotected $Z$ errors which become $X$ errors after decoding. Thus, the associated error bit at the output is $e^x(t)$ and its reliability bit is $m^{\rm rel}(t)$ due to the final swap. 

We first evaluate $\mu^x_{t}={\rm Pr}(m^{\rm rel}{(t)}=1)$. According to the algorithm, we mark a qubit in level $t$ if even one qubit in level $t-1$ is marked. And this is independent of any error because we cannot detect them, and, thus, we leave them unaddressed. Thus, $\mu^x_{t}={\rm Pr}(m^{\rm rel}{(t)}=1)= 1-(1-\mu^z_{t-1})^2 \leq 2\mu^z_{t-1}$.

Now consider $\delta^x_{t}={\rm Pr}(e^x{(t)}=1\ {\rm and} \ m^{\rm rel}{(t)}=0)$. According to the update rule, we require all $m^{\rm irrel}_i(t-1)=0$ to get $m^{\rm rel}(t)=0$, thus there will be no $\mu^z_t$ term. Furthermore, a single error, independent of messages received, causes a logical error in level $t$. Whereas, two qubit errors will cancel each other yielding no logical error. Thus, $\delta^x_{t} = {\rm Pr}(e^x(t)=1 \ {\rm and} \ m^{\rm irrel}{(t)}=0) \leq  2(\delta^z_{t-1}+p)$.

\subsection{Comparison with \textit{conservative} recursive decoder}
 We consider a more ``conservative" version of this decoder where after the correction the relevant reliability is set to $m^{\text{rel}}(t)=1$, indicating unreliability of the decoded qubit. Numerical results suggest that this modification reduces the threshold to around 
$\sim 0.3\%$, and generally results in a higher probability of logical error. The \textit{conservative} algorithm (Algorithm \ref{alg:cons local rec with 2 bits}) is described below, and the associated numerics are presented in Fig.(\ref{cons_fig}).

\begin{algorithm}[H]
\caption{\textit{Conservative} Decoding Module with two reliability bits}\label{alg:cons local rec with 2 bits}
\begin{algorithmic}
\State \textbf{begin} with $2$ qubits and their ordered \textit{reliability} bit pairs $(m_i^{\rm rel}(t-1),m_i^{\rm irrel}(t-1))$ where $i=1,2$ in a block in $t-1$\\

\If{$m^{\rm irrel}_i(t-1)=0, \ \forall i$}
    \State $m^{\rm irrel}(t)\leftarrow0$
\Else \State $m^{\rm irrel}(t)\leftarrow1$
\EndIf\\

\State \textbf{measure} $Z_1Z_2$-\textbf{syndrome} $s$ on the $2$ qubits
\If{$m^{\rm rel}_i(t-1)=0, \ \forall i$}
    \If{$s=0$} \State $m^{\rm rel}(t)\leftarrow0$
    \Else \State $m(t)\leftarrow1$
    \EndIf
\ElsIf{$m^{\rm rel}_k(t-1)=1$ for only one of the two qubits indexed $i=1,2$ that is denoted $k$}
    \If{$s=0$ } \State $m^{\rm rel}(t)\leftarrow0$
    \Else \State \textbf{correct} $k^{\rm th}$ qubit
    \State $m^{\rm rel}(t)\leftarrow1$ \#\textit{modified step}\#
    \EndIf
\Else \State $m^{\rm rel}(t)\leftarrow1$
\EndIf\\

\State \textbf{Swap} bit pair: $(m^{\rm rel}(t),m^{\rm irrel}(t)) \leftarrow (m^{\rm irrel}(t),m^{\rm rel}(t))$\\
\State output qubit $\leftarrow$ decoding of $2$ qubits in the block
\State \textbf{apply hadamard} to output qubit\\

\Return output qubit and ordered bit pair $(m^{\rm rel}(t),m^{\rm irrel}(t))$ to $t$ level.
\end{algorithmic}
\end{algorithm}

\begin{figure}[ht!]
\centering
\includegraphics[width=0.5\textwidth]{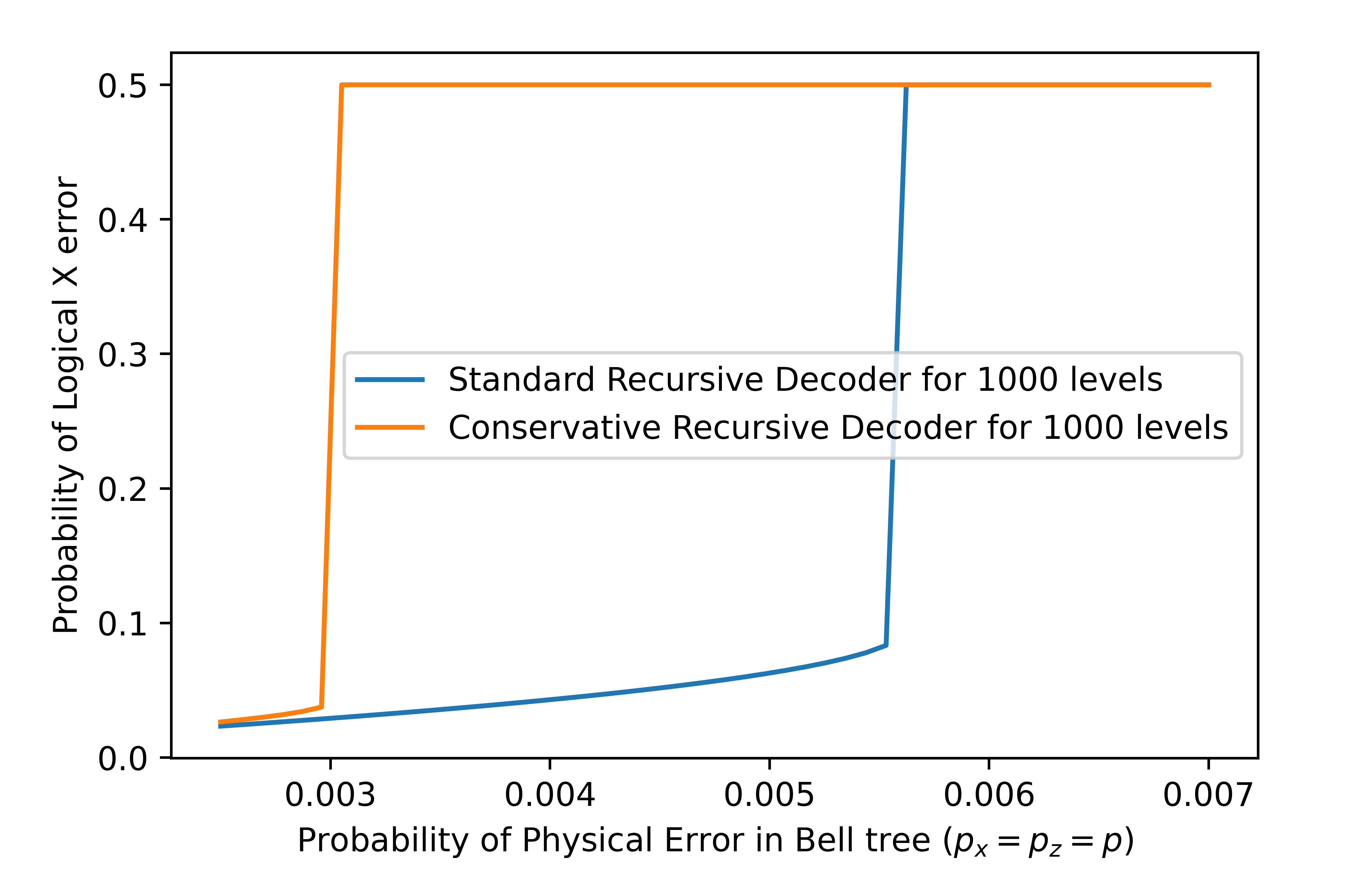} 
\caption{\textbf{Comparison between Standard Recursive decoder (Algorithm \ref{alg:local rec with 2 bits}) and Conservative Recursive decoder (Algorithm \ref{alg:cons local rec with 2 bits}) for the Bell Tree --} We observe that the conservative decoder performs worse than the standard decoder we considered in the main body of the paper. That is to say, the threshold is lower, i.e. $\sim 0.3\%$, and the logical error after correction is greater.}
\label{cons_fig}
\end{figure}

\newpage

\section{Optimal and Efficient Recovery of Noisy Concatenated Stabilizer codes}\label{BP Appendix}
In this section we detail the efficient optimal algorithm for recovering noisy stabilizer trees.  The exposition in this section is essentially in line with the presentation in \cite{poulin2006optimal} which considers concatenated codes, but without noise in the encoder.

\subsection{Clifford Encoders, Logical Errors and Syndromes}\label{clifford properties}
Consider a noisy stabilizer tree, defined in Sec. \ref{stab tree setup}, of the form,
\be\label{treee}
\mathcal{E}_{T}(\rho) = \mathcal{N}^{\otimes b^T} \circ \prod_{j=0}^{T-1}   \mathcal{V}^{\otimes b^j}\circ \mathcal{N}^{\otimes b^j}(\rho).
\ee
where $\mathcal{V}$ is a $1\rightarrow b$ qubit stabilizer isometry defined as $\mathcal{V}(\rho)=U[\rho \otimes \ketbra{0}{0}^{\otimes (b-1)}]U^\dagger$, $U$ being a Clifford unitary and $\mathcal{N}$ is a single qubit Pauli noise. Note that Eq.(\ref{treee}) has a noise before the first encoding, and also after the tree on the leaves. We redescribe the action of the decoder circuit of a general stabilizer code (which includes our tree channel) here: let error $E$ be an arbitrary tensor product of Pauli operators and the identity operator \footnote{In the context of trees/subtrees that have noise between encoders, $E$ represents the \textit{effective} error affecting the leaves that can be obtained (up to stabilizers) by commuting the local errors sampled on edges within the tree across encoders, down to its leaves.}. Then, 
\begin{align}
    U^\dagger E U=\mathcal{L}(E) \otimes \mathbf{X}^{{\rm Synd}(E)} \mathbf{Z}^\mathbf{z},
\end{align}
where $\mathcal{L}(E)$ is the logical operator acting on the logical qubit and $\mathbf{X}^{{\rm Synd}(E)}$ is the syndrome acting on the ancilla qubits. Since the ancilla is initially prepared in $\ket{0}^{\otimes b-1}$, it remains unchanged under $\mathbf{Z}^\mathbf{z}$. We measure and record the syndrome ${\rm Synd}(E)$, and thus infer the distribution $p(L|{\rm Synd}(E))$ over logical errors $L=\{I,X,Y,Z\}$ because, in general, the logical errors are correlated with the syndromes.

In the context of trees, Fig. \ref{BP_Circuit} clarifies this inversion and syndrome-acquisition procedure. Observe that this decoding circuit is essentially a `reflection' of the encoding tree. Therefore, we see that for a tree of depth $T$, we have a syndrome of length $b^{T}-1$, and thus, may need to construct and search a dictionary over $2^{b^{T}-1}$ syndrome strings, which is clearly untenable. But, thanks to the tree structure, there is a belief propagation algorithm that uses the syndrome data and computes this logical error distribution in $\mathcal{O}(b^{T-1})$ steps -- an exponential improvement!

\begin{figure}[ht!]
\centering
\includegraphics[width=0.8\textwidth]{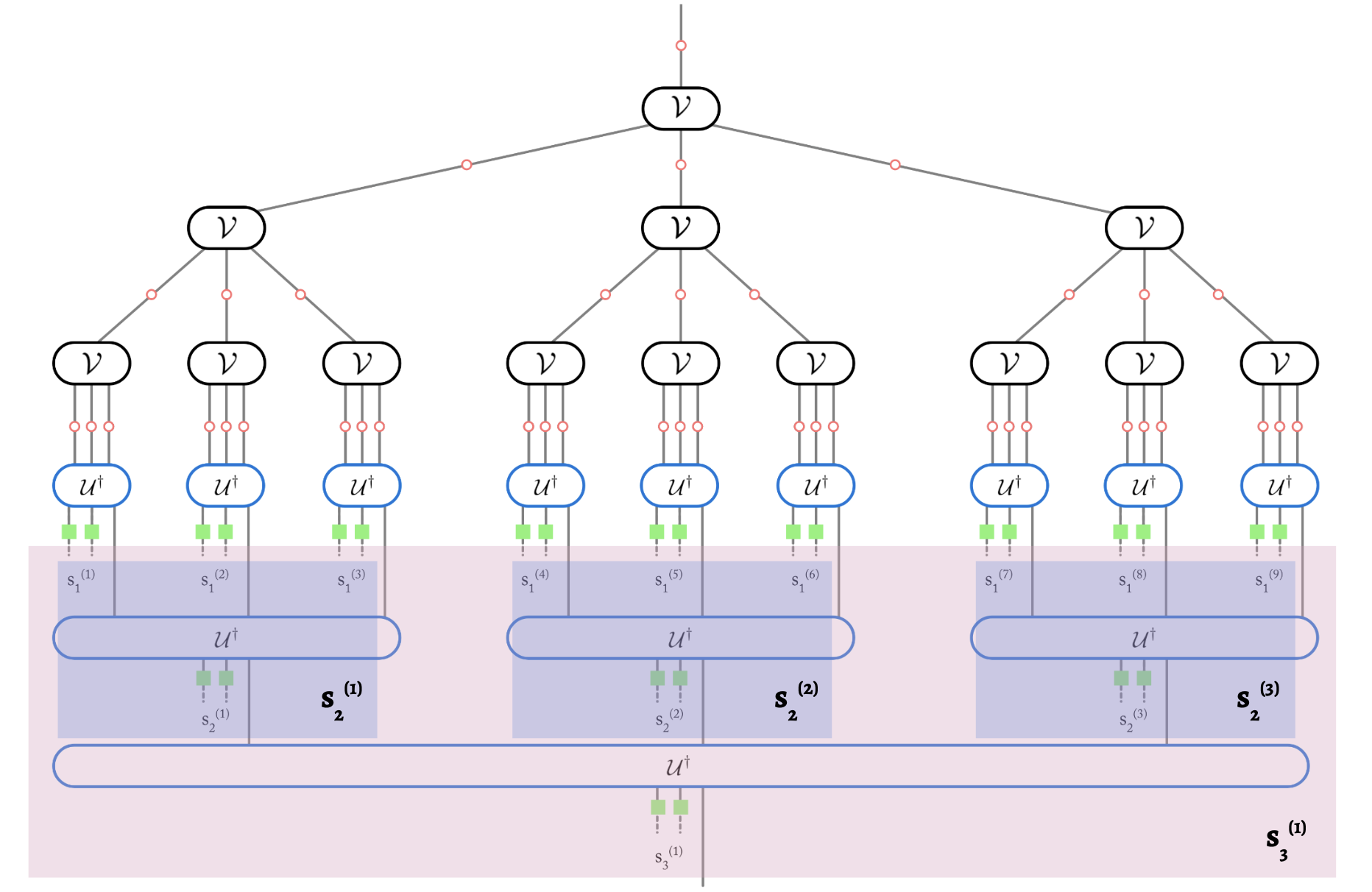} 
\caption{Figure illustrates the optimal recovery circuit. While the upper-half of the figure is a tree with stabilizer encoder $\mathcal{V}$ along with Pauli noise $\mathcal{N}$ on each edge, the lower half is the application of the inverse unitary $\mathcal{U}_T^\dagger = U_T^\dagger (.) U_T$ on the output of the tree process. Note that $\mathcal{U}$ is the Clifford unitary corresponding to the stabilizer isometry $\mathcal{V}$. Upon inversion, we measure each of the ancillary qubits in the $Z$-basis to obtain the syndrome string. The syndromes $s^{(i)}_n$ within the nested syndrome sets $\mathbf{s}^{(i)}_n$ are highlighted in the translucent boxes. Observe that the syndrome-acquisition circuit is a `reflection' of the original tree circuit.}
\label{BP_Circuit}
\end{figure}

\subsection{Ordering of syndrome strings}
The tree structure provides a natural partitioning of the syndrome string. Consider a $b$-ary tree with $T$ levels. The entire syndrome is a bit string of length $b^T-1$. Denote this by $\mathbf{s}^{(1)}_T$ (we include the $(1)$ superscript for consistency with the rest of the presentation).

As an illustrative example, we consider the natural partitions of this string $\mathbf{s}^{(1)}_T$.
This bit string can be obtained from $s_T^{(1)}\in\{0,1\}^{b-1}$, i.e., the syndrome at level $T$, combined with syndromes from the $b$ subtress from the previous levels, $T-1$. These syndromes are labeled as 
$\mathbf{s}_{T-1}^{(i)}$,
where $i=1,\cdots b$, and $T-1$ indicates that the subtree has depth $T-1$, or is `\textit{rooted} in level $T-1$'.\footnote{Note that we have essentially labeled the levels in reverse (e.g., the root is level $T$ and the leaves are level $0$) similar to Sec. \ref{local rec section}, Sec. \ref{local recovery 1 bit section} or Sec. \ref{bell tree sec}.} In summary,
 $$\mathbf{s}^{(1)}_T= s^{(1)}_T \cup (\cup_{i=1}^b \mathbf{s}_{T-1}^{(i)}) \ .$$
Recall that $s^{(1)}_T$ string is obtained by measuring the $b-1$ ancillary qubits that are outputs of the inverse encoder that corresponds `by reflection' to the first encoder in the tree. Fig. \ref{BP_Circuit} illustrates this for a 3-ary tree with depth $T=3$.

We can similarly define the syndrome string $\mathbf{s}_t^{(j)}$ for the $j^{\rm th}$ subtree rooted in level $t$. Observe that this indexing also provides a labeling for \textit{all} subtrees via their root edge position $(j,t)$. We can, thus, say two edges $(j,t)$ and $(i,t-1)$ are \textit{adjacent} when $(j,t)$ is the immediate parent edge of $(i,t-1)$, denoted by $(j,t) \sim (i,t-1)$. Thus, these satisfy the recurrence relation, $$s^{(j)}_t \cup (\cup_{i: (j,t)\sim (i,t-1)} \mathbf{s}_{t-1}^{(i)}) = \mathbf{s}^{(j)}_t.$$

This relation is used to define all labels for all the syndrome strings. Fig. \ref{BP_Circuit} clarifies this labeling. The locality structure of trees ensures that the logical error $L^{(j)}_t$, associated to the $j^{\rm th}$ subtree rooted in level $t$, only depends on $\mathbf{s}_t^{(j)}$, and is independent of the rest of the syndrome string, i.e. 
\be
p(L^{(j)}_t|\mathbf{s}_t)=p(L^{(j)}_t|\mathbf{s}_t^{(j)})\ . 
\ee
In the following, we focus on calculating this conditional probability.
\subsection{Belief Update Rule}\label{BP update}
Consider the root edge of the entire tree channel $\mathcal{E}_{T}$, i.e., level-$T$. Let us call the logical error for the entire tree $L_T$. This root qubit is firstly effected by a local Pauli noise $\mathcal{N}$, and then it is encoded in $b$ qubits via $\mathcal{V}$. Now consider the $b$ subtrees rooted at these $b$ children edges of $\mathcal{V}$. Label the logical error associated to each of these subtrees as error string $L_{T-1}=(L_{T-1}^{(1)},L_{T-1}^{(2)},...,L_{T-1}^{(b)})$ respectively. Let $\mathbf{s}_T$ label the syndrome string obtained from decoding the entire tree  (we drop the $(1)$ superscript here because there is only one string to be considered). Similarly, let $\mathbf{s}_{T-1}=(\mathbf{s}_{T-1}^{(1)},\mathbf{s}_{T-1}^{(2)},...\mathbf{s}_{T-1}^{(b)})$ indicate the syndrome set obtained from decoding the $b$ `level-${(T-1)}$' subtrees noted just earlier. We shall also consider an intermediate tree that is rooted at the edge \textit{after} the noise process $\mathcal{N}$, and whose logical error we denote by $L_T'$. Observe that the entire tree and this intermediate tree are conditioned by the same syndrome $\mathbf{s}_T$ and differ only by the effect of single-qubit $\mathcal{N}$.


Thus, the crux of the BP algorithm is an update rule from,
$$
p(L_{T-1}|\mathbf{s}_{T-1}) \rightarrow p(L_T'|\mathbf{s}_T) \rightarrow p(L_T|\mathbf{s}_T).
$$
The second update is simple as it is caused only by the single-qubit noise process $\mathcal{N}$. Then,
\begin{align}
    p(L_T|\mathbf{s}_T)=\sum_{L_T'} {N}(L_T|L_T')\ p(L_T'|\mathbf{s}_T)\ .
\end{align}
where $N(.|.)$ is defined in Eq.(\ref{pauli noise transition matrix}). The first update is derived in the spirit of Eq. 4 in \cite{poulin2006optimal}:
\begin{align}\label{first update step 1}
p(L_T'|\mathbf{s}_T) &= \sum_{L_{T-1}}\ p(L_T'|L_T, \mathbf{s}_T) \ p(L_{T-1}| \mathbf{s}_T)\ .
\end{align}
The logical error $L_T'$ is completely determined by $L_{T-1}$ by considering its logical part, i.e., $\mathcal{L}(L_{T-1})$. Thus, $p(L_T'|L_{T-1}, \mathbf{s}_T)=\delta[L_T' =  \mathcal{L}(L_{T-1})]$. Next we look at the second term. Recall that $\mathbf{s}_T = s_T \cup \mathbf{s}_{T-1} = s_T \cup (\cup_{i=1}^b \mathbf{s}_{T-1}^{(i)})$. Thus,
\begin{align}
    p(L_{T-1}| \mathbf{s}_T) &= p(L_{T-1}| s_T, \mathbf{s}_{T-1})
    =\frac{p(L_{T-1}, s_T, \mathbf{s}_{T-1})}{p(s_T, \mathbf{s}_T)}
    =\frac{p(s|L_{T-1}, \mathbf{s}_{T-1})}{p(s_T|\mathbf{s}_{T-1})}\ p(L_{T-1}|\mathbf{s}_{T-1}).
\end{align}
The syndrome $s_T$ is completely determined by the $b$-qubit error string $L_{T-1}$. Thus, $p(s_T|L_{T-1}, \mathbf{s}_{T-1})=\delta[s_T={\rm Synd}(L_{T-1})]$. By substituting these into \ref{first update step 1},
\begin{align}
    p(L_T'|\mathbf{s}_T)=\sum_{L_{T-1}} \delta[L_T' = \mathcal{L}(L_{T-1})]\ \delta[s_T={\rm Synd}(L_{T-1})] \frac{1}{p(s_T|\mathbf{s}_{T-1})} p(L_{T-1}|\mathbf{s}_{T-1}) 
\end{align}
where $\mathcal{L}(E)$ is the logical error associated to $E$ and ${\rm Synd}(E)$ is the syndrome bitstring associated to $E$ (Refer to subsection \ref{clifford properties}). An important final observation is that the logical errors associated to different \textit{non-overlapping} subtrees are independent. This is because the noise processes within the tree are local to each edge. The $b$ subtrees, whose logical errors constitute one term of $L_{T-1}$ each, are also non-overlapping. Thus, logical errors $\{L_{T-1}^{(i)}\}_i$ are independent from each other, where $i$ indexes the subtree. Recall that $\mathbf{s}_{T-1}=\{\mathbf{s}_{T-1}^{(i)}\}_i$. Thus, the first update is,
\begin{align}\label{BP update rule}
    p(E'_T|\mathbf{s}_T)=\sum_{L_{T-1}} \delta[L_T' = \mathcal{L}(L_{T-1})]\ \delta[s_T={\rm Synd}(L_{T-1})] \times \frac{1}{p(s_T|\mathbf{s}_{T-1})} \prod_i p(L^{(i)}_{T-1}|\mathbf{s}^{(i)}_{T-1})
\end{align}
Combining the first and second updates gives us,
\begin{align}
    p(L_T|\mathbf{s}_T)=\sum_{L_T'} {N}(L_T|L_T')\sum_{L_{T-1}} \delta[L_T' = \mathcal{L}(L_{T-1})] \delta[s_T={\rm Synd}(L_{T-1})] \times \frac{1}{p(s_T|\mathbf{s}_{T-1})} \prod_i p(L^{(i)}_{T-1}|\mathbf{s}^{(i)}_{T-1})
\end{align}
If we change the order of summation and define,

\begin{align}\label{BP weight final eqn}
    f_{s_T}(L_T,L_{T-1}):= \sum_{L_T'} \delta[L_T' = \mathcal{L}(L_{T-1})] \delta[s_T={\rm Synd}(L_{T-1})]\ {N}(L_T|L_T')
    = \  \delta[s_T={\rm Synd}(L_{T-1})]\ {N}(L_T|\mathcal{L}(L_{T-1}))
\end{align}

we get,
\begin{align}\label{BP final eqn}
p(L_T|\mathbf{s}_T)=\sum_{L_{T-1}}f_{s_T}(L_T,L_{T-1})\frac{1}{p(s_T|\mathbf{s}_{T-1})} \prod_i p(L^{(i)}_{T-1}|\mathbf{s}^{(i)}_{T-1}),
\end{align}
where $\mathbf{s}_T = (s_T, \mathbf{s}_{T-1})$. If we consider trees with no noise within them, ${N}(L_T|L_T')=\delta(L_T=L_T')$. Thus, $f_{s_T}(L_T,L_{T-1})=\delta[L_T = \mathcal{L}(L_{T-1})]\ \delta[s_T={\rm Synd}(L_{T-1})]$, which is equivalent to Eq.4 in \cite{poulin2006optimal}. Finally using the independence of non-overlapping subtrees, we can repeat this procedure for $p(L^{(i)}_{T-1}|\mathbf{s}^{(i)}_{T-1})$ for each $i$, and further propagate until the leaves of the tree. That is to say, $p(L^{(j)}_0)={N}(L^{(j)}_0|I)$ is set as the terminal condition of recursion for all $j$.

\textit{Note regarding noise at the root} -- The presentation above is for a tree \textit{with} noise at the root. We can very easily modify this analysis for a tree \textit{without} noise at the root: in the very last step of decoding (i.e., $t=T$), we modify Eq.(\ref{BP weight final eqn}) as 
$$f_{s_T}(L_T,L_{T-1}):= \sum_{L_T'} \delta[L_T' = \mathcal{L}(L_{T-1})] \delta[s={\rm Synd}(L_{T-1})]\ {N}(L_T|L_T')
    = \  \delta[s={\rm Synd}(L_{T-1})]\ \delta(L_T=\mathcal{L}(L_{T-1})).$$

\textit{Optimality \& Efficiency--} Firstly note that $p(L_T|\mathbf{s}_T)$ contains \textit{all} the relevant information to make a decision regarding the logical error for the entire tree. And its computation is, as noted, via a BP algorithm on a tree-like factor graph. From theorem 14.1 in \cite{mezard2009information}, we note that BP on tree-like factor graphs is exact. Thus, optimality is ensured by our BP algorithm. 
The efficiency of this algorithm is because we sum over only $\mathcal{O}(2^b)$ terms at each vertex. Thus we have $\mathcal{O}(2^b \times b^{T-1})$ terms to consider for the entire tree. This is a drastic improvement in our efficiency as compared to a double exponential search.

\subsection{Numerical Example: Classical $3$-ary tree} 
As a sanity check, we applied this optimal recovery algorithm to a classical tree with branching number 3, and bit-flip error on each edge. Note that in this set-up it suffices to just consider the group $\{I,X\}$ instead of the entire Pauli group. This case has been studied in \cite{evans2000} and we know that the asymptotic phase transition is at $(1-3^{-\frac12})/2 = \sim 0.21$. The figure below plots the overall error after recovery as a function of bit-flip error rate within the tree. 
\begin{center}
    \includegraphics[width=0.5\textwidth]{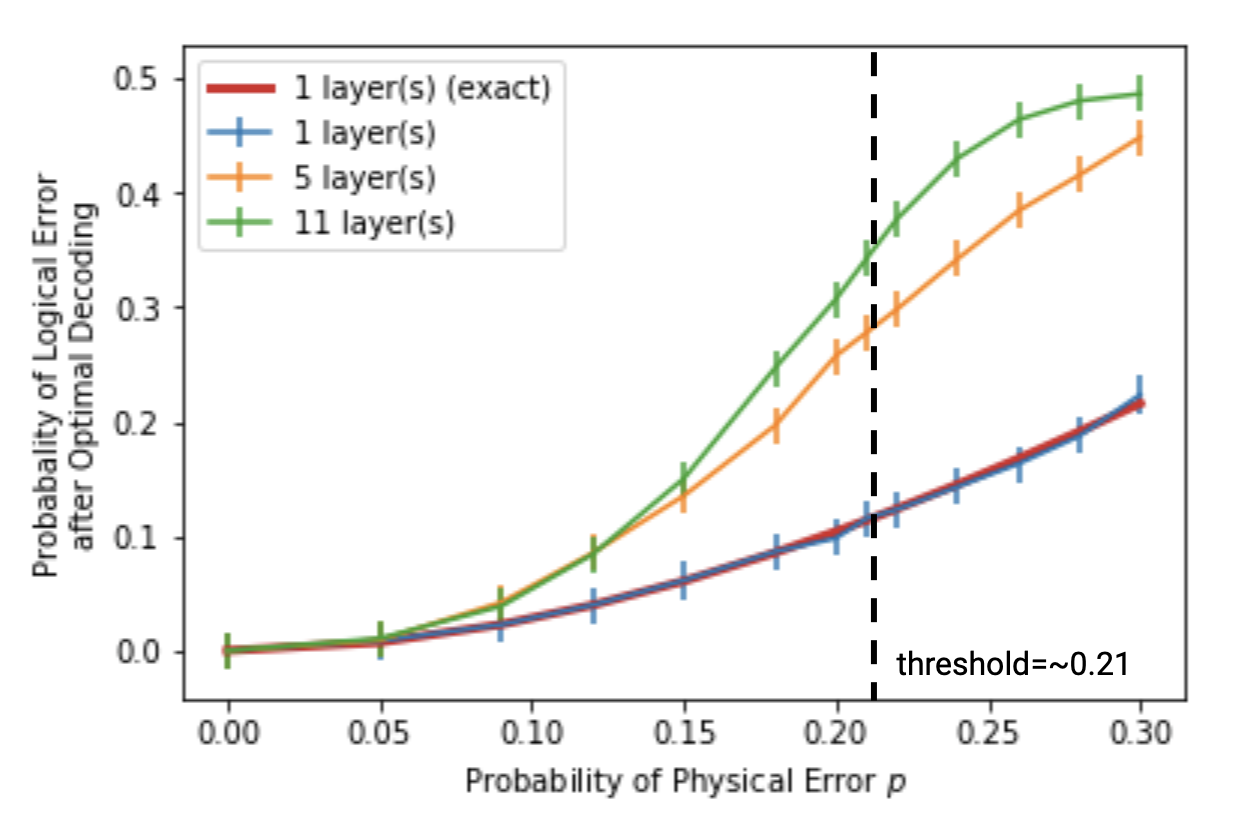} 
\end{center}
We run our Monte Carlo by randomly generating the effect of noise for given physical error settings, inverting the tree step-by-step and acquiring the syndrome strings (see Fig.\ref{BP_Circuit}), and then running the BP algorithm on them. We observe that the logical error saturates for 11 layers around $\sim 0.27$. With increasing depth, the transition progresses lower. For larger depths we should expect it to approach the optimal $\sim 0.21$. Additionally, for $T=1$ layer, we can compute the exact logical error rate, i.e., $3(1-p)p^2 + p^3$. We see that our Monte Carlo simulation matches the exact computation. 

\textit{Note regarding Estimated Mean \& Error --} We note that after randomly generating the effect of noise, the BP algorithm is a completely deterministic function, and thus, whether it is recovered correctly or not is also a deterministic function. So for a specific noise parameter setting, we have a Bernoulli random variable that outputs success or failure of the recovery procedure. We are to estimate the parameter of this Bernoulli random, that is the probability of success/failure. 
Maximum likelihood estimation on this data of $N$ runs entails that success probability is the sample mean $\mathbb{E}({\mathbf{x}})=\sum_{x\in \mathbf{x}}x/N$, and the estimate's variance is ${\rm Var}(\mathbf{x})=\mathbb{E}({\mathbf{x}})(1-\mathbb{E}({\mathbf{x}}))/N \leq 1/4N$. So, e.g., to have a standard deviation of $\mathcal{O}(0.01)$, we need $\sim 2500$ samples.

\newpage

\section{ Dephased Encoders}\label{deph appendix}
In Sec. \ref{deph tree mapping sec} we explained the idea of dephasing trees obtained from the standard (or, anti-standard) encoders of CSS codes  to obtain an equivalent  problem about the propagation of information on classical trees, which is a variant of the standard broadcasting problem studied in, e.g., \cite{evans2000}. This equivalency is based on the fact that for CSS codes $X$ and $Z$ errors can be  corrected independently, and therefore adding one type of error to the tree does not affect the  probability of logical of error for the other type. 

For completeness, here  we state and prove this property in a more formal language.


\begin{proposition}\label{prop}
Let $Z_L\in \langle iI,Z \rangle ^ {\otimes n}$ be a logical operator  of a CSS code and $|0\rangle_L$ and $|1\rangle_L$ be two orthogonal eigenstates of $Z_L$ in the code subspace, satisfying 
$|c\rangle_L=(-1)^c Z_L |c\rangle_L :\ c=0,1$.  
Suppose the qubits are subjected to noise  channel $\mathcal{N}_z\circ \mathcal{N}_x= \mathcal{N}_x \circ \mathcal{N}_z$, where  $\mathcal{N}_x$ only involves Pauli $X$ errors and  $\mathcal{N}_z$ only involves Pauli $Z$ errors (This means $Z$ and $X$ error are independent, whereas $X$ (or, $Z$) errors on different qubits can be correlated). Then,  there exists a quantum channel $\mathcal{R}$, such that
\be\label{rec}
\mathcal{R}\circ \mathcal{N}_z\circ \mathcal{N}_x(|c\rangle\langle c|_L)=\mathcal{N}_x(|c\rangle\langle c|_L) \ \ \ : c=0,1\ .
\ee
Furthermore, given state $\mathcal{N}_z\circ \mathcal{N}_x(|c\rangle\langle c|_L)$ with unknown $c=0,1$, the optimal strategy for determining the value of bit $c$ is measuring all the qubits in the $Z$ basis, followed by a  classical processing of the outcomes. 
 \end{proposition}
 
\begin{proof} 
We prove the first part of the proposition by constructing the said recovery map $\mathcal{R}$. Without loss of generality we can restrict our attention to the case of 2-dimensional codes. If the code subspace has dimension more than 2, we focus on the  2-dimensional subspace spanned by $|0\rangle_L$ and $|1\rangle_L$, which will also be a CSS code.  

In particular, there exists a logical operator $X_L\in\langle iI,X\rangle^{\otimes n}$, and
\begin{align}
S^z_j&\in\langle iI,Z\rangle^{\otimes n}: j=2,\cdots ,n_z+1\\ 
S^x_k&\in\langle iI,X\rangle^{\otimes n}: k=n_z+2,\cdots ,n
\end{align}
where $\{S^z_j\}$ and $\{S^x_k\}$ are $n_z
$ and $n_x$ independent $Z$ and $X$ stabilizers and $n=n_x+n_z+1$. Using the standard results in the theory of stabilizer codes we know that  any such code has an encoder such that
\be
|c\rangle_L=W(|c\rangle\otimes |0\rangle^{\otimes n_z}\otimes |0\rangle^{\otimes n_x})\ \ \ \ \ \  \ \ :\   c=0,1\ ,
\ee
and
\begin{align}
W Z_1 W^\dag &=Z_L\ \ , \  
W X_1 W^\dag =X_L\\
W Z_j W^\dag&=S^{z}_j\ \ \ \ \ \ : \ j=2,\cdots n_z+1\\ 
W Z_k W^\dag&=S^{x}_k\ \ \ \ \ \ : \ k=n_z+2,\cdots,  n_x+n_z+1 \ .
\end{align}
In other words, under unitary $W^\dag$ the Hilbert space of $n$ physical qubits decomposes as  
\be
(\mathbb{C}^2)^{\otimes n}=\mathbb{C}^2\otimes \mathcal{H}_{\text{z-stab}}\otimes\mathcal{H}_{\text{x-stab}}
\ee
where $\mathcal{H}_{\text{z-stab}}=(\mathbb{C}^2)^{\otimes n_z}$ and $\mathcal{H}_{\text{x-stab}}=(\mathbb{C}^2)^{\otimes n_x}$. The fact that the code is a CSS code implies that for any $X$ error $\textbf{X}^{\textbf{x}} $ with $\textbf{x}\in\{0,1\}^n$  and $Z$ error $\textbf{Z}^{\textbf{z}} $ with $\textbf{z}\in\{0,1\}^n$ it holds that  
\begin{align}
W^\dag \textbf{X}^{\textbf{x}} \textbf{Z}^{\textbf{z}} |c\rangle_L&=\pm\big(X^{x_L(\textbf{x})} Z^{z_L(\textbf{z})}|c\rangle\big)\otimes  |\mathbf{s}_z(\textbf{x})\rangle \otimes  |\mathbf{s}_x(\textbf{z})\rangle\ .
\end{align}
where $\mathbf{s}_z(\mathbf{x})$ is the syndrome string associated to $Z$ stabilizer measurements of $X$ error $\textbf{X}^{\textbf{x}} $ (and \textit{vice versa} for $\mathbf{s}_x(\mathbf{z})$). Thus,
\begin{align}
    W^\dagger(\mathcal{N}_x\circ \mathcal{E}_z(\ket{c}\bra{c}_L))W = \sum_{\mathbf{x},\mathbf{z}} p(\mathbf{x},\mathbf{z}) ({X}^{x_L(\textbf{x})} {Z}^{z_L(\textbf{z})} |c\rangle\langle c|_L {Z}^{z_L(\textbf{z})} {X}^{x_L(\textbf{x})}) \otimes \ket{\mathbf{s}_z(\mathbf{x})}\bra{\mathbf{s}_z(\mathbf{x})} \otimes \ket{\mathbf{s}_x(\mathbf{z})}\bra{\mathbf{s}_x(\mathbf{z})}
\end{align}

Recall $\ket{c}_L$ is a logical $Z$ state, i.e., ${Z}_L|c\rangle = \pm |c\rangle$. Noting the independence of $X$ and $Z$ errors, we have $p(\mathbf{x},\mathbf{z})=p(\mathbf{x})p(\mathbf{z})$. Thus,
\begin{align}
    W^\dagger(\mathcal{N}_x\circ \mathcal{N}_z(\ket{c}\bra{c}_L))W = \sum_{\mathbf{x}} p(\mathbf{x}) ({X}^{x_L(\textbf{x})} |c\rangle\langle c|_L {X}^{x_L(\textbf{x})}) \otimes \ket{\mathbf{s}_z(\mathbf{x})}\bra{\mathbf{s}_z(\mathbf{x})} \otimes \sum_{\mathbf{z}} p(\mathbf{z}) \ket{\mathbf{s}_x(\mathbf{z})}\bra{\mathbf{s}_x(\mathbf{z})}
\end{align}

We see here that $\mathbb{C}^2 \otimes \mathcal{H}_{\rm z-stab}$ is uncorrelated with $\mathcal{H}_{\rm x-stab}$. Furthermore, all the information regarding $Z$ errors is within $\mathcal{H}_{\rm x-stab}$. So replacing this subsystem with $\ket{0}^{\otimes n_z}$ allows us to discard the effect of $\mathcal{E}_z$. Thus, in summary, the recovery channel we want is,

\be
\mathcal{R}(\cdot)=W\Big(\Tr_{\text{x-stab}}(W^\dag(\cdot) W)\otimes |0\rangle\langle 0|^{\otimes n_z}\Big)W^\dag.
\ee\\

To show the second part of the proposition, suppose one measures all qubits in $\{|0\rangle,|1\rangle\}$ basis. If we ignore the outcomes of all measurements, i.e., if we dephase all qubits in the $Z$ basis, we obtain state $\mathcal{D}_z^{\otimes n}\circ\mathcal{N}_z\circ \mathcal{N}_x(|c\rangle\langle c|_L)$.  Since 
 $\mathcal{D}_z^{\otimes n}\circ\mathcal{N}_z$ only contains $Z$ errors, using the first part of the proposition, we know that there exists a recovery process that transforms state    
$\mathcal{D}_z^{\otimes n}\circ\mathcal{N}_z\circ \mathcal{N}_x(|c\rangle\langle c|_L)$ to $\mathcal{N}_z\circ \mathcal{N}_x(|c\rangle\langle c|_L)$. This means that, regardless of the figure of merit, the optimal strategy for  discriminating     
states $\mathcal{N}_z\circ \mathcal{N}_x(|0\rangle\langle 0|_L)$ and $\mathcal{N}_z\circ \mathcal{N}_x(|1\rangle\langle 1|_L)$ can start with a measurement of all qubits in the $Z$ basis, followed by a classical processing of the outcomes of $Z$ measurements (which is the only information remaining in a fully dephased state).
\end{proof}

\subsection{Dephased encoders for CSS codes}\label{CSS dephased subsec}
The dephased encoders of CSS codes has a simple description. Recall that a  CSS code can be characterized by a pair of linear classical codes $C_1\subset \{0,1\}^b$ and $C_2\subset C_1$. Since here we are focusing on codes that encode a single qubit, $|C_1|/|C_2|=2$, where $|C_1|$ and $|C_2|$ are the number of bitstrings in set $C_1$ and $C_2$, respectively.
 Elements of $C_1$ can be partitioned to two equivalency classes, namely those in $C_2$, and those that can be written as $\textbf{x}_L+C_2$, where  $\textbf{x}_L$ is a fixed bit string, which defines the logical $X$ operator. Then, $|0\rangle_L$ and $|1\rangle_L$ can be written as  
\begin{align}
|0\rangle_L&=V|0\rangle=\frac{1}{\sqrt{|C_2|}}\sum_{\textbf{w}\in C_2}|\textbf{w}\rangle\ , \\ 
|1\rangle_L&=V|1\rangle=\frac{1}{\sqrt{|C_2|}}\sum_{\textbf{w}\in C_2}|\textbf{w}+\textbf{x}_L\rangle\ .\end{align}
Then, the classical encoder defined in Eq.(\ref{class-enc}) is given by 
\begin{align}\label{CSS measured encoder general form}
    \mathbb{M}_z:&\ P_z(\textbf{w}|0)=\frac{1}{|C_2|}\ \ \  : \textbf{w} \in C_2 \nonumber\\
    &\ P_z(\textbf{w}|1)=\frac{1}{|C_2|}\ \ \  : \textbf{w} \in C_2+\textbf{x}_L\ .
\end{align}
Similarly, for logical states $|\pm\rangle_L=(|0\rangle_L\pm|1\rangle_L)/\sqrt{2}=V|\pm\rangle$ we have
\begin{align}
H^{\otimes b}|+\rangle_L&=H^{\otimes b}\frac{1}{\sqrt{|C_1|}}\sum_{\textbf{w}\in C_1}|\textbf{w}\rangle=\frac{1}{\sqrt{|C^\perp_1|}}\sum_{\textbf{w}\in C^\perp_1}|\textbf{w}\rangle\\ 
H^{\otimes b}|-\rangle_L&=H^{\otimes b}\frac{1}{\sqrt{|C_1|}}\sum_{\textbf{w}\in C_1} (-1)^{\textbf{z}_L\cdot \textbf{w}} |\textbf{w}\rangle\\ &=\frac{1}{\sqrt{|C^\perp_1|}}\sum_{\textbf{w}\in C^\perp_1}|\textbf{w}+\textbf{z}_L\rangle\ ,
\end{align}
where $H$ is the Hadamard operator. Then, by dephasing the input and output of the encoder $V$ in the $X$ basis, we obtain the classical encoder
\begin{align}\label{CSS measured encoder general x}
    \mathbb{M}_x:&\ P_x(\textbf{w}|0)=\frac{1}{|C^\perp_1|}\ \ \  : \textbf{w} \in C^\perp_1 \nonumber\\
    &\ P_x(\textbf{w}|1)=\frac{1}{|C^\perp_1|}\ \ \  : \textbf{w} \in C^\perp_1+\textbf{z}_L\ .
\end{align}

\subsection{Example: Generalized Shor Code}\label{shor dephased}
The dephased tree of generalized Shor codes $[[n^2,1,n]]$ is where we have classical encoders that are alternatingly concatenated. We label them $\mathbb{M}_1$ and $\mathbb{M}_2$, and are defined:
\begin{align}
    \mathbb{M}_1:&\ P_z(0^{\otimes n}|0)=1\nonumber\\ 
    &\ P_z(1^{\otimes n}|1)=1\\
    \mathbb{M}_2:&\ P_z(\mathbf{w}|0)=\frac{1}{2^{n-1}}, \ {\rm when}\ |\mathbf{w}|=0 \ {\rm mod \ 2} \nonumber \\
    &\ P_z(\mathbf{w}|1)=\frac{1}{2^{n-1}}, \ {\rm when}\ |\mathbf{w}|=1 \ {\rm mod \ 2}
\end{align}
where $|\mathbf{w}|$ is the Hamming weight of $\mathbf{w}$. For instance, the usual Shor-9 code corresponds to $n=3$, and the specific maps are,
\begin{align}
    \mathbb{M}_1:&\ P_z(000|0)=1\nonumber\\ 
    &\ P_z(111|1)=1\\
    \mathbb{M}_2:&\   P_z(\mathbf{w}|0)=1/4,\ {\rm when}\ \mathbf{w} \in \{000, 011,101, 110\}\nonumber\\
  &\   P_z(\mathbf{w}|1)=1/4,\ {\rm when}\ \mathbf{w} \in \{111, 100, 010, 001\}
\end{align}

Thus, the encoder becomes $\mathbb{M}_1^{\otimes n} \circ \mathbb{M}_2$, from $1\rightarrow n^2$ bits. After every encoded layer, we have local independent bit-flip error on all these bits, which corresponds to $X$ (or $Z$) error channel between concatenated Shor encoder layers.

Considering these encoders in the opposite order, i.e., $\mathbb{M}_2^{\otimes n} \circ \mathbb{M}_1$ corresponds to dephasing in the conjugate basis, i.e., $X$ to $Z$ and \textit{vice versa}.

\subsection{Summary of Open Classical Problems}\label{open q}

Study of novel dephased trees can be summarized into three succinctly stated classical tree problems, that still remain open problems. We list them in order of reducing generality. From here on, $\overline{1}$ is a bitsring with all $1$'s and `$+$' is bit-wise addition modulo 2.\\

\textit{Problem I --}
Consider a tree-like circuit.
Each \textit{vertex} is associated to an encoder that maps one bit, $b$, to $n$ bitstring, $x$, probabilistically. This encoder is, thus, characterized by the conditional probability distribution $p(x|b)$ such that $p(x|b)=p(x+\textbf{1}^n|b+1)$. This setup is general in that it also subsumes cases with binary symmetric channels (e.g., bit-flip noise channels) before or after the encoder.\\

\textit{Problem II --}\label{restrict}
Here we focus on an important special case of Problem I, where $p(x|b)$ can be realized by adding local uncorrelated bit-flip noise on the output bits which are sampled from a distribution  $q(x|b)$. Here, $q(x|b)$ has the property that $q(x|b=0)$ and $q(x|b=1)$ have disjoint support. So graphically, the vertices of the tree are associated to the aforementioned $q(x|b)$ encoder, and the edges are associated with bit-flip channels of error $\epsilon$.

More specifically, we are interested in the case where $q(x|b=0)$ is a uniform distribution over a subspace of bit strings, and  $q(x|b=1) = q(x+\textbf{1}^n|b=0)$, i.e, its complement. This restriction is interesting because the encoder by itself is not noisy (due to the disjoint support), although it might make the encoded information \textit{more} sensitive to noise. Dephased trees of Steane codes and Generalized Shor Codes satisfy said property.\\

\textit{Problem III --}
We focus on the classical version of generalized Shor codes $[[n^2,1,n]]$. Here we have two encoders from $1\rightarrow n$ bits that are alternatingly concatenated. We label them $\mathbb{M}_1$ and $\mathbb{M}_2$, and are 
defined:
\begin{align}
    \mathbb{M}_1:\ & 0 \mapsto \textbf{0}^{n}\nonumber\\
    & 1 \mapsto \textbf{1}^{n}\\
    \mathbb{M}_2:\ & 0 \mapsto {\rm uniformly\ random\ over\ even\ Hamming\ wt.}\nonumber\\
     & 1 \mapsto {\rm uniformly\ random\ over\ odd\ Hamming\ wt.}
\end{align}
The encoder associated to each vertex is $\mathbb{M}_1^{\otimes n} \circ \mathbb{M}_2$, from $1\rightarrow n^2$ bits. For instance, the usual Shor-9 code corresponds to $n=3$, and the specific maps are,
\begin{align}
    \mathbb{M}_1:\ & 0 \mapsto 000\nonumber\\
    & 1 \mapsto 111\\
    \mathbb{M}_2:\ & 0 \mapsto {\rm Unif}\{000, 011,101, 110\}\nonumber\\
     & 1 \mapsto {\rm Unif}\{111, 100, 010, 001\}
\end{align}

Note that bit-flip noise here occurs after $\mathbb{M}_1^{\otimes n} \circ \mathbb{M}_2$, not in between them. The case where we switch the order of the encoders, i.e., $\mathbb{M}_2^{\otimes n} \circ \mathbb{M}_1$, corresponds to measuring in the conjugate basis, and is also of interest. $\mathbb{M}_2$ is, in some sense, maximally sensitive to noise. That is, an error on any of the $n$ children of $\mathbb{M}_2$ is equivalent to an error on the input bit. This is an extreme illustration of the comment made in Problem II.\\ 

A thorough analysis of these classical tree problems remains an open research direction in classical network theory. 

\newpage

\end{document}